\documentclass[copyright]{eptcs}
\providecommand{\publicationstatus}{\textbf{Technical Report 5513}, NICTA, 2013\\ \url{http://www.nicta.com.au/pub?id=5513}}
\newcommand{\copyrightholders}{A. Fehnker, R.J. van Glabbeek, P. H{\"o}fner,\\{}\hspace{9.1mm} A. McIver, M. Portmann \& W.L. Tan}

\usepackage{breakurl}             
\usepackage{amsfonts,amsmath,amssymb}
\usepackage{algorithm,algorithmic}
\makeatletter
\def\theHALC@line{\thealgorithm-\theALC@line}
\def\theHALC@rem{\thealgorithm-\theALC@rem}
\makeatother
\floatname{algorithm}{Process}
\renewcommand{\listalgorithmname}{List of Processes}
\usepackage{color}
\usepackage{aodv}
\usepackage{graphicx}
\usepackage{enumerate}
\usepackage{xspace}
\usepackage{xifthen}
\usepackage{array}
\usepackage{longtable}
\LTcapwidth=\textwidth
\usepackage{colortbl}
\usepackage{wrapfig}
\usepackage{makeidx}
\usepackage{stmaryrd} 

\makeindex

\newcommand{\nocontentsline}[3]{}
\newcommand{\tocless}[2]{\bgroup\let\addcontentsline=\nocontentsline#1{#2}\egroup}
\makeatletter
    \def\@evenhead{\thepage\hfil\slshape\leftmark}%
    \def\@oddhead{{\slshape\rightmark}\hfil\thepage}%
                                                            
\makeatother

\definecolor{purple}{rgb}{.5,.0,.5}
\newcommand{\highlight}[1]{\textcolor{purple}{#1}} 

\newtheorem{common}{Common}[section]{\bfseries}{\itshape}
\newcommand{\numberedthing}[2]{%
     \newtheorem{#1}[common]{#2}{\bfseries}{\itshape}}
\numberedthing{theorem}{Theorem}
\numberedthing{cor}{Corollary}
\numberedthing{lem}{Lemma}
\numberedthing{prop}{Proposition}
\numberedthing{remark}{Remark}
\numberedthing{definition}{Definition}
\newtheorem{exam}[common]{Example}

\newenvironment{proposition}{\begin{prop}\rm}{\end{prop}}
\newenvironment{example}{\begin{exam}\rm}{\end{exam}}

\def\squareforqed{\hbox{\rlap{$\sqcap$}$\sqcup$}}
\def\endbox{\ifmmode\squareforqed\else{\unskip\nobreak\hfil
\penalty50\hskip1em\null\nobreak\hfil\squareforqed
\parfillskip=0pt\finalhyphendemerits=0\endgraf}\fi}
\newenvironment{proof}{\begin{trivlist} \item[\hspace{\labelsep}\it Proof.\/]}{\endbox\end{trivlist}}
\newenvironment{proofNobox}{\begin{trivlist} \item[\hspace{\labelsep}\it Proof.\/]}{\end{trivlist}}

\newenvironment{simpleProcess}{%
  \newcommand{\algindent}{0.1em}
  \renewcommand{\algorithmicelsif}{$+$\algorithmicif}
  \renewcommand{\algorithmiccomment}[1]{\textcolor{blue}{\hspace{\algindent}/*\,##1\,*/}}
  \algsetup{indent=0.7em}
  \begin{algorithmic}%
  }{
  \end{algorithmic}
  }

\newcommand{\awn}{AWN\xspace}
\newcommand{\bis}{\raisebox{.3ex}{$\underline{\makebox[.7em]{$\leftrightarrow$}}$}}
\newcommand{\NN}{
    \ensuremath{%
        \mathop{\rm I\mkern-2.5mu N}%
        \nolimits%
    }%
}
\newcommand{\plat}[1]{\raisebox{0pt}[0pt][0pt]{#1}} 
\newcommand{\spaces}[1]{\ #1\ }
\newcommand{\ans}{\spaces{\wedge}}
\newcommand{\ors}{\spaces{\vee}}
\newcommand{\ims}{\spaces{\Rightarrow}}
\newcommand{\iffs}{\spaces{\Leftrightarrow}}
\newcommand{\rtord}[1][\dval{dip}]{\ensuremath{\sqsubseteq_{#1}}}
\newcommand{\rtequiv}[1][\dval{dip}]{\ensuremath{\approx_{#1}}}
\newcommand{\rtsord}[1][\dval{dip}]{\ensuremath{\sqsubset_{#1}}}
\newcommand{\decremented}{\mathbin{\stackrel{\bullet}{\raisebox{0pt}[2pt]{$-$}}}1}
\newcommand{\rte}{routing table entry\xspace}
\newcommand{\rtes}{routing table entries\xspace}
\newcommand{\phrase}[1]{\index{#1}\emph{#1}}
\makeatletter
\def\comesfrom{\@transition\leftarrowfill}
\def\goesto{\@transition\rightarrowfill}
\def\ngoesto{\@transition\nrightarrowfill}
\def\Goesto{\@transition\Rightarrowfill}
\def\nGoesto{\@transition\nRightarrowfill}
\def\xmapsto{\@transition\mapstofill}
\def\nxmapsto{\@transition\nmapstofill}
\def\@transition#1{\@@transition{#1}}
\newbox\@transbox
\newbox\@arrowbox
\newbox\@downbox
\def\@@transition#1#2%
   {\setbox\@transbox\hbox
      {\vrule height 1.5ex depth .9ex width 0ex\hskip0.25em$\scriptstyle#2$\hskip0.25em}
   \ifdim\wd\@transbox<1.5em
      \setbox\@transbox\hbox to 1.5em{\hfil\box\@transbox\hfil}\fi
   \setbox\@arrowbox\hbox to \wd\@transbox{#1}
   \ht\@arrowbox\z@\dp\@arrowbox\z@
   \setbox\@transbox\hbox{$\mathop{\box\@arrowbox}\limits^{\box\@transbox}$}
   \dp\@transbox\z@\ht\@transbox 10pt
   \mathrel{\box\@transbox}}
\def\nrightarrowfill{$\m@th\mathord-\mkern-6mu%
  \cleaders\hbox{$\mkern-2mu\mathord-\mkern-2mu$}\hfill
  \mkern-6mu\mathord\not\mkern-2mu\mathord\rightarrow$}
\def\Rightarrowfill{$\m@th\mathord=\mkern-6mu%
  \cleaders\hbox{$\mkern-2mu\mathord=\mkern-2mu$}\hfill
  \mkern-6mu\mathord\Rightarrow$}
\def\nRightarrowfill{$\m@th\mathord=\mkern-6mu%
  \cleaders\hbox{$\mkern-2mu\mathord=\mkern-2mu$}\hfill
  \mkern-6mu\mathord\not\mathord\Rightarrow$}
\def\mapstofill{$\m@th\mathord\mapstochar\mathord-\mkern-6mu%
  \cleaders\hbox{$\mkern-2mu\mathord-\mkern-2mu$}\hfill
  \mkern-6mu\mathord\rightarrow$}
\def\nmapstofill{$\m@th\mathord\mapstochar\mathord-\mkern-6mu%
  \cleaders\hbox{$\mkern-2mu\mathord-\mkern-2mu$}\hfill
  \mkern-6mu\mathord\not\mkern-2mu\mathord\rightarrow$}
\makeatother 
\newcommand{\ar}[1]{\mathrel{\goesto{#1}}}            
\newcommand{\nar}[1]{\mathrel{\ngoesto{#1\;}}}        

\newcommand{\queue}[2]{
{\tiny 
 \begin{tabular}[b]{@{}|@{}c@{}|@{}}
 \mbox{}#2\\
 \hline\parbox{3.8em}{\centering{$#1$}}\\
 \hline
 \end{tabular}}
 \hspace{-0.5em} 
}
\newcommand{\store}[1]{
  \textcolor{white}{(\alph{figexmp}) }%
  \raisebox{0.2em}{\scriptsize Queues:}%
  #1%
}
\newcommand{\figurecaption}{Figure \thefigure}
\makeatletter\def\fnum@figure{\figurecaption}\makeatother
\newcounter{figexmp}
\newcounter{figexmpA}
\newboolean{multicaption}
\setboolean{multicaption}{false}
\newcommand\myCaption{}
\newcommand\myCaptionTOC{}
\newcommand\myLabel{}
\newenvironment{exampleFig}[3][]%
  {\renewcommand\myCaption{#2}%
   \ifthenelse{\isempty{#1}}{\renewcommand{\myCaptionTOC}{#2}}{\renewcommand{\myCaptionTOC}{#1}}%
   \renewcommand\myLabel{#3} 
   \setcounter{figexmp}{1}
   \begin{figure}[H]\small\begin{center}
\begin{tabular}{|@{\,}p{0.48\textwidth}@{}|@{\,}p{0.48\textwidth}|@{}}\hline}%
  {\end{tabular}\end{center}%
    	\ifmulticaption{%
		\addtocounter{figure}{-1}\renewcommand{\figurecaption}{Figure \thefigure\ (cont'd)}\tocless\caption{\myCaptionTOC}}%
	\else{%
		\caption[\myCaptionTOC]{\myCaption}\label{\myLabel}\setboolean{multicaption}{true}}%
	\fi\end{figure}}

\newlength{\shortFigShift}
\newlength{\shortmedFigShift}
\newlength{\medFigShift}
\newlength{\longFigShift}
\newlength{\leftSpecial}
\newlength{\rightSpecial}
\newlength{\LineLeft}
\newlength{\LineRight}
\newcommand{\FigLine}[7][0]{%
  \setcounter{figexmpA}{\value{figexmp}}
  \stepcounter{figexmp} 
  (\alph{figexmpA})\ \parbox[t]{0.45\textwidth}{#2\\[-1.25ex]\mbox{}}&
  (\alph{figexmp})\ \parbox[t]{0.45\textwidth}{#5\\[-1.25ex]\mbox{}}\\
  \ifthenelse{\equal{#1}{0}}{\setlength{\LineLeft}{\shortFigShift}}{}
  \ifthenelse{\equal{#1}{speciallr}}{\setlength{\LineLeft}{\leftSpecial}}{}
  \ifthenelse{\equal{#1}{l}}{\setlength{\LineLeft}{\longFigShift}}{}
  \ifthenelse{\equal{#1}{sl}}{\setlength{\LineLeft}{\medFigShift}}{}
  \ifthenelse{\equal{#1}{xsl}}{\setlength{\LineLeft}{\shortmedFigShift}}{}
  \ifthenelse{\equal{#1}{r}}{\setlength{\LineLeft}{\shortFigShift}}{}
  \ifthenelse{\equal{#1}{sr}}{\setlength{\LineLeft}{\shortFigShift}}{}
  \ifthenelse{\equal{#1}{xsr}}{\setlength{\LineLeft}{\shortFigShift}}{}
  \ifthenelse{\equal{#1}{lr}}{\setlength{\LineLeft}{\longFigShift}}{}
  \ifthenelse{\equal{#1}{lsr}}{\setlength{\LineLeft}{\longFigShift}}{}
  \ifthenelse{\equal{#1}{lxsr}}{\setlength{\LineLeft}{\longFigShift}}{}
  \ifthenelse{\equal{#1}{slr}}{\setlength{\LineLeft}{\medFigShift}}{}
  \ifthenelse{\equal{#1}{slsr}}{\setlength{\LineLeft}{\medFigShift}}{}
  \ifthenelse{\equal{#1}{slxsr}}{\setlength{\LineLeft}{\medFigShift}}{}
  \ifthenelse{\equal{#1}{xslr}}{\setlength{\LineLeft}{\shortmedFigShift}}{}
  \ifthenelse{\equal{#1}{xslsr}}{\setlength{\LineLeft}{\shortmedFigShift}}{}
  \ifthenelse{\equal{#1}{xslxsr}}{\setlength{\LineLeft}{\shortmedFigShift}}{}
  \ifthenelse{\equal{#3}{}}{}{\parbox[c]{0.47\textwidth}{\hspace{\LineLeft}\includegraphics[]{#3}}}&%
  \ifthenelse{\equal{#1}{0}}{\setlength{\LineRight}{\shortFigShift}}{}
  \ifthenelse{\equal{#1}{speciallr}}{\setlength{\LineLeft}{\rightSpecial}}{}
  \ifthenelse{\equal{#1}{l}}{\setlength{\LineRight}{\shortFigShift}}{}
  \ifthenelse{\equal{#1}{sl}}{\setlength{\LineRight}{\shortFigShift}}{}
  \ifthenelse{\equal{#1}{xsl}}{\setlength{\LineRight}{\shortFigShift}}{}
  \ifthenelse{\equal{#1}{r}}{\setlength{\LineRight}{\longFigShift}}{}
  \ifthenelse{\equal{#1}{sr}}{\setlength{\LineRight}{\medFigShift}}{}
  \ifthenelse{\equal{#1}{xsr}}{\setlength{\LineRight}{\shortmedFigShift}}{}
  \ifthenelse{\equal{#1}{lr}}{\setlength{\LineRight}{\longFigShift}}{}
  \ifthenelse{\equal{#1}{lsr}}{\setlength{\LineRight}{\medFigShift}}{}
  \ifthenelse{\equal{#1}{lxsr}}{\setlength{\LineRight}{\shortmedFigShift}}{}
  \ifthenelse{\equal{#1}{slr}}{\setlength{\LineRight}{\longFigShift}}{}
  \ifthenelse{\equal{#1}{slsr}}{\setlength{\LineRight}{\medFigShift}}{}  
  \ifthenelse{\equal{#1}{slxsr}}{\setlength{\LineRight}{\shortmedFigShift}}{}  
  \ifthenelse{\equal{#1}{xslr}}{\setlength{\LineRight}{\longFigShift}}{}
  \ifthenelse{\equal{#1}{xslsr}}{\setlength{\LineRight}{\medFigShift}}{}  
  \ifthenelse{\equal{#1}{xslxsr}}{\setlength{\LineRight}{\shortmedFigShift}}{}  
  \ifthenelse{\equal{#6}{}}{}{\parbox[c]{0.47\textwidth}{\hspace{\LineRight}\includegraphics[]{#6}}}\stepcounter{figexmp}
  \ifthenelse{\isempty{#4}\AND\isempty{#7}}%
    {}
    {%
    \ifthenelse{\isempty{#4}}{\\&~\\[-2ex]\mbox{}&}{\\&~\\[-2ex]\mbox{}\store{#4}&}
    \ifthenelse{\isempty{#7}}{}{\store{#7}}
    }
  \\
  \hline
}
\newcommand{\FigLineHalf}[4][0]{%
  \setcounter{figexmpA}{\value{figexmp}}
  \stepcounter{figexmp}
  (\alph{figexmpA})\ \parbox[t]{0.45\textwidth}{#2\\[-1.25ex]\mbox{}}\\
  \ifthenelse{\equal{#1}{0}}{\setlength{\LineLeft}{\shortFigShift}}{}
  \ifthenelse{\equal{#1}{l}}{\setlength{\LineLeft}{\longFigShift}}{}
  \ifthenelse{\equal{#1}{sl}}{\setlength{\LineLeft}{\medFigShift}}{}
  \ifthenelse{\equal{#1}{xsl}}{\setlength{\LineLeft}{\shortmedFigShift}}{}
   \parbox[c]{0.47\textwidth}{\hspace{\LineLeft}\includegraphics[]{#3}}\stepcounter{figexmp}
    \ifthenelse{\isempty{#4}}%
    {}
    {\\~\\[-2ex]\mbox{}\store{#4}}
   \\\cline{1-1}
}
\newcommand{\FigNewline}{%
  \end{tabular}\end{center}%
  	\ifmulticaption{\addtocounter{figure}{-1}\renewcommand{\figurecaption}{Figure \thefigure\ (cont'd)}\tocless\caption{\myCaptionTOC}}%
	\else{\caption[\myCaptionTOC]{\myCaption}\label{\myLabel} }%
	\fi%
  \end{figure}\newpage\setboolean{multicaption}{true}%
  \begin{figure}[H]\small\begin{center}\begin{tabular}{|@{\,}p{0.48\textwidth}@{}|@{\,}p{0.48\textwidth}|@{}}\hline
}

\newcounter{Hequation}

\makeatletter
\g@addto@macro\equation{\stepcounter{Hequation}}
\makeatother

\def\titlerunning{Modelling, Verifying  and Analysing AODV}
\def\authorrunning{A. Fehnker, R.J. van Glabbeek, P. H\"ofner, A. McIver, M. Portmann \& W.L. Tan}

\title{A Process Algebra for Wireless Mesh Networks\\
{\large used for}\\
Modelling, Verifying  and Analysing AODV}
\author{Ansgar Fehnker
\institute{NICTA\thanks{NICTA is funded by the Australian Government through the Department of Communications and the Australian\newline Research Council through the ICT Centre of Excellence Program.}
\\ Sydney, Australia}
\institute{Computer Science and Engineering\\
University of New South Wales\\
Sydney, Australia}
\and
Rob van Glabbeek
\institute{NICTA$^*$\\ Sydney, Australia}
\institute{Computer Science and Engineering\\
University of New South Wales\\
Sydney, Australia}
\and
Peter H\"ofner
\institute{NICTA$^*$\\ Sydney, Australia}
\institute{Computer Science and Engineering\\
University of New South Wales\\
Sydney, Australia}
\and
Annabelle McIver
\institute{Department of Computing\\ Macquarie University\\
Sydney, Australia}
\institute{NICTA$^*$\\ Sydney, Australia}
\and
Marius Portmann
\institute{NICTA$^*$\\ Brisbane, Australia}
\institute{Information Technology and\\
Electrical Engineering\\
University of Queensland\\
Brisbane, Australia}
\and
Wee Lum Tan\institute{NICTA$^*$\\ Brisbane, Australia}
\institute{Information Technology and\\
Electrical Engineering\\
University of Queensland\\
Brisbane, Australia}
}
\begin{document}

\maketitle
\vspace{-0.35ex}
\begin{abstract}
Route finding and maintenance are critical for the performance of
networked systems, particularly when mobility can lead to highly
dynamic and unpredictable environments; such operating contexts are
typical in wireless mesh networks.  Hence correctness and good
performance are strong requirements of routing protocols.

In this paper we propose \awn (Algebra for
Wireless Networks), a process algebra
 tailored to the modelling of Mobile Ad hoc Network (MANET) and Wireless Mesh Network (WMN)
protocols. It combines novel
treatments of local broadcast, conditional unicast and data
structures.

In this framework, we 
present a rigorous analysis of the Ad hoc On-Demand
Distance Vector (AODV) protocol, a popular routing protocol designed for
MANETs and WMNs, and one of the four protocols currently defined as an RFC (request for comments)  by the
IETF MANET working group. 

We give a complete and unambiguous specification of
this protocol, thereby formalising the RFC of AODV, the de facto standard
specification, given in English prose. In doing so, we had 
to make non-evident assumptions to resolve ambiguities occurring in that specification.
Our formalisation models the exact details of the core
functionality of AODV, such as route maintenance and
error handling, and only omits timing aspects.

The process algebra
allows us to formalise and (dis)prove crucial properties of mesh network
routing protocols such as loop freedom and packet delivery. 
We are the first to provide a detailed proof of loop
freedom of AODV\@. In contrast to evaluations using simulation or other formal methods such as model checking, 
our proof is generic and holds for any possible network scenario in terms of network topology, node mobility, traffic pattern, etc.
Due to ambiguities and contradictions the RFC specification allows several readings. For this reason, we analyse multiple interpretations. 
In fact we show for more than $5000$ interpretations whether they are loop free or not.
Thereby we demonstrate how the reasoning and proofs can relatively easily be adapted to protocol variants.

Using our  formal and unambiguous specification, we find some
shortcomings of AODV that can easily affect\linebreak[1] performance.
Examples are non-optimal routes established
by AODV and the fact that some routes are not found at all.
These problems are analysed and improvements are suggested.
As the improvements are formalised in the same process algebra,
carrying over the proofs is again relatively easy.
\end{abstract}
\advance\textheight 13.6pt
\advance\textheight 13.6pt

\newpage
\tableofcontents

\advance\textheight -25.5pt 
\newpage
\setcounter{page}{1}
\pagenumbering{arabic}
\makeatletter
    \def\@oddhead{\thepage\hfil\slshape\leftmark}%
    \def\@evenhead{{\slshape\rightmark}\hfil\thepage}%
\makeatother

\section{Introduction}\label{sec:introduction}

Wireless Mesh Networks (WMNs) have gained considerable popularity and
are increasingly deployed in a wide range of application scenarios, including emergency
response communication, intelligent  transportation systems, mining, video surveillance, etc.
They are self-organising wireless multi-hop networks that can provide broadband communication
without relying on a wired backhaul infrastructure, a benefit for rapid and low-cost network deployment.
WMNs can be considered a superset of Mobile Ad hoc Networks (MANETs), where a network consists exclusively of mobile end user devices
such as laptops or smartphones. In contrast to MANETs, WMNs typically also contain stationary
infrastructure devices called mesh routers. 
\index{Wireless Mesh Network}%
\index{WMN|see{Wireless Mesh Network}}%
\index{MANET}%

An important characteristic of WMNs is that they operate in unpredictable environments with
highly dynamic network topologies---due to node mobility and the variable nature of wireless links.
Because of this, route finding and maintenance are critical for the performance of WMNs.
\index{routing protocol}%
\index{AODV}%
\index{RFC}%
\index{node}%
Usually, a routing protocol is used to establish and maintain network connectivity through paths
between source and destination node pairs. As a consequence, the routing protocol is one of the key
factors determining the performance and reliability of WMNs. One of the most popular
routing protocols that is widely used in WMNs is the Ad hoc On-Demand Distance
Vector (AODV) routing protocol~\cite{rfc3561}. It is one of the four protocols currently
standardised by the IETF MANET working group, and it also forms the basis of new WMN
routing protocols, including HWMP in the  IEEE 802.11s wireless mesh network
standard~\cite{IEEE80211s}. The details of the AODV protocol are laid out in the request-for-comments-document (RFC 3561~\cite{rfc3561}), a de facto standard.
However, due to the use of English prose,
this specification contains ambiguities and contradictions. This can lead to significantly different
implementations of the AODV routing protocol, depending on the developer's understanding and reading
of the AODV RFC\@. In the worst case scenario, an AODV implementation may contain serious flaws,
such as routing loops.

Traditional approaches to the analysis of AODV and many other AODV-based protocols~\cite{AODVv2,IEEE80211s,AODV-ST,SBM06,PPI08}
are simulation and test-bed experiments. While these are important and valid methods for protocol
evaluation, in particular for quantitative performance evaluation, they have limitations in regards
to the evaluation of basic protocol correctness properties. Experimental evaluation is resource
intensive and time-consuming, and, even after a very long time of evaluation, only a finite set of
network scenarios can be considered---no general guarantee can be given about correct protocol
behaviour for a wide range of unpredictable deployment scenarios~\cite{Verisim}.  This problem is
illustrated by recent discoveries of limitations in AODV-like protocols that have been under intense
scrutiny over many years \cite{MK10}.

We believe that formal methods can help in this regard; they complement simulation and test-bed
experiments as methods for protocol evaluation and verification, and provide stronger and more
general assurances about protocol properties and behaviour.
The overall goal is to reduce
the ``time-to-market'' for better (new or modified) WMN protocols, and
to increase the reliability and performance of the corresponding
networks.

\index{AWN}%
\index{conditional unicast}%
\index{local broadcast}%
The {\em first contribution} of this paper is \awn (Algebra of Wireless Networks), a process algebra that provides a step
towards this goal.  It combines novel treatments of data structures, conditional unicast and local
broadcast, and allows formalisation of all important aspects of a routing protocol.
All these features are necessary to model 
``real life'' WMNs. Data structures are used to store and maintain information such as routing 
tables. The conditional unicast construct allows us to
model that a node in a network sends a message to a particular neighbour,
\index{neighbour}%
and if this fails---for example because the receiver has moved out of transmission range---error handling
is initiated. 
Finally, the local broadcast primitive, which allows a node to send messages to all 
its immediate neighbours, models
the wireless broadcast mechanism implemented by the physical and 
data link layer of wireless standards relevant for WMNs. The formalisation assumes that any broadcast message \emph{is}
received by all nodes within transmission range.%
\index{lossy broadcast}%
\footnote{\label{guaranteed receipt}In reality, communication is only
half-duplex: a single-interface network node cannot receive messages while sending and hence
messages can be lost.
However, the CSMA protocol used at the link layer---not 
modelled by {\awn}---keeps the probability of packet loss due to two nodes (within
range) sending at the same time rather low.
Since we are examining imperfect protocols, we first of all want to
establish how they behave under optimal conditions. For this reason we
abstract from probabilistic reasoning by assuming no message loss at
all, rather than working with a lossy broadcast formalism that offers
no guarantees that any message will ever arrive.}
\index{deliver}%
This abstraction enables us to interpret a failure of route discovery (see our eighth contribution below) as an imperfection in the protocol,
rather than as a result of a chosen formalism not ensuring guaranteed receipt.

\advance\textheight 25.5pt 
As a {\em second contribution}, we give a complete and accurate formal specification of the core functionality of the AODV routing protocol 
using \awn. Our model covers all core components of AODV, but none of the optional
features, and abstracts from timing issues. The algebra provides the right level of abstraction to model  key features such as 
unicast and broadcast, while abstracting from implementation-related
details. As its semantics is completely unambiguous, specifying a
protocol in such a framework enforces total precision and the removal
of any ambiguity.

The {\em third contribution} is to demonstrate how {\awn} can be used to
support reasoning about protocol behaviour and to provide rigorous proofs of key protocol
properties, using the examples of route correctness and loop freedom.
\index{route correctness}%
\index{loop freedom}%
In contrast to what can be achieved by model checking or test-bed
experiments, our proofs apply to all conceivable dynamic network topologies.
Route correctness is a minimal sanity requirement for a routing protocol; it is the property
that the routing table entries stored at a node are entirely based on information on routes to other
nodes that either is currently valid or was valid at some point in the past.
Loop freedom is a critical property for any routing protocol, but it is particularly relevant and challenging for WMNs.
Descriptions as in~\cite{Garcia-Luna-Aceves89} capture the common understanding of loop freedom:
``{A routing-table loop is a path specified in the nodes' routing tables at a particular point in time that visits the same node more than once before reaching the intended destination.}"
Packets caught in a routing loop,  until they are discarded by the IP Time-To-Live (TTL) mechanism, can quickly saturate the links and have a detrimental impact on network performance. It is therefore critical to ensure that protocols prevent routing loops.
We show that loop freedom can be guaranteed only if sequence numbers are used in a careful way, considering further rules and assumptions on the behaviour of the protocol.
The problem is, as shown in the case of AODV, that these additional rules and assumptions are not explicitly stated in the RFC, 
and that the RFC has significant ambiguities in regards to this.
To the best of our knowledge we are the first to give a complete and detailed proof of loop freedom.%
\footnote{Loop freedom of AODV has been ``proven'' at least twice~\cite{AODV99,ZYZW09}, but
 the proof in \cite{AODV99} is not correct, and the one in \cite{ZYZW09} is based on a simple
   subset of AODV only, not including the ``intermediate route reply'' feature---a most likely
   source of loops.}
This is our {\em fourth contribution}.

As  a {\em fifth contribution}, we show
details of several ambiguities and contradictions found in the AODV RFC,
and discuss which interpretations (plausible and consistent readings of the RFC) will lead to
routing loops, and which are loop free.
\index{interpretation}%
In fact we analyse more than $5000$ interpretations.
Hereby we demonstrate how our reasoning and proofs can relatively easily
be adapted to protocol variants. In particular, our {\em sixth contribution}, we demonstrate
that routing loops can be created---while fully complying with the RFC, and
making reasonable assumptions when the RFC allows different interpretations.
As our next contribution, we also analyse five 
key implementations of the AODV protocol and show that three of them can produce routing loops.

As an {\em eighth contribution}, we apply linear-time temporal logic (LTL) to formulate temporal
\index{route discovery property}%
properties of routing protocols, such as \emph{route discovery}: ``if a route discovery process is
initiated in a state where the source node is connected to the destination and during this process no (relevant)
link breaks, then the source will eventually discover a route to the destination'' and
\index{packet delivery property}%
\index{deliver}%
\emph{packet delivery}, saying that under certain circumstances a packet will surely be
delivered to its destination. We moreover show that AODV does not satisfy these properties.

In order for the last result to be meaningful, we first develop a general method to augment
a protocol specification with a \emph{fairness component} that requires that certain
 fairness properties are met, and apply this method to our
specification of AODV\@.  We also adapt the semantics of LTL in order to make a protocol
specification satisfy natural progress and justness properties. Without ensuring these
properties, temporal properties like route discovery and packet delivery would trivially
fail to hold. The same would apply if we had not assumed guaranteed receipt of broadcast
messages by nodes within transmission range (cf.\ Footnote~\ref{guaranteed receipt}).
 
Last but not least, we discuss several limitations of the AODV protocol and propose solutions to
them. We show how our formal specification can be used to analyse the proposed modifications and
show that the resulting AODV variants are loop free.

The rigorous protocol analysis discussed in this paper has the
potential to save a significant amount of time in the development and
evaluation of new network protocols, can provide increased levels of
assurance of protocol correctness, and complements simulation and
other experimental protocol evaluation approaches.

This paper is organised as follows:
\Sect{aodv} gives an informal introduction to AODV\@.
\Sect{abstractions} describes which features of the AODV protocol are
modelled in this paper, and which are not.
In \Sect{process_algebra} we introduce the process algebra AWN\@.\footnote{Major parts of this
  section have been published in ``A Process Algebra for Wireless Mesh Networks''~\cite{ESOP12}.%
  $^\textrm{\ref{crossreferrata}}$} 
 \Sect{modelling_AODV} provides a detailed formal specification of AODV\@ in AWN\@.%
\footnote{Parts of
  the specification are published in \cite{ESOP12}, in ``Automated Analysis of AODV using
  UPPAAL''~\cite{TACAS12} and in ``A Rigorous Analysis of AODV and its
  Variants''~\cite{MSWIM12}.\footnotemark}%
\footnotetext{The references in \cite{ESOP12,MSWIM12} to Prop 7.10(b), Sect.~8 and Sect.~9.1 of this paper, are now to
  Prop.~\ref{prop:msgsending}(b), Sect.~\ref{sec:properties} and Sect.~\ref{ssec:decreasingSQN}.\label{crossreferrata}}
To achieve this, we present the basic data structure needed in \Sect{types}.
In \Sect{invariants} we formally prove some properties of AODV that can be expressed as invariants,
in  particular loop freedom and route correctness.\footnote{A sketch of the loop freedom proof is given in \cite{ESOP12} and in~\cite{MSWIM12}.}

In \Sect{interpretation} we discuss and formalise many ambiguities, contradictions and cases of
unspecified behaviour in the RFC\@, and present an inventory of their plausible resolutions.
Combining the resolutions of the various ambiguities leads to 5184 possible interpretations of the RFC\@.
We show which of these interpretations lead to routing loops or other unacceptable behaviour.
For the remaining interpretations we show loop freedom and route correctness, through small
adaptations in the proofs given in \Sect{invariants}. We also analyse five implementations of AODV\@.%
\footnote{A summary of this section appeared in
  ``Sequence Numbers Do Not Guarantee Loop Freedom---{AODV} Can Yield Routing Loops'' \cite{AODVloop}.}

In \Sect{properties} we propose a general framework to ensure progress, fairness and
justness properties, and apply the proposal to augment our AODV specification with a fairness component.
Subsequently, we formulate two temporal properties (route discovery and packet delivery) that AODV-like protocols should satisfy,
and demonstrate that AODV does not enjoy these properties.
\Sect{analysingAODV} discusses several shortcomings of AODV and proposes five ways in which the
protocol can be improved. All improvements are formalised in \awn, and we show that they enjoy loop
freedom and route correctness.\footnote{Two of the improvements
  from this section are presented in \cite{MSWIM12}.}
\Sect{related work} describes related work, and in \Sect{conclusion} we summarise our findings and
point at work that is yet to be done.

\newpage

\section{Ad hoc On-Demand Distance Vector Routing Protocol}\label{sec:aodv}
AODV~\cite{rfc3561} is a widely-used routing protocol designed for
MANETs, and is one of the four protocols currently standardised by the
IETF MANET working group\footnote{\url{http://datatracker.ietf.org/wg/manet/charter/}}.
It also forms the basis of new WMN routing protocols, including the
upcoming IEEE 802.11s wireless mesh network standard~\cite{IEEE80211s}.

\subsection{Basic Protocol}

\index{route}%
\index{node}%
\index{source}%
\index{destination}%
\index{topology}%
AODV is a reactive protocol: routes are established only on demand. A
route from a source node $s$ to a destination node $d$ is a sequence
of nodes $[s,n_1,\dots,n_k,d]$, where $n_1$, $\dots$, $n_k$ are
intermediate nodes located on the path from $s$ to $d$.  Its basic
operation can best be explained using a simple example topology shown
in \Fig{topology}(a), where edges connect nodes within
transmission range. We assume node~$s$ wants to send a data packet to
node~$d$, but $s$ does not have a valid routing table entry for $d$.
\index{data packet}%
\index{routing table entry}%
\index{route discovery process}%
\index{route request (RREQ)}%
\index{RREQ message}%
\index{broadcast}%
Node $s$ initiates a route discovery mechanism by broadcasting a
route request (RREQ) message, which is received by $s$'s immediate
\index{neighbour}%
\index{destination}%
neighbours $a$ and $b$. We assume that neither $a$ nor $b$ knows a
route to the destination node $d$.\footnote{In case an intermediate
  node knows a route to $d$, it directly sends a route reply back.}
Therefore, they simply re-broadcast the message, as shown in
\Fig{topology}(b). Each RREQ message has a unique identifier
which allows nodes to ignore duplicate RREQ messages that they have
handled before.

\index{routing table}%
\index{reverse route}%
\index{forwarding}%
When forwarding the RREQ message, each intermediate node updates its
routing table and adds a ``reverse route'' entry to $s$, indicating
via which next hop the node $s$ can be reached, and the distance in
number of hops. Once the first RREQ message is received by the
\index{destination}%
destination node $d$ (we assume via $a$), $d$ also adds a reverse
route entry in its routing table, saying that node $s$ can be reached
via node $a$, at a distance of $2$ hops.

\index{route reply (RREP)}%
\index{RREP message}%
\index{unicast}%
Node $d$ then responds by sending a route reply (RREP) message back to
node $s$, as shown in \Fig{topology}(c). In contrast to the
RREQ message, the RREP is unicast, i.e., it is sent to an individual
next-hop node only. The RREP is sent from $d$ to $a$, and then to $s$,
using the reverse routing table entries created during the forwarding
\index{forwarding}%
of the RREQ message. When processing the RREP message, a node creates
\index{forward route}%
a ``forward route'' entry into its routing table. For example, upon
receiving the RREP via $a$, node $s$ creates an entry saying that $d$
can be reached via $a$, at a distance of $2$ hops. At the completion
\index{route discovery process}%
of the route discovery process, a route has been established from $s$
to $d$, and data packets can start to flow.
\index{data packet}%
\begin{figure}[t]
 \begin{center}
   \begin{tabular}[b]{r@{}l@{\hspace{13mm}}r@{}l@{\hspace{13mm}}r@{}l}
   (a)&
   \includegraphics[scale=1]{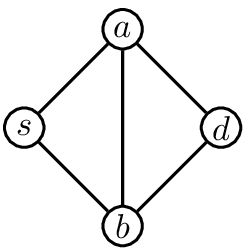}&
   (b)&
   \includegraphics[scale=1]{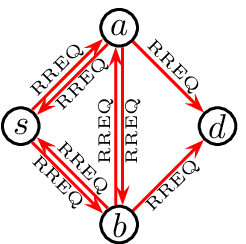}&
   (c)&
   \includegraphics[scale=1]{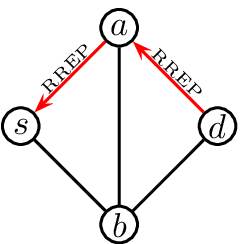}
   \end{tabular}
   \caption{Example network topology}
   \label{fig:topology}
 \end{center}
 \vspace*{-2.1ex}%
\end{figure}

\index{route error (RERR)}%
\index{RERR message}%
\index{sequence number}%
In the event of link and route breaks, AODV uses route error (RERR)
messages to inform affected nodes. Sequence numbers are another
important aspect of AODV, and are used to indicate the freshness of
routing table entries for the purpose of preventing routing loops.

\subsection{Detailed Examples}\label{ssec:detExample}

\index{sequence number}%
\index{routing table}%
\index{routing table entry}%
Each node \dval{ip} stores and maintains its own sequence number and 
its own routing table, which consists of exactly
one entry for each known destination \dval{dip}.
\index{destination}%
\index{next hop}%
\index{route}%
In this paper we represent a routing table entry
as a tuple $(\dval{dip}\comma\dval{dsn}\comma\dval{dsk}\comma\dval{flag}\comma\dval{hops}\comma
\dval{nhip}\comma\dval{pre})$, indicating that \dval{nhip} is the next hop
on a route to \dval{dip} of length \dval{hops};
\dval{dsn} is a sequence number measuring the freshness of this information.
\index{destination sequence number}%
\index{sequence number}%
\index{sequence-number-status flag}%
\index{sequence number!known}%
\index{sequence number!unknown}%
The flag \dval{dsk} indicates if the sequence number is known (\kno) or 
unknown (\unkno). In the former case the sequence number
\dval{dsn} can be used to measure the freshness; in the latter
the value of \dval{dsn} cannot be used since one cannot ``trust'' the value.
The flag \dval{flag} indicates if the route is
\index{route!valid}%
\index{route!invalid}%
\emph{valid} (\val)---it can be used to forward packets---or if it is outdated (\inval).
\index{neighbour!interested|see{precursors}}%
Finally, $\dval{pre}$ is the set of neighbours who are
``interested'' in the route to $\dval{dip}$---they are expected to use
\dval{ip} as the next hop in their own routes to \dval{dip}.
\index{next hop}%

We illustrate the AODV routing protocol in the
example of \Fig{example1}, where AODV is used to establish
a route between nodes $a$ and $c$. The small numbers inside 
the nodes denote the nodes' sequence numbers. Initially all these 
numbers are set to $1$.
\index{sequence number}%
For simplicity, we leave out the last component \dval{pre} of routing
table entries; hence each entry is a $6$-tuple here.

\index{data packet}%
\index{destination}%
\index{routing table}%
\index{routing table entry}%
\Fig{example1}(a) shows the initial state. We assume that
node $a$ wants to send a data packet to node $c$. First, $a$ checks
its routing table and finds that it does not have a (valid) routing
table entry for the destination node $c$. In fact its routing table 
is empty. Therefore it initiates a
\index{route discovery process}%
\index{RREQ message}%
route discovery process by generating a RREQ message.  For ease of
explanation, we represent the generated RREQ message as
$\rreq{\dval{hops}}{\dval{rreqid}}{\dval{dip}}{\dval{dsn}}{\dval{dsk}}{\dval{oip}}{\dval{osn}}{\dval{sip}}$,
indicating that the route request originates from node $\dval{oip}$
with sequence number $\dval{osn}$, searching for a route to
destination $\dval{dip}$ with sequence number at least $\dval{dsn}$.
\index{destination sequence number}%
\index{sequence number}%
This sequence number is taken from the entry for \dval{dip} in the
routing table maintained by node $a$.  If no entry for $\dval{dip}$ is
available, \dval{dsn} is set to $0$. If
there is no entry for $\dval{dip}$ or the sequence number is marked as unknown in the 
routing table, $\dval{dsk}$ is set to {\unkno} (``unknown''); otherwise it 
is set to {\kno} (``known'').
In addition, \dval{hops} is the number of hops the message has already
travelled from $\dval{oip}$, $\dval{rreqid}$ is the unique identifier
of the route request, and \dval{sip} denotes the sender of the
\index{sender}%
message.\footnote{Following the RFC specification of AODV, the sender
address $\dval{sip}$ is not part of the message itself; however a node
that receives a message is able to obtain it from the source IP
address field in the IP header of the message.}

\index{RREQ message}%
\index{RREQ message!originator}%
\index{sequence number}%
When generating a new RREQ message, the originator node must increment
its own sequence number before copying it into the RREQ
message. Therefore, the RREQ message from node $a$ is
$\rreq{0}{\dval{rreqid}}{c}{0}{\unkno}{a}{2}{a}$. This RREQ message is broadcast to
\index{neighbour}%
all its neighbours (\Fig{example1}(b)).\linebreak[3]\mbox{}\vspace{-13.6pt}
\begin{exampleFig}{Simple example}{fig:example1}
\FigLine[lr]%
  {$a$ wants to send a packet to $c$.}{fig/ex_std1_1}{}
  {$a$ broadcasts a new RREQ message;\\nodes $b,d$ receive the RREQ and update their RTs.}
  {fig/ex_std1_2}{}
\FigLine%
  {$d$ forwards the RREQ; node $a$ receives it;\\$b$ forwards the RREQ; nodes $a,c$ receive it.}{fig/ex_std1_3}{}
  {$c$ unicasts a RREP message to $b$.}{fig/ex_std1_4}{}
\FigNewline
\FigLineHalf%
  {$b$ unicasts the RREP to $a$.}{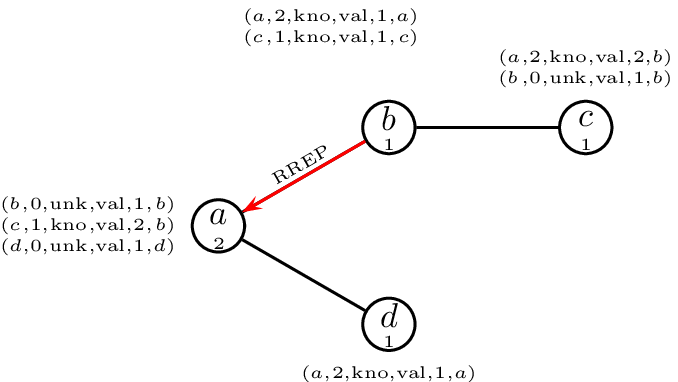}{}
\end{exampleFig}

\noindent
Nodes $b$ and $d$ receive the
request and update their routing tables to insert an entry for node
$a$. Since nodes $b$ and $d$ do not know a route to node $c$ (they
\index{destination}%
\index{routing table}%
\index{routing table entry}%
have no routing table entry with destination $c$), they both
re-broadcast the RREQ message, as shown in \Fig{example1}(c). Before
\index{forwarding}%
forwarding the message, nodes $b$ and $d$ increment the $\dval{hops}$
information in the RREQ message from $0$ to $1$, meaning that the
distance to $a$ is now $1$.

\index{RREQ message}%
The forwarded RREQ messages from nodes $b$ and $d$ are then received
by node $a$. Through these messages node $a$ knows that nodes $b$ and
$d$ are $1$-hop neighbours, but node $a$ does not know 
\index{neighbour}%
\index{sequence number!unknown}%
\index{routing table entry}%
\index{sequence-number-status flag}%
\index{RREQ message!originator}%
their sequence numbers, hence they are set to ``unknown''. Therefore, node $a$ creates routing table
entries for its neighbours, but with unknown sequence number $0$, and sequence-number-status flag set to \unkno.
Apart from this, since node $a$ is the originator of the RREQ,
it will ignore these messages.

The same RREQ message forwarded by node $b$ is also received by node
$c$. Node $c$ reacts by creating routing table entries for both its
previous-hop neighbour (node $b$) and the originator of the RREQ
 (node $a$). It then responds by generating a RREP
\index{neighbour}%
\index{RREP message}%
\index{sender}%
message. Again, for ease of explanation, we represent the RREP message
as $\rrep{\dval{hops}}{\dval{dip}}{\dval{dsn}}{\dval{oip}}{\dval{sip}}$,
where $\dval{hops}$ now indicates the distance to $\dval{dip}$.  As
before, $\dval{sip}$ is the sender of the message. Since the
\index{destination}%
\index{destination sequence number}%
\index{sequence number}%
destination node's sequence number specified in the received RREQ
message is unknown ($\dval{dsn} = 0$ and $\dval{dsk} = \unkno$), node $c$ copies its own
sequence number into the RREP message.  Hence the RREP message from
node $c$ is $\rrep{0}{c}{1}{a}{c}$.

\index{unicast}%
\index{RREQ message!originator}%
\index{routing table}%
\index{routing table entry}%
\index{route discovery process}%
\index{data packet}%
\index{forwarding}%
From node $c$, the RREP message is unicast back to its previous-hop node $b$, on the path back towards the originator node $a$ (\Fig{example1}(d)). Node $b$ processes the RREP message and updates its routing table to insert an entry for node $c$. It also increments the $\dval{hops}$ information in the RREP message from $0$ to $1$ before forwarding it to node $a$ (\Fig{example1}(e)). When node $a$ receives the RREP message, this completes the route discovery process and a route is now established from node $a$ to node $c$. Data packets from node $a$ can now be sent to node $c$.

\index{topology}%
We next describe a more interesting example of how AODV operates in a changing network topology. In this example, we will show that due to the changing network topology and subsequent updates to the routing table, a route reply message is not necessarily sent back to the node which had forwarded the route request previously.

\index{data packet}%
\index{RREQ message}%
\Fig{example2}(a) shows the initial network topology, and the initial state of the nodes in the topology. We assume that node $s$ wants to send a data packet to node $d$; hence it generates and broadcasts a route request message RREQ${}_1$ ($\rreq{0}{\dval{rreqid}}{d}{0}{\unkno}{s}{2}{s}$), as shown in \Fig{example2}(b).

\index{topology}%
Next, the network topology changes whereby node $s$ is now within
transmission range of node~$d$. This change in the network
\index{node!mobility}%
topology can be due to node mobility (i.e., node $s$ moves into 
transmission range of node $d$), or due to the improved quality
of the wireless link between nodes~$s$ and~$d$. \Fig{example2}(d)
\index{data packet}%
shows a situation where node $s$ wants to send a data packet to node
$a$, thereby generating and broadcasting a new route request message
RREQ${}_2$ ($\rreq{0}{\dval{rreqid}}{a}{0}{\unkno}{s}{3}{s}$)
\index{RREQ message}%
\index{destination}%
\index{routing table}%
\index{routing table entry}%
\index{destination sequence number}%
\index{sequence number}%
destined to node $a$. Note that RREQ${}_2$ is received by node $d$,
which results in the insertion of an entry for node $s$ with sequence
number $3$ in node $d$'s routing table. At the same time, the previous
route request RREQ${}_1$ is forwarded to node $b$ on its path towards
node $d$.
\begin{exampleFig}{An example with changing network topology}{fig:example2}
\FigLine[lr]%
  {The initial state.}{fig/ex_unexpect1_1}{}
  {$s$  broadcasts a new RREQ message destined to $d$.}{fig/ex_unexpect1_2}{}
  \multicolumn{2}{|@{\,}l@{\,}|}{
    (\alph{figexmp}) Network topology changes; $s$ moves into the transmission range of $d$}\stepcounter{figexmp}{}\\
    \hline
\FigLine[lsr]%
  {$s$  broadcasts a new RREQ message destined to $a$;\\RREQ${}_1$ is forwarded.}{fig/ex_unexpect1_3}{}
  {$d$ forwards RREQ${}_{2}$; nodes $a,b,s$ receive it;\\$b$ updates its routing table entry to $s$.}{fig/ex_unexpect1_4}{}
\multicolumn{2}{|@{\,}l@{\,}|}{
     (\alph{figexmp}) All steps that would follow as a reaction to RREQ${}_2$
     are skipped, because they are not important here.\stepcounter{figexmp}}\\
    \hline
\FigLine[slsr]%
  {$b$ forwards RREQ${}_1$; it is received by node $d$.}{fig/ex_unexpect1_5}{}
  {$d$ generates a reply to RREQ${}_1$;\\this reply is {\em not sent} back to $b$; it is sent to $s$.}{fig/ex_unexpect1_6}{}
\end{exampleFig}

\index{forwarding}%
\Fig{example2}(e) shows that node $d$ forwards RREQ${}_2$, which is received by nodes $a$, $b$, and $s$. The subsequent steps in response to RREQ${}_2$, i.e. the generation of a RREP message by node $a$, and its unicast to node $d$ and subsequent forwarding to the originator node $s$, are not shown in \Fig{example2} as they do not contribute towards the objective of this example.

\index{destination}%
\index{next hop}%
\index{routing table}%
\index{routing table entry}%
\index{RREP message}%
\index{destination sequence number}%
\Fig{example2}(g) shows that RREQ${}_1$ is forwarded by node $b$ and finally received by the destination node $d$. Since the destination sequence number for node $s$ in RREQ${}_1$ ($\dval{dsn}=2$) is older than the corresponding destination sequence number information in node $d$'s routing table entry for node $s$ ($\dval{dsn}=3$),
 the routing table entry for node $s$ is not updated. Node $d$ then generates a RREP message in response to RREQ${}_1$.  The destination node $d$ searches in its routing table for a reverse route entry for node $s$, and finds that the next hop \dval{nhip} for the route towards node $s$ is node $s$ itself. Therefore, the RREP message is not sent back to node $b$ (from which the RREQ${}_1$ message is received), but instead is sent back directly to node $s$ (\Fig{example2}(h)).

\newpage

\section{Abstractions Chosen}\label{sec:abstractions}
\index{RFC}%
Our formalisation of AODV tries to accurately model the protocol as
defined in the IETF RFC 3561 specification \cite{rfc3561}.  The model
focusses on layer $3$ of the protocol stack, i.e., the routing and
forwarding of messages and packets, and abstracts from lower layer
network protocols and mechanisms such as the Carrier Sense Multiple
Access (CSMA) protocol.  The presented formalisation includes all core
components of the protocol, but, at the moment, abstracts from timing
issues and optional protocol features. This keeps our specification
manageable. Our plan is to extend our model step by step.  Even though
our model currently does not cover all aspects, it allows us to point
to shortcomings in AODV and to discuss some possible improvements.
The model also allows us to reason about protocol behaviour and to
prove critical protocol characteristics.

In this section, we list all items that are not yet part of our formal model.

\subsection{Timing}
\index{timing}%
We abstract from \emph{all timing issues}. Surely, this is a big
decision and there are good reasons to add time as a next
step. However, this abstraction makes the verification of properties
much easier:

No entry of a routing table or route reply message has the field
\phrase{lifetime} that maintains the expiration or deletion time of the
route in AODV\@.  Informally this means that no valid route is set to
invalid due to timing, that no invalid route disappears from the
routing table (except when it is overwritten), and that we never
delete elements of the set $\rreqs$ of already seen requests
(described in \SSect{rreqs}).  In terms of the RFC that
means that \verb|ACTIVE_ROUTE_TIMEOUT|, \verb|DELETE_PERIOD| and
\verb|PATH_DISCOVERY_TIME| are set to infinity.

\subsection{Optional Protocol Features}
\index{local repair}%
A route may be \emph{locally repaired} if a link break in a valid
route occurs.  In that case, the node upstream of that break may
choose to initiate a local repair if the destination was no farther
than \verb|MAX_REPAIR_TTL| hops away.  Local repair is optional;
therefore we do not model this feature here.

To avoid unnecessary network-wide dissemination of RREQs, the
originating node should use an \emph{expanding ring search} technique.
This is again an optional feature, which is not modelled here; we can
say that the \verb|RING_TRAVERSAL_TIME| is set to infinity.

A route request may be sent multiple times.  This happens if a node,
after broadcasting a RREQ, does not receive the corresponding RREP
within a given amount of time. In that case the node may broadcast
another RREQ, up to a maximum of \verb|RREQ_RETRIES|.  Since the
default value for \verb|RREQ_RETRIES| is only two, and moreover this
whole procedure is optional, we have not modelled this resending of
RREQ messages.

If a route discovery has been attempted \verb|RREQ_RETRIES| times
without receiving any RREP, a \emph{destination unreachable message}
should be delivered to the client (application) hooked up at the
originator.  This interaction between different layers of the protocol
stack has not been modelled here since it is not a core part of the
protocol itself.

When a node wants to increment its sequence number, but the largest
possible number ($2^{32}-1$) has already been assigned to it, a
\emph{sequence number rollover} has to be accomplished.  This rollover
violates the property that sequence numbers are monotonically
increased over time; therefore it would be possible to create routing
loops.  It appears that loops as a consequence of rollover are rare in
practice and therefore we decided to model sequence numbers by the
unbounded set of natural numbers.

\emph{Interfaces}, as part of routing table entries, store information
concerning the network link, e.g., that the node is connected via
Ethernet.  This is because AODV should operate smoothly over wired as
well as wireless networks.  Here we assume that nodes have only one
type of network interface and consequently leave out this field.

Another phenomenon which may yield complications and possibly routing
loops, are node crashes.  For now, we have neither modelled crashes
nor \emph{actions after reboot}.

\index{bidirectional links}%
By default, our process algebra establishes only bidirectional
links\footnote{A bidirectional link means that if a node $b$ is in
transmission range of $a$ ($a$ can send messages to $b$), then $a$ is
also in range of $b$. A bidirectional link does \emph{not} mean that
if $a$ knows a route to $b$, then $b$ knows a route to $a$.}.  We will
point out how by a trivial change it can model unidirectional links
(\Sect{process_algebra}). We have decided not to make this our default
here, since, by doing so, fundamental properties such as route
correctness would not hold for AODV any longer (see \SSect{route correctness}).
\index{unidirectional links}%
Unidirectional links come along with ``blacklist'' sets, which we also do not model.

We further do not model the optional support for aggregate networks
and the use of AODV with other networks, as loosely discussed in
Sections 7 and 8 of the AODV RFC\@~\cite{rfc3561}.

\index{neighbour}
Finally, \emph{hello messages} can be used as an optional feature to
offer connectivity information to a node's neighbours. Since in our
model all optional parts are skipped, we do not model hello messages
either; information about $1$-hop neighbours is established by
receiving AODV control messages.

\subsection{Flags}
\index{flag}%
\index{control message}%
\index{routing table entry}%
Following the RFC~\cite{rfc3561}, AODV control messages and routing
table entries have to maintain a series of state and routing flags
such as the repair flag, the unknown sequence number flag, and the
\index{gratuitous RREP flag}%
gratuitous RREP flag.  For most of these flags there is no compulsion
to ever set them.  An exception is the \phrase{unknown sequence number
\protect\mbox{(`U')} flag}.
In some implementations, such as AODV-UU~\cite{AODVUU}, this flag
is omitted in favour of a special element denoting the unknown sequence number.
In our model, we follow
the RFC and model the sequence 
number as well as the `U' flag. We speak of a
\phrase{sequence-number-status flag}, with values ``known'' and ``unknown''.
\index{sequence number!known}%
\index{sequence number!unknown}%

Besides the `U' flag, each route request has the \emph{join} (`J'),
the \emph{repair} (`R'), the \emph{gratuitous RREP} (`G') and the
\emph{destination only} (`D') \emph{flag}.  The `J' and `R' flag are
reserved for \emph{multicast}, an optional feature not fully specified
in the RFC.  We do not model the multicast feature, and hence ignore these two flags.
The `G' flag indicates whether a gratuitous RREP should be unicast, by
an intermediate node answering the RREQ message, to the destination
node of the original message; the `D' flag indicates that only the
destination may respond to this RREQ\@. Both flags may be set when a
request is initiated. Since this is also optional, we have decided to
skip these features for the moment. However, their inclusion should be
straightforward.

A route reply carries two flags: the \emph{repair} (`R') flag, used
for the multicast feature, and the \emph{acknowledgment} (`A') \emph{flag}, which
indicates that a route reply acknowledgment message must be sent in
response to a RREP message. We do not model these flags: the former
since we do not model multicast at all; the latter since this flag is
optional. Consequently, we have no need to model the
\phrase{route reply acknowledgment (RREP-ACK)} \emph{message}, which---next to
RREQ, RREP and RERR---constitutes a fourth kind of AODV control message.

Finally, an error message only maintains the \emph{no delete} (`N')
\emph{flag}.  It is set if a node has performed a local repair. Since
we do not model local repair, we are able to abstract from that flag.

Flags pertaining to local repair, but stored in the routing tables,
are the \emph{repairable} and the \emph{being repaired flags}.  For
the same reasons, we skip these flags as well.

\newpage

\section{A Process Algebra for Wireless Mesh Routing Protocols}
\label{sec:process_algebra}
\index{AWN}%
\index{process algebra}%
In this section we propose \awn (Algebra of Wireless Networks), a process algebra for the
specification of WMN routing protocols, such as AODV\@.
It is a variant of
standard process algebras \cite{Mi89,Ho85,BK86,BB87}, adapted to the
problem at hand.  For example, it allows us to embed
\index{data structure}%
data structures.  In \awn, a WMN is modelled as an encapsulated
\index{encapsulation}%
\index{parallel composition}%
\index{node}%
parallel composition of network nodes.  On each node several
\index{sequential processes}%
sequential processes may be running in parallel.  Network nodes
\index{neighbour}%
communicate with their direct neighbours---those nodes that are in
\index{broadcast}%
\index{unicast}%
\index{transmission range}%
transmission range---using either broadcast or unicast.  Our
formalism maintains for each node the set of nodes that are currently
\index{node!mobility}%
in transmission range.  Due to mobility of nodes and variability of
wireless links, nodes can move in or out of transmission range.  The
encapsulation of the entire network inhibits communications between
network nodes and the outside world, with the exception of the receipt
\index{deliver}%
\index{data packet}%
\index{client}%
\index{application layer}%
and delivery of data packets from or to clients\footnote{The
application layer that initiates packet sending and awaits receipt of
a packet.}  of the modelled protocol that may be hooked up to various
nodes.

\subsection{A Language for Sequential Processes}

\index{variables}%
\index{data variables}%
\index{data structure}%
\index{data expressions}%
\index{data types}%
\index{data formulas}%
\index{predicate logic}%
The internal state of a process is determined, in part, by the values
of certain data variables that are maintained by that process.  To
this end, we assume a data structure with several types, variables
ranging over these types, operators and predicates. First order
predicate logic yields terms (or \emph{data expressions}) and formulas
to denote data values and statements about them.\label{pg:undefvalues}\footnote{As
    \label{fn:undefvalues}%
    operators we also allow \emph{partial} functions with the
    convention that any atomic formula containing an undefined subterm
    evaluates to {\tt false}.} Our data structure
always contains the types \tDATA, \tMSG, {\tIP} and $\pow(\tIP)$ of
\phrase{application layer data}, \phrase{messages}, \phrase{IP addresses}---or any
other node identifiers---and \emph{sets of IP addresses}. We further assume that 
there is a function $\newpktID:\tDATA \times \tIP \rightarrow \tMSG$
that generates a message with new application layer data for a particular destination. 
\index{newpkt@$\newpktID$}%
The purpose of this function is to inject data to the protocol; details will be given later.

In addition, we assume a type \keyw{SPROC} of \phrase{sequential processes},
and a collection of \phrase{process names}, each being an operator of type
$\keyw{TYPE}_1 \times \cdots \times \keyw{TYPE}_n \rightarrow \keyw{SPROC}$
for certain data types $\keyw{TYPE}_i$.
Each process name $X$ comes with a \phrase{defining equation}
\[X(\keyw{var}_1\comma\ldots\comma\keyw{var}_n) \stackrel{{\it def}}{=} p\ ,\]
in which, for each $i=1,\ldots,n$, $\keyw{var}_i$ is a variable of type
$\keyw{TYPE}_i$ and $p$ a \emph{sequential process expression}
defined by the grammar below. $p$ may contain the variables
$\keyw{var}_i$ as well as $X$; however, all occurrences of data
variables in $\p$ have to be \phrase{bound}.\footnote{An occurrence of a
  data variable in $p$ is \emph{bound} if it is one of the variables
  $\keyw{var}_i$, a variable {\msg} occurring in a subexpression
  $\receive{\msg}.\q$, a variable \keyw{var} occurring in a subexpression
  $\assignment{\keyw{var}:=\dexp{exp}}\q$, or an occurrence in a
  subexpression $\cond{\varphi}\q$ of a variable occurring free in $\varphi$.
  Here $\q$ is an arbitrary sequential process expression.}
The choice of the underlying data structure and the process names with
their defining equations can be tailored to any particular application
of our language; our decisions made for modelling AODV are presented
in Sections~\ref{sec:types} and~\ref{sec:modelling_AODV}.  The process
names are used to denote the processes that feature in this
application, with their arguments $\keyw{var}_i$ binding the current
values of the data variables maintained by these processes.

The \phrase{sequential process expressions} are given by the following grammar:
$$\begin{array}[t]{@{}l@{}}
\SP ::= X(\dexp{exp}_1\comma\ldots\comma\dexp{exp}_n) ~\mid~ \cond{\varphi}\SP ~\mid~ \assignment{\keyw{var}:=\dexp{exp}} \SP
  ~\mid~ \SP+\SP ~\mid ~ \alpha.\SP ~\mid~ \unicast{\dexp{dest}}{\dexp{ms}}.\SP \prio \SP \\
\alpha ::=
  \broadcastP{\dexp{ms}} ~\mid~ \groupcastP{\dexp{dests}}{\dexp{ms}}
  ~\mid~ \send{\dexp{ms}} ~\mid~ \deliver{\dexp{data}} ~\mid~ \receive{\msg}
\end{array}$$
\index{assignment}%
Here $X$ is a process name, $\dexp{exp}_i$ a data expression of the
same type as $\keyw{var}_i$, $\varphi$ a data formula,
$\keyw{var}\mathop{:=}\dexp{exp}$ an assignment of a data expression
\dexp{exp} to a variable \keyw{var} of the same type, \dexp{dest},
\dexp{dests}, \dexp{data} and \dexp{ms} data expressions of types
{\tIP}, $\pow(\tIP)$, {\tDATA} and {\tMSG}, respectively, and $\msg$ a
data variable of type \tMSG.

\index{data values}%
Given a valuation of the data variables by concrete data values, the
sequential process $\cond{\varphi}\p$ acts as $\p$ if $\varphi$
evaluates to {\tt true}, and deadlocks if $\varphi$ evaluates to
{\tt false}. In case $\varphi$ contains free variables that are not
yet interpreted as data values, values are assigned to these variables
in any way that satisfies $\varphi$, if possible.
The sequential process $\assignment{\keyw{var}:=\dexp{exp}}\p$
acts as $\p$, but under an updated valuation of the data variable \keyw{var}.
The sequential process $\p+\q$ may act either as $\p$ or as
$\q$, depending on which of the two processes is able to act at all.  In a
context where both are able to act, it is not specified how the choice
is made. The sequential process $\alpha.\p$ first performs the action
$\alpha$ and subsequently acts as $\p$.  The action
\index{broadcast}%
$\broadcastP{\dexp{ms}}$ broadcasts (the data value bound to the
expression) $\dexp{ms}$ to the other network nodes within transmission range,
\index{unicast}%
\index{conditional unicast}%
whereas $\unicast{\dexp{dest}}{\dexp{ms}}.\p \prio \q$ is a sequential process
that tries to unicast the message $\dexp{ms}$ to the destination
\index{destination}%
\dexp{dest}; if successful it continues to act as $\p$ and otherwise
as $\q$. In other words, $\unicast{\dexp{dest}}{\dexp{ms}}.\p$ is
prioritised over $\q$; only if the action $\unicast{\dexp{dest}}{\dexp{ms}}$
is not possible, the alternative $q$ will happen. It models an abstraction
\index{acknowledgment}%
of an acknowledgment-of-receipt mechanism that is typical for unicast
communication but absent in broadcast communication, as implemented by
the link layer of relevant wireless standards such as IEEE 802.11.
\index{groupcast}%
The process $\groupcastP{\dexp{dests}}{\dexp{ms}}.\p$ tries
to transmit \dexp{ms} to all destinations $\dexp{dests}$, and proceeds
as $\p$ regardless of whether any of the transmissions is successful.
Unlike {\bf unicast} and  {\bf broadcast}, the expression {\bf groupcast} 
does not have a unique counterpart in networking.
Depending on the protocol and the implementation it can 
be an iteratively unicast, a broadcast, or a multicast;
thus  {\bf groupcast} abstracts from implementation details. The
\index{send@\textbf{send}}%
action $\send{\dexp{ms}}$ synchronously transmits a message to another
process running on the same network node; this action can occur only
when this other sequential process is able to receive the message.
\index{receive@\textbf{receive}}%
The sequential process $\receive{\msg}.\p$ receives any message $m$
(a data value of type \tMSG) either from another node, from another
sequential process running on the same node or from the client hooked
up to the local node.  It then proceeds as $\p$, but with the data
variable $\msg$ bound to the value $m$. 
The submission of data from a client
is modelled by the receipt of a message $\newpkt{\dval{d}}{\dval{dip}}$,
where the function $\newpktID$ generates a message containing the
data $\dval{d}$ and the intended destination $\dval{dip}$. 
Data is delivered to the client by \deliver{\dexp{data}}.%
\index{deliver}%
\begin{table}[t]
$$\begin{array}{@{}r@{~}l@{\qquad}l@{}}
  \xi,\broadcastP{\dexp{ms}}.\p &\ar{\broadcastP{\xi(\dexp{ms})}} \xi,\p
\\[8pt]
  \xi,\groupcastP{\dexp{dests}}{\dexp{ms}}.\p &\ar{\groupcastP{\xi(\dexp{dests})}{\xi(\dexp{ms})}} \xi,\p
\\[8pt]
  \xi,\unicast{\dexp{dest}}{\dexp{ms}}.\p \prio \q &\ar{\unicast{\xi(\dexp{dest})}{\xi(\dexp{ms})}} \xi,\p
\\[8pt]
  \xi,\unicast{\dexp{dest}}{\dexp{ms}}.\p \prio \q &\ar{\neg\unicast{\xi(\dexp{dest})}{\xi(\dexp{ms})}} \xi,\q
\\[8pt]
  \xi,\send{\dexp{ms}}.\p &\ar{\send{\xi(\dexp{ms})}} \xi,\p
\\[8pt]
  \xi,\deliver{\dexp{data}}.\p &\ar{\deliver{\xi(\dexp{data})}} \xi,\p
\\[8pt]
  \xi,\receive{\keyw{msg}}.\p &\ar{\receive{m}} \xi[\keyw{msg}:=m],\p
  & \mbox{\small($\forall m\in \tMSG$)}
\\[8pt]
  \xi,\assignment{\keyw{var}:=\dexp{exp}}\p &\ar{\tau} \xi[\keyw{var}:=\xi(\dexp{exp})],\p
\\[8pt]\multicolumn{2}{c}{\displaystyle
  \frac{\emptyset[\keyw{var}_i:=\xi(\dexp{exp}_i)]_{i=1}^n,\p \ar{a} \xii,\p'}
  {\xi,X(\dexp{exp}_1,\ldots,\dexp{exp}_n) \ar{a} \xii,\p'}
  ~\mbox{(\small$X(\keyw{var}_1,\ldots,\keyw{var}_n) \stackrel{{\it def}}{=} \p$)}}
  & \mbox{(\small$\forall a\in \act$)}
\\[15pt]\displaystyle
  \frac{\xi,\p \ar{a} \xii,\p'}{\xi,\p+\q \ar{a} \xii,\p'} \qquad
  \frac{\xi,\q \ar{a} \xii,\q'}{\xi,\p+\q \ar{a} \xii,\q'} \qquad\mbox{}
  &\displaystyle
  \frac{ \xi \stackrel{\varphi}{\rightarrow}\xii}
       {\xi,\cond{\varphi}\p \ar{\tau} \xii,\p}
  & \mbox{(\small$\forall a\in \act$)}
\vspace{-1.2ex}
\end{array}$$
\caption[Structural operational semantics for sequential process expressions]
    {\em Structural operational semantics for sequential process expressions}
\label{tab:sos sequential}
\vspace{-1.5ex}
\end{table}

The internal state of a sequential process described by an expression
$\p$ in this language is determined by $\p$, together with a
\phrase{valuation} $\xi$ associating data values $\xi(\keyw{var})$ to
the data variables \keyw{var} maintained by this process.
\index{closed}%
Valuations naturally extend to \emph{$\xi$-closed} data
expressions---those in which all variables are either bound or in the
domain of $\xi$. The structural operational semantics of
\index{structural operational semantics}%
Table~\ref{tab:sos sequential} is in the style of Plotkin \cite{Pl04}
and describes how one internal state can evolve into another by
performing an \phrase{action}.\footnote{Eight of the transition rules
  feature statements of the form $\xi(\dexp{exp})$ where \dexp{exp} is
  a data expression. Here the application of the rule depends on
  $\xi(\dexp{exp})$ being defined. In case $\xi(\dexp{exp})$ is\label{partial}
  undefined---either because \dexp{exp} contains a variable that is
  not in the domain of $\xi$ or because \dexp{exp} contains a partial
  function that is given an argument for which it is not defined---the
  transition cannot be taken, possibly leading to a deadlock of the represented process.}
The set $\act$ of actions consists of
\index{broadcast}%
\index{groupcast}%
\index{unicast}%
\index{receive@\textbf{receive}}%
\index{send@\textbf{send}}%
\index{deliver}%
$\broadcastP{m}$,
$\groupcastP{D}{m}$,
$\unicast{\dval{dip}}{m}$,
$\neg\unicast{\dval{dip}}{\dval{m}}$,
$\send{m}$,
$\deliver{\dval{d}}$,
$\receive{m}$
\index{internal actions}%
and internal actions~$\tau$, for each choice of $m \mathop{\in}\tMSG$,
$\dval{dip}\mathop{\in}\tIP$, $D\mathop{\in}\pow(\tIP)$ and
$\dval{d}\mathop{\in}\tDATA$.  Here, $\neg\unicast{\dval{dip}}{\dval{m}}$
denotes a failed unicast. Moreover $\xi[\keyw{var}:= v]$ denotes the
valuation that assigns the value $v$ to the variable \keyw{var}, and
agrees with $\xi$ on all other variables. The empty valuation
$\emptyset$ assigns values to no variables. Hence
$\emptyset[\keyw{var}_i:=v_i]_{i=1}^n$ is the valuation that
\emph{only} assigns the values $v_i$ to the variables $\keyw{var}_i$
for $i=1,\ldots,n$. The rule for process names in
\index{transition}%
Table~\ref{tab:sos sequential} (Line~$9$) says that a process, named
$X$, has the same  transitions as the body $p$ of its defining equation.
In CCS \cite{Mi89}, such a rule is $\displaystyle\frac{p \ar{a} p'}{X \ar{a} p'}$\,.
\vspace{-1ex}
Adding data variables as arguments of process names would yield
$\displaystyle\frac{\xi,\p \ar{a} \xii,\p'}{\xi,X(\keyw{var}_1\comma\ldots\comma\keyw{var}_n) \ar{a} \xii,\p'}$\,.
\vspace{3pt}
However, a sequential process expression may call a process name with
data expressions filled in for these variables. This necessitates a
translation from a given valuation $\xi$ of the variables that may
occur in these data expressions to a new valuation $\xi^\#$ of the
variables $\keyw{var}_i$ that occur in the defining equation of $X$:
\[\displaystyle\frac{\xi^\#,X(\keyw{var}_1\comma\ldots\comma\keyw{var}_n) \ar{a} \xii,\p'}
  {\xi,X(\dexp{exp}_1\comma\ldots\comma\dexp{exp}_n) \ar{a} \xii,\p'}\ .\]
Here $\xi^\#(\keyw{var}_i) = \xi(\dexp{exp}_i)$. Moreover, in defining
$\xi^\#$ we drop all bindings of variables other than the $\keyw{var}_i$.

\begin{example}\label{ex:recursion}
Given the defining equation\vspace{-3pt}
\[X(\keyw{numa}) ~\stackrel{{\it def}}{=}~\send{\keyw{numa}+1}\,.\,
  \receive{\keyw{numb}}\,.\,X(\keyw{numa}+\keyw{numb})\]
and the valuation given by $\xi(\keyw{numa})=3$ and
$\xi(\keyw{numb})=4$, with \keyw{numa} and \keyw{numb} data variables
of type $\NN$, we have\vspace{-1ex}
\[\xi, X(\keyw{numa}+\keyw{numb}) \ar{\send{8}}
\xii, \receive{\keyw{numb}}\,.\,X(\keyw{numa}+\keyw{numb})\ ,\]
where $\xii(\keyw{numa})=7$ and $\xii(\keyw{numb})$ is undefined.
\end{example}

An alternative and more traditional rule for process names would be
$\displaystyle \frac{\xi,\p[\dexp{exp}_i/\keyw{var}_i]_{i=1}^n \ar{a} \xii,\p'}
{\xi,X(\dexp{exp}_1\comma\ldots\comma\dexp{exp}_n) \ar{a} \xii,\p'}$
where $\p[\dexp{exp}_i/\keyw{var}_i]_{i=1}^n$ denotes the expression $\p$
in which each variable $\keyw{var}_i$ is replaced
by the expression $\dexp{exp}_i$, for $i=1,\ldots,n$.
This would modify the derivation of \Ex{recursion} into
\[\xi, X(\keyw{numa}+\keyw{numb}) \ar{\send{8}}
\xi, \receive{\keyw{numc}}\,.\,X(\keyw{numa+numb}+\keyw{numc})\ ,\]
\index{alpha@$\alpha$-conversion}%
in which one applies \emph{$\alpha$-conversion} when renaming the
argument \keyw{numb} of \textbf{receive} into \keyw{numc} to avoid a
name clash.  In this paper we avoid casual application of
$\alpha$-conversion, since in our invariant proofs in \Sect{invariants}
we track the value of variables that are identified by name only.
With this in mind we formulated our rule for process names.

\index{choice operator}%
The rules defining the choice operator (Table~\ref{tab:sos sequential}, Line~$10$) are standard and imply immediately that 
$+$ is associative.

Finally, \plat{$\xi\stackrel{\varphi}{\rightarrow}\xii$} says that
$\xii$ is an extension of $\xi$, i.e., a valuation that agrees with
$\xi$ on all variables on which $\xi$ is defined, and valuates the
other variables occurring free in $\varphi$, such that the formula
$\varphi$ holds under $\xii$. All variables not free in $\varphi$ and
not evaluated by $\xi$ are also not evaluated by $\xii$.

\begin{example}
Let $\xi(\keyw{numa})=7$ and $\xi(\keyw{numb})$, $\xi(\keyw{numc})$
be undefined. Then the sequential process given by the pair
\index{transition}%
$\xi,[\keyw{numa} = \keyw{numb}+\keyw{numc}]p$ admits several transitions of the form
$$\xi,[\keyw{numa} = \keyw{numb}+\keyw{numc}]p \ar{\tau} \xii,p$$
such as the one with $\xii(\keyw{numb})=2$ and $\xii(\keyw{numc})=5$.
On the other hand, $\xi,[\keyw{numa} = \keyw{numb}+8]p$ admits no transitions, since $\keyw{numb}\in\NN$.
\end{example}

\subsection{A Language for Parallel Processes}

\index{parallel process expressions}%
\emph{Parallel process expressions} are given by the grammar
$$PP ~::=~ \xi,\SP ~\mid~ PP \parl PP$$
where $\SP$ is a sequential process expression and $\xi$ a valuation.
An expression $\xi,\p$ denotes a sequential process expression
equipped with a valuation of the variables it maintains.
\index{parallel composition}%
The process $P\parl Q$ is a parallel composition of $P$ and $Q$,
running on the same network node. As formalised in
Table~\ref{tab:sos}, an action $\receive{\dval{m}}$ of $P$
\index{receive@\textbf{receive}}%
\index{send@\textbf{send}}%
\index{internal actions}%
synchronises with an action $\send{\dval{m}}$ of $Q$ into an internal
action $\tau$. These receive actions of $P$ and send actions of $Q$
cannot happen separately. All other actions of $P$ and $Q$, including
receive actions of $Q$ and send actions of $P$, occur interleaved in
$P\parl Q$. Thus, in an expression $(P\parl Q)\parl R$, for example,
the send and receive actions of $Q$ can communicate only with $P$ and
$R$, respectively, but the receive actions of $R$, as well as the send
actions of $P$, remain available for communication with the
environment.  Therefore, a parallel process expression denotes a
parallel composition of sequential processes $\xi,P$ with information
flowing from right to left. The variables of different sequential
processes running on the same node are maintained separately, and thus
\index{shared variables}%
cannot be shared.

\begin{table}[t]
$$\begin{array}{@{}r@{~}l@{}}
\displaystyle
  \frac{P \ar{a} P'}{P\parl Q \ar{a} P'\parl Q}
  \quad\mbox{\small($\forall a\neq \receive{m}$)}
  \qquad\mbox{}
  &\displaystyle
  \frac{Q \ar{a} Q'}{P\parl Q \ar{a} P\parl Q'}
  \quad\mbox{\small($\forall a\neq \send{m}$)}
\\[16pt]\multicolumn{2}{c}{\displaystyle
  \frac{P \ar{\receive{m}} P'\qquad Q \ar{\send{m}} Q'}
       {P\parl Q \ar{\tau} P'\parl Q'}
    \quad\mbox{\small($\forall m\in\tMSG$)}}
\vspace{-1.2ex}
\end{array}$$
\caption[Structural operational semantics for parallel process expressions]
    {\em Structural operational semantics for parallel process expressions}
\label{tab:sos}
\vspace{-1ex}
\end{table}

Instead of introducing the novel operator $\parl$, we could have used
the partially synchronous parallel composition operator $\|$ of ACP \cite{BK86},
$|$ of CCS \cite{Mi89} or $\|_A$ of CSP \cite{OH86}. However, those
operators are normally used in conjunction with restriction and/or
concealment operators, which are not needed when using~$\parl$.
\index{restriction operator}%
\index{encapsulation operator}%
In ACP a \emph{restriction} or \emph{encapsulation} operator is used to
prevent read and send actions of the components of a parallel
composition to occur by themselves, without synchronising with an
action from another other component. Furthermore, a \emph{concealment} or
\index{concealment operator}%
\index{abstraction operator}%
\emph{abstraction} operator is used to convert the results of
successful synchronisation into internal actions, thereby making sure
that they will not take part in further synchronisations with the
environment. In CCS, the concealment operator is not needed, as the
parallel composition directly produces internal actions as the results
of synchronisation; however, the restriction operator is indispensable.
In CSP, on the other hand, the restriction operator is made redundant
by incorporating its function within the parallel composition. In this
framework matching read and send actions have the same name,
which is also the name of the result of their synchronisation. This
makes the concealment operator indispensable. It appears to be
impossible to combine the ideas of CCS and CSP directly to make both the
restriction and the concealment operator redundant, while maintaining
associativity of the parallel composition.  Our operator $\parl$ is
the first that does not need such auxiliary operators, sacrificing
commutativity, but not associativity, to make this possible.

Though $\parl$ only allows information flow in one direction, it reflects reality of 
WMNs. Usually two sequential processes run on the same node:
$
P \parl Q
$.
The main process $P$ deals with all protocol details  of the node, e.g., message handling 
and maintaining the data such as routing tables.
The process $Q$ manages the queueing of messages as they arrive; it is always able to
receive a message even if $P$ is busy. 
\index{message queueing}%
The use of message queueing in combination with $\parl$  is crucial, since
otherwise incoming messages would be lost when the process is busy dealing with other
messages\footnote{assuming that one employs the optional augmentation of Section~\ref{ssec:non-blocking}},
which would not be an accurate model of what happens in real implementations.

\subsection{A Language for Networks}\label{ssec:networks}

\index{node}%
We model network nodes in the context of a wireless mesh network by
\phrase{node expressions} of the form $\dval{ip}:PP:R$. Here $\dval{ip}
\in \tIP$ is the \emph{address} of the node, $PP$ is a parallel process
\index{IP addresses}%
expression, and $R\in\pow(\tIP)$ is the \emph{range} of the node---the
\index{transmission range}%
set of nodes that are currently within transmission range of $\dval{ip}$.

\begin{table}[hbt]
$$\begin{array}{@{}c@{\qquad}c@{}}
\displaystyle
  \frac{P \ar{\broadcastP{m}} P'}
  {\rule[13pt]{0pt}{1pt}
   \dval{ip}:P:R \ar{\colonact{R}{\starcastP{m}}} \dval{ip}:P':R}
&\displaystyle
  \frac{P \ar{\groupcastP{D}{m}} P'}
  {\rule[13pt]{0pt}{1pt}
   \dval{ip}:P:R \ar{\colonact{R\cap D}{\starcastP{m}}} \dval{ip}:P':R}
\\[22pt]\displaystyle
  \frac{P \ar{\unicast{\dval{dip}}{m}} P'\qquad \dval{dip}\in R}
  {\rule[13pt]{0pt}{1pt}
   \dval{ip}:P:R \ar{\colonact{\{\dval{dip}\}}{\starcastP{m}}} \dval{ip}:P':R}
&\displaystyle
  \frac{P \ar{\neg\unicast{\dval{dip}}{\dval{m}}} P'\qquad \dval{dip}\not\in R}
  {\rule[13pt]{0pt}{1pt}
   \dval{ip}:P:R \ar{\tau} \dval{ip}:P':R}
\\[22pt]\displaystyle
  \frac{P \ar{\deliver{\dval{d}}} P'}
  {\rule[13pt]{0pt}{1pt}
   \dval{ip}:P:R \ar{\colonact{\dval{ip}}{\deliver{\dval{d}}}} \dval{ip}:P':R}
&\displaystyle
  \frac{P \ar{\receive{m}} P'}
  {\rule[13pt]{0pt}{1pt}
   \dval{ip}:P:R \ar{\colonact{\{\dval{ip}\}\neg\emptyset}{\listen{m}}} \dval{ip}:P':R}
\\[20pt]\displaystyle
  \frac{P \ar{\tau} P'}
  {\dval{ip}:P:R \ar{\tau} \dval{ip}:P':R}
&
  \dval{ip}:P:R \ar{\colonact{\emptyset\neg\{\dval{ip}\}}{\listen{m}}} \dval{ip}:P:R
\\[20pt]\displaystyle
  \dval{ip}:P:R \ar{\textbf{connect}(\dval{ip},\dval{ip}')} \dval{ip}:P:R\cup\{\dval{ip}'\}
&
  \dval{ip}:P:R \ar{\textbf{disconnect}(\dval{ip},\dval{ip}')} \dval{ip}:P:R-\{\dval{ip}'\}
\\[8pt]\displaystyle
  \dval{ip}:P:R \ar{\textbf{connect}(\dval{ip}',\dval{ip})} \dval{ip}:P:R\cup\{\dval{ip}'\}
&
  \dval{ip}:P:R \ar{\textbf{disconnect}(\dval{ip}',\dval{ip})} \dval{ip}:P:R-\{\dval{ip}'\}
\\[8pt]\displaystyle
  \frac{\dval{ip} \not\in \{\dval{ip}'\!,\dval{ip}''\}}
  {\rule[13pt]{0pt}{1pt}
   \dval{ip}:P:R \ar{\textbf{connect}(\dval{ip}'\!,\dval{ip}'')} \dval{ip}:P:R}
&\displaystyle
  \frac{\dval{ip} \not\in \{\dval{ip}'\!,\dval{ip}''\}}
  {\rule[13pt]{0pt}{1pt}
   \dval{ip}:P:R \ar{\textbf{disconnect}(\dval{ip}'\!,\dval{ip}'')} \dval{ip}:P:R}
\end{array}$$
\vspace{-1ex}
\caption[Structural operational semantics for node expressions]
    {\em Structural operational semantics for node expressions}
\label{tab:sos node}
\vspace{2ex}
\end{table}

\index{network}%
A \phrase{partial network} is then modelled by a \phrase{parallel
  composition} $\|$ of node expressions, one for every node in the
network, and a \emph{complete network} is a partial network within an
\phrase{encapsulation operator} $[\_]$ that limits the communication of
network nodes and the outside world to the receipt and the delivery
of data packets to and from the application layer attached to the
modelled protocol in the network nodes.
\index{network expressions}%
This yields the following grammar for network expressions:\vspace{-3pt}
\[N ::= [M] \qquad\qquad M::= ~~ \dval{ip}:PP:R ~~\mid~~ M \| M\ .\vspace{-3pt}\]

\begin{table}[t]\vspace{-2ex}
$$\begin{array}{@{}c@{\hspace{2.5mm}}c@{}}
\displaystyle
  \frac{M \ar{\colonact{R}{\starcastP{m}}} M' \quad N \ar{\colonact{H\neg K}{\listen{m}}} N'}
  {\rule[13pt]{0pt}{1pt}
   M \| N \ar{\colonact{R}{\starcastP{m}}} M' \| N'}
  \mbox{\footnotesize
  $\left(\begin{array}{@{}c@{}}H\subseteq R\\K \cap R = \emptyset\end{array}\right)$}
  &{\displaystyle
  \frac{M \ar{\colonact{H\neg K}{\listen{m}}} M' \quad N \ar{\colonact{R}{\starcastP{m}}} N'}
  {\rule[13pt]{0pt}{1pt}
   M \| N \ar{\colonact{R}{\starcastP{m}}} M' \| N'}
  \mbox{\footnotesize
  $\left(\begin{array}{@{}c@{}}H\subseteq R\\K \cap R = \emptyset\end{array}\right)$}}\\[22pt]
\multicolumn{2}{c}{\displaystyle
  \frac{M \ar{\colonact{H\neg K}{\listen{m}}} M' \quad N \ar{\colonact{H'\neg K'}{\listen{m}}} N'}
  {\rule[13pt]{0pt}{1pt}
   M \| N \ar{\colonact{(H\cup H')\neg(K\cup K')}{\listen{m}}} M' \| N'}}
\\[22pt]\displaystyle
  \frac{M \ar{\colonact{\dval{ip}}{\deliver{\dval{d}}}} M'}{\rule[13pt]{0pt}{1pt}
   M \| N \ar{\colonact{\dval{ip}}{\deliver{\dval{d}}}} M' \| N}
  \qquad
  \frac{N \ar{\colonact{\dval{ip}}{\deliver{\dval{d}}}} N'}{\rule[13pt]{0pt}{1pt}
   M \| N \ar{\colonact{\dval{ip}}{\deliver{\dval{d}}}} M \| N'}
&\displaystyle
  \frac{M \ar{\tau} M'}{M \| N \ar{\tau} M' \| N}
  \qquad
  \frac{N \ar{\tau} N'}{M \| N \ar{\tau} M \| N'}
\\[22pt]\displaystyle
  \frac{M \ar{\textbf{connect}(\dval{ip},\dval{ip}')} M' \quad
        N \ar{\textbf{connect}(\dval{ip},\dval{ip}')} N'}
  {\rule[13pt]{0pt}{1pt}
   M \| N \ar{\textbf{connect}(\dval{ip},\dval{ip}')} M' \| N'}
&\displaystyle
  \frac{M \ar{\textbf{disconnect}(\dval{ip},\dval{ip}')} M' \quad
        N \ar{\textbf{disconnect}(\dval{ip},\dval{ip}')} N'}
  {\rule[13pt]{0pt}{1pt}
   M \| N \ar{\textbf{disconnect}(\dval{ip},\dval{ip}')} M' \| N'}
\\[22pt]\displaystyle
  \frac{M \ar{\textbf{connect}(\dval{ip},\dval{ip}')} M'}{\rule[13pt]{0pt}{1pt}
   [M] \ar{\textbf{connect}(\dval{ip},\dval{ip}')} [M']}
\qquad
  \frac{M \ar{\textbf{disconnect}(\dval{ip},\dval{ip}')} M'}{\rule[13pt]{0pt}{1pt}
   [M] \ar{\textbf{disconnect}(\dval{ip},\dval{ip}')} [M']}
&\displaystyle
  \frac{M \ar{\colonact{R}{\starcastP{m}}} M'}{[M] \ar{\tau} [M']}
\qquad
  \frac{M \ar{\tau} M'}{[M] \ar{\tau} [M']}
\\[22pt]\displaystyle
  \frac{M \ar{\colonact{\dval{ip}}{\deliver{\dval{d}}}} M'}{\rule[13pt]{0pt}{1pt}
   [M] \ar{\colonact{\dval{ip}}{\deliver{\dval{d}}}} [M']}
&\displaystyle
  \frac{M \ar{\colonact{\{\dval{ip}\}\neg K}{\listen{\newpkt{\dval{d}}{\dval{dip}}}}} M'}
  {\rule[13pt]{0pt}{1pt}
   [M] \ar{\colonact{\dval{ip}}{\textbf{newpkt}(\dval{d},\dval{dip})}} [M']}
\end{array}$$
\vspace{-3ex}
\caption[Structural operational semantics for network expressions]
    {\em Structural operational semantics for network expressions}
\label{tab:sos network}
\vspace{-3pt}
\end{table}

The operational semantics of node and network expressions of
Tables~\ref{tab:sos node} and~\ref{tab:sos network} uses transition labels
\index{starcast@\textbf{*cast}}%
\index{arrive@\textbf{arrive}}%
\index{connect@\textbf{connect}}%
\index{disconnect@\textbf{disconnect}}%
\index{newpkt@$\newpktID$}%
\index{deliver}%
$\colonact{R}{\starcastP{m}}$,
$\colonact{H\neg K}{\listen{m}}$,
$\textbf{connect}(\dval{ip},\dval{ip}')$,
$\textbf{disconnect}(\dval{ip},\dval{ip}')$,
$\colonact{\dval{ip}}{\textbf{newpkt}(\dval{d},\dval{dip})}$,
$\colonact{\dval{ip}}{\deliver{\dval{d}}}$\linebreak[4]
and $\tau$.
As before, $m\in\tMSG$, $d\in\tDATA$, $R\in\pow(\tIP)$, and $\dval{ip},\dval{ip}'\in\tIP$.
\index{IP addresses}%
Moreover, $H,K\in\pow(\tIP)$ are sets of IP addresses.
\index{starcast@\textbf{*cast}}%
The action $\colonact{R}{\starcastP{m}}$ casts a message $m$  that can
be received by the set $R$ of network nodes.  We do not distinguish
whether this message has been broadcast, groupcast or unicast---the
differences show up merely in the value of $R$. Recall that $D\in\pow(\tIP)$
denotes a set of intended destinations, and $\dval{dip}\in\tIP$ a single
destination. A failed unicast attempt on the part of its process is
modelled as an internal action $\tau$ on the part of a node expression.
\index{send@\textbf{send}}%
The action $\send{m}$ of a process does not give rise to any action of
the corresponding node---this action of a sequential process cannot
occur without communicating with a receive action of another sequential
process running on the same node.

\index{arrive@\textbf{arrive}}%
The action $\colonact{H\neg K}{\listen{m}}$ states that the message $m$
simultaneously arrives at all addresses $\dval{ip}\mathbin\in H$, and fails to
arrive at all addresses $\dval{ip}\mathbin\in K$.
The rules of Table~\ref{tab:sos network} let a $\colonact{R}{\starcastP{m}}$-action of one node
synchronise with an $\listen{m}$ of all other nodes, where this
$\listen{m}$ amalgamates the arrival of message $m$ at the nodes in
the transmission range $R$ of the $\starcastP{m}$, and the non-arrival at the
other nodes. The rules for $\listen{m}$ in Table~\ref{tab:sos node}
state that arrival of a message at a node happens if and only if the
node receives it, whereas non-arrival can happen at any time.
This embodies our assumption that, at any time, any message that is
transmitted to a node within range of the sender is actually received by that node.
(The eighth rule in Table~\ref{tab:sos node}, having no
  premises, may appear to say that any node \dval{ip} has the option to
  disregard any message at any time. However, the encapsulation
  operator (below) prunes away all such disregard-transitions that do
  not synchronise with a cast action for which \dval{ip} is out of range.)

\index{deliver}%
Internal actions $\tau$ and the action $\colonact{\dval{ip}}{\deliver{\dval{d}}}$ are simply
inherited by node expressions from the processes that run on these
nodes, and are interleaved in the parallel composition of nodes that
makes up a network. Finally, we allow actions $\textbf{connect}(\dval{ip},\dval{ip}')$
\index{connect@\textbf{connect}}%
\index{disconnect@\textbf{disconnect}}%
and $\textbf{disconnect}(\dval{ip},\dval{ip}')$ for $\dval{ip},\dval{ip}'\in \tIP$
modelling a change in network topology. Each node needs to synchronise
with such an action. These actions can be thought of as occurring
nondeterministically, or as actions instigated by the environment of
the modelled network protocol. In this formalisation node $\dval{ip}'$ is in
the range of node $\dval{ip}$, meaning that $\dval{ip}'$ can receive
messages sent by $\dval{ip}$, if and only if $\dval{ip}$ is in the
range of $\dval{ip}'$. To break this symmetry, one just skips the last
four rules of Table~\ref{tab:sos node} and replaces the synchronisation
rules for \textbf{connect} and \textbf{disconnect} in
Table~\ref{tab:sos network} by interleaving rules (like the ones for
\textbf{deliver} and $\tau$).
\label{pg:sym}

The main purpose of the encapsulation operator is to ensure that no
messages will be received that have never been sent. In a parallel
composition of network nodes, any action $\receive{\dval{m}}$ of one
of the nodes \dval{ip} manifests itself as an action $\colonact{H\neg
K}{\listen{\dval{m}}}$ of the parallel composition, with $\dval{ip}\in
H$. Such actions can happen (even) if within the parallel composition
they do not communicate with an action $\starcastP{\dval{m}}$ of
another component, because they might communicate with a
$\starcastP{\dval{m}}$ of a node that is yet to be added to the
parallel composition. However, once all nodes of the network are
accounted for, we need to inhibit unmatched arrive actions,
as otherwise our formalism would allow any node at any time to receive
any message. One exception however are those arrive actions
\index{newpkt@$\newpktID$}%
that stem from an action $\receive{\newpkt{\dval{d}}{\dval{dip}}}$ of a
sequential process running on a node, as those actions represent
communication with the environment. Here, we use the function
$\newpktID$, which we assumed to exist.\footnote{%
\newcommand{\npkt}{\textbf{newpkt}\xspace}
To avoid the function $\newpktID$ we could have introduced a new
primitive \npkt, which is dual to \textbf{deliver}.}
\index{injection}%
It models the injection of new data $\dval{d}$ for destination $\dval{\dip}$.

\index{starcast@\textbf{*cast}}%
\index{deliver}%
The encapsulation operator passes through internal actions, as well as
delivery of data to destination nodes, this being an interaction
with the outside world. $\starcastP{m}$-actions are declared internal
\index{internal actions}%
actions at this level; they cannot be steered by the outside world.
\index{connect@\textbf{connect}}%
\index{disconnect@\textbf{disconnect}}%
The connect and disconnect actions are passed through in
Table~\ref{tab:sos network}, thereby placing them under control of the
environment; to make them nondeterministic, their rules should have a
$\tau$-label in the conclusion, or alternatively
$\textbf{connect}(\dval{ip},\dval{ip}')$ and $\textbf{disconnect}(\dval{ip},\dval{ip}')$
should be thought of as internal actions. Finally, actions
\index{arrive@\textbf{arrive}}%
$\listen{m}$ are simply blocked by the encapsulation---they
cannot occur without synchronising with a $\starcastP{m}$---except for
\index{newpkt@$\newpktID$}%
$\colonact{\{\dval{ip}\}\neg K}{\listen{\newpkt{\dval{d}}{\dval{dip}}}}$
with $\dval{d}\in\tDATA$ and $\dval{dip}\in \tIP$. This action
represents new data \dval{d} that is submitted by a
client of the modelled protocol to node $\dval{ip}$, for delivery at
destination \dval{dip}.

\subsection{Results on the Process Algebra}

\index{data structure}%
\index{variables}%
\index{data values}%
Our process algebra admits translation into one without data
structures (although we cannot \emph{describe} the target algebra without
using data structures).  The idea is to replace any variable by all
possible values it can take. Formally, processes $\xi,p$ are replaced
by $\mathcal{T}_\xi(p)$, where $\mathcal{T}_\xi$ is defined inductively by

\begin{tabular}{@{}l}
$\mathcal{T}_\xi(\broadcastP{\dexp{ms}}\,.\,p)= \broadcastP{\xi(\dexp{ms})}\,.\,\mathcal{T}_\xi(p)$\ ,\\[0.5mm]
$\mathcal{T}_\xi(\groupcastP{\dexp{dests}}{\dexp{ms}}\,.\,p)= \groupcastP{\xi(\dexp{dests})}{\xi(\dexp{ms})}\,.\,\mathcal{T}_\xi(p)$\ ,\\[0.5mm]
$\mathcal{T}_\xi(\unicast{\dexp{dest}}{\dexp{ms}}\,.\,p \prio q)= \unicast{\xi(\dexp{dest})}{\xi(\dexp{ms})}\,.\,\mathcal{T}_\xi(p) \prio \mathcal{T}_\xi(q)$\ ,\\[0.5mm]
$\mathcal{T}_\xi(\send{\dexp{ms}}\,.\,p)= \send{\xi(\dexp{ms})}\,.\,\mathcal{T}_\xi(p)$\ ,\\[0.5mm]
$\mathcal{T}_\xi(\deliver{\dexp{data}}\,.\,p)= \deliver{\xi(\dexp{data})}\,.\,\mathcal{T}_\xi(p)$\ ,\\[0.5mm]
$\mathcal{T}_\xi(\receive{\msg}\,.\,p)= \sum_{m\in\tMSG}\receive{m}\,.\,\mathcal{T}_{\xi[\msg:=m]}(p)$\ ,\\[0.5mm]
$\mathcal{T}_\xi(\assignment{\keyw{var}:=\dexp{exp}}p)= \tau\,.\,\mathcal{T}_{\xi[\keyw{var}:=\xi(\dexp{exp})]}(p)$\ ,\\[0.5mm]
$\mathcal{T}_\xi([\varphi]p)=\sum_{ \{\zeta\mid\xi \stackrel{\varphi}{\rightarrow}\xii\}} \tau\,.\,\mathcal{T}_{\xii}(p)$\ ,\\[0.5mm]
$\mathcal{T}_\xi(p+q)=\mathcal{T}_\xi(p) + \mathcal{T}_\xi(q)$\ ,\\[0.5mm]
$\mathcal{T}_\xi(X(\dexp{exp}_1,\ldots,\dexp{exp}_n))= X_{\xi(\dexp{exp}_1),\ldots,\xi(\dexp{exp}_n)}$\ .\\[1mm]
\end{tabular}

\noindent
The last equation requires the introduction of a process name $X_{\vec{v}}$
for every substitution instance $\vec{v}$ of the arguments of $X$.
The resulting process algebra has a structural operational semantics in the
\emph{de Simone} format\index{de Simone format}, generating the same transition system---up to strong
\index{strong bisimilarity}%
bisimilarity, $\bis$ ---as the original. Only the rules for sequential process
expressions are different; these are displayed in \Tab{sos sequential modified}.
It follows that $\bis\,$, and many other
semantic equivalences, are congruences on our language.

\begin{table}[t]
$$\begin{array}{@{}r@{~}l@{\qquad}r@{~}l@{}}
  \broadcastP{\dval{m}}.\p &\ar{\broadcastP{\dval{m}}} \p
&
  \send{\dval{m}}.\p &\ar{\send{\dval{m}}} \p
\\[8pt]
  \groupcastP{D}{\dval{m}}.\p &\ar{\groupcastP{D}{\dval{m}}} \p
&
  \deliver{d}.\p &\ar{\deliver{d}} \p
\\[8pt]
  \unicast{\dval{dip}}{\dval{m}}.\p \prio \q &\ar{\unicast{\dval{dip}}{\dval{m}}} \p
&
  \receive{m}.\p &\ar{\receive{m}} \p
\\[8pt]
  \unicast{\dval{dip}}{\dval{m}}.\p \prio \q &\ar{\neg\unicast{\dval{dip}}{\dval{m}}} \q
&
  \tau.\p &\ar{\tau} \p
\end{array}$$
$$\frac{\p \ar{a} \p'}
  {X \ar{a} \p'}
  ~~~\mbox{(\small$X \stackrel{{\it def}}{=} \p$)}
\qquad
  \frac{\p \ar{a} \p'}{\p+\q \ar{a} \p'} \qquad
  \frac{\q \ar{a} \q'}{\p+\q \ar{a} \q'} \qquad
  \frac{ \p_i \ar{a} \p'}{\sum_{i\in I}\p_i \ar{a} \p'}
\qquad
  \mbox{(\small$\forall a\in \act$)}
$$
\vspace{-1.2ex}
\caption[Operational semantics for sequential processes
  after elimination of data structures]
    {\em Structural operational semantics for sequential processes
  after elimination of data structures}
\label{tab:sos sequential modified}
\vspace{-1.5ex}
\end{table}

\begin{theorem}
Strong bisimilarity is a congruence for all operators of our language.
\end{theorem}
This is a deep result that usually takes many pages to establish (e.g.,~\cite{SRS10}).
Here we get it directly from the existing theory on structural
operational semantics, as a result of carefully designing our
language within the disciplined framework described by de Simone~\cite{dS85}.
\endbox
\begin{theorem}
\index{associativity}%
$\parl$ is associative, and $\|$ is associative and commutative, up to $\bis\,$.
\end{theorem}

\begin{proof}
The operational rules for these operators fit a format presented in \cite{CMR08},
guaranteeing associativity up to~$\bis$.
The \emph{ASSOC-de Simone format} of \cite{CMR08} applies to all 
transition system specifications (TSSs) in de
Simone format, and allows $7$ different types of rules (named $1$--$7$) for the operators in question.
Our TSS is in De Simone format; the three rules for $\parl$ of
Table~\ref{tab:sos} are of types $1$, $2$ and $7$, respectively.
To be precise, it has rules $1_a$ for $a \in \act -                      
\{\receive{m}\mid m\mathop\in\tMSG\}$, rules $2_a$ for $a \in \act -                      
\{\send{m}\mid m\mathop\in\tMSG\}$, and rules $7_{(a,b)}$ for
$(a,b)\in\{(\receive{m},\send{m})\mid m\mathop\in\tMSG\}$.
Moreover, the partial \phrase{communication function}
$\gamma:\act\times\act\rightharpoonup \act$ is given by
$\gamma(\receive{m},\send{m})=\tau$.
The main result of \cite{CMR08} is that an operator is guaranteed to
be associative, provided that $\gamma$ is associative and six
conditions are fulfilled. In the absence of rules of types 3, 4, 5
and 6, five of these conditions are trivially fulfilled, and the
remaining one reduces to\vspace{-1ex}
$$7_{(a,b)} \ims (1_a \Leftrightarrow 2_b)
           \ans (2_a \Leftrightarrow 2_{\gamma(a,b)})
           \ans (1_b \Leftrightarrow 1_{\gamma(a,b)})\ .$$
Here $1_a$ says that rule $1_a$ is present, etc.
This condition is met for $\parl$ because the antecedent holds only
when taking $(a,b)=(\receive{m},\send{m})$ for some $m\mathop\in\tMSG$.
In that case $1_a$ is false, $2_b$ is false, and $2_a$, $2_\tau$,
$1_b$ and $1_\tau$ are true. Moreover, $\gamma(\gamma(a,b),c)$ and
$\gamma(a,\gamma(b,c))$ are never defined, thus making $\gamma$
trivially associative.
The argument for $\|$ being associative proceeds likewise.
Here the only non-trivial condition is the associativity of $\gamma$,
given by
$$\gamma(\colonact{R}{\starcastP{m}},\colonact{H\neg K}{\listen{m}})=
 \gamma(\colonact{H\neg K}{\listen{m}},\colonact{R}{\starcastP{m}})
 = \colonact{R}{\starcastP{m}}\ ,$$ provided $H\subseteq R$ and $K\cap R =\emptyset$, and
$$\gamma(\colonact{H\neg K}{\listen{m}},\colonact{H'\neg K'}{\listen{m}})
 =\colonact{(H\cup H')\neg(K\cup K')}{\listen{m}}\ .$$
Commutativity of $\|$ follows by symmetry.
\end{proof}

\subsection{Optional Augmentation to Ensure Non-Blocking Broadcast}
\label{ssec:non-blocking}

Our process algebra, as presented above, is intended for networks in which
each node is \phrase{input enabled} \cite{LT89}, meaning that it is
always ready to receive any message, i.e., able to engage in the
transition $\receive{m}$ for any $m\in \tMSG$. In our model of AODV
(Section~\ref{sec:modelling_AODV}) we will ensure this by equipping each node
with a message queue that is always able to accept messages for later
handling---even when the main sequential process is currently busy.
This makes our model \phrase{non-blocking}, meaning that no sender can
be delayed in transmitting a message simply because one of the
potential recipients is not ready to receive it.

However, the operational semantics does allow blocking if one would
(mis)use the process algebra to model nodes that are not input enabled.
This is a logical consequence of insisting that any broadcast
message \emph{is} received by all nodes within transmission range.

Since the possibility of blocking can regarded as a bad property of
broadcast formalisms, one may wish to take away the expressiveness of
the language that allows modelling a blocking broadcast. This is the
purpose of the following optional augmentations of our operational
semantics.

The first possibility is the addition of the rule
$$\frac{P \nar{\receive{m}}}  {\rule[13pt]{0pt}{1pt}
  \dval{ip}:P:R \ar{\colonact{\{\dval{ip}\}\neg\emptyset}{\listen{m}}} \dval{ip}:P:R}\;.$$
It states that a message may arrive at a node \dval{ip} regardless whether
the node is ready to receive it; if it is not ready, the message is simply
ignored, and the process running on the node remains in the same state.

\newcommand{\discard}{\textbf{ignore}(m)} A variation on the same idea
stems from the \emph{Calculus of Broadcasting Systems} (CBS) \cite{CBS}.  It
consists in eliminating the negative premise in the above rule in
favour of actions $\discard\in\act$---in \cite{CBS} called \emph{discard}
\index{discard actions}%
actions $w\!:$ ---which can be performed by a process exactly when
it is not ready to do a $\receive{m}$. The rule above then becomes
$$\frac{P \ar{\discard} P'}  {\rule[13pt]{0pt}{1pt}
  \dval{ip}:P:R \ar{\colonact{\{\dval{ip}\}\neg\emptyset}{\listen{m}}} \dval{ip}:P':R}$$
and we need the extra rules:
$$\begin{array}{@{}r@{~}l@{}}
  \xi,\broadcastP{\dexp{ms}}.\p &\ar{\discard} \xi,\broadcastP{\dexp{ms}}.\p
\\[8pt]
  \xi,\groupcastP{\dexp{dests}}{\dexp{ms}}.\p &\ar{\discard}
  \xi,\groupcastP{\dexp{dests}}{\dexp{ms}}.\p
\\[8pt]
  \xi,\unicast{\dexp{dest}}{\dexp{ms}}.\p \prio \q &\ar{\discard}
  \xi,\unicast{\dexp{dest}}{\dexp{ms}}.\p \prio \q
\\[8pt]
  \xi,\send{\dexp{ms}}.\p &\ar{\discard}
  \xi,\send{\dexp{ms}}.\p
\\[8pt]
  \xi,\deliver{\dexp{data}}.\p &\ar{\discard}
  \xi,\deliver{\dexp{data}}.\p
\\[8pt]
  \xi,\assignment{\keyw{var}:=\dexp{exp}}\p &\ar{\discard}
  \xi,\assignment{\keyw{var}:=\dexp{exp}}\p
\\[8pt]
  \xi,\cond{\varphi}\p &\ar{\discard} \xi,\cond{\varphi}\p
\\[8pt]\multicolumn{2}{c}{\displaystyle
  \frac{\xi,\p \ar{\discard} \xi,\p' \quad \xi,\q \ar{\discard} \xi,\q'}
   {\rule[13pt]{0pt}{1pt} \xi,\p+\q \ar{\discard} \xi,\p'+\q'}}
\vspace{-1ex}
\end{array}$$
for all $m\in \tMSG$. Furthermore, the first rule for $\parl$
from Table~\ref{tab:sos node} is replaced by
$$\frac{P \ar{a} P'}{P\parl Q \ar{a} P'\parl Q}
  \quad\mbox{\small($\forall a\neq \receive{m},~\discard$)}.$$
These rules ensure that for all $P$ and $m$ we always have
~$P \ar{\discard} Q ~\Leftrightarrow~ (Q=P \wedge P\nar{\receive{m}})$.
After elimination of the data structures as described in
\SSect{non-blocking}, this operational semantics is again in the
de Simone format.

Either of these two optional augmentations of our semantics gives rise
to the same transition system. Moreover, when modelling networks in
which all nodes are input enabled---as we do in this paper---the added
rule for node expressions will never be used, and the resulting
transition system is the same whether we use augmentation or not.

\subsection{Illustrative Example}
\renewcommand{\a}{a}
\renewcommand{\b}{b}
\newcommand{\mymsg}[2]{\keyw{mg}(#1,#2)}
To illustrate the use of our process algebra \awn, we consider a network
of two nodes $\a$ and $\b$ ($\a,\b\in\tIP$)
on which the same process is running, although starting in different states.
The process describes a simply (toy-)protocol: whenever a new data packet 
for destination \dval{dip} ``appears'',\footnote{In
this small example, we assume that new data packets just
  appear ``magically''; of course one could use the message
$\newpkt{\data}{\dip}$ instead.}
the data is broadcast through the network until it finally reaches \dval{dip}. 
A node alternates between broadcasting, and receiving and handling a message.
The \dval{data} stemming from a message received by node \dval{ip} will be delivered to
the application layer if the message is destined for \dval{ip} itself. Otherwise the node
forwards the message. 
Every message travelling through the network and handled by the protocol 
has the form $\mymsg{\dval{data}}{\dval{dip}}$, where $\dval{data}\in\tDATA$ is the data to be sent 
and $\dval{dip}\in\tIP$ is its destination.
The behaviour of each node can be modelled by:%
\newcommand{\XP}{\keyw{X}}%
\newcommand{\YP}{\keyw{Y}}%
\begin{simpleProcess}
	\item[$\XP(\ip\,\comma\,\data\comma\dip)$]\hspace{-\labelsep}\ $\stackrel{{\it def}}{=}
 	\textbf{broadcast}(\mymsg\data\dip).\YP(\ip)$
	\item[$\YP(\ip)$]\hspace{-\labelsep}
		$\stackrel{{\it def}}{=} \textbf{receive}(\keyw{m}).%
		([\keyw{m} \mathord= \mymsg\data\dip\wedge\dip\mathord=\ip] \textbf{deliver}(\data).\YP(\ip)$\\ 
		\hspace{7.25em}
		$+ [\keyw{m} \mathord= \mymsg\data\dip\wedge\dip\mathord{\not=}\ip] \XP(\ip\,\comma\,\data\comma\dip))$\ .
\end{simpleProcess}%
\vspace{2pt}
If a node is in a state $\XP(\dval{ip}\,\comma\,\dval{data}\comma\dval{dip})$, where $\dval{ip}\in\tIP$
is the node's stored value of its ow IP address,
it will broadcast $\mymsg{\dval{data}}{\dval{dip}}$ and continue in state $\YP(\dval{ip})$, 
meaning that all information about the message is dropped.
If a node in state $\YP(\dval{ip})$ receives a message $m$---a value that will be assigned to the variable
$\keyw{m}$---it has two ways to continue: process [$\keyw{m} \mathord= \mymsg\data\dip\wedge\dip\mathord=\ip$] \textbf{deliver}(\data).\YP(\ip) is enabled if 
the incoming message has the form $\mymsg{\dval{data}}{\dval{dip}}$
and the node itself is the destination of the
message ($\dip\mathord=\ip$). In
that case the data distilled from $m$ will be delivered to the application layer, and the process returns
to $\YP(\dval{ip})$. Alternatively, if [$\keyw{m} \mathord= \mymsg\data\dip\wedge\dip\mathord{\not=}\ip$], the process continues as
$\XP(\dval{ip}\,\comma\,\dval{data}\comma\dval{dip})$, which will then broadcast another
message with contents $\dval{data}$ and $\dval{dip}$.
Note that calls to processes use expressions as parameters.

Let us have a look at three scenarios.
First, assume that the nodes $\a$ and $\b$ are within transmission range of each other; node $\a$ in state
$\XP(\a\,\comma\,d\comma\a)$, and node $\b$ in $\YP(\b)$. This is formally expressed as
$[\colonact{\a}{\XP(\a\,\comma\,d\comma\a)}\mathop{:}\{\b\}\|\,\colonact{\b}{\YP(\b)}\mathop{:}\{\a\}]$,
although for compactness of presentation, we just write
$[\XP(\a\,\comma\,d\comma\a)\,\|\,\YP(\b)]$ below.
In this case, node $a$ broadcasts the message $\mymsg{d}{\a}$ and
continues as $\YP(\a)$. Node $\b$ receives the message, and continues (after evaluation of the message) as $\XP(\b\,\comma\,d\comma\a)$. Next $\b$
broadcasts (forwards) the message, and continues as $\YP(\b)$, while node $\a$ receives $\mymsg{d}{\b}$, and, due to evaluation, \textbf{deliver}s $d$ and continues as $\YP(\a)$.  Formally,
\index{transition}%
we get transitions from one state to the other:
\newcommand{\sm}[1]{\mbox{$\scriptstyle #1$}}
\[
[\XP(\a\,\comma\,d\comma\a)\,\|\,\YP(\b)]
	\ar{{\sm{\a}:\textbf{*cast}\sm{(\mymsg{{d}}{\a})}}}%
	\ar{\tau}
 [\YP(\a)\,\|\,\XP(\b\,\comma\,{d}\comma\a)]
	\ar{{\sm{\b}:\textbf{*cast}\sm{(\mymsg{{d}}{\a})}}}
	\ar{\tau}
	\ar{{\sm{\a}:\textbf{deliver}\sm{({d})}}} 
[\YP(\a)\,\|\,\YP(\b)].
\]
Here, the $\tau$-transitions are the actions of evaluating one of the two guards of a process $\YP$,
and we left out three intermediate expressions.

Second, assume that the nodes are not within transmission range,
with the initial process of $\a$ and $\b$ the same as above; formally
[$\colonact{\a}{\XP(\a\,\comma\,d\comma\a)}\mathop{:}\emptyset\,\|\,\colonact{\b}{\YP(\b)}\mathop{:}\emptyset$].
As before, node $a$ broadcasts $\mymsg{d}{\a}$ and continues in $\YP(\a)$; but this
time the message is not received by any node; hence
no message is forwarded or delivered and both nodes end up running process~$\YP$.

For the last scenario, we assume that $a$ and $b$ are
within transmission range and that
they have the initial states $\XP(\a\,\comma\,d\comma\b)$ and $\XP(\b\,\comma\,e\comma\a)$.
Without the augmentation of \SSect{non-blocking},
the network expression $[\XP(\a\,\comma\,d\comma\b)\,\|\,\XP(\b\,\comma\,e\comma\a)]$ admits no transitions at all;
neither node can broadcast its message, because the other node is not listening.
With the optional augmentation,
assuming that node $a$ sends first:
\[
[\XP(\a\,\comma\,d\comma\b)\,\|\,\XP(\b\,\comma\,e\comma\a)] 
	\ar{{\sm{\a}:\textbf{*cast}\sm{((\mymsg{d}{\b})}}} 
[\YP(a)\,\|\,\XP(\b\,\comma\,e\comma\a)]
	\ar{{\sm{\b}:\textbf{*cast}\sm{(\mymsg{e}{\a})}}}
	\ar{\tau}
\ar{{\sm{\a}:\textbf{deliver}\sm{(e)}}} [\YP(\a)\,\|\,\YP(\b)].
\]
Unfortunately, node $\b$ is initially in a state where it cannot receive a message,
so $\a$'s message ``remains unheard'' and $\b$ will never deliver that message.
To avoid this behaviour, and ensure that both messages get delivered,
as happens in real WMNs, a message queue can be
introduced (see \SSect{message_queue}). Using a message queue, 
the optional augmentation is not needed, since any node is always in a state where it can receive a message.

\section{Data Structure for AODV}\label{sec:types}
\index{data structure}%
\index{data types}%
In this section we set out the basic data structure needed for the
detailed formal specification of AODV.  As well as describing
\emph{types} for the information handled at the nodes during the
execution of the protocol we also define functions which will be used
to describe the precise intention---and overall effect---of the
various update mechanisms in an AODV implementation. The definitions
are grouped roughly according to the various ``aspects" of AODV and
the host network.

\subsection{Mandatory Types}\label{ssec:ip}
\index{application layer data}%
\index{messages}%
\index{IP addresses}%
As stated in the previous section, the data structure always consists
of application layer data, messages, IP addresses and sets of IP addresses.

\begin{enumerate}[(a)]
\item\label{tDATA} The ultimate purpose of AODV is to deliver
  \emph{application layer data}.  The type $\tDATA$ describes a
  set of application layer data items. An item of data is thus a particular element of that
  set, denoted by the variable $\data\in\tDATA$.
\item \emph{Messages} are used to send information via the network. In
  our specification we use the variable $\msg$ of the type $\tMSG$.
  \index{control message}%
  \index{route request (RREQ)}%
  \index{route reply (RREP)}%
  \index{route error (RERR)}%
  \index{data packet}%
  We distinguish AODV control messages (route request, route
  reply, and route error) as well as \emph{data packets}: messages for sending
  application layer data (see \SSect{messages}).
\item
  The type $\tIP$ describes a set of IP addresses or, more generally, a
  \index{node identifiers}%
  \emph{set of node identifiers}.  In the RFC 3561~\cite{rfc3561},
  $\tIP$ is defined as the set of all IP addresses.  
  We assume that
  each node has a unique identifier $\dval{ip}\in\tIP$.
  Moreover, in our  model, each node \dval{ip} maintains a variable {\ip}
  which always has the value \dval{ip}.
  In any AODV control message, the variable {\sip} holds the IP
  address of the sender, and if the message is part of the \phrase{route
  discovery process}---a route request or route reply message---we use
  {\oip} and {\dip} for the origin and destination of the route
  sought.  Furthermore, {\rip} denotes an unreachable destination and
  {\nhip} the next hop on some route.
\end{enumerate}

\subsection{Sequence Numbers}\label{ssec:sequence numbers}

As explained in \Sect{aodv}, any node maintains its own \phrase{sequence number}---
the value of the variable {\sn}---and
\index{routing table}%
a routing table whose entries describe routes to other nodes. The value of {\sn} increases over time.
AODV equips each routing table entry  with a
sequence number to constitute a measure approximating the
\index{fresh}%
relative freshness of the information held---a smaller number denotes
older information.  All sequence numbers of routes to
$\dval{dip}\in \tIP$ stored in routing tables are ultimately derived from
\dval{dip}'s own sequence number at the time such a route was discovered.

We denote the set of sequence numbers by $\tSQN$ and assume it to be
totally ordered.  By default we take $\tSQN$ to be $\NN$, and use
standard functions such as $\max$.  The initial sequence number of any
node is $1$.  We reserve a special element $0\in\tSQN$ to be used 
for the sequence number of a route, whose semantics is that no
sequence number for that route is known.  Sequence numbers are
incremented by the function\vspace{-1ex}
\index{inc@\keyw{inc}}%
\hypertarget{inc}{
\[\begin{array}{r@{\hspace{0.5em}}c@{\hspace{0.5em}}l}
\fninc:\tSQN&\to&\tSQN\\
\inc{\dval{sn}}&=&
\left\{\begin{array}{ll}
  \dval{sn}+1&\mbox{if } \dval{sn}\not=0\\
  \dval{sn}&\mbox{otherwise}\ .
\end{array}\right.
\end{array}\]
}
The variables $\osn$, $\dsn$ and $\rsn$ of type $\tSQN$ are used to denote the
sequence numbers of routes leading to the nodes $\oip$, $\dip$ and
$\rip$.

\index{sequence number!known}%
\index{sequence number!unknown}%
\index{sequence-number-status flag}%
AODV tags sequence numbers of routes as ``known'' or ``unknown''. This
indicates whether the value of the sequence number can be trusted. The
sequence-number-status flag is set to unknown (\unkno) when a routing table entry is
updated with information that is not equipped with a sequence number
itself.  In such a case the old sequence number of the entry is
maintained; hence the value {\unkno} does not indicate
that no sequence number for the entry is known.
Here we use the set $\tSQNK=\{\kno,\unkno\}$ for the possible values
of the sequence-number-status flag; we use the variable $\keyw{dsk}$
to range over type $\tSQNK$.

\subsection{Modelling Routes}

\index{connected}%
\index{transmission range}%
In a network, pairs $(\dval{ip}_0, \dval{ip}_k) \in \tIP \times \tIP$
of nodes are considered to be ``connected" if $\dval{ip}_0$ can send
to $\dval{ip}_k$ directly, i.e., $\dval{ip}_0$ is in transmission range of
$\dval{ip}_k$ and vice versa. We say that such nodes are connected by
a single \phrase{hop}. When $\dval{ip}_0$ is not connected to $\dval{ip}_k$
then messages from $\dval{ip}_0$ directed to $\dval{ip}_k$ need to be
``routed" through intermediate nodes. We say that a \phrase{route}
(from $\dval{ip}_0$ to $\dval{ip}_k$) is made up of a sequence
$[\dval{ip}_0,\dval{ip}_1,\dval{ip}_2,\dots,\dval{ip}_{k-1},\dval{ip}_k]$,
where $(\dval{ip}_{i}, \dval{ip}_{i+1})$, $i=0,\dots, k\mathord-1$, are
connected pairs; the \emph{length} or \phrase{hop count} of the route is
the number of single hops, and any node $\dval{ip}_i$ needs only to
\index{next hop}%
\index{destination}%
know the ``next hop" address $\dval{ip}_{i+1}$ in order to be able to
route messages intended for the final destination $\dval{ip}_{k}$.

In operation, routes to a particular destination are requested and,
when finally established, need to be re-evaluated in regard to their
\index{validity status}%
\index{route!invalid}%
``validity". Routes may become \emph{invalid} if one of the pairs
$(\dval{ip}_i, \dval{ip}_{i+1})$ in the hop-to-hop sequence gets
disconnected. Then AODV may be reinvoked, as the need arises, to
discover alternative routes.  Meanwhile, an invalid route remains
invalid until fresh information is received which establishes a valid
replacement route.

In addition to the next hop and hop count, AODV also ``tags" a route
with its validity, sequence number and sequence-number status.
For every route, a node moreover stores a list of \phrase{precursors},
modelled as a set of $\tIP$ addresses. This set collects all nodes
which are currently potential users of the route, and are located one
hop further ``upstream''.  When the interest of other
nodes emerges, these nodes are added to the precursor
list\footnote{The RFC does not mention a situation where nodes are
dropped from the list, which seems curious.};  the main purpose of
recording this information is to inform those nodes when the route
becomes invalid.

In summary, following the RFC, a routing table entry (or entry for short) is given by $7$ components:

\begin{enumerate}[(a)]
\item The destination IP address, which is an element of $\tIP$;
\index{destination}%
\item The destination sequence number---an element of $\tSQN$;
\index{destination sequence number}%
\item The sequence-number-status flag---an element of the set $\tSQNK=\{\kno,\unkno\}$;
\index{sequence-number-status flag}%
\item A flag tagging the route as being valid or invalid---an element
  of the set $\tFLAG= \{\val,\inval\}$. We use the variable $\flag$
  to range over type $\tFLAG$;
\index{validity status}%
\index{route!valid}%
\index{route!invalid}%
\item The hop count, which is an element of $\NN$.  The variable
  $\hops$ ranges over the type $\NN$ and we make use of the standard
  function $+1$;
\index{hop count}%
\item The next hop, which is again an element of $\tIP$; and
\index{next hop}%
\item A precursor list, which is modelled as an element of $\pow(\tIP)$.%
\index{precursor list}%
\footnote{The word ``precursor list'' is used in the RFC, but no properties of lists are used.}
  We use the variable $\pre$ to range over $\pow(\tIP)$.
\end{enumerate}
We denote the type of routing table entries by $\tROUTE$, use
the variable \keyw{r}, and define a generation function
\[\begin{array}{r@{\hspace{0.5em}}c@{\hspace{0.5em}}l}	
 (\_\comma\_\comma\_\comma\_\comma\_\comma\_\comma\_\,): \tIP \times \tSQN \times\tSQNK \times \tFLAG
\times \NN \times \tIP \times \pow(\tIP) &\rightarrow& \tROUTE\ .
 \end{array}\]
A tuple $(\dval{dip}\comma\dval{dsn}\comma\dval{dsk}\comma\dval{flag}\comma\dval{hops}\comma\dval{nhip}\comma\dval{pre})$
\index{route}%
describes a route to $\dval{dip}$ of length $\dval{hops}$ and validity
$\dval{flag}$; the very next node on this route is $\dval{nhip}$; the
last time the entry was updated the destination sequence number was
$\dval{dsn}$; \dval{dsk} denotes whether the sequence number
is ``outdated'' or can be used to reason about freshness of the route.
Finally, $\dval{pre}$ is a set of all neighbours who are
``interested'' in the route to $\dval{dip}$.  A node being
``interested" in the route is somewhat sketchily defined as one which
has previously used the current node to route messages to
$\dval{dip}$. Interested nodes are recorded in case the route to
$\dval{dip}$ should ever become invalid, so that they may subsequently
be informed.  We use projections $\pi_{1},\dots\pi_{7}$ to select the
corresponding component from the $7$-tuple: For example,
$\pi_6:\tROUTE\to\tIP$ determines the next hop.

\subsection{Routing Tables}\label{ssec:rt}

{
\renewcommand{\ip}{\dval{ip}}
\renewcommand{\dip}{\dval{dip}}
\renewcommand{\oip}{\dval{oip}}
\renewcommand{\sip}{\dval{sip}}
\renewcommand{\rip}{\dval{rip}}
\renewcommand{\rt}{\dval{rt}}
  \newcommand{\nrt}{\dval{nrt}}
\renewcommand{\route}{\dval{r}}
  \newcommand{\s}{\dval{s}}
  \newcommand{\nr}{\dval{nr}}
  \newcommand{\ns}{\dval{ns}}
\renewcommand{\osn}{\dval{osn}}
\renewcommand{\dsn}{\dval{dsn}}
\renewcommand{\rsn}{\dval{rsn}}
\renewcommand{\flag}{\dval{flag}}
\renewcommand{\hops}{\dval{hops}}
\renewcommand{\nhip}{\dval{nhip}}
\renewcommand{\pre}{\dval{pre}}
  \newcommand{\npre}{\dval{npre}}
\renewcommand{\dests}{\dval{dests}}
\renewcommand{\rreqid}{\dval{rreqid}}
\renewcommand{\rreqs}{\dval{rreqs}}

\index{routing table}%
\index{routing table entry}%
Nodes store all their information about routes in their \emph{routing
tables}; a node \dval{ip}'s routing table consists of a set of routing
table entries, exactly one for each known destination.  Thus, a
routing table is defined as a set of entries, with the restriction
that each has a different destination $\dip$, i.e., the first
component of each entry in a routing table is unique.\footnote{As an
alternative to restricting the set, we could have defined routing
tables as partial functions from $\tIP$ to $\tROUTE$, in which case it
makes more sense to define an entry as a $6$-tuple, not including the
the destination IP as the first component.}  Formally, we define the
type $\tRT$ of routing tables by
\[\begin{array}{r@{\hspace{0.5em}}c@{\hspace{0.5em}}l}	
\tRT &:=& \{\rt\mid \rt\in\pow(\tROUTE)\ans\forall \route,\s\in\rt:\route\not=\s\Rightarrow\pi_{1}(\route)\not=\pi_{1}(\s)\}\ .
\end{array}\]
In the specification and implementation of AODV during route finding
nodes choose between alternative routes if necessary to ensure
that only one route per destination ends up in their routing table.
In our model, each node \dval{ip} maintains a variable \keyw{rt},
whose value is the current routing table of the node. 

In the formal model (and indeed in any AODV implementation) we need to
extract the components of the entry for any given destination from a
routing table. To this end, we define the following partial
functions---they are partial because the routing table need not have
an entry for the given destination.
We begin by selecting the entry in a routing table corresponding to a given destination $\dip$:
\hypertarget{selroute}{
\[\begin{array}{@{}r@{\hspace{0.5em}}c@{\hspace{0.5em}}l@{}}
		\fnselroute:\tRT\times\tIP&\rightharpoonup&\tROUTE\\
		\selr{\rt}{\dip}&:=&
		\left\{
		  \begin{array}{ll}
		   \route&\mbox{if } \route\in\rt\ans \pi_{1}(\route)=\dip\\
		   \mbox{undefined}&\mbox{otherwise}\ .
		\end{array}\right.\!\!\!\!
\end{array}\]
}
Through
the projections $\pi_{1},\dots,\pi_{7}$, defined above, we can now
select the components of a selected entry:
\begin{enumerate}[(a)]	
\item The \phrase{destination sequence number} relative to the destination $\dip$:
\index{sqn@\keyw{sqn}}%
\hypertarget{sqn}{
\[\begin{array}{r@{\hspace{0.5em}}c@{\hspace{0.5em}}l}
		    	  \fnsqn : \tRT\times\tIP&\to& \tSQN\\
	    \sqn{\rt}{\dip}&:=&
	    \left\{
	    \begin{array}{ll}
	        \pi_{2}(\selr{\rt}{\dip}) &\mbox{if }\selr{\rt}{\dip}\mbox{ is defined}\\
	        0&\mbox{otherwise}
	    \end{array}\right.
	     \end{array}\]
}
\item The \emph{``known'' status} of the sequence number of a route:
\index{sequence-number-status flag}%
\index{sqnf@\keyw{sqnf}}%
\hypertarget{sqnf}{
\[\begin{array}{r@{\hspace{0.5em}}c@{\hspace{0.5em}}l}
		    	  \fnsqnf : \tRT\times\tIP&\to& \tSQNK\\
	    \sqnf{\rt}{\dip}&:=&
	    \left\{
	    \begin{array}{ll}
	        \pi_{3}(\selr{\rt}{\dip}) &\mbox{if }\selr{\rt}{\dip}\mbox{ is defined}\\
	        \unkno&\mbox{otherwise}
	    \end{array}\right.
	     \end{array}\]
}
\item The \phrase{validity status} of a recorded route:
\index{flag@\keyw{flag}}%
\hypertarget{status}{
\[\begin{array}{r@{\hspace{0.5em}}c@{\hspace{0.5em}}l}
              \fnstatus :\ \tRT\times\tIP&\rightharpoonup& \tFLAG\\
	     \status{\rt}{\dip}&:=& \pi_{4}(\selr{\rt}{\dip})
	\end{array}
\]	
}
\item The \phrase{hop count} of the route from the current node (hosting $\rt$) to $\dip$:
\index{dhops@\keyw{dhops}}%
\hypertarget{dhops}{
\[\begin{array}{r@{\hspace{0.5em}}c@{\hspace{0.5em}}l}
	       \fndhops :\tRT\times\tIP&\rightharpoonup& \NN\\
	     \dhops{\rt}{\dip}&:=&  \pi_{5}(\selr{\rt}{\dip})\\[-2ex]
  \end{array}\]
}\pagebreak[4]
\item The \emph{identity of the next node on the route to} $\dip$ (if such a route is known):
\index{next hop}%
\index{nhop@\keyw{nhop}}%
\hypertarget{nhop}{
\[\begin{array}{r@{\hspace{0.5em}}c@{\hspace{0.5em}}l}
 	\fnnhop :\tRT\times\tIP&\rightharpoonup& \tIP\\
         \nhop{\rt}{\dip}&:=& \pi_{6}(\selr{\rt}{\dip})
\end{array}\]
}
\item The set of \phrase{precursors} or neighbours interested in
  using the route from $\ip$ to $\dip$:
\index{precs@\keyw{precs}}%
\hypertarget{precs}{
\[\begin{array}{r@{\hspace{0.5em}}c@{\hspace{0.5em}}l}
	    \fnprecs : \tRT\times\tIP&\rightharpoonup&\pow(\tIP)\\
	     \precs{\rt}{\dip}&:=& \pi_{7}(\selr{\rt}{\dip})\\
  \end{array}\]
}
\end{enumerate}
The domain of these partial functions changes during the operation of
AODV as more routes are discovered and recorded in the routing table
$\rt$.  The first two functions are extended to be total functions:
whenever there is no route to $\dip$ inside the routing table under
consideration, the sequence number is set to ``unknown''  $(0)$ and 
\index{sequence number!unknown}%
\index{sequence-number-status flag}%
the sequence-number-status flag is set to ``unknown'' $(\unkno)$, respectively. In the same
style each partial function could be turned into a total one. However,
in the specification we use these functions only when they are defined.

We are not only interested in information about a single route, but
also in general information on a routing table:
\begin{enumerate}[(a)]	
\item The set of destination IP addresses for \emph{valid} routes in $\rt$ is given by
\index{route!valid}%
\index{vD@\keyw{vD}}%
\[\begin{array}{r@{\hspace{0.5em}}c@{\hspace{0.5em}}l}
	      \fnakD :\tRT&\to& \pow(\tIP)\\
	    \akD{\rt}&:=& \{\dip\mid (\dip\comma*\comma*\comma\val\comma*\comma*\comma*)\in\rt\}
\end{array} \]
	
\item The set of destination IP addresses for \emph{invalid} routes in $\rt$ is
\index{route!invalid}%
\index{iD@\keyw{iD}}%
\[\begin{array}{r@{\hspace{0.5em}}c@{\hspace{0.5em}}l}    	
	      \fnikD :\tRT&\to& \pow(\tIP)\\
	      \ikD{\rt}&:=& \{\dip\mid (\dip\comma*\comma*\comma\inval\comma*\comma*\comma*)\in\rt\}
  \end{array} \]

\item Last, we define the set of destination IP addresses for \emph{known} routes by
\index{route!known}%
\index{kD@\keyw{kD}}%
\[\begin{array}{r@{\hspace{0.5em}}c@{\hspace{0.5em}}l}    	
	      \fnkD :\tRT&\to& \pow(\tIP)\\
	      \kD{\rt}&:=&\akD{\rt}\cup\ikD{\rt} = \{\dip\mid (\dip\comma*\comma*\comma*\comma*\comma*\comma*)\in\rt\}
\end{array}\]
\end{enumerate}
Obviously, the partial functions $\fnselroute$, {\fnstatus},
{\fndhops}, {\fnnhop} and {\fnprecs} are defined for {\rt} and {\dip}
exactly when $\dip \in \kD{\rt}$.

\subsection{Updating Routing Tables}

Routing tables can be updated for three principal reasons. The first
is when a node needs to adjust its list of precursors relative to a
given destination; the second is when a received request or response
carries information about network connectivity; and the last when
information is received to the effect that a previously valid route
should now be considered invalid. We define an update function for
each case.

\subsubsection{Updating Precursor Lists}

\index{precursors}%
\index{addpre@\keyw{addpre}}%
Recall that the precursors of a given node $\ip$ relative to a
particular destination $\dip$ are the nodes that are ``interested" in a route to {\dip}
via {\ip}.  The function $\fnaddprec$ takes a routing table entry and a set of IP addresses
$\npre$ and updates the entry by adding $\npre$ to the list of precursors
already present:
\[\begin{array}{@{}r@{\hspace{0.5em}}c@{\hspace{0.5em}}l@{}}
\fnaddprec : \tROUTE\times \pow(\tIP) &\to& \tROUTE\\
\addprec{(\dip\comma\dsn\comma\dval{dsk}\comma\flag\comma\hops\comma\nhip\comma\pre)}{\npre} &:=& (\dip\comma\dsn\comma\dval{dsk}\comma\flag\comma\hops\comma\nhip\comma\pre\cup\npre)\ .
\end{array}\]

Often it is necessary to add  precursors to an entry of a given
routing table. For that, we define the function $\fnaddprecrt$, which takes a routing table \rt, a destination {\dip} and a set of IP addresses
$\npre$ and updates the entry with destination $\dip$ by adding $\npre$ to the list of precursors
already present. It is only 
defined if an entry for destination $\dip$ exists.
\index{addpreRT@\keyw{addpreRT}}%
\hypertarget{addprert}{
\[\begin{array}{@{}r@{\hspace{0.5em}}c@{\hspace{0.5em}}l@{}}
\fnaddprecrt : \tRT\times \tIP\times \pow(\tIP) &\rightharpoonup& \tRT\\
\addprecrt{\rt}{\dip}{\npre} &:=& (\rt - \{\selr{\rt}{\dip}\})\cup\{\addprec{\selr{\rt}{\dip}}{\npre}\}\ .
\end{array}
\]
}
Formally, we remove the entry with destination $\dip$ from the routing table and 
insert a new entry for that destination. This new entry is the same as before---only the precursors have been added.

\subsubsection{Inserting New Information in Routing Tables}\label{sssec:update}

\index{update@\keyw{update}}%
If a node gathers new information about a route to a
destination \dip, then it updates its routing table depending on
its existing information on a route to \dip.
If no route to {\dip} was known at all, it inserts a new entry
in its routing table  recording the information received.
If it already has some (partial) information then it may update
this information, depending on whether the new route is fresher or
shorter than the one it has already.
We define an update function $\upd{\rt}{\route}$ of a routing table 
$\rt$ with an entry $\route$ only when
 $\route$ is valid, i.e., $\pi_{4}(\route)=\val$,
$\pi_{2}(\route)=0\Leftrightarrow\pi_{3}(\route)=\unkno$,
and $\pi_{3}(\dval{r})=\unkno\ims\pi_{5}(\dval{r})=1$.\footnote{
After we have introduced our specification for AODV in Section~\ref{sec:modelling_AODV},
we will justify that this definition is sufficient.}\vspace{-1ex}
\hypertarget{update}{
\[\begin{array}{@{}l@{\hspace{0.5em}}c@{\hspace{0.5em}}l@{}}
\multicolumn{3}{l}{\fnupd : \tRT\times\tROUTE\ \ \rightharpoonup\ \ \tRT}\label{df:update}\\
\upd{\rt}{\hspace{-1pt}\route}&:=& \left\{
\begin{array}{@{\,}ll@{}}
\rt\cup\{\route\} & \mbox{if }  \pi_{1}(\route)\not\in\kD{\rt}\\[1mm]
\nrt\cup\{\nr\}&\mbox{if }  \pi_{1}(\route)\in\kD{\rt} \wedge  \sqn{\rt}{\pi_{1}(\route)}<\pi_{2}(\route)\\[1mm]
\nrt\cup\{\nr\}&\mbox{if }  \pi_{1}(\route)\in\kD{\rt} \wedge \sqn{\rt}{\pi_{1}(\route)}=\pi_{2}(\route) \wedge \dhops{\rt}{\pi_{1}(\route)}>\pi_{5}(\route)\\[1mm]
\nrt\cup\{\nr\}&\mbox{if }  \pi_{1}(\route)\in\kD{\rt} \wedge \sqn{\rt}{\pi_{1}(\route)}=\pi_{2}(\route) \wedge \status{\rt}{\pi_{1}(\route)}=\inval\\[1mm]
\nrt\cup\{\nr'\}&\mbox{if } \pi_{1}(\route)\in\kD{\rt} \wedge  \pi_3(\route)=\unkno\\[1mm]
\nrt\cup\{\ns\}&\mbox{otherwise\ ,}
\end{array}
\right.\vspace{-1ex}
\end{array}\]
}
where $\s:=\selr{\rt}{\pi_{1}(\route)}$ is the current entry in the
routing table for the destination of $\route$ (if it exists), and
$\nrt := \rt -\{\s\}$ is the routing table without that entry.
The entry $\nr:=\addprec{\route}{\pi_{7}(\s)}$ is identical to~$\route$ except
that the precursors from $\s$ are added and $\ns:=\addprec{\s}{\pi_{7}(\route)}$
is generated from $\s$ by adding the precursors from $\route$. 
Lastly, 
$\nr'$ is identical to $\nr$ except that the sequence number is replaced by the one from 
the route $s$. More precisely,
$\nr':=(\dip_{\nr}\comma\pi_{2}(\s)\comma\dval{dsk}_{\nr}\comma\flag_{\nr}\comma\hops_{\nr}\comma\nhip_{\nr}\comma\pre_{\nr})$ if
$\nr=(\dip_{\nr}\comma*\comma\dval{dsk}_{\nr}\comma\flag_{\nr}\comma\hops_{\nr}\comma\nhip_{\nr}\comma\pre_{\nr})$.
In the situation where $\sqn{\rt}{\pi_{1}(\route)}=\pi_{2}(\route)$ both routes $\nr$ and $\nr'$ are equal.
Therefore, though the cases of the above definition are not
mutually exclusive, the function is well defined.

The first case describes the situation where the routing table does not contain any
information on a route to $\dip$. 
The second case models the situation where the new route has a greater
sequence number. As a consequence all the information from the incoming information 
is copied into the routing table. In the third and fourth case the sequence numbers 
are the same and cannot be used to identify better information. 
Hence other measures are used. The route inside the routing table is only replaced if 
either the new hop count is strictly smaller---a shorter route has been found---or if the route inside 
the routing table is marked as invalid.  The fifth case deals with the situation where a
new route to a known destination has been found without any
  information on its sequence number ($\pi_{2}(\route)=0\wedge\pi_{3}(\route)=\unkno$).
In that case the routing table entry to that destination is always
updated, but the existing sequence number is maintained, and marked as ``unknown''.

Note that we do not update if we receive a new entry where the sequence number and
the hop count are identical to the current entry in the routing table. Following the
RFC, the time period (till the valid route becomes invalid) should be reset; however
at the moment we do not model timing aspects.

\subsubsection{Invalidating Routes}\label{sssec:invalidate}

\index{invalidate@\keyw{invalidate}}%
Invalidating routes is a main feature of AODV; if a route is not valid
any longer its validity flag has to be set to invalid. By doing
this, the stored information about the route, such as the sequence number or the hop count,
remains accessible. The process of invalidating a routing table entry follows four rules:
(a) any sequence number is incremented by $1$, except
(b) the truly unknown sequence number ($\dval{sqn}=0$, which will only
occur if $\dval{dsk}=\unkno$) is not incremented,
(c) the validity flag of the entry is set to \inval, and
(d) an invalid entry cannot be invalidated again.
However, in exception to (a) and (b), when the invalidation is in
response to an error message, this message also contains a new (and
already incremented) sequence number for each destination to be invalidated.

The function for invalidating routing table entries takes as arguments a routing
table and a set of destinations $\dests\in\pow(\tIP\times\tSQN)$.
Elements of this set are $(\rip,\rsn)$-pairs that not only identify an
unreachable destination $\rip$, but also a sequence number that
describes the freshness of the faulty route.  As for routing tables,
we restrict ourselves to sets that have at most one entry for each
destination; this time we formally define {\dests} as a \emph{partial
function} from $\tIP$ to $\tSQN$, i.e.\ a subset of $\tIP\times\tSQN$
satisfying 
\[(\rip,\rsn),(\rip,\rsn')\in \dests \Rightarrow \rsn=\rsn'\ .
\]
We use the variable \keyw{dests} to range over such sets.
When invoking {\fninv} we either distil {\dests} from an error
message, or determine {\dests} as a set of pairs $(\rip,\inc{\sqn{\rt}{\rip}}$,
where the operator {\fninc} (from \SSect{sequence numbers}) takes care of (a) and (b).
Moreover, we will distil or construct {\dests} in such a way that it only lists
destinations for which there is a valid entry in the routing table---this takes care
of (d).
\hypertarget{invalidate}{
\[\begin{array}{@{}r@{\ \ }c@{\ \ }l@{}}
\multicolumn{3}{l}{\fninv : \tRT\times(\tIP\rightharpoonup\tSQN) \to\tRT}\\
\inv{\rt}{\dests}&:=& \{\route\,|\,\route\in\rt\ans (\pi_{1}(\route),*)\not\in\dests\}\\
&\cup&\{(\pi_{1}(\route),\highlight\rsn,\pi_{3}(\route),\inval,\pi_{5}(\route),\pi_{6}(\route),\pi_{7}(\route))\mid\route\in\rt\ans(\pi_{1}(r),\rsn)\in\dests\}
\end{array}\]
}

\noindent
All entries in the routing table for a destination $\rip$ in $\dests$
are modified.
The modification replaces the value {\val} by {\inval} and
the sequence number in the entry by the corresponding sequence number from $\dests$.

Copying the sequence number from $\dests$ leaves the possibility
that the destination sequence number of an entry is decreased, which violates one
of the fundamental assumption of AODV and may yield unexpected behaviour (cf.~\SSect{decreasingSQN}).
To guarantee an increase of the sequence number, \highlight{\rsn} in Line $3$ of the above
definition could be replaced by taking the maximum of  the sequence number
that was already in the routing table $+1$, and the sequence number from
$\dests$, i.e., $\highlight{\max(\inc{\pi_2(r)},\rsn)}$. 

\subsection{Route Requests}\label{ssec:rreqs}

\index{route request (RREQ)}%
A route request---RREQ---for a destination $\dip$ is initiated by a
node (with routing table $\rt$) if this node wants to transmit a data
packet to $\dip$ but there is no valid entry for $\dip$ in the routing
table, i.e.\ $\dip \mathbin{\not\in} \akD{\rt}$. When a new route request is sent
out it contains the identity of the originating node $\oip$, and a
\phrase{route request identifier} (RREQ ID); the type of all such identifiers
 is denoted by $\tRREQID$, and the variable \keyw{rreqid}
ranges over this type. This information does not change, even when
the request is re-broadcast by any receiving node that does not
already know a route to the requested destination. In this way any
request still circulating through the network can be uniquely
identified by the pair $(\oip, \rreqid)\in\tIP\times\tRREQID$. For our
specification we set $\tRREQID=\NN$. In our model, each node
maintains a variable \keyw{rreqs} of type
\[
		\pow(\tIP\times\tRREQID)
\]
of sets of such pairs to store the sets of route requests seen
by the node so far. Within this set, the node records the requests
it has previously initiated itself.
\index{nrreqid@\keyw{nrreqid}}%
\hypertarget{nrreqid}{To ensure a fresh {\rreqid} for each new RREQ it generates,
the node {\ip} applies the following function:
\[\begin{array}{r@{\hspace{0.5em}}c@{\hspace{0.5em}}l}    	
	     	  \fnnrreqid:\pow(\tIP\times\tRREQID)\times\tIP&\to& \tRREQID\\
		\nrreqid{\rreqs}{\ip}&:=&\max\{n\mid(\ip,n)\in\rreqs\}+1\ ,
\end{array}\]
where we take the maximum of the empty set to be $0$.}

\subsection{Queued Packets}\label{ssec:store}
Strictly speaking the task of sending data packets is not regarded as
part of the AODV protocol---however, failure to send a packet because
either a route to the destination is unknown, or a previously known
route has become invalid, prompts AODV to be activated. In our
modelling we describe this interaction between packet sending and
AODV, providing the minimal infrastructure for our specification.

If a new packet is submitted by a client of AODV to a node, it may
need to be stored until a route to the packet's destination has been
found and the node is not busy carrying out other AODV tasks. We use a
queue-style data structure for modelling the store of packets at a
node, noting that at each node there may be many data queues, one for
each destination. In general, we denote queues of type $\tTYPE$ by
$[\tTYPE]$, denote the empty queue by $[\,]$, and make use of the
standard (partial) functions
\index{head@\keyw{head}}%
\index{tail@\keyw{tail}}%
\index{append@\keyw{append}}%
$\hypertarget{head}{\fnhead}:[\tTYPE]\rightharpoonup\tTYPE$,
$\fntail:[\tTYPE]\rightharpoonup[\tTYPE]$ and
$\fnappend:\tTYPE\times[\tTYPE]\rightarrow[\tTYPE]$ that return the
``oldest'' element in the queue, remove the ``oldest'' element, and
add a packet to the queue, respectively. 

The data type\vspace{-1ex}
\[\begin{array}{r@{\hspace{0.5em}}c@{\hspace{0.5em}}l}	
\tQUEUES &:=& \left\{\dval{store}\,\left|
\begin{array}{@{~}l@{}}\dval{store}\in\pow(\tIP\times\tPendingRREQ\times[\tDATA])\ans\mbox{}\\
\big((\dip,p,q),(\dip,p',q')\in\dval{store} \Rightarrow p=p' \wedge q=q'\big)
\end{array}\right.\right\}
\end{array}\]
\index{queued data}%
describes stores of enqueued data
packets for various destinations, where $\tPendingRREQ:=\{\pen,\nonpen\}$. An element $(\dip, p, q) \in \tIP\times\tPendingRREQ\times[\tDATA]$
\index{request-required flag}%
denotes the queue $q$ of packets destined for $\dip$; the \penFlag $p$ is 
{\nonpen} if a new route discovery process for $\dip$ still needs to be initiated, i.e., a route request message needs to be sent.
The value {\pen} indicates that such a RREQ message has
been sent already, and either the reply is still pending or a route to $\dip$ has been established.
The flag is set to {\nonpen} when a routing table entry is invalidated. 

As for routing tables, we require that there is at most one entry for
every IP address.
In our model, each node maintains a variable {\queues} of type
{\tQUEUES} to record its current store of data packets.

\renewcommand{\queues}{\dval{store}} 
\renewcommand{\data}{\dval{d}}

We define some functions for inspecting a store:
\begin{enumerate}[(a)]
\item Similar to $\fnselroute$, we need a function that is able to extract the queue for a given destination.
\[\begin{array}{r@{\hspace{0.5em}}c@{\hspace{0.5em}}l}
\fnselqueue:\tQUEUES\times\tIP &\to&  [\tDATA]\\
\selq{\queues}{\dip}&:=&\left\{
\begin{array}{ll}
  q&\mbox{if } (\dip,*,q)\in\queues\\
  {[\,]}&\mbox{otherwise}
\end{array}\right.
\end{array}\]
\item We define a function $\fnqD$ to extract the destinations for which there are unsent packets:
\index{qD@\keyw{qD}}%
\[\begin{array}{r@{\hspace{0.5em}}c@{\hspace{0.5em}}l}
\fnqD:\tQUEUES&\to&\pow(\tIP)\\
	\qD{\queues}&:=&
	\{\dip\mid(\dip,*,*)\in\queues\}\ .
\end{array}\]
\end{enumerate}
Next, we define operations for adding and removing data packets from a store.
\begin{enumerate}[(a)]
\setcounter{enumi}{2}
\item \hypertarget{add}{
Adding a data packet for a particular destination to a store is defined by:
\index{add@\keyw{add}}%
\[\begin{array}{r@{\hspace{0.5em}}c@{\hspace{0.5em}}l}
\multicolumn{3}{l}{\fnadd:\tDATA\times\tIP\times\tQUEUES \to\tQUEUES}\\
\add{\data}{\dip}{\queues}&:=&\left\{
\begin{array}{ll}
  \queues\cup\{(\dip\comma\nonpen\comma\append{\data}{[\,]})\}&\mbox{if } (\dval{dip},*,*)\notin\queues\\
  \queues-\{(\dip,p,q)\}\\
  \phantom{\queues}\cup\{(\dip\comma p\comma\append{\data}{q})\}&\mbox{if } (\dval{dip},p,q)\in\queues\ .
\end{array}\right.
\end{array}\]}
Informally, the process selects 
the entry $(\dip\comma p\comma q)\in\queues\in\tQUEUES$,
where $\dip$ is the destination of the application layer data
$\data$,
and appends {\data} to the queue $q$ of {\dip} in that triple;
the \penFlag $p$ remains unchanged.
In case there is no entry for {\dip} in \queues,
the process creates a new queue $[\data]$ of stored packets that only contains the
data packet under consideration and inserts it---together with $\dip$---into the store;
the \penFlag is set to \nonpen, since a route request needs to be sent.

\item \hypertarget{drop}{
To delete the oldest packet for a particular destination from a store , we define:
\hypertarget{drop}{
\index{drop@\keyw{drop}}%
\[\begin{array}{r@{\hspace{0.5em}}c@{\hspace{0.5em}}l}
\fndrop:\tIP\times\tQUEUES &\rightharpoonup&\tQUEUES\\
\drop{\dip}{\queues}&:=&\left\{
\begin{array}{ll}
  \queues-\{(\dip,*,q)\}&\mbox{if } \tail{q}=[\,]\\
  \queues-\{(\dip,p,q)\}\\
  \phantom{\queues}\;\!\cup\{(\dip,p,\tail{q})\}&\mbox{otherwise\ ,}
\end{array}\right.
\end{array}\]
}
\noindent where $q=\selq{\queues}{\dip}$ is the selected queue for destination $\dip$.
If $\dip\not\in\qD{\queues}$ then $q={[\,]}$. Therefore $\tail{q}$ and hence also $\drop{\dip}{\queues}$ is undefined.
Note that if $\data$ is the last queued packet for a specific
destination, the whole entry for the destination is removed from $\queues$.}
\end{enumerate}
\noindent In our model of AODV we use only {\fnadd} and {\fndrop} to update a store. 
This ensures that the store will never contain a triple $(\dip,*,[\,])$ with an empty data queue, i.e.,
\begin{equation}\label{non-empty}
\dip\in \qD{\queues} \Rightarrow \selq{\queues}{\dip}\not={[\,]}\ .
\end{equation}
Finally, we define operations for reading and manipulating the \penFlag of a queue.
\begin{enumerate}[(a)]
\setcounter{enumi}{4}
\item We define a partial function $\fnfD$ to extract the flag for a destination for which there are unsent packets:
\hypertarget{qflag}{
\[\begin{array}{r@{\hspace{0.5em}}c@{\hspace{0.5em}}l}
\fnfD:\tQUEUES\times\tIP&\rightharpoonup&\tPendingRREQ\\
\fD{\queues}{\dip}&:=&\left\{
\begin{array}{ll}
  p&\mbox{if } (\dip,p,*)\in\queues\\
  \mbox{undefined}&\mbox{otherwise}\ .
\end{array}\right.
\end{array}\]}

\item We define functions $\fnsetrrf$ and $\fnunsetrrf$ to change the \penFlag.
After a route request for destination $\dip$ has been initiated, the \penFlag for $\dip$ has to be set to $\pen$.
\hypertarget{setrrf}{
\index{unsetRRF@\keyw{unsetRRF}}%
\[\begin{array}{@{}r@{\hspace{0.5em}}c@{\hspace{0.5em}}l@{}}
\fnunsetrrf:\tQUEUES\times\tIP&\to&\tQUEUES\\
\unsetrrf{\queues}{\dip}&:=&\left\{
\begin{array}{ll}
  \queues-\{(\dip,*,q)\}\cup\{(\dip,\pen,q)\}&\mbox{if } \{(\dip,*,q)\}\in\queues\\
  \queues&\mbox{otherwise\ .}
\end{array}\right.
\end{array}\]}
In case that 
there is no queued data for destination {\dip}, the {\queues} remains unchanged.

Whenever a route is invalidated the corresponding \penFlag has to be set to \nonpen;
this indicates that the protocol might need to initiate a new route
discovery process.
Since the function \hyperlink{invalidate}{$\fninv$} invalidates sets of routing table entries, we 
define a function with a set of destinations
$\dests\in\pow(\tIP\times\tSQN)$ as one of its arguments (annotated with sequence numbers, which are not used here).
\hypertarget{setrrf}{
\index{setRRF@\keyw{setRRF}}%
\[\begin{array}{r@{\hspace{0.5em}}c@{\hspace{0.5em}}l}
\fnsetrrf:\tQUEUES\times(\tIP\rightharpoonup\tSQN)&\to&\tQUEUES\\
\setrrf{\queues}{\dests}&:=&
  \{(\dip,p,q) \mid (\dip,p,q) \in \queues \wedge (\dip,*) \notin\dests\}\\
  &\cup&\{(\dip,\nonpen,q) \mid (\dip,p,q) \in \queues \wedge (\dip,*) \in\dests\}\;.
\end{array}\]}

\end{enumerate}

\subsection{Messages and Message Queues}\label{ssec:messages}
\index{messages}%
Messages are the main ingredient of any routing protocol.
The message types used in the AODV protocol are route request,
route reply, and route error. To generate theses
messages, we use functions
\vspace{-0.6ex}
\[\begin{array}{l}
\rreqID:\NN \times \tRREQID \times \tIP \times \tSQN \times \tSQNK\times \tIP \times \tSQN \times \tIP \rightarrow \tMSG\\
\rrepID:\NN \times \tIP \times \tSQN \times \tIP \times \tIP \rightarrow \tMSG\\
\rerrID:(\tIP\rightharpoonup\tSQN) \times \tIP \rightarrow \tMSG\ .\footnotemark
\end{array}\vspace{-0.6ex}\]
\index{RREQ message}%
\index{rreq@\keyw{rreq}}%
\index{RREP message}%
\index{rrep@\keyw{rrep}}%
\index{RERR message}%
\index{rerr@\keyw{rerr}}%
\footnotetext{The ordering of the arguments follows the RFC.}%
The function $\rreq{\hops}{\rreqid}{\dip}{\dsn}{\dval{dsk}}{\oip}{\osn}{\sip}$
generates a route request. Here, $\hops$ indicates the hop count from
the originator $\oip$---that, at the time of sending, had the sequence
number $\osn$---to the sender of the message $\sip$; $\rreqid$
uniquely identifies the route request;  $\dsn$ is the least level
of freshness of a route to dip that is acceptable to \oip---it has been
obtained by incrementing the latest sequence number received in the
past by {\oip} for a route towards $\dip$; and \dval{dsk}
  indicates whether we can trust that number. In case no 
sequence number is known, $\dsn$ is set to $0$ and $\dval{dsk}$ to $\unkno$.
By $\rrep{\hops}{\dip}{\dsn}{\oip}{\sip}$ a route reply message is obtained.
Originally, it was generated by
 $\dip$---where $\dsn$ denotes the
sequence number of $\dip$ at the time of sending---and  is destined
for $\oip$; the last sender of the message was the node with IP
address $\sip$ and the distance between $\dip$ and $\sip$ is given by $\hops$.
The error message is generated by $\rerr{\dests}{\sip}$, where
$\dests:\tIP\rightharpoonup\tSQN$ is the list of unreachable
destinations and $\sip$ denotes the sender.
Every unreachable destination $\rip$ comes together with the incremented last-known sequence number $\rsn$.

Next to these AODV control messages, we use for our specification
also data packets: messages that carry application layer data.
\index{newpkt@$\newpktID$}%
\index{pkt@\keyw{pkt}}%
\[\begin{array}{l}
\newpktID:\tDATA \times \tIP \rightarrow \tMSG\\
\pktID:\tDATA \times \tIP \times \tIP \rightarrow \tMSG
\end{array}\]
Although these messages are not part of the protocol itself, they are
necessary to initiate error messages, and to trigger the route discovery process.
$\newpkt{\data}{\dip}$ generates a message containing new application layer data
$\data$ destined for a particular destination $\dip$. Such a
message is submitted to a node by a client of the AODV protocol
hooked up to that node. The function $\pkt{\data}{\dip}{\sip}$
generates a message containing application layer data {\data},
that is sent by the sender $\sip$
to the next hop on the route towards $\dip$.

\index{message queueing}%
All messages received by a particular node are first stored in a queue
(see Section~\ref{ssec:message_queue} for a detailed description). To
model this behaviour we use a message queue, denoted by the variable
$\msgs$ of type $[\tMSG]$.  As for every other queue, we will freely
use the functions $\fnhead$, $\fntail$ and $\fnappend$.
}

\subsection{Summary}
The following table describes the entire data structure we use.
\vspace{-\abovedisplayskip}
\begin{center}
\setlength{\tabcolsep}{4.0pt}
{\small
\begin{longtable}{@{}|l|l|l|@{}}
\hline
\textbf{Basic Type} & \textbf{Variables} & \textbf{Description}\\
\hline
 \tIP			&\ip, \dip, \oip, \rip, \sip, \nhip	&node identifiers\\
 \tSQN		&\dsn, \osn, \rsn, \sn			&sequence numbers\\
 \tSQNK		&\keyw{dsk}        			&sequence-number-status flag\\
 \tFLAG		&\flag					&route validity\\
 \NN			&\hops					&hop counts\\
 \tROUTE  	& \keyw{r} 						&routing table entries\\
 \tRT			&\rt						&routing tables\\		
 \tRREQID		&\rreqid					&request identifiers\\
 \tPendingRREQ&		                                &\penFlag\\
 \tDATA 		&\data               				&application layer data\\
 \tQUEUES       	&\queues					&store of queued data packets\\
 \tMSG		&\msg					&messages\\
\hline
\newpage
\multicolumn{3}{l}{\mbox{}}\\[-2ex]
\hline
\textbf{Complex Type} & \textbf{Variables} & \textbf{Description}\\
\hline
${[\tTYPE]}$						&			&queues with elements of type \tTYPE\\
\quad[\tMSG]						&\msgs		&message queues\\
$\pow(\tTYPE)$			        		&			&sets consisting of elements of type \tTYPE\\
\quad$\pow(\tIP)$					&$\pre$     	&sets of identifiers (precursors, destinations, \dots)\\
\quad$\pow(\tIP\times\tRREQID)$	        	&\rreqs		&sets of request identifiers  with originator IP\\
$\tTYPE_1\rightharpoonup\tTYPE_2$       &              	 	&partial functions from $\tTYPE_1$ to $\tTYPE_2$\\
\quad$\tIP\rightharpoonup\tSQN$         	&\dests         	&sets of destinations with sequence numbers\\
\hline
\hline
\multicolumn{2}{|l|}{\textbf{Constant/Predicate}}& \textbf{Description}\\
\hline
\multicolumn{2}{|l|}{$0:\tSQN,~1:\tSQN$}&
	unknown, smallest sequence number\\
\multicolumn{2}{|l|}{$\mathord{<} \subseteq \tSQN\times\tSQN$}&
	strict order on sequence numbers\\
\multicolumn{2}{|l|}{$\kno,\unkno:\tSQNK$}&
	constants to distinguish  known and unknown sqns\\
\multicolumn{2}{|l|}{$\val,\inval:\tFLAG$}&
	constants to distinguish  valid and invalid routes\\
\multicolumn{2}{|l|}{$\pen,\nonpen:\tPendingRREQ$}&
	constants indicating whether a RREQ is required\\
\multicolumn{2}{|l|}{$0:\NN,~1:\NN,~\mathord{<} \subseteq \NN\times\NN$}&
	standard constants/predicates of natural numbers\\
\multicolumn{2}{|l|}{${[\,]}:{[\tTYPE]},~\emptyset:\pow(\tTYPE)$}&
	empty queue, empty set\\
\multicolumn{2}{|l|}{$\mathord{\in}\subseteq\tTYPE\times\pow(\tTYPE)$}&
	membership, standard set theory\\
\hline
\hline
\multicolumn{2}{|l|}{\textbf{Function}} & \textbf{Description}\\
\hline
\multicolumn{2}{|l|}{$\fnhead:[\tTYPE]\rightharpoonup\tTYPE$}&
	returns the ``oldest'' element in the queue\\
\multicolumn{2}{|l|}{$\fntail:[\tTYPE]\rightharpoonup[\tTYPE]$}&
	removes the ``oldest'' element in the queue\\
\multicolumn{2}{|l|}{$\fnappend:\tTYPE\times[\tTYPE]\rightarrow[\tTYPE]$}&
	inserts a new element into the queue\\
\multicolumn{2}{|l|}{$\fndrop:\tIP\times\tQUEUES \rightharpoonup\tQUEUES$}&
	deletes a packet from the queued data packets\\
\multicolumn{2}{|l|}{$\fnadd:\tDATA\times\tIP\times\tQUEUES \to\tQUEUES$}&
	adds a packet to the queued data packets\\
\multicolumn{2}{|l|}{$\fnunsetrrf:\tQUEUES\times\tIP\to\tQUEUES$}&
	set the \penFlag to \pen\\
\multicolumn{2}{|l|}{$\fnsetrrf:\tQUEUES\times(\tIP\rightharpoonup\tSQN)\to\tQUEUES$}&
	set the \penFlag to \nonpen\\
\multicolumn{2}{|l|}{$\fnselqueue:\tQUEUES\times\tIP \rightarrow [\tDATA]$}&
	selects the data queue for a particular destination\\
\multicolumn{2}{|l|}{$\fnfD:\tQUEUES\times\tIP\rightharpoonup\tPendingRREQ$}&
	selects the flag for a destination from the store\\
\multicolumn{2}{|l|}{$\fnselroute:\tRT\times\tIP \rightharpoonup \tROUTE$}&
	selects the route for a particular destination\\
\multicolumn{2}{|l|}{$(\_\comma\_\comma\_\comma\_\comma\_\comma\_\comma\_\,)\!: \tIP
        \mathord\times \tSQN \mathord\times\tSQNK \mathord\times\tFLAG \mathord\times \NN
        \mathord\times \tIP \mathord\times \pow(\tIP) \mathop\rightarrow \tROUTE$}&
	generates a routing table entry\\
\multicolumn{2}{|l|}{$\fninc:\tSQN \rightarrow \tSQN$}&
	increments the sequence number\\
\multicolumn{2}{|l|}{$\max:\tSQN\times\tSQN \to\tSQN$}&
	returns the larger sequence number\\	
\multicolumn{2}{|l|}{$\fnsqn:\tRT \times \tIP \to \tSQN$}&
	returns the sequence number of a particular route\\	
\multicolumn{2}{|l|}{$\fnsqnf:\tRT \times \tIP \to \tSQNK$}&
	determines whether the sequence number is known\\
\multicolumn{2}{|l|}{$\fnstatus:\tRT\times\tIP\rightharpoonup\tFLAG$}&
	returns the validity of a particular route\\
\multicolumn{2}{|l|}{$+1:\NN \rightarrow \NN$}&
	increments the hop count\\
\multicolumn{2}{|l|}{$\fndhops:\tRT \times \tIP \rightharpoonup \NN$}&
	returns the hop count of a particular route\\
\multicolumn{2}{|l|}{$\fnnhop:\tRT \times \tIP \rightharpoonup \tIP$}&
	returns the next hop of a particular route\\
\multicolumn{2}{|l|}{$\fnprecs:\tRT \times \tIP \rightharpoonup \pow(\tIP)$}&
	returns the set of precursors of a particular route\\
\multicolumn{2}{|l|}{$\fnakD, \fnikD, \fnkD:\tRT \rightarrow\pow(\tIP)$}&
	returns the set of valid, invalid, known destinations\\
\multicolumn{2}{|l|}{$\fnqD:\tQUEUES \rightarrow \pow(\tIP)$}&
	returns the set of destinations with unsent packets\\
\multicolumn{2}{|l|}{$\cap,~\cup,~\bigcup\{\ldots\},~\dots$}&
	standard set-theoretic functions\\
\multicolumn{2}{|l|}{$\fnaddprec : \tROUTE\times \pow(\tIP) \to \tROUTE$}&
	adds a set of precursors to a routing table entry\\	
\multicolumn{2}{|l|}{$\fnaddprecrt : \tRT\times\tIP\times \pow(\tIP) \rightharpoonup \tRT$}&
	adds a set of precursors to an entry inside a table\\	
\multicolumn{2}{|l|}{$\fnupd:\tRT \times \tROUTE \rightharpoonup \tRT$}&
	updates a routing table with a route (if fresh enough)\\
\multicolumn{2}{|l|}{$\fninv:\tRT \times (\tIP\rightharpoonup\tSQN) \rightarrow \tRT$}&
	invalidates a set of routes within a routing table\\
\multicolumn{2}{|l|}{$\fnnrreqid: \pow(\tIP\times\tRREQID) \times \tIP \rightarrow \tRREQID$}&
	generates a new route request identifier\\
\multicolumn{2}{|l|}{$\newpktID:\tDATA \times \tIP \rightarrow \tMSG$}&
	generates a message with new application layer data\\
\multicolumn{2}{|l|}{$\pktID:\tDATA \times \tIP \times \tIP \rightarrow \tMSG$}&
	generates a message containing application layer data\\
\multicolumn{2}{|l|}{$\rreqID:\NN \mathord\times \tRREQID \mathord\times \tIP \mathord\times \tSQN\mathord\times\tSQNK \mathord\times \tIP \mathord\times \tSQN \mathord\times \tIP \rightarrow \tMSG$}&
	generates a route request\\
\multicolumn{2}{|l|}{$\rrepID:\NN \times \tIP \times \tSQN \times \tIP \times \tIP \rightarrow \tMSG$}&
	generates a route reply\\
\multicolumn{2}{|l|}{$\rerrID:(\tIP\rightharpoonup\tSQN) \times \tIP \rightarrow \tMSG$}&
	generates a route error message\\
\hline
\multicolumn{3}{l}{\mbox{}}\\[-1ex]
\caption[Data structure of AODV]{\em Data structure of AODV}
\vspace{-11ex}
\end{longtable}
}
\end{center}


\newpage
\section{Modelling AODV}\label{sec:modelling_AODV}
In this section, we present a specification of the AODV protocol
using process algebra.
The model includes a mechanism to describe the delivery of data
packets; though this is not part of the protocol itself it is
necessary to trigger any AODV activity.  Our model consists of $7$
processes, named $\AODV$, $\NEWPKT$, $\PKT$, $\RREQ$, $\RREP$, $\RERR$ and
$\QMSG$:

\begin{itemize}
	\item The {basic process} $\AODV$ reads a message from the message queue and,
		depending on the type of the message, calls other
		processes. When there is no message handling going on,
		the process initiates the transmission of queued data packets or generates
		a new route request (if packets are stored for a destination, no route 
		to this destination is known and no route request for this destination is pending).
        \item The processes $\NEWPKT$ and $ \PKT$ describe all actions
                performed by a node when a data packet is
                received.  The former process handles a newly injected
		packet. The latter describes all actions performed
		when a node receives data from another node via the protocol. This
		includes accepting the packet (if the
                node is the destination), forwarding the packet (if
                the node is not the destination) and sending an error
                message (if forwarding fails).
	\item The process $\RREQ$ models all events that might occur
                after a route request has been received. This includes
                updating the node's routing table, forwarding the
                route request as well as the initiation of a route
                reply if a route to the destination is known.
	\item Similarly, the $\RREP$ process describes the reaction
		of the protocol to an incoming route reply.
	\item The process $\RERR$ models the part of AODV which handles error messages.
                In particular, it describes the modification and
		forwarding of the AODV error message.
	\item The last process $\QMSG$ concerns message handling. Whenever a message is received,
          \index{message queueing}%
	  it is first stored in a message queue. If the corresponding node is able to handle a message
	  it pops the oldest message from the queue and handles it. An example where a node is not ready
	  to process an incoming message immediately is when it is already handling a message.
\end{itemize}

In the remainder of the section, we provide a formal specification for
each of these processes and explain them step by step. Our
specification can be split into three parts: the brown lines describe
updates to be performed on the node's data, e.g., its routing table;
the black lines are other process algebra constructs
(cf.\ \Sect{process_algebra}); and the blue lines are
ordinary comments.

\subsection{The Basic Routine}\label{ssec:proc_aodv}

The {basic process} $\AODV$ either reads a message from the
corresponding queue, sends a queued data packet if a route to the destination has been established, 
or initiates a new route discovery process in case of queued
data packets with invalid or unknown routes.
This process maintains five data variables, {\ip}, {\sn}, {\rt}, {\rreqs} and
{\queues}, in which it stores its own identity, its own sequence number, its current routing
table, the list of route requests seen so far, and its current store
of queued data packets that await transmission (cf.\ \Sect{types}).

The message handling is described in Lines~\ref{aodv:line2}--\ref{aodv:line20}.
First, the message has to be read from the queue of stored messages
(\receive{\msg}). After that, the process $\AODV$ checks the type of
the message and calls a process that can handle the message: in case
\index{injection}%
\index{data packet}%
\index{control message}%
of a newly injected data packet, the process $\NEWPKT$ is called; 
in case of  an incoming data packet, the process $\PKT$ is
called; in case that the incoming message is an AODV control message
(route request, route reply or route error), the node updates its routing table.
More precisely, if there is no entry to the message's sender $\sip$, the
receiver-node creates an entry with the unknown sequence number $0$
and hop count $1$; in case  there is already a 
routing table entry $(\sip,\dsn,*,*,*,*,\pre)$, then this entry is updated to
$(\sip,\dsn,\unkno,\val,1,\sip,\pre)$
(cf. Lines~\ref{aodv:line10}, \ref{aodv:line14} and~\ref{aodv:line18}). 
Afterwards, the processes $\RREQ$, $\RREP$ and $\RERR$ are called, respectively.

  \algsetup{linenodelimiter=.,linenosize=\tiny}
  \begin{algorithm}[H]
    {\footnotesize
      \caption{The basic routine}
      \label{pro:aodv}
      \begin{algorithmic}[1]
\DEFPROCESS{\AODV}{\ip\comma\sn\comma\rt\comma\rreqs\comma\queues}
	\IFempty
		\receiveL{\msg}\ .																															\label{aodv:line2}
		\COMLINE{depending on the message, the node calls different processes}	
		\PAR						\label{aodv:line3}
		\IF[new DATA packet]{$\msg = \newpkt{\data}{\dip}$}																	\label{aodv:line4}
			\newpktP{\data}{\dip}{\ip}{\sn}{\rt}{\rreqs}{\queues}																	\label{aodv:line5}
		\ELSIF[incoming DATA packet]{$\msg = \pkt{\data}{\dip}{\oip}$}  \label{aodv:line6}
			\pktP{\data}{\dip}{\oip}{\ip}{\sn}{\rt}{\rreqs}{\queues}																	\label{aodv:line7}
		\ELSIF[RREQ]{$\msg = \rreq{\hops}{\rreqid}{\dip}{\dsn}{\dsk}{\oip}{\osn}{\sip}$}								\label{aodv:line8}
			\COMLINE{update the route to \sip\ in \rt}																					\label{aodv:line9}
			\UPD{\rt:=\upd{\rt}{(\sip,0,\unkno,\val,1,\sip,\emptyset)}}																\label{aodv:line10}
			\COMMENT{$0$ is used since no sequence number is known}%
			\rreqP{\hops}{\rreqid}{\dip}{\dsn}{\dsk}{\oip}{\osn}{\sip}{\ip}{\sn}{\rt}{\rreqs}{\queues}					\label{aodv:line11}
		\ELSIF[RREP]{$\msg = \rrep{\hops}{\dip}{\dsn}{\oip}{\sip}$}															\label{aodv:line12}
			\COMLINE{update the route to \sip\ in \rt}																					\label{aodv:line13}
			\UPD{\rt:=\upd{\rt}{(\sip,0,\unkno,\val,1,\sip,\emptyset)}}																\label{aodv:line14}
			\rrepP{\hops}{\dip}{\dsn}{\oip}{\sip}{\ip}{\sn}{\rt}{\rreqs}{\queues}												\label{aodv:line15}
		\ELSIF[RERR]{$\msg = \rerr{\dests}{\sip}$}																					\label{aodv:line16}
			\COMLINE{update the route to \sip\ in \rt}																					\label{aodv:line17}
			\UPD{\rt:=\upd{\rt}{(\sip,0,\unkno,\val,1,\sip,\emptyset)}}																\label{aodv:line18}
			\rerrP{\dests}{\sip}{\ip}{\sn}{\rt}{\rreqs}{\queues}																			\label{aodv:line19}
		\ENDIFii
		\ENDPAR																																		\label{aodv:line20}
		\ELSIF[send a queued data packet if a valid route is known]{$\mbox{Let } \dip\in\qD{\queues}\cap\akD{\rt}$}				\label{aodv:line22}
			\UPD{\data:=\head{\selq{\queues}{\dip}}}													\label{aodv:line23}
			\STARTPRIO
			 	\unicast{\nhop{\rt}{\dip}}{\pkt{\data}{\dip}{\ip}}\ . 											\label{aodv:line24}
				\UPD{\queues:=\drop{\dip}{\queues}}\COMMENT{drop {\data} from the {\queues} for {\dip} if the transmission was successful}													\label{aodv:line26}
				\aodvL{\ip}{\sn}{\rt}{\rreqs}{\queues}			\label{aodv:line27}
			 \PRIO
				\COMspec{an error is produced and the routing table is updated}							\label{aodv:line29}
				\UPD{\dests:=\{(\rip,\inc{\sqn{\rt}{\rip}})\,|\,\rip\in\akD{\rt}\ans \nhop{\rt}{\rip}=\nhop{\rt}{\dip}\}}		\label{aodv:line30}
				\UPD{\rt:=\inv{\rt}{\dests}}															\label{aodv:line32}
				\UPD{\queues:=\setrrf{\queues}{\dests}}\label{aodv:line32a}
				\UPD{\pre:=\bigcup\{\precs{\rt}{\rip}\,|\,(\rip,*)\in\dests\}}									\label{aodv:line31}
				\UPD{\dests:=\{(\rip,\rsn)\,|\,(\rip,\rsn)\in\dests\ans \precs{\rt}{\rip}\not=\emptyset\}}				\label{aodv:line31a}
				\groupcast{\pre}{\rerr{\dests}{\ip}}\ .
				\aodv{\ip}{\sn}{\rt}{\rreqs}{\queues}													\label{aodv:line33}
 		 	\ENDPRIO
		\ELSIF[a route discovery process is initiated]{$\mbox{Let } \dip\in\qD{\queues}-\akD{\rt}\ans\fD{\queues}{\dip}=\nonpen$}	\label{aodv:line34}		
			\UPD{\queues:=\unsetrrf{\queues}{\dip}}\COMMENT{set \penFlag to \pen}				\label{aodv:line35}		
			\UPD{\sn:=\inc{\sn}}\COMMENT{increment own sequence number}								\label{aodv:line36}
			\COMLINE{update \rreqs\ by adding $(\ip,\nrreqid{\rreqs}{\ip})$}								\label{aodv:line37}
			\UPD{\rreqid:=\nrreqid{\rreqs}{\ip}}							\label{aodv:line38a}
			\UPD{\rreqs := \rreqs\cup\{(\ip,\rreqid)\}}							\label{aodv:line38b}
			\broadcast{\rreq{$0$}{\rreqid}{\dip}{\sqn{\rt}{\dip}}{\sqnf{\rt}{\dip}}{\ip}{\sn}{\ip}}\ .					\label{aodv:line39}
			\aodv{\ip}{\sn}{\rt}{\rreqs}{\queues}														\label{aodv:line40}			
	\ENDIFii

	\end{algorithmic}
    }
  \end{algorithm}

\index{queued data}%
The second part of $\AODV$ (Lines~\ref{aodv:line22}--\ref{aodv:line33}) initiates the sending of a data packet.
For that, it has to be checked if there is a queued data packet for a destination that has a known and valid route
in the routing table ($\qD{\queues}\cap\akD{\rt}\not=\emptyset$). In case that there is more than one destination with
stored data and a known route, an arbitrary destination is chosen and denoted by $\dip$
(Line~\ref{aodv:line22}).\footnote{Although the word ``let'' is not part of the syntax, we add it to stress
the nondeterminism happening here.}\linebreak
Moreover $\data$ is set to the first
queued data packet from the application layer that should be sent
($\data:=\fnhead(\selq{\queues}{\dip})$).\footnote{Following the RFC,
data packets waiting for a route should be buffered ``first-in, first-out'' (FIFO).} 
This data packet is unicast to the next hop on the route to $\dip$.
If the unicast is successful, the data packet $\data$ is removed from $\queues$ (Line~\ref{aodv:line26}).
Finally, the process calls itself---stating that the node is ready for
handling a new message, initiating the sending of another packet towards a destination, etc.
In case the unicast is not successful, the data packet has not been transmitted.
Therefore $\data$ is not removed from $\queues$. Moreover, the node knows that the link
to the next hop on the route to $\dip$ is faulty and, most probably, broken.
\index{error handling}%
An error message is initiated. Generally, route error and link breakage processing requires the
following steps: (a) invalidating existing routing table entries, (b) listing affected destinations,
(c) determining which neighbours may be affected (if any), and
(d) delivering an appropriate AODV error message to such neighbours \cite{rfc3561}.
Therefore, the process determines all valid destinations $\dests$ that have
this unreachable node as next hop (Line~\ref{aodv:line30}) and marks
the routing table entries for these destinations
as invalid (Line~\ref{aodv:line32}), while incrementing their sequence numbers (Line~\ref{aodv:line30}).
In Line~\ref{aodv:line32a}, we set, for all invalidated routing table entries, the \penFlag to \nonpen,
thereby indicating that a new route discovery process may need to be initiated.
In Line~\ref{aodv:line31} the recipients of the error message
\index{precursors}%
are determined. These are the precursors of the invalidated
destinations, i.e., the neighbouring nodes listed as having a route to one of the
affected destinations passing through the broken link.
\index{error message|see{RERR message}}%
Finally, an error message is sent to them (Line~\ref{aodv:line33}),
listing only those invalidated destinations with a non-empty set of
precursors (Line~\ref{aodv:line31a}).

\index{route discovery process}%
The third and final part of $\AODV$ (Lines~\ref{aodv:line34}--\ref{aodv:line40}) initiates a route
discovery process. 
This is done when there is  at least one queued data
packet for a destination without a valid
routing table entry, that is not 
waiting for a reply in response to a route request process initiated before. 
Following the RFC, the process generates a new route request.
This is achieved in four steps:
First, the \penFlag is set to {\pen} (Line~\ref{aodv:line35}), meaning that no further
route discovery processes for this destination need to be initiated.%
\footnote{The RFC does not describe packet handling in detail; hence
the \penFlag is not part of the RFC's RREQ generation process.}
Second, the node's own sequence number is increased by~$1$
(Line~\ref{aodv:line36}). Third, by determining
$\nrreqid{\rreqs}{\ip}$, a new route request identifier is created and
stored---together with the node's $\ip$---in the set $\rreqs$ of route
requests already seen (Line~\ref{aodv:line38b}). Fourth, the message
itself is sent (Line~\ref{aodv:line39}) using broadcast. In contrast to
\textbf{unicast}, transmissions via \textbf{broadcast} are not checked
on success.  The information inside the message follows strictly the
RFC. In particular, the hop count is set to $0$, the route request identifier previously created is used, etc.
This ends the initiation of the route discovery process.

\subsection{Data Packet Handling}\label{ssec:proc_pkt}

The processes $\NEWPKT$ and $\PKT$ describe all actions performed by a
node when a data packet is injected by a client hooked up to the
local node or received via the protocol, respectively.  For
the process $\PKT$, this includes the acceptance
(if the node is the destination), the forwarding (if the node is not
the destination), as well as the sending of an error message in case
something went wrong.
The process $\NEWPKT$ does not include the initiation of a new route request; this is part of the process $\AODV$.
Although packet handling itself
is not part of AODV, it is necessary to include it in our
formalisation, since a failure to transmit a data packet triggers AODV
activity.

The process $\NEWPKT$ first checks whether the node is the intended addressee
of the data packet. If this is the
case, it delivers the data and returns to the basic routine $\AODV$.
If the node is not the intended destination ($\dip\not=\ip$, Line~\ref{newpkt:line4}),
the $\data$ is added to the data queue for {\dip} (Line~\ref{newpkt:line5}),\footnote{If no
data for destination $\dip$ was already queued, the function \hyperlink{add}{$\fnadd$} creates a
fresh queue for $\dip$, and set the request-required flag to $\nonpen$; otherwise, the
request-required flag keeps the value it had already.}
which finishes the handling of a newly injected data packet.
The further handling of queued data (forwarding it to the next hop on the way to the destination in
case a valid route to the destination is known,
and otherwise initiating a new route request if still required)
is the responsibility of the main process {\AODV}.

  \algsetup{linenodelimiter=.,linenosize=\tiny}
  \begin{algorithm}[H]
    {\footnotesize
      \caption{Routine for handling a newly injected data packet}
      \label{pro:newpkt}
      \begin{algorithmic}[1]
\DEFPROCESS{\NEWPKT}{\data\comma\dip\,\comma\,\ip\comma\sn\comma\rt\comma\rreqs\comma\queues}
	\IF[the DATA packet is intended for this node]{$\dip=\ip$}															\label{newpkt:line2}
		\deliverL{\data}\ .
		\aodv{\ip}{\sn}{\rt}{\rreqs}{\queues}																						\label{newpkt:line3}
	\ELSIF[the DATA packet is not intended for this node]{$\dip\not=\ip$}										\label{newpkt:line4}
		\UPD{\queues:=\add{\data}{\dip}{\queues}}	\ .
		\aodv{\ip}{\sn}{\rt}{\rreqs}{\queues}																						\label{newpkt:line5}	
	\ENDIFii

	\end{algorithmic}
    }
  \end{algorithm}

Similar to $\NEWPKT$, the process $\PKT$ first checks whether it is the intended addressee
of the data packet. If this is the
case, it delivers the data and returns to the basic routine $\AODV$.
If the node is not the intended destination
($\dip\not=\ip$, Line~\ref{pkt2:line4}) more activity is needed.

In case that the node has a valid route to the $\data$'s destination $\dip$
($\dip\in\akD{\rt}$), it forwards the packet using a unicast to the
next hop $\nhop{\rt}{\dip}$ on the way to $\dip$.
Similar to the unicast of the process $\AODV$, it has to be checked
whether the transmission is successful: no further action is necessary
if the transmission succeeds, and the node returns to the basic
routine $\AODV$. If the transmission fails, the link to the next hop
$\nhop{\rt}{\dip}$ is assumed to be broken. As before, all
destinations $\dests$ that are reached via that broken link are
determined (Line~\ref{pkt2:line9}) and all precursors interested in at
least one of these destinations are informed via an error message
(Line~\ref{pkt2:line14}).  Moreover, all the routing table entries
using the broken link have to be invalidated in the node's routing
table $\rt$ (Line~\ref{pkt2:line10}), and all corresponding \penFlag{}s
are set to \nonpen\ (Line~\ref{pkt2:line11}).

In case that the node has no valid route to the destination $\dip$ ($\dip\not\in\akD{\rt}$),
the data packet is lost and possibly an error message is sent. If 
there is an (invalid) route to the \data's destination {\dip} in
the routing table (Line~\ref{pkt2:line18}), the possibly affected
neighbours can be determined and the error message is sent to these
precursors (Line \ref{pkt2:line20}). If there is no information about a
route towards $\dip$ nothing happens (and the basic process {\AODV} is called again).

  \algsetup{linenodelimiter=.,linenosize=\tiny}
  \begin{algorithm}[H]
    {\footnotesize
      \caption{Routine for handling a received data packet}
      \label{pro:pkt}
      \begin{algorithmic}[1]
\DEFPROCESS{\PKT}{\data\comma\dip\comma\oip\,\comma\,\ip\comma\sn\comma\rt\comma\rreqs\comma\queues}
	\IF[the DATA packet is intended for this node]{$\dip=\ip$}																				\label{pkt2:line2}
		\deliverL{\data}\ .
		\aodv{\ip}{\sn}{\rt}{\rreqs}{\queues}																											\label{pkt2:line3}
	\ELSIF[the DATA packet is not intended for this node]{$\dip\not=\ip$}															\label{pkt2:line4}
	 	\PAR
			\IF[valid route to \dip]{$\dip\in\akD{\rt}$}																									\label{pkt2:line5}
				\COMLINE{forward packet}
				\STARTPRIO
					\unicast{\nhop{\rt}{\dip}}{\pkt{\data}{\dip}{\oip}}\ . \aodv{\ip}{\sn}{\rt}{\rreqs}{\queues}						\label{pkt2:line7}
				\PRIO
					\COMspec{If the packet transmission is unsuccessful, a RERR message is generated}		
					\UPD{\dests:=\{(\rip,\inc{\sqn{\rt}{\rip}})\,|\,\rip\in\akD{\rt}\ans \nhop{\rt}{\rip}=\nhop{\rt}{\dip}\}}			\label{pkt2:line9}
					\UPD{\rt:=\inv{\rt}{\dests}}																												\label{pkt2:line10}
					\UPD{\queues:=\setrrf{\queues}{\dests}}																						\label{pkt2:line11}
					\UPD{\pre:=\bigcup\{\precs{\rt}{\rip}\,|\,(\rip,*)\in\dests\}}																	\label{pkt2:line12}
					\UPD{\dests:=\{(\rip,\rsn)\,|\,(\rip,\rsn)\in\dests\ans \precs{\rt}{\rip}\not=\emptyset\}}							\label{pkt2:line13}
					\groupcast{\pre}{\rerr{\dests}{\ip}}\ . \aodv{\ip}{\sn}{\rt}{\rreqs}{\queues}											\label{pkt2:line14}
				\ENDPRIO
			\ELSIF[no valid route to \dip]{$\dip\not\in\akD{\rt}$}																				\label{pkt2:line15}
				\COMLINE{no local repair occurs; data is lost}																					\label{pkt2:line16}
				\PAR	
					\IF[invalid route to \dip]{$\dip\in\ikD{\rt}$}																						\label{pkt2:line18}
						\COMLINE{if the route is invalid, a RERR is sent to the precursors}
						\groupcast{\precs{\rt}{\dip}}{\rerr{\{(\dip,\sqn{\rt}{\dip})\}}{\ip}}\ .											
						\aodv{\ip}{\sn}{\rt}{\rreqs}{\queues}																							\label{pkt2:line20}
					\ELSIF[route not in \rt]{$\dip\not\in\ikD{\rt}$}																					\label{pkt2:line21}
						\aodvL{\ip}{\sn}{\rt}{\rreqs}{\queues}																							\label{pkt2:line22}
					\ENDIFii
				\ENDPAR																																			\label{pkt2:line23}									
			\ENDIFii
		\ENDPAR																																					\label{pkt2:line24}
	\ENDIFii

	\end{algorithmic}
    }
  \end{algorithm}

\subsection{Receiving Route Requests}\label{ssec:proc_rreq}
The process $\RREQ$ models all events that may occur after a route
request has been received.

The process first reads the unique identifier $(\oip,\rreqid)$ of the route request received.
If this pair is already stored in the node's data $\rreqs$, the route request has been handled
before and the message can silently be ignored (Lines~\ref{rreq:line2}--\ref{rreq:line3}).

  \algsetup{linenodelimiter=.,linenosize=\tiny}
  \begin{algorithm}[H]
    {\footnotesize
      \caption{RREQ handling}
      \label{pro:rreq}
      \begin{algorithmic}[1]
\DEFPROCESS{\RREQ}{\hops\comma\rreqid\comma\dip\comma\dsn\comma\dsk\comma\oip\comma\osn\comma\sip\,\comma\,\ip\comma\sn\comma\rt\comma\rreqs\comma\queues}
	\IF[the RREQ has been received previously]{$(\oip\,\comma\,\rreqid)\in\rreqs$}																							\label{rreq:line2}
		\aodvL{\ip}{\sn}{\rt}{\rreqs}{\queues} \COM{silently ignore RREQ, i.e. do nothing}																			\label{rreq:line3}
	\ELSIF[the RREQ is new to this node]{$(\oip\,\comma\,\rreqid)\not\in\rreqs$}																								\label{rreq:line4}
		\UPD{\rt:=\upd{\rt}{(\oip,\osn,\kno,\val,\hops+1,\sip,\emptyset)}}	\COMMENT{update the route to \oip\ in \rt}									\label{rreq:line6}
		\UPD{\rreqs:=\rreqs\cup\{(\oip,\rreqid)\}}		\COMMENT{update \rreqs\ by adding $(\oip\,\comma\,\rreqid)$}											\label{rreq:line8}
		\PAR																																																\label{rreq:line9}
		\IF[this node is the destination node]{$\dip=\ip$}																															\label{rreq:line10}
			\UPD{\sn:=\max(\sn,\dsn)}	\COMMENT{update the sqn of \ip}																									\label{rreq:line12}
			\COMLINE{unicast a RREP towards \oip\ of the RREQ}																											\label{rreq:line13}
			\STARTPRIO
					\unicast{\nhop{\rt}{\oip}}{{\rrep{$0$}{\dip}{\sn}{\oip}{\ip}}}\ . 																								\label{rreq:line14a}								
					\aodv{\ip}{\sn}{\rt}{\rreqs}{\queues}																																	\label{rreq:line14}
				\PRIO
					\COMspec{If the transmission is unsuccessful, a RERR message is generated}																\label{rreq:line15}
					\UPD{\dests:=\{(\rip,\inc{\sqn{\rt}{\rip}})\,|\,\rip\in\akD{\rt}\ans \nhop{\rt}{\rip}=\nhop{\rt}{\oip}\}}												\label{rreq:line16}
					\UPD{\rt:=\inv{\rt}{\dests}}																																					\label{rreq:line18}			
					\UPD{\queues:=\setrrf{\queues}{\dests}}																															\label{rreq:line18a}
					\UPD{\pre:=\bigcup\{\precs{\rt}{\rip}\,|\,(\rip,*)\in\dests\}}																										\label{rreq:line17}
					\UPD{\dests:=\{(\rip,\rsn)\,|\,(\rip,\rsn)\in\dests\ans \precs{\rt}{\rip}\not=\emptyset\}}																\label{rreq:line17a}
					\groupcast{\pre}{\rerr{\dests}{\ip}}\ .																																	
					\aodv{\ip}{\sn}{\rt}{\rreqs}{\queues}																																	\label{rreq:line19}
				\ENDPRIO	
		\ELSIF[this node is not the destination node]{$\dip\not=\ip$}																											\label{rreq:line20}
			\PAR																																															\label{rreq:line21}
			\IF[$\!$valid route to \dip\ that is fresh enough$\!$]{$\!\dip\mathbin\in\akD{\rt} \wedge \dsn \mathbin\leq  \sqn{\rt}{\!\dip} \wedge\sqnf{\rt}{\!\dip}\mathbin=\kno\!$}		\label{rreq:line22}
					\COMLINE{update \rt\ by adding precursors}																														\label{rreq:line23}
					\UPD{\rt := \addprecrt{\rt}{\dip}{\{\sip\}}}																																\label{rreq:line24}
					\UPD{\rt := \addprecrt{\rt}{\oip}{\{\nhop{\rt}{\dip}\}}}																												\label{rreq:line25}
				\COMLINE{unicast a RREP towards the \oip\ of the RREQ}
				\STARTPRIO
					\unicast{\nhop{\rt}{\oip}}{\rrep{\dhops{\rt}{\dip}}{\dip}{\sqn{\rt}{\dip}}{\oip}{\ip}}\ .\\																	\label{rreq:line26}
					\aodvL{\ip}{\sn}{\rt}{\rreqs}{\queues}																																	\label{rreq:line26a}
				\PRIO
					\COMspec{If the transmission is unsuccessful, a RERR message is generated}
					\UPD{\dests:=\{(\rip,\inc{\sqn{\rt}{\rip}})\,|\,\rip\in\akD{\rt}\ans \nhop{\rt}{\rip}=\nhop{\rt}{\oip}\}}												\label{rreq:line28}
					\UPD{\rt:=\inv{\rt}{\dests}}																																					\label{rreq:line30}			
					\UPD{\queues:=\setrrf{\queues}{\dests}}																															\label{rreq:line30a}		
					\UPD{\pre:=\bigcup\{\precs{\rt}{\rip}\,|\,(\rip,*)\in\dests\}}																										\label{rreq:line29}
					\UPD{\dests:=\{(\rip,\rsn)\,|\,(\rip,\rsn)\in\dests\ans \precs{\rt}{\rip}\not=\emptyset\}}																\label{rreq:line29a}
					\groupcast{\pre}{\rerr{\dests}{\ip}}\ . 																																	\label{rreq:line31}
					\aodv{\ip}{\sn}{\rt}{\rreqs}{\queues}
				\ENDPRIO
			\ELSIF[$\!$no valid route that is fresh enough$\!$]{$\dip\mathbin{\not\in}\akD{\rt} \vee \sqn{\rt}{\!\dip} <  \dsn \vee\sqnf{\rt}{\!\dip}\mathbin=\unkno$}					\label{rreq:line32}
				\COMLINE{no further update of \rt}
				\broadcast{\rreq{$\hops+1$}{\rreqid}{\dip}{\max(\sqn{\rt}{\dip}\comma\dsn)}{\dsk}{\oip}{\osn}{\ip}}\ .										\label{rreq:line34}
				\aodvL{\ip}{\sn}{\rt}{\rreqs}{\queues}				\label{rreq:line35}
			\ENDIFii
			\ENDPAR																																													\label{rreq:line36}
		\ENDIFii
		\ENDPAR																																														\label{rreq:line37}
	\ENDIFii

	\end{algorithmic}
    }
  \end{algorithm}

\vspace{-2mm}
If the received message is new to this node
($(\oip,\rreqid)\not\in\rreqs$, Line~\ref{rreq:line4}), the node
establishes a route of length $\hops\mathord+1$ back to the originator $\oip$
of the message. If this route is ``better'' than the route to $\oip$
in the current routing table, the routing table is updated by this
route (Line~\ref{rreq:line6}).  Moreover the unique identifier has to
be added to the set $\rreqs$ of already seen (and handled) route
requests (Line~\ref{rreq:line8}).

After these updates the process checks if the node is the intended
destination ($\dip=\ip$, Line~\ref{rreq:line10}). In that case, a
route reply must be initiated: first, the node's sequence number
is---according to the RFC---set to the maximum of the current sequence
number and the destination sequence number in the RREQ packet
(Line~\ref{rreq:line12}).%
\interfootnotelinepenalty=10000%
\footnote{\label{Chakares-increment-when-issuing-RREP}%
According to I.~Chakeres on the IETF MANET mailing list
(\url{http://www.ietf.org/mail-archive/web/manet/current/msg02589.html})
Line~\ref{rreq:line12} ought to be replaced by $\assignment{\sn := \max(\sn,\inc{\dsn})}$.}
Then the reply is unicast to the next hop on
the route back to the originator {\oip} of the route request. The
content of the new route reply is as follows: the hop count is set to
$0$, the destination and originator are copied from the route request
received and the destination's sequence number is the node's own
sequence number \sn; of course the sender's IP of this message has to be set
to the node's $\ip$. As before (cf.\ Sections~\ref{ssec:proc_aodv} and
\ref{ssec:proc_pkt}), the process invalidates the corresponding routing table entries, sets \penFlag{}s and sends an error message to all
relevant precursors if the unicast transmission fails (Lines~\ref{rreq:line16}--\ref{rreq:line19}).

If the node is not the destination $\dip$ of the message but an intermediate hop
along the path from the originator to the destination, it is allowed to generate
a route reply only if the information in its own routing table is fresh enough. This means
that
(a) the node has a valid route to the destination,
(b) the destination sequence number in the node's existing routing table entry
for the destination ($\sqn{\rt}{\dip}$) is greater than or equal to
the requested destination sequence number $\dsn$ of the message and
(c) the sequence number $\sqn{\rt}{\dip}$ is known, i.e., $\sqnf{\rt}{\dip}=\kno$.
If these three conditions are satisfied---the check is done in Line~\ref{rreq:line22}---the
node generates a new route reply and sends it to the next hop on the
way back to the originator {\oip} of the received route
request.\footnote{This next hop will often, 
  but not always, be $\sip$; see \Fig{example2} in \Sect{aodv}.}. To this end, it copies the sequence number for the
destination $\dip$ from the routing table $\rt$ into the destination sequence number field of the RREP message and
it places its distance in hops from the destination ($\dhops{\rt}{\dip}$) in the corresponding field of the new reply
(Line~\ref{rreq:line26}). The unicast might fail, which
causes the usual error handling (Lines~\ref{rreq:line28}--\ref{rreq:line31}).
Just before transmitting the unicast, the intermediate node updates the forward route entry to
$\dip$ by placing the last hop node ($\sip$)\footnote{This is a mistake in the RFC; it should have
  been $\nhop{\rt}{\oip}$.} into the precursor list for the forward route entry (Line~\ref{rreq:line24}).
Likewise, it updates the reverse route entry to {\oip} by placing the first hop $\nhop{\rt}{\dip}$
towards $\dip$ in the precursor list for that entry
(Line~\ref{rreq:line25}).\footnote{Unless the \emph{gratuitous RREP flag}
is set, which we do not model in this paper, this update is rather useless,
as the precursor $\nhop{\rt}{\dip}$ in general is not aware that it
has a route to $\oip$.}

If the node is not the destination and there is either no route to the
destination $\dip$ inside the routing table or the route is not fresh enough,
the route request received has to be forwarded. This happens in Line~\ref{rreq:line34}.
The information inside the forwarded request is mostly copied from the request received.
Only the hop count is increased by $1$ and the destination sequence number is set
to the maximum of the destination sequence number in the RREQ packet
and the current sequence number for $\dip$ in the routing table.
In case $\dip$ is an unknown destination, $\sqn{\rt}{\dip}$ returns the
unknown sequence number $0$.

\subsection{Receiving Route Replies}\label{ssec:rrep}
The process $\RREP$ describes the reaction of the protocol to an incoming route reply.
Our model first checks if a forward routing table entry is going to be
created or updated (Line~\ref{rrep:line3}). This is the case if (a) the
node has no known route to the destination, or 
(b) the destination sequence number in the node's existing routing table entry
for the destination ($\sqn{\rt}{\dip}$) is smaller than the destination sequence number $\dsn$ in the RREP message, or
(c) the two destination sequence numbers are equal and, in addition,
either the incremented hop count of the RREP received is strictly smaller than the
one in the routing table, or the entry for $\dip$ in the routing table
is invalid. 
Hence Line~\ref{rrep:line3} could be replaced~by
\begin{algorithmic}[3]%
  {\footnotesize
	\STATE\hspace{-1em}\algorithmicif$\dip\not\in\kD{\rt} \vee \sqn{\rt}{\dip}\mathbin<\dsn \vee (\sqn{\rt}{\dip}\mathbin=\dsn\wedge(\dhops{\rt}{\dip}\mathbin>\hops\mathord+1 \vee\status{\rt}{\dip}\mathbin=\inval))$
\algorithmicthen\ .\,\footnote{
In case $\dip\not\in\kD{\rt}$, the terms
  $\dhops{\rt}{\dip}$ and $\status{\rt}{\dip}$ are not defined.
  In such a case, according to the convention of Footnote~\ref{fn:undefvalues} in
  \Sect{process_algebra}, the atomic formulas
  $\dhops{\rt}{\dip}\mathbin>\hops\mathord+1$ and $\status{\rt}{\dip}\mathbin=\inval$
  evaluate to {\tt false}. However, in case one would use lazy evaluation of the outermost disjunction,
  the evaluation of the expression would be  independent of the choice
  of a convention for interpreting undefined terms appearing in formulas.\label{fn:lazy}
  }}
 \end{algorithmic}\pagebreak[3]

  \algsetup{linenodelimiter=.,linenosize=\tiny}
  \begin{algorithm}[H]
    {\footnotesize
      \caption{RREP handling}
      \label{pro:rrep}
      \begin{algorithmic}[1]
\DEFPROCESS{\RREP}{\hops\comma\dip\comma\dsn\comma\oip\comma\sip\,\comma\,\ip\comma\sn\comma\rt\comma\rreqs\comma\queues}
	\IF[the routing table has to be updated]{$\rt\not=\upd{\rt}{(\dip\comma\dsn\comma\kno\comma\val\comma\hops+1\comma\sip\comma\emptyset)}$}						\label{rrep:line3}
		\UPD{\rt:=\upd{\rt}{(\dip\comma\dsn\comma\kno\comma\val\comma\hops+1\comma\sip\comma\emptyset)}}			\label{rrep:line5}
		\PAR\label{rrep:line5a}
		\IF[this node is the originator of the corresponding RREQ]{$\oip = \ip$}																	\label{rrep:line6}
			\COMLINE{a packet may now be sent; this is done in the process \AODV}
			\aodvL{\ip}{\sn}{\rt}{\rreqs}{\queues}																													\label{rrep:line8}
		\ELSIF[this node is not the originator; forward RREP]{$\oip \not= \ip$}																	\label{rrep:line9}
			\PAR
				\IF[valid route to \oip]{$\oip\in\akD{\rt}$}																											\label{rrep:line11}
					\COMLINE{add next hop towards $\oip$ as precursor and forward the route reply}									\label{rrep:line12}									
					\UPD{\rt := \addprecrt{\rt}{\dip}{\{\nhop{\rt}{\oip}\}}}																						\label{rrep:line12a}									
					\UPD{\rt := \addprecrt{\rt}{\nhop{\rt}{\dip}}{\{\nhop{\rt}{\oip}\}}}																		\label{rrep:line12b}
					\STARTPRIO
						\unicast{\nhop{\rt}{\oip}}{\rrep{$\hops+1$}{\dip}{\dsn}{\oip}{\ip}}\ .															\label{rrep:line13}
						\aodvL{\ip}{\sn}{\rt}{\rreqs}{\queues}																										\label{rrep:line14}
					\PRIO
						\COMspec{If the transmission is unsuccessful, a RERR message is generated}
						\UPD{\dests:=\{(\rip,\inc{\sqn{\rt}{\rip}})\,|\,\rip\in\akD{\rt}\ans \nhop{\rt}{\rip}=\nhop{\rt}{\oip}\}}					\label{rrep:line16}
						\UPD{\rt:=\inv{\rt}{\dests}}																														\label{rrep:line18}							
												\UPD{\queues:=\setrrf{\queues}{\dests}}																		\label{rrep:line16a}
						\UPD{\pre:=\bigcup\{\precs{\rt}{\rip}\,|\,(\rip,*)\in\dests\}}																			\label{rrep:line17}
						\UPD{\dests:=\{(\rip,\rsn)\,|\,(\rip,\rsn)\in\dests\ans \precs{\rt}{\rip}\not=\emptyset\}}									\label{rrep:line17a}
						\groupcast{\pre}{\rerr{\dests}{\ip}}\ .\ \aodv{\ip}{\sn}{\rt}{\rreqs}{\queues} 												\label{rrep:line20}
					\ENDPRIO
				\ELSIF[no valid route to \oip]{$\oip\not\in\akD{\rt}$}																						\label{rrep:line21}
					\aodvL{\ip}{\sn}{\rt}{\rreqs}{\queues}																											\label{rrep:line21a}
				\ENDIFii																																						\label{rrep:line22}
			\ENDPAR																																							\label{rrep:line23}
		\ENDIFii	
		\ENDPAR																																								\label{rrep:line23a}
	\ELSIF[the routing table is not updated]{$\rt=\upd{\rt}{(\dip\comma\dsn\comma\kno\comma\val\comma\hops+1\comma\sip\comma\emptyset)}$}							\label{rrep:line25}
		\aodvL{\ip}{\sn}{\rt}{\rreqs}{\queues}																														\label{rrep:line26}
	\ENDIFii
	\end{algorithmic}
    }
  \end{algorithm}

\vspace{-2mm}

In case that one of these conditions is true, the routing table is
updated in Line~\ref{rrep:line5}.
If the node is the intended addressee of the route reply ($\oip=\ip$) 
the protocol returns to its basic process $\AODV$.
Otherwise ($\oip\not=\ip$) the message should be forwarded.
Following the RFC~\cite{rfc3561},
``If the current node is not the node indicated by the Originator IP Address
in the RREP message AND a forward route has been created or
updated {[\dots]}, the node consults its route table entry
for the originating node to determine the next hop for the RREP
packet, and then forwards the RREP towards the originator using the
information in that route table entry.''
This action needs a valid route to the originator $\oip$ of the route
request to which the current message is a reply ($\oip\in\akD{\rt}$, Line~\ref{rrep:line11}).
The content of the RREP message to be sent is mostly copied from the RREP received;
only the sender has to be changed (it is now the node's $\ip$) and the
hop count is incremented.
Prior to the unicast, the node $\nhop{\rt}{\oip}$, to which the message is
sent, is added to the list of precursors for the routes to $\dip$
(Line~\ref{rrep:line12a}) and to the next hop on the route to $\dip$ (Line~\ref{rrep:line12b}).
Although not specified in the RFC, it would make sense to also add a precursor to the reverse
route by 
$\assignment{\rt := \addprecrt{\rt}{\oip}{\{\nhop{\rt}{\dip}\}}}$\label{pg:small_precursor_improvement}.
As usual, if the unicast fails, the affected routing table entries are invalidated and the precursors of all
routes using the broken link are determined and an error message is sent (Lines~\ref{rrep:line16}--\ref{rrep:line20}).
In the unlikely situation that a reply should be forwarded but no valid route is known by the node,
nothing happens.
Following the RFC, no precursor has to be notified and no error message
has to be sent---even if there is an invalid route.

If a forward routing table entry is not created nor updated, 
the reply is silently ignored and the basic process is called (Lines~\ref{rrep:line25}--\ref{rrep:line26}).

\subsection{Receiving Route Errors}

The process
$\RERR$ models the part of AODV which handles error messages.
An error message consists of a set $\dests$ of pairs of
an unreachable destination IP address $\rip$ and
the corresponding unreachable destination sequence number $\rsn$.

  \algsetup{linenodelimiter=.,linenosize=\tiny}
  \begin{algorithm}[H]
    {\footnotesize
      \caption{RERR handling}
      \label{pro:rerr}
      \begin{algorithmic}[1]
\DEFPROCESS{\RERR}{\dests\comma\sip\,\comma\,\ip\comma\sn\comma\rt\comma\rreqs\comma\queues}
		\COMLINE{invalidate broken routes}												\label{rerr:line4}
		\UPD{\dests:=\{(\rip,\rsn)\,|\,(\rip,\rsn)\in\dests\ans\rip\in\akD{\rt}\ans \nhop{\rt}{\rip}=\sip\ans \sqn{\rt}{\rip}<\rsn\}}			\label{rerr:line2}
		\UPD{\rt:=\inv{\rt}{\dests}}															\label{rerr:line5}
		\UPD{\queues:=\setrrf{\queues}{\dests}}\label{rerr:line5a}
		\COMLINE{forward the RERR to all precursors for \rt\ entries for broken connections}			\label{rerr:line1}
		\UPD{\pre:=\bigcup\{\precs{\rt}{\rip}\,|\,(\rip,*)\in\dests\}}									\label{rerr:line3}
		\UPD{\dests:=\{(\rip,\rsn)\,|\,(\rip,\rsn)\in\dests\ans \precs{\rt}{\rip}\not=\emptyset\}}				\label{rerr:line3a}
		\groupcast{\pre}{\rerr{\dests}{\ip}}\ . \aodv{\ip}{\sn}{\rt}{\rreqs}{\queues}						\label{rerr:line6}

	\end{algorithmic}
    }
  \end{algorithm}

\vspace{-2mm}

If a node receives an AODV error message from a neighbour for one or
more valid routes, it has---under some conditions---to invalidate the entries for those routes
in its own routing table and forward the error message.  The node
compares the set $\dests$ of unavailable destinations from the
incoming error message with its own entries in the routing table.  If
the routing table lists a valid route with a $(\rip,\rsn)$-combination
from $\dests$ and if the next hop on this route is the sender $\sip$
of the error message, this entry may be affected by the error message. 
In our formalisation, we have added the requirement
\highlight{$\sqn{\rt}{\rip}<\rsn$}, saying that
the entry is affected by the error message only if the ``incoming'' sequence number
is larger than the one stored in the routing table, meaning
that it is based on fresher information.\footnote{This additional
    requirement is in the spirit of Section 6.2 of the RFC~\cite{rfc3561} on
    updating routing table entries, but in contradiction with Section
    6.11 of the RFC on handling $\RERR$ messages. In
    \Sect{interpretation} we will show that the reading of Section
    6.11 of the RFC gives rise to routing loops.}
In this case, the entry has to be invalidated and all precursors of this particular route have
to be informed. This has to be done for all affected routes.

In fact, the process first determines all $(\rip,\rsn)$-pairs that have effects on its
own routing table and that may have to be forwarded as content of a new error message
(Line~\ref{rerr:line2}). After that, all entries to unavailable destinations are invalidated
(Line~\ref{rerr:line5}), and 
as usual when routing table entries are invalidated, the \penFlag{}s are set to {\nonpen} (Line~\ref{rerr:line5a}).
In Line~\ref{rerr:line3} the set of all precursors (affected
neighbours) of the unavailable destinations are summarised in the set $\pre$.
Then, the set {\dests} is ``thinned out'' to only those destinations that have at least one precursor---
only these destinations are transmitted in the forwarded error message
(Line~\ref{rerr:line3a}).
Finally, the message is sent (Line~\ref{rerr:line6}).

\subsection{The Message Queue and Synchronisation}\label{ssec:message_queue}

We assume that any message sent by a node \dval{sip} to a node
\dval{ip} that happens to be within transmission range of \dval{sip}
is actually received by \dval{ip}.  For this reason, \dval{ip} should
always be able to perform a receive action, regardless of which state
it is in.
However, the main process {\AODV} that runs on the node \dval{ip} can
reach a state, such as {\PKT}, {\RREQ}, {\RREP} or {\RERR}, in which
it is not ready to perform a receive action. For this reason we
introduce a process $\QMSG$, modelling a message queue,
\begin{table}[htb]\vspace{-2ex}
  \algsetup{linenodelimiter=.,linenosize=\tiny}
  \begin{algorithm}[H]
    {\footnotesize
      \caption{Message queue}
      \label{pro:queues}
      \begin{algorithmic}[1]
\DEFPROCESS{\QMSG}{\msgs}
	\IFempty
		\COMLINE{store incoming message at the end of \msgs}			\label{queues:line1}
		\receiveL{\msg}\ . 			
		\Qmsg{\append{\msg}{\msgs}}								\label{queues:line2}
	\ELSIF[the queue is not empty]{$\msgs\not=[\,]$}					\label{queues:line3}
		\PAR
		\COMLINE{pop top message and send it to another sequential process}								\label{queues:line4}
		\sendL{\head{\msgs}}\ .\ \Qmsg{\tail{\msgs}}					\label{queues:line5}
		\COMLINE{or receive and store an incoming message}				\label{queues:line7}
		\STATE $+$\,  \receive{\msg}\ . \Qmsg{\append{\msg}{\msgs}}                     \label{queues:line8}
		\ENDPAR
	\ENDIFii

	\end{algorithmic}
    }
  \end{algorithm}

\vspace{-2ex}\end{table}
that runs in parallel with {\AODV} or any other process that might be called.
Every incoming message is first stored in this queue, and piped from
there to the process {\AODV}, whenever {\AODV} is ready to handle a
new message. The process {\QMSG} is always ready to receive a new
message, even when {\AODV} is not.  The whole parallel process running
on a node is then given by an expression of the form
\[
(\xi,\aodv{\ip}{\sn}{\rt}{\rreqs}{\queues})\ \parl\ (\xii,\Qmsg{\msgs})
\ .\]

\subsection{Initial State}\label{ssec:initial}

\index{initial state}%
To finish our specification, we have to define an initial state.  The
initial network expression is an encapsulated parallel composition of
node expressions $\dval{ip}:P:R$, where the (finite) number of nodes
and the range $R$ of each node expression is left unspecified (can be
anything). However, each node in the parallel composition is required
to have a unique IP address \dval{ip}.  The initial process $P$ of
\dval{ip} is given by the expression
$(\xi,\aodv{\ip}{\sn}{\rt}{\rreqs}{\queues})\ \parl\ (\xii,\Qmsg{\msgs})$,
with
\begin{equation}\label{eq:initialstate_rt}
\xi(\ip)=\dval{ip}
\ans
\xi(\sn)=1
\ans
\xi(\rt)=\emptyset
\ans
\xi(\rreqs)=\emptyset
\ans
\xi(\queues)=\emptyset
\ans
\xii(\msgs)=[\,]\ .
\end{equation}
This says that initially each node is correctly informed about its own
identity; its own sequence number is initialised with $1$ and its
routing table, the list of RREQs seen, the store of queued data
packets as well as the message queue are empty.

\section{Invariants}\label{sec:invariants}

\newcommand{\hopsc}{\dval{hops}_c}
\newcommand{\dipc}{\dval{dip}_{\hspace{-1pt}c}}
\newcommand{\ripc}{\dval{rip}_{\hspace{-1pt}c}}
\newcommand{\oipc}{\dval{oip}_{\hspace{-1pt}c}}
\newcommand{\rreqidc}{\dval{rreqid}_{c}}
\newcommand{\dsnc}{\dval{dsn}_c}
\newcommand{\rsnc}{\dval{rsn}_c}
\newcommand{\destsc}{\dval{dests}_c}
\newcommand{\osnc}{\dval{osn}_c}
\newcommand{\ipc}{\dval{ip}_{\hspace{-1pt}c}}
\newcommand{\xiN}[2][N]{\xi_{#1}^{#2}}
\newcommand{\zetaN}[2][N]{\zeta_{#1}^{#2}}
\newcommand{\RN}[2][N]{R_{#1}^{#2}}
\newcommand{\PN}[2][N]{P_{#1}^{#2}}

Using our process algebra for wireless mesh networks and the proposed
model of AODV we can now formalise and prove crucial properties of
AODV\@.  In this section we verify properties that can be expressed as
\index{invariants}%
invariants, i.e., statements that hold all the time when the protocol
is executed.

The most important invariant we establish is \emph{loop freedom}; most
prior results can be regarded as stepping stones towards this goal.
Next to that we also formalise and discuss \emph{route correctness}.

\subsection{State and Transition Invariants}\label{ssec:transition invariants}

A \emph{(state) invariant} is a statement that holds for all reachable
\index{state}%
states of our model.  Here states are network expressions, as defined
in \SSect{networks}.  An invariant is usually verified by showing that it holds
\index{transition}%
for all possible initial states, and that, for any transition
$N\ar{\ell}N'$ between (encapsulated) network expressions derived by our operational
semantics, if it holds for state $N$ then it also holds for state $N'$.

Besides (state) invariants, we also establish statements we call
\emph{transition invariants}.  A transition invariant is a statement
that holds for each reachable transition $N\ar{\ell}N'$ between (encapsulated)
network expressions derived by the operational semantics (Table~\ref{tab:sos network}).
In establishing a transition invariant for a
particular transition, we usually assume it has already been obtained
for all \emph{prior transitions}, those that occurred beforehand.
Since the transition system generated by our operational semantics may
have cycles, we need to give a well-founded definition of
``beforehand''.  To this end we treat a statement about a transition
as one about a \emph{transition occurrence}, defined as a path in our
transition system, stating in an initial state, and ending with the
transition under consideration. This way the induction is performed on
the length of such a path. We speak of \hypertarget{induction-on-reachability}
{\phrase{induction on reachability}}.

To facilitate formalising transition invariants, we present a taxonomy of
the transitions that can be generated by our operational semantics,
along with some notation: the label $\ell$ of a transition $N\ar{\ell}N'$ can be either
$\textbf{connect}(\dval{ip},\dval{ip}')$,
$\textbf{disconnect}(\dval{ip},\dval{ip}')$,
$\colonact{\dval{ip}}{\textbf{newpkt}(\dval{d},\dval{dip})}$,
$\colonact{\dval{ip}}{\deliver{\dval{d}}}$ or $\tau$.
We are most interested in the last case.
A transition $N\ar{\tau}N'$ either arises from a transition
$\colonact{R}{\starcastP{m}}$ performed by a network node \dval{ip},
synchronising with receive actions of all nodes $\dval{dip}\in R$
in transmission range, or stems from a $\tau$-transition of a network node \dval{ip}.

In the former case, we write $N\ar{R:\starcastP{\dval{m}}}_\dval{ip}N'$.
This means that $N=[M]$ and $N'=[M']$ are network expressions such
that $M\ar{R:\starcastP{m}}M'$, and the cast action is performed by
node \dval{ip}.  This transition originates from an action
$\broadcastP{\dexp{ms}}$, $\groupcastP{\dexp{dests}}{\dexp{ms}}$, or
$\unicast{\dexp{dest}}{\dexp{ms}}$ (cf.\ \Sect{process_algebra}). Each
such action can be identified by a line number in one of the processes
of \Sect{modelling_AODV}.

In the latter case, a $\tau$-transition of a node \dval{ip} stems
either from a failed unicast, an evaluation $[\varphi]$, an assignment
$\assignment{\keyw{var}:=\dexp{exp}}$, or a synchronisation of two
actions $\send{\dexp{ms}}$ and $\receive{\msg}$ performed by
sequential processes running on that node.  In our model these
processes are \AODV\ and \QMSG, and these actions can also be
identified by line numbers in the processes of \Sect{modelling_AODV}.

The following observations are crucial in establishing many of our invariants.
\begin{prop}\label{prop:preliminaries}\rm~
\begin{enumerate}[(a)]
\item\label{it:preliminariesi} With the exception of new packets that
  are submitted to a node by a client of AODV, every message received and handled by the
  main routine of AODV has to be sent by some node before.\label{before}
  More formally, we consider an arbitrary
  path $N_0\ar{\ell_1}N_1\ar{\ell_2} \ldots \ar{\ell_k} N_k$ with
  $N_0$ an initial state in our model of AODV\@. If the transition
  $N_{k-1}\ar{\ell_k} N_k$ results from a synchronisation involving the
  action $\receive{\msg}$ from Line~\ref{aodv:line2} of
  Pro.~\ref{pro:aodv}---performed by the node \dval{ip}---, where the
  variable {\msg} is assigned the value $m$, then either
  $m=\newpkt{\dval{d}}{\dval{dip}}$ or one of the $\ell_i$ with $i<k$
  stems from an action $\starcastP{m}$ of a node $\dval{ip}'$ of the network.
\item\label{it:preliminariesii} No node can receive a message directly from itself.
  Using the formalisation above, we must have $\dval{ip}\neq\dval{ip}'$.
\end{enumerate}
\end{prop}
\begin{proof}
The only way Line~\ref{aodv:line2} of Pro.~\ref{pro:aodv} can be
executed, is through a synchronisation of the main process \AODV\ with
the message queue \QMSG\ (Pro.~\ref{pro:rerr}) running on the same
node. This involves the action $\send{m}$ of \QMSG\@.  Here $m$ is
popped from the message queue {\msgs}, which started out empty. So at
some point \QMSG\ must have performed the action $\receive{m}$. However,
this action is blocked by the encapsulation operator $[\_]$ of
Table~\ref{tab:sos network}, except when $m$ has the form
$\newpkt{\dval{d}}{\dval{dip}}$ or when it synchronises with an action
$\starcastP{m}$ of another node.
\end{proof}
At first glance Part\eqref{it:preliminariesii} does not seem to reflect reality. Of course, an application running on a local node has to be able to send 
data packets to another application running on the same node. However,
in any practical implementation, when a node sends a message to itself, the message 
will be delivered to the corresponding application on the local node without ever being ``seen'' by AODV 
or any other routing protocol. Therefore,
from AODV's perspective, no node can receive a message (directly) from itself.

\subsection{Notions and Notations}
Before formalising and proving invariants, we introduce some useful notions and notations.

All processes except $\QMSG$ maintain the five data variables {\ip}, {\sn},
{\rt}, {\rreqs} and {\queues}. Next to that $\QMSG$ maintains the variable $\msgs$.
Hence, these $6$ variables can be evaluated at any time.
Moreover, every node expression in the transition system looks like
\[
\dval{ip}:\left(\xi,P\ \parl\ \xii,\Qmsg{\msgs}\right):R
\ ,\]
where $P$ is a state in one of the following sequential processes:

\begin{tabular}{@{}l}
$\aodv{\ip}{\sn}{\rt}{\rreqs}{\queues}$\ ,\\[0.5mm]
$\newpktPL{\data}{\dip}{\ip}{\sn}{\rt}{\rreqs}{\queues}$\ ,\\[0.5mm]
$\pktPL{\data}{\dip}{\oip}{\ip}{\sn}{\rt}{\rreqs}{\queues}$\ ,\\[0.5mm]
$\rreqPL{\hops}{\rreqid}{\dip}{\dsn}{\dsk}{\oip}{\osn}{\sip}{\ip}{\sn}{\rt}{\rreqs}{\queues}$\ ,\\[0.5mm]
$\rrepPL{\hops}{\dip}{\dsn}{\oip}{\sip}{\ip}{\sn}{\rt}{\rreqs}{\queues}$ or\\[0.5mm]
$\rerrPL{\dests}{\sip}{\ip}{\sn}{\rt}{\rreqs}{\queues}$\ .\\[0.5mm]
\end{tabular}

\noindent Hence the state of the transition system for a node $\dval{ip}$
is determined by
the process $P$,
the range $R$, and
the two valuations $\xi$ and $\xii$.
If a network consists of a (finite) set $\IP\subseteq\tIP$ of nodes, a
reachable network expression $N$ is an encapsulated parallel composition
of node expressions---one for each $\dval{ip}\in\IP$.
In this section, we assume $N$ and $N'$ to be reachable
network expressions in our model of AODV\@.
To distill current information about a node from $N$,
we define the following projections:

\begin{tabular}{@{}l@{\,$:=$\,}l@{\ where\ \,}l@{\,:\,}c@{\,:\,}l@{\,\ is a node expression of $N$}l}
$\PN{\dval{ip}}$       &$ P$,        & $\dval{ip}$ & $(*,P\parl *,*)$       & $*$   &\ ,\\[0.5mm]
$\RN{\dval{ip}}$       &$R$,         & $\dval{ip}$ & $(*,*\parl *,*)$        & $R$  &\ ,\\[0.5mm]
$\xiN{\dval{ip}}$       &$ \xi$,      & $\dval{ip}$  & $(\xi,*\parl *,*)$     & $*$   &\ ,\\[0.5mm]
$\zetaN{\dval{ip}}$   &$\zeta$,   & $\dval{ip}$   & $(*,*\parl \zeta,*)$ & $*$   &\ .
\end{tabular}

\noindent
For example,
$\PN{\dval{ip}}$ determines the sequential process the node is currently working in,
$\RN{\dval{ip}}$ denotes the set of all nodes currently within transmission range of $\dval{ip}$, and 
$\xiN{\dval{ip}}(\rt)$ evaluates the current routing table maintained by node \dval{ip} in the network expression $N$.
In the forthcoming proofs, when discussing the effects of an action,
identified by a line number in one of the processes of our model, $\xi$ denotes the current valuation
\plat{$\xiN{\dval{ip}}$}, where \dval{ip} is the address of
the local node, executing the action under consideration, and
$N$ is the network expression obtained right before this action occurs, corresponding
with the line number under consideration. When consider the effects of
several actions, corresponding to several line numbers, $\xi$ is
always interpreted most locally. For instance, in the proof of \Prop{msgsending}\eqref{prop:msgsendingRREQ},
case \hyperlink{forinstance}{\textbf{Pro.~\ref*{pro:rreq}, Line~\ref*{rreq:line34}}}, we write
\begin{quote}
  Hence \ldots $\ipc:=\xi(\ip)=\dval{ip}$ and $\xiN{\ipc}=\xi$ (by \Eq{uniqueidwithxi}).
  At Line~\ref{rreq:line6} we update the routing table using
  $\dval{r}:=\xi(\oip,\osn,\kno,\val,\hops\mathord+1,\sip,\emptyset)$
  as new entry.  The routing table does not change between
  Lines~\ref{rreq:line6} and~\ref{rreq:line34}; nor do the values of
  the variables {\hops}, {\oip} and {\osn}.
\end{quote}
Writing $N_k$ for a network expression in which the local node \dval{ip} is about to
execute Line~$k$, this passage can be reworded as
\begin{quote}
  Hence \ldots $\ipc:=\xiN[N_{\ref*{rreq:line34}}]{\dval{ip}}(\ip)=\dval{ip}$ and
  $\xiN[N_{\ref*{rreq:line34}}]{\ipc}=\xiN[N_{\ref*{rreq:line34}}]{\dval{ip}}$
  (by \Eq{uniqueidwithxi}).\\
  $\xiN[N_{\ref*{rreq:line8}}]{\dval{ip}}(\rt)\begin{array}[t]{@{~:=~}l@{}}
   \xiN[N_{\ref*{rreq:line6}}]{\dval{ip}}(\upd{\rt}{(\oip,\osn,\kno,\val,\hops\mathord+1,\sip,\emptyset)})\\
   \upd{\xiN[N_{\ref*{rreq:line6}}]{\dval{ip}}(\rt)}{(\xiN[N_{\ref*{rreq:line6}}]{\dval{ip}}(\oip),\xiN[N_{\ref*{rreq:line6}}]{\dval{ip}}(\osn),\ldots)}.
   \end{array}$\\
  \makebox[10pt][l]{$\xiN[N_{\ref*{rreq:line8}}]{\dval{ip}}(\rt)=\xiN[N_{\ref*{rreq:line34}}]{\dval{ip}}(\rt) \wedge
   \xiN[N_{\ref*{rreq:line6}}]{\dval{ip}}(\hops)=\xiN[N_{\ref*{rreq:line34}}]{\dval{ip}}(\hops) \wedge
   \xiN[N_{\ref*{rreq:line6}}]{\dval{ip}}(\oip)  =\xiN[N_{\ref*{rreq:line34}}]{\dval{ip}}(\oip) \wedge
   \xiN[N_{\ref*{rreq:line6}}]{\dval{ip}}(\osn)  =\xiN[N_{\ref*{rreq:line34}}]{\dval{ip}}(\osn)$.}
\end{quote}
In all of case \hyperlink{forinstance}{\textbf{Pro.~\ref*{pro:rreq}, Line~\ref*{rreq:line34}}}, through the
statement of the proposition, $N$ is bound to $N_{\ref*{rreq:line34}}$,
so that $\xiN{\dval{ip}}=\xiN[N_{\ref*{rreq:line34}}]{\dval{ip}}$.

In \SSect{rt} we have defined functions that work on evaluated
routing tables $\xiN{\dval{ip}}(\rt)$, such as $\fnnhop$.
To ease readability, we abbreviate
\plat{$\nhop{\xiN{\dval{ip}}(\rt)}{\dval{dip}}$} by \plat{$\nhp{\dval{ip}}$}.
Similarly, we use 
\plat{$\sq{\dval{ip}}$}, \plat{$\dhp{\dval{ip}}$}, \plat{$\sta{\dval{ip}}$}, \plat{$\sr{\dval{ip}}$}, 
\plat{$\kd{\dval{ip}}$}, \plat{$\akd{\dval{ip}}$}
and \plat{$\ikd{\dval{ip}}$} for
\plat{$\sqn{\xiN{\dval{ip}}(\rt)}{\dval{dip}}$}, \plat{$\dhops{\xiN{\dval{ip}}(\rt)}{\dval{dip}}$},
\plat{$\status{\xiN{\dval{ip}}(\rt)}{\dval{dip}}$,}
 \plat{$\selr{\xiN{\dval{ip}}(\rt)}{\dval{ip}}$},
\plat{$\kD{\xiN{\dval{ip}}(\rt)}$}, \plat{$\akD{\xiN{\dval{ip}}(\rt)}$} and \plat{$\ikD{\xiN{\dval{ip}}(\rt)}$}, respectively.

\subsection{Basic Properties}
In this section we show some of the most fundamental invariants for AODV. 
The first one is already stated in the RFC~\cite[Sect. 3]{rfc3561}.
\begin{prop}\rm
\label{prop:invarianti_itemiii}
Each sequence number of any given node $\dval{ip}$ increases monotonically, i.e., never decreases, 
and is never unknown.
  That is, for $\dval{ip}\mathbin\in\IP$, if $N \ar{\ell} N'$
  then $1\leq\xiN{\dval{ip}}(\sn)\leq\xiN[N']{\dval{ip}}(\sn)$.
\end{prop}

\begin{proof}~
In all initial states the invariant is satisfied, as 
all sequence numbers of all nodes are set to $1$
(see~\eqref{eq:initialstate_rt} in Section~\ref{ssec:initial}).
The Processes \ref{pro:aodv}--\ref{pro:queues} of
\Sect{modelling_AODV} change a node's sequence number
only through the functions
{\fninc} and $\max$.
This occurs at two places only:
\begin{description}
  \item[Pro.~\ref{pro:aodv}, Line~\ref{aodv:line36}:]
    Here $\xiN{\dval{ip}}(\sn)\leq
    \inc{\xiN{\dval{ip}}(\sn)} = \xiN[N']{\dval{ip}}(\sn)$.
  \item[Pro.~\ref{pro:rreq}, Line~\ref{rreq:line12}:]
    Here $\xiN{\dval{ip}}(\sn)\leq
    \max({\xiN{\dval{ip}}(\sn)},*) = \xiN[N']{\dval{ip}}(\sn)$.
\end{description}
From this and the fact that all sequence numbers are initialised with
$1$ we get $1\leq\xiN{\dval{ip}}(\sn)$.
\end{proof}

The proof strategy used above can be generalised.

\begin{remark}\label{rem:remark}
Most of the forthcoming proofs can be done by showing the statement
for each initial state and then checking all locations in the
processes where the validity of the invariant is possibly changed.
Note that routing table entries are only changed by the functions
\hyperlink{update}{$\fnupd$}, \hyperlink{invalidate}{$\fninv$}
or \hyperlink{addprert}{$\fnaddprecrt$}.
Thus we have to show that an invariant dealing with routing tables is
satisfied after the execution of these functions if it was valid
before.  In our proofs, we go through all occurrences of these functions.
In case the invariant does not make statements about precursors, 
the function $\fnaddprecrt$ need not be considered.
\end{remark}

\begin{prop}\label{prop:destinations maintained}\rm
  The set of known destinations of a node increases monotonically.
  That is, for $\dval{ip}\mathbin\in\IP$,\linebreak[2] if $N \ar{\ell} N'$
  then $\kd{\dval{ip}}\subseteq\kd[N']{\dval{ip}}$.
\end{prop}
\begin{proof}
None of the functions used to change routing tables removes an entry altogether.
\end{proof}

\begin{prop}\label{prop:rreqs increase}\rm
The set of already seen route requests of a node
increases monotonically.
  That is, for $\dval{ip}\mathbin\in\IP$, if $N \ar{\ell} N'$
  then $\xiN{\dval{ip}}(\rreqs)\subseteq\xiN[N']{\dval{ip}}(\rreqs)$.
\end{prop}
\begin{proof}
None of the functions used in the specification ever removes an entry 
from \rreqs.
\end{proof}

\begin{prop}\label{prop:dsn increase}\rm
In each node's routing table, the
  sequence number for any given destination increases monotonically, i.e., never decreases.
  That is, for $\dval{ip},\dval{dip}\mathbin\in\IP$, if $N \ar{\ell} N'$
  then $\sq{\dval{ip}}\leq\fnsqn_{N'}^\dval{ip}(\dval{dip})$.
\end{prop}
\begin{proofNobox}
The only function that can decrease a sequence number is 
\hyperlink{invalidate}{$\fninv$}. 
When invalidating routing table entries using the function $\inv{\rt}{\dests}$, sequence numbers are copied from {\dests} to the corresponding entry in \rt. 
It is sufficient to show that for all \plat{$(\dval{rip},\dval{rsn})\in\xiN{\dval{ip}}(\dests)$}
$\sq[\dval{rip}]{\dval{ip}}\leq\dval{rsn}$, as all other sequence numbers in routing table
entries remain unchanged.
\begin{description}
  \item[Pro.~\ref{pro:aodv}, Line~\ref{aodv:line32}; Pro.~\ref{pro:pkt}, Line~\ref{pkt2:line10}; Pro.~\ref{pro:rreq}, Lines~\ref{rreq:line18}, \ref{rreq:line30}; Pro.~\ref{pro:rrep}, Line~\ref{rrep:line18}:]~\\
  The set {\dests} is constructed immediately before the invalidation procedure. For $(\dval{rip},\dval{rsn})\in\xiN{\dval{ip}}(\dests)$, we have
$
   \sq[\dval{rip}]{\dval{ip}} \leq \inc{\sq[\dval{rip}]{\dval{ip}}} = \dval{rsn}.
$
\item[Pro.~\ref{pro:rerr}, Line~\ref{rerr:line5}:] When constructing {\dests} in Line~\ref{rerr:line2}, the side condition $\xiN[N_{\ref*{rerr:line2}}]{\dval{ip}}(\sqn{\rt}{\rip})<\xiN[N_{\ref*{rerr:line2}}]{\dval{ip}}(\rsn)$ is taken into account, which immediately yields the claim for $(\dval{rip},\dval{rsn})\in\xiN{\dval{ip}}(\dests)$.\endbox
\end{description}
\end{proofNobox}

Our next invariant tells that each node is correctly informed about
its own identity.
\begin{prop}\rm\label{prop:self-identification}
For each $\dval{ip}\in\IP$ and each reachable state $N$ we have
$\xiN{\dval{ip}}(\ip)=\dval{ip}$.
\end{prop}
\begin{proof}
According to \SSect{initial} the claim is assumed
to hold for each initial state, and none of our processes has an
assignment changing the value of the variable $\ip$.
\end{proof}
This proposition will be used implicitly in many of the proofs to follow.
In particular, for all $\dval{ip}',\dval{ip}''\mathbin{\in}\IP$%
\begin{equation}\label{eq:uniqueidwithxi}
\xiN{\dval{ip}'}(\ip)=\dval{ip}'' \ims \dval{ip}'=\dval{ip}''\ans \xiN{\dval{ip}'}=\xiN{\dval{ip}''}\ .
\end{equation}

Next, we show that every AODV control message contains the IP address of the sender.
\begin{prop}\label{prop:ip=ipc}\rm If an AODV control message is
  sent by node $\dval{ip}\in\IP$, the node sending this message
  identifies itself correctly:
\begin{equation*}
	N\ar{R:\starcastP{m}}_{\dval{ip}}N' \ims \dval{ip}=\ipc\ ,
\end{equation*}
where the message $m$ is either
$\rreq{*}{*}{*}{*}{*}{*}{*}{\ipc}$,
$\rrep{*}{*}{*}{*}{\ipc}$, or
$\rerr{*}{\ipc}$.
\end{prop}
The proof is straightforward: whenever such a message is sent in
one of the processes of \Sect{modelling_AODV}, $\xi(\ip)$ is set as the last argument.
\endbox

\begin{cor}\rm\label{cor:sipnotip}
At no point will the variable $\sip$ maintained by node \dval{ip} have
the value \dval{ip}.
\begin{equation*}
\xiN{\dval{ip}}(\sip)\neq\dval{ip}
\end{equation*}
\end{cor}
\begin{proof}
The value of $\sip$ stems, through Lines~\ref{aodv:line8},~\ref{aodv:line12} or~\ref{aodv:line16}
of Pro.~\ref{pro:aodv}, from an incoming AODV control message of the form
$\xiN{\dval{ip}}(\rreq{*}{*}{*}{*}{*}{*}{*}{\sip})$,
$\xiN{\dval{ip}}(\rrep{*}{*}{*}{*}{\sip})$, or
$\xiN{\dval{ip}}(\rerr{*}{\sip})$
(Pro.~\ref{pro:aodv}, Line~\ref{aodv:line2}); the value of $\sip$ is never changed.
By \Prop{preliminaries}, this message must have been sent before by a node $\dval{ip}'\neq\dval{ip}$.
By \Prop{ip=ipc}, $\xiN{\dval{ip}}(\sip)=\dval{ip}'$.
\end{proof}

\begin{prop}\rm\label{prop:positive hopcount}
All routing table entries have a hop count greater or equal than $1$.
\begin{equation}\label{eq:inv_length}
(*,*,*,*,\dval{hops},*,*)\in\xiN{\dval{ip}}(\rt) \ims \dval{hops}\geq1
\end{equation}
\end{prop}

\begin{proofNobox}
All initial states trivially satisfy the invariant since all routing tables are empty.
The functions \hyperlink{invalidate}{\fninv} and \hyperlink{addprert}{\fnaddprecrt} do not affect
the invariant, since they do not change the hop count of a routing table entry.
Therefore, we only have to look at the application calls of \hyperlink{update}{\fnupd}.
In each case, if the update does not change the routing table entry beyond its precursors
({the last clause} of \hyperlink{update}{\fnupd}), the invariant is trivially
preserved; hence we examine the cases that an update actually occurs.
	\begin{description}
		\item[Pro.~\ref{pro:aodv}, Lines~\ref{aodv:line10}, \ref{aodv:line14}, \ref{aodv:line18}:]
			All these updates have a hop count equals to $1$; hence the invariant is preserved.
		\item[Pro.~\ref{pro:rreq}, Line~\ref{rreq:line6}; Pro.~\ref{pro:rrep}, Line~\ref{rrep:line5}:]
			Here, $\xi(\hops)+1$ is used for the update. Since $\xi(\hops)\in\NN$, the invariant is maintained.
			\endbox
	\end{description}
\end{proofNobox}

\begin{prop}\label{prop:starcastNew}\rm~
\begin{enumerate}[(a)]
\item If a route request with hop count $0$ is sent by a node
  $\ipc\in\IP$ , the sender must be the originator.
	\begin{equation}\label{inv:starcast_i}
	N\ar{R:\starcastP{\rreq{0}{*}{*}{*}{*}{\oipc}{*}{\ipc}}}_{\dval{ip}}N' \ims \oipc=\ipc(=\dval{ip})
	\end{equation}
\item If a route reply with hop count $0$ is sent by a node
  $\ipc\in\IP$, the sender must be the destination.
	\begin{equation}\label{inv:starcast_i_rrep}
        N\ar{R:\starcastP{\rrep{0}{\dipc}{*}{*}{\ipc}}}_{\dval{ip}}N' \ims \dipc=\ipc(=\dval{ip})    
	\end{equation}
\end{enumerate}
\end{prop}

\begin{proofNobox}~
\begin{enumerate}[(a)]
\item We have to check that the consequent holds whenever a route request is sent. In all the processes there
are only two locations where this happens.
\begin{description}
	\item[\Pro{aodv}, Line~\ref{aodv:line39}:]
		A request with content 
		$
		\xi(0\comma*\comma*\comma*\comma*\comma\ip\comma*\comma\ip)
		$
		 is sent. Since the sixth and the eighth component are
                 the same ($\xi(\ip)$), the claim holds.
	\item[Pro.~\ref{pro:rreq}, Line~\ref{rreq:line34}:]
	The message has the form $\rreq{\xi(\hops)\mathord+1}{*}{*}{*}{*}{*}{*}{*}$.
	Since $\xi(\hops)\in\NN$, $\xi(\hops)+1\not=0$ and hence the
        antecedent does not hold.
\end{description}
\item We have to check that the consequent holds whenever a route reply is sent. In all the processes there
are only three locations where this happens.
\begin{description}
	\item[\Pro{rreq}, Line~\ref{rreq:line14a}:]
		A reply with content 
		$
		\xi(0\comma\dip\comma*\comma*\comma\ip)
		$
		is sent. By Line~\ref{rreq:line10} we have
                $\xi(\dip)=\xi(\ip)$, so the claim holds.
	\item[Pro.~\ref{pro:rreq}, Line~\ref{rreq:line26}:]
	The message has the form $\rrep{\dhops{\rt}{\dip}}{\!*}{\!*}{\!*}{\!*}$.
        By \Prop{positive hopcount}, $\dhops{\rt}{\dip}>0$, so the
        antecedent does not hold.
	\item[Pro.~\ref{pro:rrep}, Line~\ref{rrep:line13}:]
	The message has the form $\rrep{\xi(\hops)\mathord+1}{\!*}{\!*}{\!*}{\!*}$.
	Since $\xi(\hops)\in\NN$, $\xi(\hops)+1\not=0$ and hence the
        antecedent does not hold.
\endbox
\end{description}
\end{enumerate}
\end{proofNobox}

\begin{prop}\rm\label{prop:rte}~
\begin{enumerate}[(a)]
\item\label{it_a} Each routing table entry
with $0$ as its destination sequence number has a sequence-number-status flag valued unknown.
\begin{equation}\label{eq:unk_sqn}
(\dval{dip},0,\dval{f},*,*,*,*)\in\xiN{\dval{ip}}(\rt) \ims \dval{f}=\unkno
\end{equation}
\item\label{it_d} Unknown sequence numbers can only occur at $1$-hop connections.
\begin{equation}\label{eq:inv_viia}
(*,*,\unkno,*,\dval{hops},*,*)\in\xiN{\dval{ip}}(\rt) \ims \dval{hops}=1
\end{equation}
\item\label{it_c} $1$-hop connections must contain the destination as next hop.
\begin{equation}\label{eq:inv_vii}
(\dval{dip},*,*,*,1,\dval{nhip},*)\in\xiN{\dval{ip}}(\rt)\ims \dval{dip}=\dval{nhip}
\end{equation}
\item\label{it_e}
If the sequence number $0$ occurs within a routing table entry, the
hop count as well as the next hop can be determined.
\begin{equation}\label{eq:inv_viib}
(\dval{dip},0,\dval{f},*,\dval{hops},\dval{nhip},*)\in\xiN{\dval{ip}}(\rt)\ims \dval{f}=\unkno \ans \dval{hops}=1\ans \dval{dip}=\dval{nhip}
\end{equation}
\end{enumerate}
\end{prop}
\begin{proofNobox}At the initial states all routing tables are empty. Since \hyperlink{invalidate}{$\fninv$} and \hyperlink{addprert}{$\fnaddprecrt$}
  change neither the sequence-number-status flag, nor the next hop or the hop count of a routing
  table entry, and---by \Prop{dsn increase}---cannot decrease the sequence number of a destination,
  we only have to look at the application calls of \hyperlink{update}{\fnupd}.
  As before, we only examine the cases that an update actually occurs.
\begin{enumerate}[(a)]
\item
 Function calls of the form \hyperlink{update}{\upd{\dval{rt}}{\dval{r}}} always preserve the invariant:
in case {\fnupd} is
  given an argument for which it is not defined, the process algebra blocks and no change
  of the routing table is performed (cf.\ Footnote~\ref{partial} in \Sect{process_algebra});
in case one of the first four clauses in the definition of \hyperlink{update}{\fnupd} is used, this
follows because $\upd{\dval{rt}}{\dval{r}}$ is defined only when $\pi_{2}(\dval{r})=0\Leftrightarrow\pi_{3}(\dval{r})=\unkno$;
 in
case the fifth clause is used it follows because $\pi_{3}(\dval{r})=\unkno$; and in case the last
clause is used, it follows by induction, since the invariant was already valid before the update.
\item\begin{description}
		\item[Pro.~\ref{pro:aodv}, Lines~\ref{aodv:line10}, \ref{aodv:line14}, \ref{aodv:line18}:]
			All these updates have an unknown sequence number and hop count equal to $1$.
                        By Clause 5 of
			\hyperlink{update}{\fnupd}, these sequence-number-status flag and hop count
			are transferred literally into the routing table; hence the invariant is preserved.
		\item[Pro.~\ref{pro:rreq}, Line~\ref{rreq:line6} and Pro.~\ref{pro:rrep}, Line~\ref{rrep:line5}:]
		  In these updates the sequence-number-status flag is set to {\kno}. By the definition of
		  \hyperlink{update}{\fnupd}, this value ends up in the routing table.
		  Hence the assumption of the invariant to be proven
		  is not satisfied. 
	\end{description}
\item 	\begin{description}
		\item[Pro.~\ref{pro:aodv}, Lines~\ref{aodv:line10}, \ref{aodv:line14}, \ref{aodv:line18}:]
			The new entries ($\xi(\sip,0,\unkno,\val,1,\sip,\emptyset)$) satisfy the invariant; even if the routing table is actually
			updated with one of the new routes, the invariant holds afterwards.
		\item[Pro.~\ref{pro:rreq}, Line~\ref{rreq:line6}; Pro.~\ref{pro:rrep}, Line~\ref{rrep:line5}:]
			The route which might be inserted into the routing table
			has hop count $\hops\mathord+1$, $\hops\in\NN$.
		 	It can only be equal to $1$ if the received message had hop count $\hops=0$.
			In that case
                        Invariant~\eqref{inv:starcast_i}, resp.~\eqref{inv:starcast_i_rrep}, guarantees that the invariant remains unchanged.
	\end{description}
\item Immediate from Parts~\eqref{it_a} to \eqref{it_c}. \endbox
\end{enumerate}
\end{proofNobox}

\begin{prop}\label{prop:msgsendingii}~\rm
\begin{enumerate}[(a)]
\item Whenever an originator sequence number is sent as part of a route request message, it is known, i.e.,
it is greater or equal than $1$.
	\begin{equation}\label{inv:starcast_sqni}	
		N\ar{R:\starcastP{\rreq{*}{*}{*}{*}{*}{*}{\osnc}{*}}}_{\dval{ip}}N' \ims \osnc \geq1
	\end{equation}
\item Whenever a destination sequence number is sent as part of a route reply message, it is known, i.e.,
it is greater or equal than $1$.
 	\begin{equation}\label{inv:starcast_sqnii}
		N\ar{R:\starcastP{\rrep{*}{*}{\dsnc}{*}{*}}}_{\dval{ip}}N'\ims\dsnc\geq1
 	\end{equation}
\end{enumerate}
\end{prop}
\begin{proofNobox}~
\begin{enumerate}[(a)]
\item
We have to check that the consequent holds whenever a route request is sent.
\begin{description}
	\item[Pro.~\ref{pro:aodv}, Line~\ref{aodv:line39}:] A route request is initiated. 
		The originator sequence number is a copy of the node's own sequence 
		number, i.e., $\osnc=\xi(\sn)$. By \Prop{invarianti_itemiii}, 
		we get $\osnc\geq 1$.
	\item[Pro.~\ref{pro:rreq}, Line~\ref{rreq:line34}:] Here, $\osnc:=\xi(\osn)$. $\xi(\osn)$ 
		is not changed within Pro.~\ref{pro:rreq}; it stems, through Line~\ref{aodv:line8}
		of Pro.~\ref{pro:aodv}, from an incoming RREQ message (Pro.~\ref{pro:aodv}, Line~\ref{aodv:line2}).
		For this incoming RREQ message, using
		\Prop{preliminaries}(\ref{it:preliminariesi}) and
		\hyperlink{induction-on-reachability}{induction on reachability}, the
		invariant holds and hence the claim follows immediately. 
\end{description}
\item We have to check that the consequent holds whenever a route reply is sent.
\begin{description}
	\item[Pro.~\ref{pro:rreq}, Line~\ref {rreq:line14}:] The destination initiates a route reply. 
		The sequence number is a copy of the node's own sequence number, i.e., 
		$\dsnc=\xi(\sn)$. By \Prop{invarianti_itemiii}, we get $\dsnc\geq 1$.
	\item[Pro.~\ref{pro:rreq}, Line~\ref{rreq:line26}:] 
		The sequence number used for the message is copied from the routing table;
		its value is $\dsnc:=\sqn{\xi(\rt)}{\xi(\dip)}$. 
		By Line~\ref{rreq:line22}, we know that $\status{\xi(\rt)}{\xi(\dip)}=\kno$ and hence,
                by Invariant~\Eq{unk_sqn}, $\dsnc\geq 1$. Thus the invariant is maintained.
	\item[Pro.~\ref{pro:rrep}, Line~\ref{rrep:line13}:] Here, $\dsnc:=\xi(\dsn)$. $\xi(\dsn)$ 
		is not changed within Pro.~\ref{pro:rrep}; it stems, through Line~\ref{aodv:line12}
		of Pro.~\ref{pro:aodv}, from an incoming RREP message (Pro.~\ref{pro:aodv}, Line~\ref{aodv:line2}).
		For this incoming RREP message the invariant holds and hence the claim follows immediately.
\endbox
\end{description}
\end{enumerate}
\end{proofNobox}

\begin{prop}\label{prop:msgsending}\rm~
\begin{enumerate}[(a)]
\item\label{prop:msgsendingRREQ}
If a route request is sent (forwarded) by a node $\ipc$ different from
the originator of the request 
then the content of $\ipc$'s routing table
must be fresher or at least as good as the information inside the message.
	\begin{equation}\label{inv:starcast_ii}	
	\begin{array}{rcl}
	  &&N\ar{R:\starcastP{\rreq{\hopsc}{*}{*}{*}{*}{\oipc}{\osnc}{\ipc}}}_{\dval{ip}}N'
          \ans\ipc\neq\oipc\\
	  &\Rightarrow&
	  \oipc\in\kd{\ipc}
	  \ans\big(\sq[\oipc]{\ipc}>\osnc
	  \\
	  &&\ors (\sq[\oipc]{\ipc}=\osnc \ans \dhp[\oipc]{\ipc}\leq\hopsc\ans \sta[\oipc]{\ipc}=\val)\big)
	\end{array}
	\end{equation}
\item
If a route reply is sent by a node $\ipc$, different from
the destination of the route, then the content of $\ipc$'s routing table
must be consistent with the information inside the message.
 	\begin{equation}\label{inv:starcast_iv}
 	\begin{array}{@{}rcl@{}}
 	  &&N\ar{R:\starcastP{\rrep{\hopsc}{\dipc}{\dsnc}{*}{\ipc}}}_{\dval{ip}}N'
          \ans\ipc\neq\dipc\\
 	  &\Rightarrow&
	  \dipc\in\kd{\ipc}
 	  \ans \sq[\dipc]{\ipc} = \dsnc\ans  \dhp[\dipc]{\ipc}=\hopsc\ans \sta[\dipc]{\ipc}=\val
 	\end{array}
 	\end{equation}
\end{enumerate}
\end{prop}
\begin{proofNobox}~
\begin{enumerate}[(a)]
\item
We have to check all cases where a route request is sent:
\begin{description}
	\item[Pro.~\ref{pro:aodv}, Line~\ref{aodv:line39}:]
		A new route request is initiated with $\ipc=\oipc:=\xi(\ip)=\dval{ip}$.
		Here the antecedent of \eqref{inv:starcast_ii} is not satisfied.
	\item[Pro.~\ref{pro:rreq}, Line~\ref{rreq:line34}:]
                \hypertarget{forinstance}{The broadcast message has the form
                \[\xi(\rreq{\hops\mathord+1}{\rreqid}{\dip}{\max(\sqn{\rt}{\dip},\dsn)}{\dsk}{\oip}{\osn}{\ip})\ .\]
                Hence $\hopsc:=\xi(\hops)\mathord+1$,
                $\oipc:=\xi(\oip)$, $\osnc:=\xi(\osn)$,
                $\ipc:=\xi(\ip)=\dval{ip}$ and $\xiN{\ipc}=\xi$
                (by \Eq{uniqueidwithxi}).
		At Line~\ref{rreq:line6} we update the routing table using
		$\dval{r}\mathbin{:=}\xi(\oip,\osn,\kno,\hspace{-.5pt}\val,\hops\mathord+1,\sip,\emptyset)$ as new entry.
		The routing table does not change between
                	Lines~\ref{rreq:line6}
                	and~\ref{rreq:line34}; nor do the values of the variables
                	{\hops}, {\oip} and {\osn}.}
		If the new (valid) entry is inserted into the routing table, then
                	one of the first four cases in the definition of \hyperlink{update}{$\fnupd$}
                	must have applied---the fifth case cannot apply, since $\pi_3(\dval{r})=\kno$.
                	Thus, using that $\oipc\neq\ipc$,
		$$\begin{array}{r@{~=~}l@{~=~}l@{~=~}l}
		\sq[\oipc]{\ipc} & \sqn{\xi(\rt)}{\xi(\oip)}& \xi(\osn) & \osnc\\
		\dhp[\oipc]{\ipc} & \dhops{\xi(\rt)}{\xi(\oip)} & \xi(\hops)+1 & \hopsc\\
		\sta[\oipc]{\ipc}& \status{\xi(\rt)}{\xi(\oip)} & \xi(\val) & \val\;.
		\end{array}$$
		
		In case the new entry is not inserted into the routing
		table (the sixth case of \hyperlink{update}{$\fnupd$}), we have
		$\sq[\oipc]{\ipc}=\sqn{\xi(\rt)}{\xi(\oip)}\geq\xi(\osn)=\osnc$,
		and in case that		
		$\sq[\oipc]{\ipc}=\osnc$ we see that
		$\dhp[\oipc]{\ipc} = \dhops{\xi(\rt)}{\xi(\oip)}\leq \xi(\hops)\mathop{+}1 = \hopsc$
		and moreover $\sta[\oipc]{\ipc}=\val$.
		Therefore the invariant holds.
\end{description}
\item We have to check all cases where a route reply is sent.
\begin{description}
	\item[Pro.~\ref{pro:rreq}, Line~\ref {rreq:line14}:]
		A new route reply with
		$\ipc:=\xi(\ip)=\dval{ip}$ is initiated.
		Moreover, by Line~\ref{rreq:line10}, $\dipc:=\xi(\dip)=\xi(\ip)=\dval{ip}$ and 
		thus $\ipc=\dipc$.
		Hence, the antecedent of \eqref{inv:starcast_iv} is not satisfied.
	\item[Pro.~\ref{pro:rreq}, Line~\ref{rreq:line26}:]
	        We have $\ipc:=\xi(\ip)=\dval{ip}$, so	$\xiN{\ipc}=\xi$.
		This time, by Line~\ref{rreq:line20}, 
		$\dipc:= \xi(\dip)\neq\xi(\ip)=\ipc$.
		By Line~\ref{rreq:line22} there is a valid routing table entry for $\dipc:=\xi(\dip)$.
		\[\begin{array}{r@{~:=~}c@{~}l}
		\dsnc&\sqn{\xi(\rt)}{\xi(\dip)}&=~\sq[\dipc]{\ipc}\ ,
		\\
		\hopsc&\dhops{\xi(\rt)}{\xi(\dip)}&=~\dhp[\dipc]{\ipc}\ .
		\end{array}\]
	\item[Pro.~\ref{pro:rrep}, Line~\ref{rrep:line13}:]
		The RREP message has the form
		\[\xi(\rrep{\hops\mathop{+}1}{\dip}{\dsn}{\oip}{\ip})\ .\]
		Hence $\hopsc:=\xi(\hops)\mathord+1$,
                $\dipc:=\xi(\dip)$, $\dsnc:=\xi(\dsn)$,
                $\ipc:=\xi(\ip)=\dval{ip}$ and \mbox{$\xiN{\ipc}=\xi$}.
		Using $(\xi(\dip),\xi(\dsn),\kno,\val,\xi(\hops)\mathord{+}1,\xi(\sip),\emptyset)$ as
		new entry, the routing table is updated at Line~\ref{rrep:line5}.
		With exception of its precursors, which are irrelevant here, the routing table
		does not change between Lines~\ref{rrep:line5} and \ref{rrep:line13};
		nor do the values of the variables {\hops}, {\dip} and {\dsn}.
		Line~\ref{rrep:line3} guarantees that
		during the update in Line~\ref{rrep:line5},
		the new entry is inserted into the routing table,
		so\\
		\mbox{}\hfill$\begin{array}[b]{r@{~=~}l@{~=~}l@{~=~}l}
		\sq[\dipc]{\ipc}& \sqn{\xi(\rt)}{\xi(\dip)} & \xi(\dsn) & \dsnc\\
		\dhp[\dipc]{\ipc} & \dhops{\xi(\rt)}{\xi(\dip)} & \xi(\hops)+1 & \hopsc\\
		\sta[\dipc]{\ipc} & \status{\xi(\rt)}{\xi(\dip)} & \xi(\val) & \val\;.
		\end{array}$\hfill\mbox{\endbox}
\end{description}
\end{enumerate}
\end{proofNobox}

\begin{prop}\rm\label{prop:starcastrerr} Any sequence number appearing in a route error message
 stems from an invalid destination
 and is equal to the sequence number for that destination in
 the sender's routing table at the time of sending.
\begin{equation}\label{inv:starcast_rerr}
N\ar{R:\starcastP{\rerr{\destsc}{\ipc}}}_{\dval{ip}}N'\ans (\ripc,\rsnc)\in\destsc
\ims 
\ripc\in\ikd{\dval{ip}} \ans \rsnc = \sq[\ripc]{\dval{ip}}
\end{equation}
\end{prop}

\begin{proofNobox}
We have to check that the consequent holds whenever a route error message is sent.
In all the processes there are only seven locations where this happens.
\begin{description}
\item[Pro.~\ref{pro:aodv}, Line~\ref{aodv:line33}:]
The set $\destsc$ is constructed in Line~\ref{aodv:line31a} as a subset of 
$\xiN[N_{\ref*{aodv:line31a}}]{\dval{ip}}(\dests)=\xiN[N_{\ref*{aodv:line32}}]{\dval{ip}}(\dests)$.
For each pair $(\ripc,\rsnc)\in\xiN[N_{\ref*{aodv:line32}}]{\dval{ip}}(\dests)$
one has $\ripc=\xiN[N_{\ref*{aodv:line30}}]{\dval{ip}}(\rip)\in\fnakD_{N_{\ref*{aodv:line30}}}^{\dval{ip}}$.
Then in Line~\ref{aodv:line32}, using the function \hyperlink{invalidate}{$\fninv$},
$\status{\xi(\rt)}{\ripc}$ is set to $\inval$ and 
$\sqn{\xi(\rt)}{\ripc}$ to $\rsnc$.
Thus we obtain $\ripc\in\ikd{\dval{ip}}$ and
$\sq[\ripc]{\dval{ip}} = \rsnc$.
\item[Pro.~\ref{pro:pkt}, Line~\ref{pkt2:line14};
      Pro.~\ref{pro:rreq}, Lines~\ref{rreq:line19},~\ref{rreq:line31};
      Pro.~\ref{pro:rrep}, Line~\ref{rrep:line20};
      Pro.~\ref{pro:rerr}, Line~\ref{rerr:line6}:]
Exactly as above.
\item[Pro.~\ref{pro:pkt}, Line~\ref{pkt2:line20}:]
The set $\destsc$ contains only one single element. Hence $\ripc:=\xiN{\dval{ip}}(\dip)$ and $\rsnc:=\xiN{\dval{ip}}(\sqn{\rt}{\dip})$.
By Line~\ref{pkt2:line18}, we have $\ripc=\xiN{\dval{ip}}(\dip)\in\ikd{\dval{ip}}$. The remaining claim follows by 
$
\rsnc=\xiN{\dval{ip}}(\sqn{\rt}{\dip})=\sqn{\xiN{\dval{ip}}(\rt)}{\xiN{\dval{ip}}(\dip)} = \sq[\ripc]{\dval{ip}}.
$\endbox
\end{description}
\end{proofNobox}

\subsection{Well-Definedness}
We have to ensure that our specification of AODV is 
actually well defined. Since many functions introduced in \Sect{types} are only partial,
it has to be checked that these functions are either defined when they are used, 
or are subterms of atomic formulas. In the latter case, those formula would evaluate 
to {\tt false} (cf.\ Footnote~\ref{fn:undefvalues} on Page~\pageref{fn:undefvalues}).

The first proposition shows that the functions defined in \Sect{types} respect the data structure.
In fact, these properties are required (or implied) by our data structure.
\begin{prop}\label{prop:invarianti}\rm~
\begin{enumerate}[(a)]
\item\label{prop:invarianti_itemii} In each routing table
 there is at most one entry for each destination.
\item\label{prop:invarianti_itemiv} In each store of queued data
  packets there is at most one data queue for each destination.
\item\label{prop:invarianti_dests} Whenever a set of pairs $(\dval{rip},\dval{rsn})$ is assigned
  to the variable {\dests} of type $\tIP\rightharpoonup\tSQN$,
or to the first argument of the function $\rerrID$, this
  set is a partial function, i.e., there is at most one entry
  $(\dval{rip},\dval{rsn})$ for each destination $\dval{rip}$.
\end{enumerate}
\end{prop}

\begin{proofNobox}~
\begin{enumerate}[(a)]
\item In all initial states the invariant is satisfied, as a routing
table starts out empty
(see \eqref{eq:initialstate_rt} in Section~\ref{ssec:initial}).
None of the Processes \ref{pro:aodv}--\ref{pro:queues} of
\Sect{modelling_AODV} changes a routing table directly;
the only way a routing table can be changed is through the functions
\hyperlink{update}{$\fnupd$}, \hyperlink{invalidate}{$\fninv$} and
\hyperlink{addprert}{$\fnaddprecrt$}. The latter two
only change the sequence number, the validity status and the precursors of an existing route.
This kind of update has no effect on the invariant.
The first function inserts a new entry into a routing table only if the
destination is unknown, that is, if no entry for this destination
already exists in the routing table; otherwise the existing entry is replaced.
Therefore the invariant is maintained.
\item In any initial state the invariant is satisfied, as each store of
  queued data packets starts out empty. In Processes
  \ref{pro:aodv}--\ref{pro:queues} of \Sect{modelling_AODV}
  a store is updated only through the functions \hyperlink{add}{\fnadd} and \hyperlink{drop}{\fndrop}.
  These functions respect the invariant.
\item This is checked by inspecting all assignments to $\dests$ in
  Processes \ref{pro:aodv}--\ref{pro:queues}.
\begin{description}
\item[Pro.~\ref{pro:aodv}, Line~\ref{aodv:line16}:]
  The message $\xi(\msg)$ is received in Line~\ref{aodv:line2}, and
  hence, by \Prop{preliminaries}(\ref{it:preliminariesi}), sent by
  some node before. The content of the message does not change during
  transmission, and we assume there is only one way to read a message
  $\xi(\msg)$ as $\rerr{\xi(\dests)}{\xi(\sip)}$. By induction, we may
  assume that when the other node composed the message, a partial
  function was assigned to the first argument $\xi(\dests)$ of $\rerrID$.
\item[Pro.~\ref{pro:aodv}, Line~\ref{aodv:line30}; 
	Pro.~\ref{pro:pkt}, Line~\ref{pkt2:line9}; 
	Pro.~\ref{pro:rreq}, Lines~\ref{rreq:line16},~\ref{rreq:line28}; 
	Pro.~\ref{pro:rrep}, Line~\ref{rrep:line16}:]
  The assigned sets have the form
  $\{(\xi(\rip),\inc{\sqn{\xi(\rt)}{\xi(\rip)}})\mid \dots\})$.
  Since $\fninc$ and $\fnsqn$ are functions, for each $\xi(\rip)$ there is only
  one pair $(\xi(\rip),\inc{\sqn{\xi(\rt)}{\xi(\rip)}})$.
\item[Pro.~\ref{pro:aodv}, Line~\ref{aodv:line31a}; 
	Pro.~\ref{pro:pkt}, Line~\ref{pkt2:line13}; 
	Pro.~\ref{pro:rreq}, Lines~\ref{rreq:line17a},~\ref{rreq:line29a}; 
	Pro.~\ref{pro:rrep}, Line~\ref{rrep:line17a}; 
	Pro.~\ref{pro:rerr}, Line~\ref{rerr:line3a}:]
In each of these cases a set $\xi(\dests)$ constructed four lines before is used to construct a new set.
By the invariant to be proven, these sets are already partial functions. 
From these sets some values are removed.
Since subsets of partial functions are again partial functions, the claim follows immediately.
\item[Pro.~\ref{pro:rerr}, Line~\ref{rerr:line2}:] Similar to the previous case except that the set $\xi(\dests)$ to be thinned 
out is not constructed before but stems from an incoming RERR message.
\item[Pro.~\ref{pro:pkt}, Lines~\ref{pkt2:line20}:]
The set is explicitly given and consists of only one element; thus the claim is trivial.
\endbox
\end{description}
\end{enumerate}
\end{proofNobox}
Property~\eqref{prop:invarianti_itemii}
is stated in the RFC~\cite{rfc3561}.

\begin{prop}\label{prop:selr_welldefined}\rm
In our specification of AODV, the functions \hyperlink{selroute}{$\fnselroute$} and
\hyperlink{status}{$\fnstatus$}
are only used when they are defined.
\end{prop}
\begin{proof}In our entire specification we do not use
these functions at all; they are only used for defining other functions.
\end{proof}

\begin{prop}\label{prop:dhops_well_defined}\rm
In our specification of AODV, the function \hyperlink{dhops}{$\fndhops$}
is only used when it is defined.
\end{prop}
\begin{proofNobox}
The function $\dhops{\dval{rt}}{\dval{dip}}$ is defined iff $\dval{dip}\in\kD{\dval{rt}}$.
\begin{description}
\item[\Pro{rreq}, Line~\ref{rreq:line26}:]
By Line~\ref{rreq:line22} $\xi(\dip)\mathbin\in\akD{\xi(\rt)}\mathbin\subseteq\kD{\xi(\rt)}$; so $\dhops{\xi(\rt)}{\xi(\dip)}$ is defined.\hspace*{-10pt}\endbox
\end{description}
\end{proofNobox}

\begin{prop}\label{prop:nhop_well_defined}\rm
In our specification of AODV, the function \hyperlink{nhop}{$\fnnhop$} is either used within formulas or if it is defined; hence it is only used in a meaningful way. 
\end{prop}
\begin{proofNobox}
As in \Prop{dhops_well_defined}, the function $\nhop{\dval{rt}}{\dval{dip}}$ is defined iff $\dval{dip}\in\kD{\dval{rt}}$. 
\begin{description}
\item[\Pro{aodv}, Line~\ref{aodv:line30};
    \Pro{pkt}, Line~\ref{pkt2:line9}; 
    \Pro{rreq}, Lines~\ref{rreq:line16}, \ref{rreq:line28}; 
    \Pro{rrep}, Line~\ref{rrep:line16}; 
    \Pro{rerr}, Line~\ref{rerr:line2}:] The function is used within a formula.
\item[\Pro{aodv}, Line~\ref{aodv:line24}:] Line~\ref{aodv:line22} states $\xi(\dip)\in\akD{\xi(\rt)}$; hence $\nhop{\xi(\rt)}{\xi(\dip)}$ is defined.
\item[\Pro{pkt}, Line~\ref{pkt2:line7}:] By Line~\ref{pkt2:line5}, $\xi(\dip)\in\akD{\xi(\rt)}$.
\item[\Pro{rreq}, Lines~\ref{rreq:line14a}, \ref{rreq:line26}:] In Line~\ref{rreq:line6} the entry for destination $\xi(\hspace{-.2pt}\oip\hspace{-.2pt})$ is updated;
by this $\xi(\hspace{-.2pt}\oip\hspace{-.2pt})\mathbin\in\kD{\xi(\hspace{-.2pt}\rt\hspace{-.2pt})}$.
\item[\Pro{rreq}, Line~\ref{rreq:line25}:] By Line~\ref{rreq:line22} $\xi(\dip)\in\akD{\xi(\rt)}$.
\item[\Pro{rrep}, Lines~\ref{rrep:line12a},  \ref{rrep:line13}:]  By Line~\ref{rrep:line11} $\xi(\oip)\in\akD{\xi(\rt)}$.
\item[\Pro{rrep}, Line~\ref{rrep:line12b}:] In Line~\ref{rrep:line5} the entry for destination $\xi(\dip)$ is updated;
by this $\xi(\dip)\mathbin\in\kD{\xi(\rt)}$. By Line~\ref{rrep:line11} $\xi(\oip)\in\akD{\xi(\rt)}$.
\end{description}
  If $\fnnhop$ is used within a formula, then
  $\nhop{\dval{rt}}{\dval{rip}}$ may not be defined, namely
  if $\dval{rip}\not\in\kD{\dval{rt}}$. In such a case, according to the convention of
  Footnote~\ref{fn:undefvalues} in \Sect{process_algebra},  the atomic formula
  in which this term occurs evaluates to {\tt false}, and thereby is
  defined properly.\endbox
\end{proofNobox}
If one chooses to use lazy evaluation for conjunction, then \hyperlink{nhop}{$\fnnhop$} is only used where it is defined.

\begin{prop}\label{prop:precs_well_defined}\rm
In our specification of AODV, the function \hyperlink{precs}{$\fnprecs$} is only used when it is defined.
\end{prop}
\begin{proofNobox}
As in \Prop{dhops_well_defined}, the function $\precs{\dval{rt}}{\dval{dip}}$ is defined iff $\dval{dip}\in\kD{\dval{rt}}$. 
\begin{description}
\item[\Pro{aodv}, Line~\ref{aodv:line31}; 
	\Pro{pkt}, Line~\ref{pkt2:line12}; 
	\Pro{rreq}, Lines~\ref{rreq:line17}, \ref{rreq:line29}; 
	\Pro{rrep}, Line~\ref{rrep:line17}:] 
Three lines before the $\fnprecs$ is used, the set $\xi(dests)$ is created containing only pairs $(\xi(\rip),*)$ with $\xi(\rip)\in\akD{\xi(\rt)}$.
\item[\Pro{aodv}, Line~\ref{aodv:line31a}; 
	\Pro{pkt}, Line~\ref{pkt2:line13}; 
	\Pro{rreq}, Lines~\ref{rreq:line17a}, \ref{rreq:line29a}; 
	\Pro{rrep}, Line~\ref{rrep:line17a}:] 
Four lines before the $\fnprecs$ is used, the set $\xi(dests)$ is created containing only pairs $(\xi(\rip),*)$ with $\xi(\rip)\in\akD{\xi(\rt)}$.
\item[\Pro{pkt}, Line~\ref{pkt2:line20}:] Line~\ref{pkt2:line18} states that $\xi(\dip)\in\ikD{\xi(\rt)}\subseteq\kD{\xi(\rt)}$.
\item[\Pro{rerr}, Line~\ref{rerr:line3}:] Similar to \Pro{aodv}, Line~\ref{aodv:line31}; the set $\xi(\dests)$ is created under the assumption 
$\xi(\rip)\in\akD{\xi(\rt)}$ in Line~\ref{rerr:line2}.
\endbox
\end{description}
\end{proofNobox}

\begin{prop}\label{prop:upd_well_defined}\rm
In our specification of AODV, the function \hyperlink{update}{$\fnupd$} is only used when it is defined. 
\end{prop}
\begin{proofNobox}
$\upd{\dval{rt}}{\dval{r}}$ is defined only under the assumptions $\pi_{4}(\dval{r})\mathbin=\val$,
\mbox{$\pi_{2}(\dval{r})\mathbin=0 \Leftrightarrow \pi_{3}(\dval{r})\mathbin=\unkno$} and $\pi_{3}(r)=\unkno\Rightarrow\pi_{5}(r)=1$.
In \Pro{aodv}, Lines~\ref{aodv:line10}, \ref{aodv:line14} and \ref{aodv:line18}, the entry
$\xi(\sip,0,\unkno,\val,1,\sip, \emptyset)$ is used as second argument, which obviously satisfies the 
assumptions. The function is used at four other locations: 
\begin{description}
\item[\Pro{rreq}, Line~\ref{rreq:line6}:]
Here, the entry $\xi(\oip, \osn, \kno, \val, \hops + 1, \sip, \emptyset)$ is
used as $\dval{r}$ to update the routing table. 
This entry fulfils $\pi_{4}(\dval{r})=\val$. 
Since $\pi_{3}(\dval{r})=\kno$, it remains to show that $\pi_{2}(\dval{r})=\xi(\osn)\geq1$.
The sequence number $\xi(\osn)$ stems, through Line~\ref{aodv:line8}
of Pro.~\ref{pro:aodv}, from an incoming RREQ message and is not changed within Pro.~\ref{pro:rreq}.
Hence, by Invariant~\eqref{inv:starcast_sqni}, $\xi(\osn)\geq 1$.
\item[\Pro{rrep}, Lines~\ref{rrep:line3},~\ref{rrep:line5},~\ref{rrep:line25}:]
                  The update is similar to the one of Pro.~\ref{pro:rreq}, Line~\ref{rreq:line6}. 
                  The only difference is that the information stems from an incoming RREP message and 
                  that a routing table entry to $\xi(\dip)$ (instead of  $\xi(\oip)$) is established. 
                  Therefore, the proof is similar to the one of Pro.~\ref{pro:rreq}, Line~\ref{rreq:line6}; instead 
                  of Invariant~\eqref{inv:starcast_sqni} we use Invariant~\eqref{inv:starcast_sqnii}. \endbox
\end{description}
\end{proofNobox}

\begin{prop}\label{prop:addpreRT_well_defined}\rm
In our specification of AODV, the function \hyperlink{addprert}{$\fnaddprecrt$} is only used when it is defined.
\end{prop}
\begin{proofNobox}
It suffices to check that for any call $\addprecrt{\dval{rt}}{\dval{dip}}{*}$ the destination has an entry in the routing table, i.e., $\dval{dip}\in\kD{\dval{rt}}$.
\begin{description}
\item[Pro.~\ref{pro:rreq}, Line~\ref{rreq:line24}:] Line~\ref{rreq:line22} shows that $\xi(\dip)\in\akD{\xi(\rt)}\subseteq\kD{\xi(\rt)}$.
\item[Pro.~\ref{pro:rreq}, Line~\ref{rreq:line25}:] In Line~\ref{rreq:line6} an entry to $\xi(\oip)$ is updated. 
In case there was no entry before, it is inserted; hence we know $\xi(\oip)\in\kD{\xi(\rt)}$.
\item[Pro.~\ref{pro:rrep}, Line~\ref{rrep:line12a}:] Similar to the previous case: Line~\ref{rrep:line5} updates a 
routing entry to $\xi(\dip)$.
\item[Pro.~\ref{pro:rrep}, Line~\ref{rrep:line12b}:] Line~\ref{rrep:line5} updates the routing table entry with destination 
$\xi(\dip)$. By Line~\ref{rrep:line3} it is known that the entry $\xi(\dip,\dsn,\kno,\val,\hops+1,\sip,\emptyset)$
is inserted; hence $\nhop{\xi(\rt)}{\xi(\dip)} = \xi(\sip)$. A routing table entry for $\xi(\sip)$ exists by Line~\ref{aodv:line14} of \Pro{aodv}.
\endbox
\end{description}
\end{proofNobox}

\begin{prop}\label{prop:headtail_well_defined}\rm
In our specification of AODV, the functions \hyperlink{head}{$\fnhead$ and $\fntail$} are only used when they are defined.
\end{prop}
\begin{proofNobox}These functions are defined if the list given as argument is non-empty. 
\begin{description}
\item[Pro.~\ref{pro:aodv}, Line~\ref{aodv:line23}:] The function {\fnhead}
	tries to return the first element of $\selq{\xi(\queues)}{\xi(\dip)}$,
         which is, by Line~\ref{aodv:line22}
	($\xi(\dip)\in\qD{\xi(\queues)}$) and (\ref{non-empty}), not empty.
\item[Pro.~\ref{pro:queues}, Line~\ref{queues:line5}:] Here, the functions work on
the list $\xi(\msgs)$; Line~\ref{queues:line3} shows that \mbox{$\xi(\msgs)\not=[\,]$}.
\endbox
\end{description}
\end{proofNobox}

\begin{prop}\label{prop:drop_well_defined}\rm
In our specification of AODV, the function \hyperlink{drop}{$\fndrop$} is only used when it is defined.
\end{prop}
\begin{proof}
The function $\fndrop$ is only used in Pro.~\ref{pro:aodv}, Line~\ref{aodv:line26}.
It tries to delete the oldest packet queued for destination $\xi(\dip)$; the function is defined
if at least one packet for $\xi(\dip)$ is stored in $\xi(\queues)$---this is guaranteed by
Line~\ref{aodv:line22}, which states $\xi(\dip)\in\qD{\xi(\queues)}$, and (\ref{non-empty}).
\end{proof}

\begin{prop}\label{prop:qflag_well_defined}\rm
In our specification of AODV, the function \hyperlink{qflag}{$\fnfD$} is only used within formulas.\endbox
\end{prop}
The function is called only in Pro.~\ref{pro:aodv} in Line~\ref{aodv:line34} using
$\fD{\xi(\queues)}{\xi(\dip)}$. Again, if one would use lazy evaluation for
  conjunction, then \hyperlink{nhop}{$\fnfD$}
  is used where it is defined.

\subsection{The Quality of Routing Table Entries}\label{ssec:quality}

In this section we define a total preorder $\rtord$ on routing table
entries for a given destination \dval{dip}. Entries are ordered
by the \phrase{quality} of the information they provide. This order will be
defined in such a way that 
(a) the quality of a node's routing table entry for \dval{dip}
will only increase over time, and 
(b)  the quality of valid routing table entries along a route to \dval{dip} strictly increases every hop
(at least prior to reaching \dval{dip}).
This order allows us 
to prove \emph{loop freedom} of AODV in the next section.

A main ingredient in the definition of the quality preorder is the
sequence number of a routing table entry. A higher sequence number
denotes fresher information. However, it generally is not the
case that along a route to \dval{dip} found by AODV the sequence
numbers are only increasing. This is since AODV increases the
sequence number of an entry at an intermediate node when invalidating
it.  To ``compensate'' for that we introduce the concept of a \phrase{net
sequence number}. It is defined by a function $\fnnsqn:\tROUTE\to\tSQN$
\[\begin{array}{@{}r@{\hspace{0.5em}}c@{\hspace{0.5em}}l@{}}
\fnnsqn(\dval{r})&:=&\left\{
\begin{array}{ll}
\pi_2(\dval{r}) &\mbox{if } \pi_4(\dval{r})=\val
                             \ors\pi_2(\dval{r})=0\\
\pi_2(\dval{r})-1&\mbox{otherwise}\ .
\end{array}\right.
\end{array}\]
For $n\in\NN$ define $n\decremented:=\max(n\mathord-1,0)$,
so that $\hyperlink{inc}{\fninc}(n)\decremented=n$.
Then $\fnnsqn(r)\mathbin=\pi_2(r)\decremented$ if $\pi_4(r)\mathbin=\inval$.

To model increase in quality, we define $\rtord$ by first comparing the net sequence numbers of
two entries---a larger net sequence number denotes fresher
and higher quality information. In case the net sequence numbers are
equal, we decide on their hop counts---the entry with the least hop
count is the best. This yields the following lexicographical order:

Assume two routing table entries $\dval{r},\dval{r}'\in\tROUTE$ with
$\pi_1(\dval{r})=\pi_1(\dval{r})=\dval{dip}$. Then
\[
\dval{r} \rtord \dval{r}'\ :\Leftrightarrow\
\fnnsqn(\dval{r}) < \fnnsqn(\dval{r}') \ors 
\left(\fnnsqn(\dval{r})=\fnnsqn(\dval{r}')\ans
   \pi_5(\dval{r}) \geq \pi_5(\dval{r}')
\right)\ .
\]

To reason about AODV, net sequence numbers and the quality preorder is
lifted to routing tables. As for \hyperlink{sqn}{$\fnsqn$} we define a total function 
to determine net sequence numbers.
\hypertarget{nsqn}{
\[\begin{array}{r@{\hspace{0.5em}}r@{\hspace{0.5em}}l}
  \fnnsqn : \tRT\times\tIP&\to& \tSQN\\
  \nsqn{\dval{rt}}{\dval{dip}}&:=&
    \left\{
      \begin{array}{ll}
        \fnnsqn(\selr{\dval{rt}}{\dval{dip}})&\mbox{if }\selr{\dval{rt}}{\dval{dip}}\mbox{ is defined}\\
        0&\mbox{otherwise}
      \end{array}
    \right.\\[2ex]
  &=&
    \left\{
      \begin{array}{ll}
        \sqn{\dval{rt}}{\dval{dip}}&\mbox{if }\status{\dval{rt}}{\dval{dip}}=\val \ors \sqn{\dval{rt}}{\dval{dip}}=0\\
        \sqn{\dval{rt}}{\dval{dip}}-1&\mbox{otherwise}\ .
      \end{array}
    \right.
  \end{array}\]
}
If two routing tables $\dval{rt}$ and $\dval{rt}'$ have a routing table entry to $\dval{dip}$, i.e., $\dval{dip}\in\kD{\dval{rt}}\cap\kD{\dval{rt}'}$,
 the preorder can be lifted as well.
\[\begin{array}{r@{\hspace{0.5em}}r@{\hspace{0.5em}}l}
\dval{rt} \rtord \dval{rt}' &:\Leftrightarrow&
\selr{\dval{rt}}{\dval{dip}} \rtord \selr{\dval{rt}'}{\dval{dip}}\\
&\Leftrightarrow&\nsqn{\dval{rt}}{\dval{dip}} < \nsqn{\dval{rt}'}{\dval{dip}} \ors\\
&&\big(\nsqn{\dval{rt}}{\dval{dip}} = \nsqn{\dval{rt}'}{\dval{dip}} \ans
\dhops{\dval{rt}}{\dval{dip}} \geq \dhops{\dval{rt}'}{\dval{dip}}\big)
\end{array}
\]
\noindent For all destinations $\dval{dip}\in\IP$, 
the relation $\rtord$ on routing tables with an entry for \dval{dip} is total preorder.
The equivalence relation induced by $\rtord$ is denoted by $\rtequiv$.

As with \fnsqn, we shorten $\fnnsqn$:
$
\nsq{\dval{ip}} := \nsqn{\xiN{\dval{ip}}(\rt)}{\dval{dip}}.
$
Note that
\begin{equation}\label{eq:sqn_vs_nsqn}
\sq{\dval{ip}}\decremented\leq \nsq{\dval{ip}}\leq \sq{\dval{ip}}\ .
\end{equation}
After setting up this notion of quality, we now show that routing tables, when modified by AODV, 
never decrease their quality.
\pagebreak[3]

\begin{prop}\label{prop:qual}\rm~
Assume a routing table $\dval{rt}\in\tRT$ with $\dval{dip}\in\kD{\dval{rt}}$.
\begin{enumerate}[(a)]
\item\label{it:qual_upd}
An \hyperlink{update}{\fnupd} of $\dval{rt}$ can only increase the quality of the routing table.
That is, for all routes \dval{r} such that \upd{\dval{rt}}{\dval{r}} is defined
(i.e., $\pi_{4}(\dval{r})=\val$,
 $\pi_{2}(\dval{r})=0\Leftrightarrow\pi_{3}(\dval{r})=\unkno$ and
 $\pi_{3}(\dval{r})=\unkno\Rightarrow\pi_{5}(\dval{r})=1$)
we have\vspace{-1ex}\vspace{-3pt}
  \begin{equation}\label{eq:qual_upd}
    \dval{rt}\rtord\upd{\dval{rt}}{\dval{r}}\ .
\vspace{-2pt}
  \end{equation}
\item\label{it:qual_inv}
An \hyperlink{invalidate}{$\fninv$} on $\dval{rt}$ does not change the
quality of the routing table if, for each $(\dval{rip},\dval{rsn})\in\dval{dests}$,
\dval{rt} has a valid entry for \dval{rip}, and
  \begin{itemize}
  \item \dval{rsn} is the by one incremented sequence number from the routing table, or 
  \item both \dval{rsn} and the sequence number in the routing table are $0$.
  \end{itemize}
That is, for all partial functions $\dval{dests}$ (subsets of $\tIP\times\tSQN$)
\vspace{-2pt}
  \begin{equation}\label{eq:qual_inv}
  \begin{array}{rcl}
  &&\big((\dval{rip},\dval{rsn})\in\dval{dests} \ims \dval{rip}\in\akD{\dval{rt}} \ans
  \dval{rsn}=\inc{\sqn{\dval{rt}}{\dval{rip}}}\big) \\
  &\Rightarrow&
    \dval{rt}\rtequiv\inv{\dval{rt}}{\dval{dests}}\ .
  \end{array}
\vspace{-2pt}
  \end{equation}
\item\label{it:qual_addpre}
If precursors are added to an entry of $\dval{rt}$, the quality of the routing table
does not change.
 That is, for all $\dval{dip}\in\IP$ and sets of precursors $\dval{npre}\in\pow(\IP)$,
\vspace{-2pt}
  \begin{equation}\label{eq:qual_addpre}
    \dval{rt}\rtequiv\addprecrt{\dval{rt}}{\dval{dip}}{\dval{npre}}\ .
\vspace{-2pt}
  \end{equation}
\end{enumerate}
\end{prop}

\begin{proofNobox}For the proof we denote the routing table after the update by $\dval{rt}'$.
  \begin{enumerate}[(a)]
 \item By assumption, there is an entry $(\dval{dip},\dval{dsn}_{\dval{rt}},*,\dval{f}_{\dval{rt}},\dval{hops}_{\dval{rt}},*,*)$ 
 for $\dval{dip}$ in $\dval{rt}$. In case $\pi_{1}(r) \not=\dval{dip}$ the quality of the routing table w.r.t.\ \dval{dip}
stays the same, since the entry for \dval{dip} is not changed.
 
We first assume that $\dval{r} := (\dval{dip},0,\unkno,\val,1,*,*)$. This 
means that the Clause 5 in the definition of \hyperlink{update}{\fnupd} 
is used. The updated routing table entry to $\dval{dip}$ has the form $(\dval{dip},\dval{dsn}_{\dval{rt}},\unkno,\val,1,*,*)$.
So\vspace{-2pt}
 \begin{eqnarray*}
 &\nsqn{\dval{rt}}{\dval{dip}} \leq \sqn{\dval{rt}}{\dval{dip}}=\dval{dsn}_{\dval{rt}} =  \nsqn{\dval{rt}'}{\dval{dip}}\ ,\mbox{ and}\\
 &\dhops{\dval{rt}}{\dval{dip}} = \dval{hops}_{\dval{rt}} \geq 1 =  \dhops{\dval{rt}'}{\dval{dip}}\ .
\vspace{-2pt}
 \end{eqnarray*}
The first inequality holds by~\eqref{eq:sqn_vs_nsqn}; the penultimate step by Invariant~\eqref{eq:inv_length}.

Next, we assume that the sequence number is known and therefore the 
route used for the update has the form $\dval{r} = (\dval{dip},\dval{dsn},\kno,\val,\dval{hops},*,*)$ with $\dval{dsn}\geq1$.
After the performed update the routing entry for $\dval{dip}$ either has the form 
$(\dval{dip},\dval{dsn}_{\dval{rt}},*,\dval{f}_{\dval{rt}},\dval{hops}_{\dval{rt}},*,*)$ or 
$(\dval{dip},\dval{dsn},\kno,\val,\dval{hops},*,*)$.
In the former case the invariant is trivially preserved;
in the latter, we know, by definition of \fnupd, that either
(i) $\dval{dsn}_{\dval{rt}}<\dval{dsn}$, 
(ii) $\dval{dsn}_{\dval{rt}}=\dval{dsn} \wedge \dval{hops}_{\dval{rt}}>\dval{hops}$, or
(iii) $\dval{dsn}_{\dval{rt}}=\dval{dsn} \wedge \dval{f}_{\dval{rt}}=\inval$ holds. 
We complete the proof of the invariant by a case distinction.
\begin{description}\hfuzz 50pt
\item[(i) holds:]
First, $\nsqn{\dval{rt}}{\dval{dip}}\leq\dval{dsn}_{\dval{rt}}<\dval{dsn}=\sqn{\dval{rt}'}{\dval{dip}}=\nsqn{\dval{rt}'}{\dval{dip}}$. Since $\dval{dsn}_{\dval{rt}}$ is strictly smaller than $\nsqn{\dval{rt}'}{\dval{dip}}$, there is nothing more to prove.
\item[(iii) holds:]
We have $\nsqn{\dval{rt}}{\dval{dip}}=\dval{dsn}_{\dval{rt}}\decremented<\dval{dsn}=\sqn{\dval{rt}'}{\dval{dip}}=\nsqn{\dval{rt}'}{\dval{dip}}$.
The inequality holds since either $\dval{dsn}_{\dval{rt}}\decremented=0<1\leq\dval{dsn}$ or 
$\dval{dsn}_{\dval{rt}}\decremented=\dval{dsn}_{\dval{rt}}-1<\dval{dsn}_{\dval{rt}}=\dval{dsn}$.
\item[(ii) holds but (iii) does not:] Then $\dval{f}_{\dval{rt}}=\val$.
In this case the update does not change the net sequence number for $\dval{dip}$:\\[-1ex]
\centerline{$\nsqn{\dval{rt}}{\dval{dip}}=\dval{dsn}_{\dval{rt}}=\dval{dsn}=\nsqn{\dval{rt}'}{\dval{dip}}$\ .} 
By (ii), the hop count decreases:\\
\centerline{$\dhops{\dval{rt}}{\dval{dip}} = \dval{hops}_{\dval{rt}}>\dval{hops} = \dhops{\dval{rt}'}{\dval{dip}}$\ .}
\end{description}
\item
Assume that \hyperlink{invalidate}{$\fninv$} modifies an
entry of the form $(\dval{rip},\dval{dsn},*,\dval{flag},*,*,*)$.
Let $(\dval{rip},\dval{rsn})\mathbin\in\dval{dests}$; then $\dval{flag}\mathbin=\val$ and
the update results in the
entry $(\dval{rip},\inc{\dval{dsn}},*,\inval,*,*,*)$.\linebreak[2]
By definition of net sequence numbers,\vspace{-2pt}
\[
\nsqn{\dval{rt}}{\dval{rip}} = \sqn{\dval{rt}}{\dval{rip}} =
  \dval{dsn} =
  \inc{\dval{dsn}}\decremented 
   = \nsqn{\dval{rt}'}{\dval{rip}}\,.
\vspace{-2pt}
\]
Since the hop count is not changed by $\fninv$, we also have
$\dhops{\dval{rt}}{\dval{rip}} =\dhops{\dval{rt}'}{\dval{rip}}$,
and hence $\dval{rt}\rtequiv\inv{\dval{rt}}{\dval{dests}}$.
 \item The function \hyperlink{addprert}{$\fnaddprecrt$} only modifies a set of precursors; 
 it does not change the sequence number, the validity, the flag, nor the hop count 
 of any entry of the routing table $\dval{rt}$.
\endbox
\end{enumerate}
\end{proofNobox}
We can apply this result to obtain the following theorem.

\begin{theorem}\label{thm:state_quality}\rm
In AODV, the quality of routing tables can only be increased, never decreased.

Assume $N \ar{\ell}N'$ and $\dval{ip},\dval{dip}\mathbin\in\IP$.
If $\dval{dip}\in\kd{\dval{ip}}$, then $\dval{dip}\in\kd[N']{\dval{ip}}$ and
\[\xiN{\dval{ip}}(\rt)\rtord\xiN[N']{\dval{ip}}(\rt)\ .\]
\end{theorem}
\begin{proofNobox}
If $\dval{dip}\in\kd{\dval{ip}}$, then
$\dval{dip}\in\kd[N']{\dval{ip}}$ follows by \Prop{destinations maintained}.
To show $\xiN{\dval{ip}}(\rt)\rtord\xiN[N']{\dval{ip}}(\rt)$,
by Remark~\ref{rem:remark} and \Prop{qual}(\ref{it:qual_upd}) and~(\ref{it:qual_addpre})
it suffices to check all calls of \hyperlink{invalidate}{$\fninv$}.
\begin{description}
\item[Pro.~\ref{pro:aodv}, Line~\ref{aodv:line32}; 
	\Pro{pkt}, Line~\ref{pkt2:line10}; 
	Pro.~\ref{pro:rreq}, Lines~\ref{rreq:line18}, \ref{rreq:line30}; 
	Pro.~\ref{pro:rrep}, Line~\ref{rrep:line18}:]~\\
By construction of {\dests} (immediately before the invalidation call) 
$(\dval{rip},\dval{rsn})\in\xiN{\dval{ip}}(\dests) \ims  \dval{rip}\in\akD{\xiN{\dval{ip}}(\rt)} \ans \dval{rsn}=\inc{\sqn{\xiN{\dval{ip}}(\rt)}{\dval{rip}}}$
and hence, by \Prop{qual}\eqref{it:qual_inv},
$\xiN{\dval{ip}}(\rt)\rtequiv\inv{\xiN{\dval{ip}}(\rt)}{\xiN{\dval{ip}}(\dests)} =\xiN[N']{\dval{ip}}(\rt)$.
\item[Pro.~\ref{pro:rerr}, Line~\ref{rerr:line5}:]
Assume that \hyperlink{invalidate}{$\fninv$} modifies an
entry of the form $(\dval{rip},\dval{dsn},*,\dval{flag},*,*,*)$.
Let $(\dval{rip},\dval{rsn})\mathbin\in\dval{dests}$; then
the update results in the entry $(\dval{rip},\dval{rsn},*,\inval,*,*,*)$.
Moreover, by Line~\ref{rerr:line2} of Pro.~\ref{pro:rerr}, $\dval{flag}\mathbin=\val$.
By definition of net sequence numbers,
\[
\nsqn{\xiN{\dval{ip}}(\rt)}{\dval{rip}} = \sqn{\xiN{\dval{ip}}(\rt)}{\dval{rip}}
  \leq \dval{rsn}\decremented 
  = \nsqn{\xiN[N']{\dval{ip}}(\rt)}{\dval{rip}}\,.
\]
The second step holds, since, by Line~\ref{rerr:line2},
$\sqn{\xiN[N_{\ref*{rerr:line2}}]{\dval{ip}}(\rt)}{\dval{rip}} <\dval{rsn}$.
Since the hop count is not changed by $\fninv$, we also have
$\dhops{\xiN{\dval{ip}}(\rt)}{\dval{rip}} =\dhops{\xiN[N']{\dval{ip}}(\rt)}{\dval{rip}}$,
and therefore $\xiN{\dval{ip}}(\rt)\rtord\xiN[N']{\dval{ip}}(\rt)$.
\endbox
\end{description}
\end{proofNobox}
\Thm{state_quality} states in particular that if $N \ar{\ell}N'$ then
$\nsq{\dval{ip}} \leq \fnnsqn_{N'}^{\dval{ip}}(\dval{dip})$.

\begin{prop}\rm\label{prop:inv_nsqn}
If, in a reachable network expression $N$, a node $\dval{ip}\mathop\in\IP$ has a
routing table entry to $\dval{dip}$, then also the next hop
\dval{nhip} towards \dval{dip}, if not \dval{dip} itself, has a
routing table entry to $\dval{dip}$, and the net sequence number of
the latter entry is at least as large as that of the former.
\begin{equation}\label{eq:inv_ix}
\dval{dip}\in\kd{\dval{ip}}\ans\dval{nhip}\not=\dval{dip}
\ims	\dval{dip}\in\kd{\dval{nhip}} \ans \nsq{\dval{ip}}\leq \nsq{\dval{nhip}}\ ,
\end{equation}
where $\dval{nhip}:=\nhp{\dval{ip}}$ is the IP address of the next hop.
\end{prop}

\begin{proofNobox} As before, we first check the initial states
  of our transition system and then check all locations in
  Processes~\ref{pro:aodv}--\ref{pro:queues} where a routing table might
  be changed. For an initial network expression, the invariant holds
  since all routing tables are empty.
	
A modification of \plat{$\xiN{\dval{nhip}}(\rt)$} is harmless, as it
can only increase $\kd{\dval{nhip}}$ (cf.\ \Prop{destinations maintained})
as well as $\nsq{\dval{nhip}}$ (cf.\ \Thm{state_quality}).

Adding precursors to $\xiN{\dval{ip}}(\rt)$ does not harm since the 
invariant does not depend on precursors.
It remains to examine all calls of \hyperlink{update}{$\fnupd$} and
\hyperlink{invalidate}{$\fninv$} to \plat{$\xiN{\dval{ip}}(\rt)$}.
Without loss of generality we restrict attention to those applications of $\fnupd$
or $\fninv$ that actually modify the entry for \dval{dip}, beyond its
precursors; if $\fnupd$ only adds some precursors in the routing
table, the invariant---which is assumed to hold before---is maintained.
If $\fninv$ occurs,
the next hop \dval{nhip} is not changed.
Since the invariant has to hold before the execution, it follows that 
$\dval{dip}\in\kd{\dval{nhip}}$ also holds
after execution.

\begin{description}
\item[Pro.~\ref{pro:aodv}, Lines~\ref{aodv:line10}, \ref{aodv:line14}, \ref{aodv:line18}:]
	The entry $\xi(\sip\comma0\comma\unkno\comma\val\comma1\comma\sip\comma\emptyset)$ is used for the update; 
	its destination is $\dval{dip}:=\xi(\sip)$.
	Since $\dval{dip}=\xi(\sip)=\nhp[\xi(\sip)]{\dval{ip}}=\nhp[\dval{dip}]{\dval{ip}}=\dval{nhip}$, the antecedent of the invariant to be proven is not satisfied.
\item[Pro.~\ref{pro:aodv}, Line~\ref{aodv:line32}; 
	Pro.~\ref{pro:pkt}, Line~\ref{pkt2:line10}; 
	Pro.~\ref{pro:rreq}, Lines~\ref{rreq:line18}, \ref{rreq:line30}; 
	Pro.~\ref{pro:rrep}, Line~\ref{rrep:line18}:]~\\
        In each of these cases, the precondition of  \eqref{eq:qual_inv} is satisfied by 
        the executions of the line immediately before the call of {\fninv} 
        (Pro.~\ref{pro:aodv}, Line~\ref{aodv:line30},
        Pro.~\ref{pro:pkt}, Line~\ref{pkt2:line9}; Pro.~\ref{pro:rreq}, Lines~\ref{rreq:line16}, \ref{rreq:line28}; Pro.~\ref{pro:rrep}, Line~\ref{rrep:line16}).
        Thus, the quality of the routing table w.r.t.\ \dval{dip}, and thereby the net
        sequence number of the routing table entry for \dval{dip},
	remains unchanged. Therefore the invariant is maintained.
\item[Pro.~\ref{pro:rreq}, Line~\ref{rreq:line6}:] 
        We assume that the entry $\xi(\oip,\osn,\kno,\val,\hops+1,\sip,*)$ is inserted into $\xi(\rt)$.
	So $\dval{dip}:=\xi(\oip)$, $\dval{nhip}:=\xi(\sip)$,
        $\nsq{\dval{ip}}:=\xi(\osn)$ and $\dhp{\dval{ip}}:=\xi(\hops)+1$.\linebreak[2]
        This information is distilled from a received
	route request message (cf.\ Lines~\ref{aodv:line2} and~\ref{aodv:line8}
	of Pro.~\ref{pro:aodv}).
	By \Prop{preliminaries} this message was sent before, say in state $N^\dagger$;
        by \Prop{ip=ipc} the sender of this message is $\xi(\sip)$.

	By Invariant~\eqref{inv:starcast_ii}, with $\ipc:=\xi(\sip)=\dval{nhip}$,
	~$\oipc:=\xi(\oip)=\dval{dip}$, ~$\osnc:=\xi(\osn)$~ and ~$\hopsc:=\xi(\hops)$,
        and using that $\ipc = \dval{nhip} \neq \dval{dip} = \oipc$, we get that
	$\dval{dip}\in \kd[N^\dagger]{\dval{nhip}}$ and
\begin{eqnarray*}
&\fnsqn_{N^\dagger}^{\dval{nhip}}(\dval{dip})~=~\fnsqn_{N^\dagger}^{\ipc}(\oipc) ~>~ \osnc ~=~ \xi(\osn)\ , \mbox{ or}\\
&\fnsqn_{N^\dagger}^{\dval{nhip}}(\dval{dip})~=~\xi(\osn) \ans
 \fnstatus_{N^\dagger}^{\dval{nhip}}(\dval{dip})~=~\val\ .
\end{eqnarray*}
We first assume that the first line holds.
Then, by \Thm{state_quality} and \Eq{sqn_vs_nsqn},
\[
\nsq{\dval{nhip}}
\geq
\fnnsqn_{N^\dagger}^{\dval{nhip}}(\dval{dip})
\geq
\fnsqn_{N^\dagger}^{\dval{nhip}}(\dval{dip})\decremented
\geq \xi(\osn)=\nsq{\dval{ip}}\ .
\]

We now assume the second line to be valid. 
From this we conclude
$$\nsq{\dval{nhip}}\geq
\fnnsqn_{N^\dagger}^{\dval{nhip}}(\dval{dip})
=\fnsqn_{N^\dagger}^{\dval{nhip}}(\dval{dip})
=\xi(\osn)=\nsq{\dval{ip}}.$$

\item[Pro.~\ref{pro:rrep}, Line~\ref{rrep:line5}:]
                  The update is similar to the one of Pro.~\ref{pro:rreq}, Line~\ref{rreq:line6}. 
                  The only difference is that the information stems from an incoming RREP message and 
                  that a routing table entry to $\xi(\dip)$ (instead of  $\xi(\oip)$) is established. 
                  Therefore, the proof is similar to the one of Pro.~\ref{pro:rreq}, Line~\ref{rreq:line6}; instead 
                  of Invariant~\eqref{inv:starcast_ii} we use Invariant~\eqref{inv:starcast_iv}.
\item[Pro.~\ref{pro:rerr}, Line~\ref{rerr:line5}:]
        Let $N_{\ref*{rerr:line5}}$ and $N$ be the network expressions right before
        and right after executing Pro.~\ref{pro:rerr}, Line~\ref{rerr:line5}.
	The entry for destination $\dval{dip}$ can be affected 
	only if $(\dval{dip},\dval{dsn})\in\xiN[N_{\ref*{rerr:line2}}]{\dval{ip}}(\dests)$ for some $\dval{dsn}\in\tSQN$.
	In that case, by Line~\ref{rerr:line2},
        \plat{$(\dval{dip},\dval{dsn})\in\xiN[N_{\ref*{rerr:line2}}]{\dval{ip}}(\dests)$},
	\plat{$\dval{dip}\in\akd[N_{\ref*{rerr:line2}}]{\dval{ip}}$}, and
        \plat{$\fnnhop_{N_{\ref*{rerr:line2}}}^{\dval{ip}}(\dval{dip})=\xiN[N_{\ref*{rerr:line2}}]{\dval{ip}}(\sip)$}.
By definition of \hyperlink{invalidate}{\fninv},
$\sq{\dval{ip}} = \dval{dsn}$
and $\sta{\dval{ip}}=\inval$, so
$$\nsq{\dval{ip}}=\sq{\dval{ip}}\decremented=\dval{dsn}\decremented\;.$$
Hence we need to show that $\dval{dsn}\decremented \leq \nsq{\dval{nhip}}$.

        The values $\xiN[N_{\ref*{rerr:line2}}]{\dval{ip}}(\dests)$ and
        $\xiN[N_{\ref*{rerr:line2}}]{\dval{ip}}(\sip)$ stem from a received route 
	error message (cf.\ Lines~\ref{aodv:line2} and~\ref{aodv:line16} of Pro.~\ref{pro:aodv}).
	By \Prop{preliminaries}\eqref{it:preliminariesi}, a transition
	labelled $\colonact{R}{\starcastP{\rerr{\destsc}{\ipc}}}$ with $\destsc:=\xiN[N_{\ref*{rerr:line2}}]{\dval{ip}}(\dests)$
	and $\ipc:=\xiN[N_{\ref*{rerr:line2}}]{\dval{ip}}(\sip)$ must have occurred before, say in state $N^\dagger$.  
	By \Prop{ip=ipc}, the node casting this message is
	\plat{$\ipc=
	\xiN[N_{\ref*{rerr:line2}}]{\dval{ip}}(\sip)=
	\fnnhop_{N_{\ref*{rerr:line2}}}^{\dval{ip}}(\dval{dip})=
	\nhp{\dval{ip}}=\dval{nhip}$}.
	The penultimate equation holds since the next hop to $\dval{dip}$
	is not changed during the execution of Pro.~\ref{pro:rerr}.

By \Prop{starcastrerr} we have 
\plat{$\dval{dip}\in\fnikD_{N^\dagger}^{\dval{nhip}}$} and 
\plat{$\dval{dsn} \leq \sqn{\xiN[N^\dagger]{\dval{nhip}}(\rt)}{\dval{dip}}$}. 
Hence\vspace{-1ex}
$$\nsq{\dval{nhip}}\geq
\fnnsqn_{N^\dagger}^{\dval{nhip}}(\dval{dip}) =
\nsqn{\xiN[N^\dagger]{\dval{nhip}}(\rt)}{\dval{dip}} =
\sqn{\xiN[N^\dagger]{\dval{nhip}}(\rt)}{\dval{dip}}\decremented \geq
\dval{dsn}\decremented\ ,$$ where the first inequality follows by \Thm{state_quality}.
\endbox
\end{description}
\end{proofNobox}

\noindent
To prove loop freedom we will show that on any route established by AODV the quality 
of routing tables increases when going from one node to the next hop. Here, the 
preorder is not sufficient, since we need a strict increase in quality. 
Therefore, on routing tables $\dval{rt}$ and $\dval{rt}'$ that both have an entry to $\dval{dip}$, i.e., $\dval{dip}\in\kD{\dval{rt}}\cap\kD{\dval{rt}'}$,
 we define a relation $\rtsord$ by
\[
\dval{rt} \rtsord \dval{rt}'\ :\Leftrightarrow\  \dval{rt} \rtord \dval{rt}' \ans \dval{rt} \not\rtequiv \dval{rt}'\ .
\]

\begin{cor}\label{cor:strictord}\rm
The relation $\rtsord$ is irreflexive and transitive. 
\end{cor}

\begin{theorem}\rm
\label{thm:inv_a}
The quality of the routing table entries for a destination \dval{dip} is strictly increasing
along a route towards \dval{dip},
 until it reaches either \dval{dip} or a node with an invalid routing table entry to \dval{dip}.
\begin{equation}\label{eq:inv_x}
\dval{dip}\in\akd{\dval{ip}}\cap \akd{\dval{nhip}} \ans\dval{nhip}\not=\dval{dip}
\ims \xiN{\dval{ip}}(\rt)\rtsord \xiN{\dval{nhip}}(\rt)\ ,
\end{equation}
where $N$ is a reachable network expression and $\dval{nhip}:=\nhp{\dval{ip}}$ is the IP address of the next hop.
\end{theorem}

\begin{proofNobox}
As before, we first check the initial states of our transition system
and then check all locations in Processes~\ref{pro:aodv}--\ref{pro:queues}
where a routing table might be changed. For an initial network
expression, the invariant holds since all routing tables are empty.
Adding precursors to $\xiN{\dval{ip}}(\rt)$ or $\xiN{\dval{nhip}}(\rt)$
does not affect the invariant, since the invariant does not depend on
precursors, so it suffices to examine all modifications to $\xiN{\dval{ip}}(\rt)$
or $\xiN{\dval{nhip}}(\rt)$ using \hyperlink{update}{$\fnupd$} or
\hyperlink{invalidate}{$\fninv$}. Moreover, without loss of generality we restrict
attention to those applications of $\fnupd$ or $\fninv$ that actually
modify the entry for \dval{dip}, beyond its precursors; if $\fnupd$
only adds some precursors in the routing table, the invariant---which
is assumed to hold before---is maintained. 

Applications of {\fninv} to either $\xiN{\dval{ip}}(\rt)$ or
$\xiN{\dval{nhip}}(\rt)$ lead to a network state in which the
antecedent of \Eq{inv_x} is not satisfied.
Now consider an application of $\fnupd$ to \plat{$\xiN{\dval{nhip}}(\rt)$.}
We restrict attention to the case that the antecedent of \Eq{inv_x} is satisfied right after the
update, so that right before the update we have 
$\dval{dip}\in\akd{\dval{ip}} \wedge \dval{nhip}\not=\dval{dip}$.
In the special case that $\sq{\dval{nhip}}=0$ right before the update, we have
$\nsq{\dval{nhip}}=0$ and thus $\nsq{\dval{ip}}=0$ by Invariant~\Eq{inv_ix}.
Since \plat{$\sta{\dval{ip}}=\val$}, this implies $\sq{\dval{ip}}=0$.
By \Prop{rte}(\ref{it_e}) we have $\dval{nhip}=\dval{dip}$,
contradicting our assumptions. It follows that right before the update $\sq{\dval{nhip}}>0$, and
hence $\nsq{\dval{nhip}}<\sq{\dval{nhip}}$.

An application of $\fnupd$ to \plat{$\xiN{\dval{nhip}}(\rt)$}
that changes \plat{$\sta{\dval{nhip}}$} from {\inval} to {\val} cannot
decrease the sequence number of the entry to \dval{dip} and hence
strictly increases its net sequence number.
Before the $\fnupd$ we had
$\nsq{\dval{ip}}\leq \nsq{\dval{nhip}}$ by Invariant~\eqref{eq:inv_ix},
so afterwards we must have
\mbox{$\nsq{\dval{ip}} < \nsq{\dval{nhip}}$}, and hence 
$\xiN{\dval{ip}}(\rt)\rtsord \xiN{\dval{nhip}}(\rt)$.
An $\fnupd$ to $\xiN{\dval{nhip}}(\rt)$ that maintains
$\sta{\dval{nhip}}=\val$ can only increase the quality of the entry to
\dval{dip} (cf.\ \Thm{state_quality}), and hence maintains Invariant \Eq{inv_x}.

It remains to examine the $\fnupd$s to $\xiN{\dval{ip}}(\rt)$.
\begin{description}
\item[Pro.~\ref{pro:aodv}, Lines~\ref{aodv:line10}, \ref{aodv:line14}, \ref{aodv:line18}:]
	The entry $\xi(\sip\comma0\comma\unkno\comma\val\comma1\comma\sip\comma\emptyset)$ is used for the update; 
	its destination is $\dval{dip}:=\xi(\sip)$.
	Since $\dval{dip}=\nhp[\dval{dip}]{\dval{ip}}=\dval{nhip}$,
	the antecedent of the invariant to be proven is not satisfied.
\item[Pro.~\ref{pro:rreq}, Line~\ref{rreq:line6}:]\hypertarget{729Pro3Line4}{ }\label{pg:729Pro3Line4}
        We assume that  the entry $\xi(\oip,\osn,\kno,\val,\hops+1,\sip,*)$ is inserted into $\xi(\rt)$.
	So $\dval{dip}:=\xi(\oip)$, $\dval{nhip}:=\xi(\sip)$,
        $\nsq{\dval{ip}}:=\xi(\osn)$ and $\dhp{\dval{ip}}:=\xi(\hops)+1$.\linebreak[2]
        This information is distilled from a received
	route request message (cf.\ Lines~\ref{aodv:line2} and~\ref{aodv:line8}
	of Pro.~\ref{pro:aodv}).
	By \Prop{preliminaries} this message was sent before, say in state $N^\dagger$;
        by \Prop{ip=ipc} the sender of this message is $\xi(\sip)$.

	By Invariant~\eqref{inv:starcast_ii}, with $\ipc:=\xi(\sip)=\dval{nhip}$,
	~$\oipc:=\xi(\oip)=\dval{dip}$, ~$\osnc:=\xi(\osn)$~ and ~$\hopsc:=\xi(\hops)$,
        and using that $\ipc = \dval{nhip} \neq \dval{dip} = \oipc$, we get that
\begin{eqnarray*}
&\fnsqn_{N^\dagger}^{\dval{nhip}}(\dval{dip})~=~\fnsqn_{N^\dagger}^{\ipc}(\oipc) ~>~ \osnc ~=~ \xi(\osn)\ , \mbox{ or}\\
&\fnsqn_{N^\dagger}^{\dval{nhip}}(\dval{dip})~=~\xi(\osn) \ans
 \fndhops_{N^\dagger}^{\dval{nhip}}(\dval{dip})\leq\xi(\hops) \ans
 \fnstatus_{N^\dagger}^{\dval{nhip}}(\dval{dip})~=~\val\ .
\end{eqnarray*}
We first assume that the first line holds.
Then, by the assumption $\dval{dip}\in\akD{\xiN{\dval{nhip}}(\rt)}$,
the definition of net sequence numbers, and \Prop{dsn increase},
\vspace{-.5ex}
\[
\nsq{\dval{nhip}}=\sq{\dval{nhip}} \geq
\fnsqn_{N^\dagger}^{\dval{nhip}}(\dval{dip})
>\xi(\osn)=\nsq{\dval{ip}}\ .
\vspace{-.5ex}
\]
and hence $\xiN{\dval{ip}}(\rt)\rtsord \xiN{\dval{nhip}}(\rt)$.

We now assume the second line to be valid. 
From this we conclude
\vspace{-.5ex}
\[ \fnnsqn_{N^\dagger}^{\dval{nhip}}(\dval{dip})
  =\fnsqn_{N^\dagger}^{\dval{nhip}}(\dval{dip})
  =\xi(\osn)
  =\nsq{\dval{ip}}\ .
\vspace{-.5ex}
\]

Moreover,\hfill $\fndhops_{N^\dagger}^{\dval{nhip}}(\dval{dip})~\leq~\xi(\hops) ~<~ \xi(\hops)+1~=~\dhp{\dval{ip}}$~.\hfill\hfill\mbox{}\\
Hence $\xiN{\dval{ip}}(\rt)\rtsord \xiN[N^\dagger]{\dval{nhip}}(\rt)$.
Together with \Thm{state_quality} and the transitivity of $\rtord$\vspace{-2pt}
this yields $\xiN{\dval{ip}}(\rt)\rtsord \xiN{\dval{nhip}}(\rt)$.
	
\item[Pro.~\ref{pro:rrep}, Line~\ref{rrep:line5}:]
                  The update is similar to the one of Pro.~\ref{pro:rreq}, Line~\ref{rreq:line6}. 
                  The only difference is that the information stems from an incoming RREP message and 
                  that a routing table entry to $\xi(\dip)$ (instead of  $\xi(\oip)$) is established. 
                  Therefore, the proof is similar to the one of Pro.~\ref{pro:rreq}, Line~\ref{rreq:line6}; instead 
                  of Invariant~\eqref{inv:starcast_ii} we use Invariant~\eqref{inv:starcast_iv}.
\endbox
\end{description}
\end{proofNobox}

\subsection{Loop Freedom}\label{ssec:loop freedom}
\index{loop freedom}%
The ``na\"ive'' notion of loop freedom is a term that informally means
that ``a packet never goes round in cycles without (at some point)
being delivered".  This dynamic definition is not only hard to formalise, 
it is also too restrictive a requirement for AODV\@. There are situations where 
packets are sent in cycles, but which are not considered harmful. 
This can happen when the topology keeps changing.%

\subsubsection*{Loops within a Topology that Changes Forever}
\begin{wrapfigure}[10]{r}{0.35\textwidth}
 \vspace{-8.5ex}
\centering
\includegraphics[scale=1.3]{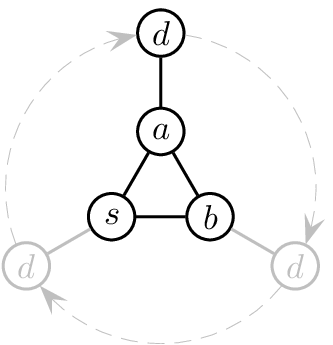}
\caption{Loop considered harmless}\label{fig:good loop}
\vspace*{-.5ex}
\end{wrapfigure}
The following example shows that data packets can travel in cycles 
without being delivered. However, it is our belief that this 
example is not a loop that should be avoided by a good routing protocol. 

The example consists of
a ``core'' network built up by the nodes $s$, $a$ and $b$, which form a ring topology. 
All links between these three nodes are stable. Node $d$ is part of the network and keeps moving around the 
core such that it is always connected to only one of the nodes at a time; see \Fig{good loop}.
In the initial state $d$ is connected to $a$ and node $s$
wants to send a data packet to $d$.

Since $s$ does not have a routing table entry to $d$, it generates and broadcasts a
RREQ message, which is received by $d$ via node $a$ (\Fig{example_dynamic_loop} (a)).\footnote{The
``snapshots'' in this figure are slightly different from the ones presented before;
in Figures~\ref{fig:example1} and \ref{fig:example2} (as well as in \ref{fig:loopdec}),
each snapshot presents the system in a state after an AODV control
message or data packet has been received and already partly handled
(e.g., the routing tables are updated). Here, the subfigures describe
the system when each message has either been handled completely or has been received and stored in
the buffer, but {\em not} yet handled.} 
In \Fig{example_dynamic_loop}(b), $d$ sends a RREP message back to $s$ (via $a$). 

Since $s$ now has a routing table entry for $d$, it sends the data packet to
 $a$---the next hop on the route to $d$ (\Fig{example_dynamic_loop}(c)).
In the meantime, node $d$ has moved away from node $a$, and is now connected to node $b$. 
In \Fig{example_dynamic_loop}(d), node $a$ detects the link break
(e.g.\ while trying to send the data packet from node
$s$ to node $d$), and proceeds to do a local repair.\footnote{Even
  though we do not model the local repair feature, we use it here to
  illustrate scenarios where data packets can travel in
  cycles. It is easy to modify the presented example into one without
  local repair; however the modified example would require error handling and hence would
  be longer.} The data packet is buffered while waiting for the local repair
  process to complete. To repair the link break, node $a$ generates a new RREQ message,
  which is received by node $d$ via node $b$.

In \Fig{example_dynamic_loop}(e), node $d$ sends a RREP message back
to node $a$ (via node $b$), thus enabling node $a$ to repair
its routing table entry to node $d$.\footnote{We
  simplify the description of the local repair process here. Further
  details are available in the RFC~\cite{rfc3561}.}

With a valid entry in its routing table for node $d$, node $a$ can now send the buffered data packet to node $b$---the next hop on the route towards node $d$ (\Fig{example_dynamic_loop}(f)). If node $d$ now moves away from node $b$ and into the transmission range of node $s$, the events of Parts (d)--(f) will repeat. This will continue as long as the destination node $d$ keeps moving ``around" nodes $s$, $a$ and $b$. The data packet will then travel through a loop $a$--$b$--$s$--$a$. Though this is a loop, it is not undesirable behaviour since the data packet is always travelling on the shortest path towards node $d$; it is due to the movement of node $d$ that the data packet is never delivered.
\endbox
{
\setlength{\medFigShift}{32.2pt}
\setlength{\shortmedFigShift}{.1pt}
\begin{exampleFig}{A ``dynamic loop''}{fig:example_dynamic_loop}
	\FigLine[slsr]%
	{$s$ broadcasts a new RREQ message destined to $d$.}{fig/ex_loop_dynamic4nodesii_2}{}
	{$d$ updates its RT and unicasts a RREP.}{fig/ex_loop_dynamic4nodesii_3}{}
	\FigLine[xslxsr]%
	{The topology changes;\\$s$ sends the data packet to $a$.}
	{fig/ex_loop_dynamic4nodesii_4}{}
	{$a$ detects the link break;\\ it initiates new RREQ (local repair).}{fig/ex_loop_dynamic4nodesii_5}{}
	\FigLine[xslxsr]%
	{$d$ updates its RT and unicasts a RREP back to $a$.}{fig/ex_loop_dynamic4nodesii_6}{}
	{$a$ forwards data packet to $b$;\\the topology changes.}{fig/ex_loop_dynamic4nodesii_7}{}
\end{exampleFig}
}
\noindent
Due to this dynamic behaviour, the sense of loop freedom is much
better captured by a static invariant,
saying that at any given
time the collective routing tables of the nodes do not admit a loop.
Such a requirement does not rule out the dynamic loop exemplified
above. However, in situations where the topology remains stable sufficiently long
it does guarantee that packets will not keep going around in cycles.
In the above example the packet would actually be delivered as soon as the topology 
stops changing---it does not matter when.

\newcommand{\RG}[2]{\mathcal{R}_{#1}(#2)}
To this end we define the \phrase{routing graph} of network expression $N$ with respect to
destination~$\dval{dip}$ by $\RG{N}{\dval{dip}}\mathop{:=}\linebreak[1](\IP,E)$, where
all nodes of the network form the set of vertices and there is an
arc $({\dval{ip}},{\dval{ip}}')\in E$ iff $\dval{ip}\mathop{\not=}\dval{dip}$ and
$
(\dval{dip}\comma\nosp{*}\comma\nosp{*}\comma\nosp{\val}\comma\nosp{*}\comma\nosp{\dval{ip}'}\comma\nosp{*})\mathop{\in}\xiN{\dval{ip}}(\rt).
$

An arc in a routing graph states that $\dval{ip}'$ is the next hop on
a valid route to $\dval{dip}$ known by $\dval{ip}$; a path in a routing
graph describes a route towards $\dval{dip}$ discovered by AODV\@.
\index{loop freedom}%
We say that a network expression $N$ is \emph{loop free} if the
corresponding routing graphs $\RG{N}{\dval{dip}}$ are loop free, for
all $\dval{dip}\mathop{\in}\IP$. A routing protocol, such as AODV, is
\emph{loop free} iff all reachable network expressions are loop free.

Using this definition of a routing graph, \Thm{inv_a} states that 
along a path towards a destination \dval{dip} in the routing
graph of a reachable network expression $N$, until it reaches either
\dval{dip} or a node with an invalided routing table entry to dip,
the quality of the routing table entries for \dval{dip} is strictly increasing.
From this, we can immediately conclude
\begin{theorem}\rm\label{thm:loop free}
The specification of AODV given in \Sect{modelling_AODV} is loop free.
\end{theorem}
\begin{proof}
If there were a loop in a routing graph $\RG{N}{\dval{dip}}$, then for
any edge $(\dval{ip},\dval{nhip})$ on that loop one has
$\xiN{\dval{ip}}(\rt)\rtsord\xiN{\dval{nhip}}(\rt)$, by \Thm{inv_a}. 
Thus, by transitivity of $\rtsord$, one has
$\xiN{\dval{ip}}(\rt)\rtsord\xiN{\dval{ip}}(\rt)$, which
contradicts the irreflexivity of {\rtsord} (cf.\ \Cor{strictord}).
\end{proof}

\hypertarget{end}{According to \Thm{loop free} any route to a destination \dval{dip}
established by AODV---i.e.\ a path in $\RG{N}{\dval{dip}}$---ends after finitely many
hops. There are three possible ways in which it could end:
\begin{enumerate}[(1)]
\item by reaching the destination,\label{eq:success}
\item by reaching a node with an invalid entry to \dval{dip}, or\label{eq:invalid}
\item by reaching a node without any entry to \dval{dip}.\label{eq:failure}
\end{enumerate}
\Eq{success} is what AODV attempts to accomplish, whereas
\Eq{invalid} is an unavoidable
due to link breaks in a dynamic topology. It follows directly from \Prop{inv_nsqn} that \Eq{failure} can never occur.}

\subsection{Route Correctness}\label{ssec:route correctness}

\newcommand{\CG}[1]{\mathcal{C}_{#1}}
The creation of a routing table entry at node \dval{ip} for destination \dval{dip} is no guarantee
that a route from \dval{ip} to \dval{dip} actually exists. The entry is created based on information
gathered from messages received in the past, and at any time link breaks may occur. The best one could require
of a protocol like AODV is that routing table entries are based on information that was valid at some
point in the past. This is the essence of what we call \phrase{route correctness}.

We define a \phrase{history} of an AODV-like protocol as a sequence $H=N_0 N_1 \ldots N_k$ of network expressions,
where $N_0$ is an initial state of the protocol, and for $1\leq i\leq k$ there is a transition $N_{i-1}\ar{\ell}N_i$;
we call $H$ a history \emph{of} the state $N_k$.
The \phrase{connectivity graph} of a history $H$ is $\CG{H}\mathop{:=}(\IP,E)$, where
the nodes of the network form the set of vertices and there is an
arc $({\dval{ip}},{\dval{ip}}')\in E$ iff ${\dval{ip}}'\in \RN[N_i]{\dval{ip}}$ for some $0\leq i
\leq k$, i.e.\ if at some point during that history node $\dval{ip}'$ was in transmission range of \dval{ip}.
A protocol satisfies the property \phrase{route correctness} if 
for every history $H$ of a reachable state $N$
and for every routing table entry $(\dval{dip}\comma\nosp{*}\comma\nosp{*}\comma\nosp{*}\comma\nosp{\dval{hops}}\comma\nosp{\dval{nhip}}\comma\nosp{*})\mathop{\in}\xiN[N]{\dval{ip}}(\rt)$
there is a path $\dval{ip}\rightarrow\dval{nhip}\rightarrow\cdots\rightarrow\dval{dip}$
in $\CG{H}$ from $\dval{ip}$ to $\dval{dip}$ with \dval{hops} hops and (if $\dval{hops}>0$) next hop \dval{nhip}.%
\footnote{A path with $0$ hops consists of a single node only.}

\begin{theorem}\rm\label{thm:route correctness}
Let $H$ be a history of a network state $N$.
\begin{enumerate}[(a)]
\item For each routing table entry $(\dval{dip},*,*,*,\dval{hops},\dval{nhip},*)\mathop{\in}\xiN[N]{\dval{ip}}(\rt)$
there is a path $\dval{ip}\rightarrow\dval{nhip}\rightarrow\cdots\rightarrow\dval{dip}$
in $\CG{H}$ from $\dval{ip}$ to $\dval{dip}$ with \dval{hops} hops and (if $\dval{hops}>0$) next hop \dval{nhip}.
\item For each route request sent in state $N$ there is a corresponding path in the connectivity graph of $H$.
	\begin{equation}\label{eq:rcreq}
	\begin{array}{rcl}
	  &&N\ar{R:\starcastP{\rreq{\hopsc}{*}{*}{*}{*}{\oipc}{*}{\ipc}}}_{\dval{ip}}N'\\
	  &\Rightarrow&
          \mbox{there is a path $\ipc\rightarrow\cdots\rightarrow\oipc$ in $\CG{H}$ from $\ipc$ to $\oipc$ with $\hopsc$ hops}
	\end{array}
	\end{equation}
\item For each route reply sent in state $N$ there is a corresponding path in the connectivity graph of $H$.
 	\begin{equation}\label{eq:rcrep}
 	\begin{array}{@{}rcl@{}}
 	  &&N\ar{R:\starcastP{\rrep{\hopsc}{\dipc}{*}{*}{\ipc}}}_{\dval{ip}}N'\\
 	  &\Rightarrow&
          \mbox{there is a path $\ipc\rightarrow\cdots\rightarrow\dipc$ in $\CG{H}$ from $\ipc$ to $\dipc$ with $\hopsc$ hops}
 	\end{array}
 	\end{equation}

\end{enumerate}
\end{theorem}

\begin{proofNobox}
In the course of running the protocol, the set of edges $E$ in the connectivity graph $\CG{H}$ only increases,
so the properties are invariants. We prove them by simultaneous induction.
\begin{enumerate}[(a)]
\item In an initial state the invariant is satisfied because the routing tables are empty.
  Since routing table entries can never be removed, and the functions $\fnaddprecrt$ and $\fninv$ do
  not affect $\dval{hops}$ and $\dval{nhip}$, it suffices to check all application calls of \hyperlink{update}{$\fnupd$}.
  In each case, if the update does not change the routing table entry beyond its precursors
  (the last clause of \hyperlink{update}{\fnupd}), the invariant is trivially
  preserved; hence we examine the cases that an update actually occurs.
\begin{description}
\item[Pro.~\ref{pro:aodv}, Lines~\ref{aodv:line10}, \ref{aodv:line14}, \ref{aodv:line18}:]
The update changes the entry into $\xi(\sip,*,\unkno,\val,1,\sip,*)$; hence \mbox{$\dval{hops}\mathbin=1$} and
$\dval{nhip}=\dval{dip}:=\xi(\sip)$. The value $\xi(\sip)$ stems through
Lines~\ref{aodv:line8},~\ref{aodv:line12} or~\ref{aodv:line16} of Pro.~\ref{pro:aodv} from an
incoming AODV control message.	By \Prop{preliminaries} this message was sent before, say in state $N^\dagger$;
by \Prop{ip=ipc} the sender of this message is $\xi(\sip)=\dval{nhip}$. Since in state $N^\dagger$
the message must have reached the queue of incoming messages of node \dval{ip}, it must be that
\plat{${\dval{ip}}\mathbin\in \RN[N^\dagger]{\dval{nhip}}$}.
In our formalisation of \awn the connectivity graph is always symmetric: ${\dval{nhip}}\mathbin\in
\RN[N^\dagger]{\dval{ip}}$ iff \plat{${\dval{ip}}\mathbin\in \RN[N^\dagger]{\dval{nhip}}$}.
It follows that $(\dval{ip},\dval{nhip})\in E$, so there is a 1-hop path in $\CG{H}$ from $\dval{ip}$ to $\dval{dip}$.
\item[Pro.~\ref{pro:rreq}, Line~\ref{rreq:line6}:] 
Here $\dval{dip}:=\xi(\oip)$, $\dval{hops}:=\xi(\hops)\mathord+1$ and $\dval{nhip}:=\xi(\sip)$. 
These values stem from an incoming RREQ message, which must have been sent beforehand, say in state $N^\dagger$.
As in the previous case we obtain $(\dval{ip},\dval{nhip})\in E$.
By Invariant~\Eq{rcreq}, with $\oipc:=\xi(\oip)=\dval{dip}$, $\hopsc:=\xi(\hops)$ and $\ipc:=\xi(\sip)=\dval{nhip}$, 
there is a path $\dval{nhip}\rightarrow\cdots\rightarrow\dval{dip}$ in $\CG{H}$ from $\ipc$ to $\oipc$ with $\hopsc$ hops.
It follows that there is a path $\dval{ip}\rightarrow\dval{nhip}\rightarrow\cdots\rightarrow\dval{dip}$
in $\CG{H}$ from $\dval{ip}$ to $\dval{dip}$ with \dval{hops} hops and next hop \dval{nhip}.
\item[Pro.~\ref{pro:rrep}, Line~\ref{rrep:line5}:]
Here $\dval{dip}:=\xi(\dip)$, $\dval{hops}:=\xi(\hops)\mathord+1$ and $\dval{nhip}:=\xi(\sip)$. 
The reasoning is exactly as in the previous case, except that we deal with an incoming RREP message
and use Invariant~\Eq{rcrep}.
\end{description}

\item We check all occasions where a route request is sent.
\begin{description}
	\item[Pro.~\ref{pro:aodv}, Line~\ref{aodv:line39}:]
		A new route request is initiated with $\ipc=\oipc:=\xi(\ip)=\dval{ip}$
                and $\hopsc:=0$.
                Indeed there is a path in $\CG{H}$ from $\ipc$ to $\oipc$ with $0$ hops.
	\item[Pro.~\ref{pro:rreq}, Line~\ref{rreq:line34}:]
                The broadcast message has the form
                \[\xi(\rreq{\hops\mathord+1}{\rreqid}{\dip}{\max(\sqn{\rt}{\dip},\dsn)}{\dsk}{\oip}{\osn}{\ip})\ .\]
                Hence $\hopsc:=\xi(\hops)\mathord+1$, $\oipc:=\xi(\oip)$ and $\ipc:=\xi(\ip)=\dval{ip}$.
                The values $\xi(\hops)$ and $\xi(\oip)$ stem through Line~\ref{aodv:line8} of Pro.~\ref{pro:aodv} from an
                incoming RREQ message of the form
                \[\xi(\rreq{\hops}{\rreqid}{\dip}{\dsn}{\dsk}{\oip}{\osn}{\sip})\ .\]
                By \Prop{preliminaries} this message was sent before, say in state $N^\dagger$;
                by \Prop{ip=ipc} the sender of this message is $\dval{sip}:=\xi(\sip)$. 
                By induction, using Invariant~\eqref{eq:rcreq}, there is a path $\dval{sip}\rightarrow\cdots\rightarrow\oipc$ in
                $\CG{H^{\dagger}} \subseteq \CG{H}$ from $\dval{sip}$ to $\oipc$ with $\xi(\hops)$ hops.
                It remains to show that there is a $1$-hop path from $\dval{ip}$ to
                $\dval{sip}$. In state $N^\dagger$ the message sent by $\dval{sip}$ must have
                reached the queue of incoming messages of node \dval{ip}, and therefore $\dval{ip}$
                was in transmission range of $\dval{sip}$, i.e., \plat{${\dval{ip}}\mathbin\in \RN[N^\dagger]{\dval{sip}}$}.
                Since the connectivity graph of \awn is always symmetric (cf.\ Tables~\ref{tab:sos node}
                and \ref{tab:sos network}, and explanation on Page~\pageref{pg:sym}),
                \plat{${\dval{ip}}\mathbin\in \RN[N^\dagger]{\dval{sip}}$} holds as well. Hence it follows that $(\dval{ip},\dval{sip})\in E$.
\end{description}
\item We check all occasions where a route reply is sent.
\begin{description}
	\item[Pro.~\ref{pro:rreq}, Line~\ref {rreq:line14}:]
		A new route reply with $\hopsc:=0$ and
		$\ipc:=\xi(\ip)=\dval{ip}$ is initiated.
		Moreover, by Line~\ref{rreq:line10}, $\dipc:=\xi(\dip)=\xi(\ip)=\dval{ip}$.
                Thus there is a path in $\CG{H}$ from $\ipc$ to $\dipc$ with $0$ hops.
	\item[Pro.~\ref{pro:rreq}, Line~\ref{rreq:line26}:]
                We have $\dipc:=\xi(\dip)$, $\hopsc:=\dhp[\dipc]{\dval{ip}}$ and $\ipc:=\xi(\ip)=\dval{ip}$.
		By Line~\ref{rreq:line22} there is a routing table entry
                \plat{$(\dipc,*,*,*,\hopsc,*,*)\mathop{\in}\xiN[N]{\dval{ip}}(\rt)$}.
                Hence by Invariant~(a), which we may assume to hold when using simultaneous
                induction, there is a path $\dval{ip}\rightarrow\cdots\rightarrow\dipc$
                in $\CG{H}$ from $\dval{ip}=\ipc$ to $\dipc$ with $\hopsc$ hops.
	\item[Pro.~\ref{pro:rrep}, Line~\ref{rrep:line13}:]
		The RREP message has the form
		$\xi(\rrep{\hops\mathop{+}1}{\dip}{\dsn}{\oip}{\ip})$
                and the proof goes exactly as for Pro.~\ref{pro:rreq}, Line~\ref{rreq:line34} of Part (b),
                by using $\dipc:=\xi(\dip)$ instead of $\oipc:=\xi(\oip)$, and an incoming RREP
                message instead of an incoming RREQ message.
\endbox
\end{description}
\end{enumerate}
\end{proofNobox}
\Thm{route correctness}(a) says that the AODV protocol is route correct.
For the proof it is essential that we use the version of \awn were a node $\dval{ip}'$ is in
the range of node $\dval{ip}$, meaning that $\dval{ip}'$ can receive
messages sent by $\dval{ip}$, if and only if $\dval{ip}$ is in the range of $\dval{ip}'$.
If \awn is modified so as to allow asymmetric connectivity graphs, as indicated in \SSect{networks},
it is trivial to construct a 2-node counterexample to route correctness.

A stronger concept of route correctness requires that for each
$(\dval{dip},*,*,*,\dval{hops},\dval{nhip},*)\mathop{\in}\xiN{\dval{ip}}(\rt)$
\begin{itemize}\vspace{-0.5ex}
\item either $\dval{hops}=0$ and $\dval{dip}=\dval{ip}$,\vspace{-1ex}
\item or $\dval{hops}=1$ and $\dval{dip}=\dval{nhip}$ and there is a $N^\dagger$ in $H$ such that
  ${\dval{nhip}}\mathbin\in \RN[N^\dagger]{\dval{ip}}$,\vspace{-1ex}
\item or $\dval{hops}\mathbin>1$ and there is a $N^\dagger$ in $H$ with \plat{${\dval{nhip}}\mathbin\in \RN[N^\dagger]{\dval{ip}}$}
and \plat{$(\dval{dip},*,*,\val,\dval{hops}\mathord-1,*,*)\mathbin\in\xiN[N^\dagger]{\dval{nhip}}(\rt)$}.\vspace{-0.5ex}
\end{itemize}
It turns out that this stronger form of route correctness does not hold for AODV\@.

\subsection{Further Properties}
We conclude this section by proving a few more properties of AODV; these will be used
  later in the paper and/or shed some light on how AODV operates.
\subsubsection{Queues}
\begin{prop}\rm A node $\dval{ip}\in\IP$ never queues data packets intended for itself.
\begin{equation}\label{eq:inv_v}
\dval{ip}\not\in\qD{\xiN{\dval{ip}}(\queues)}
\end{equation}
\end{prop}
\begin{proof} We first show the claim for the initial states;
afterwards we go through our specification (step by step) and look at all locations
where the store of an arbitrary node $\dval{ip}\in\IP$ can be changed.

In an initial network expression all sets of queued data are empty.
There is only one place where a new destination is added to \queues,
namely Pro.~\ref{pro:newpkt}, Line~\ref{newpkt:line5}.  Here,
$\xi(\dip)$ is added as new queued destination. However,
Line~\ref{newpkt:line4} shows that $\xi(\dip)\not=\xi(\ip)$.
\end{proof}

\subsubsection{Route Requests and RREQ IDs}
A transition $N\ar{R:\starcastP{\rreq{*}{\dval{rreqid}}{\dval{dip}}{*}{*}{\dval{oip}}{\dval{osn}}{*}}}N'$
that stems from \Pro{aodv}, Line~\ref{aodv:line39} marks the
initiation of a \emph{new} route request. Each such transition that
stems from \Pro{rreq}, Line~\ref{rreq:line34}, which
is the only alternative, marks the \phrase{forwarding} of a route request.
In this case, the variables {\rreqid}, {\dip}, {\oip} and {\osn},
which supply the values \dval{rreqid}, \dval{dip}, \dval{oip} and
\dval{osn}, get these values in \Pro{aodv}, Line~\ref{aodv:line8};
nowhere else is the value of these variables set or changed.
Hence the values mentioned are copied directly from another RREQ
message, read in \Pro{aodv}, Line~\ref{aodv:line2}.
By \Prop{preliminaries}(\ref{before}), this message has to be sent
before; and this is the message that is forwarded.
\index{route request (RREQ)}%
Now a \emph{route request} can be defined as an equivalence class of
route request messages (transitions in our operational semantics),
namely by considering a forwarded RREQ message to belong to the same
route request as the message being forwarded.

\begin{prop}\label{prop:invarianti_itemi}\rm
A route request is uniquely determined by the pair
$(\dval{oip}, \dval{rreqid})$ of the originator IP address and its
route request identifier.\end{prop}

\begin{proof}
  As argued above, each forwarded RREQ message carries the same
  pair $(\dval{oip}, \dval{rreqid})$ as the message being forwarded.
  It remains to show that each new route request is initiated with a
  different pair $(\dval{oip}, \dval{rreqid})$.

The broadcast message id is determined by the function \hyperlink{nrreqid}{$\fnnrreqid$}.
At the initial state the function $\fnnrreqid$ will return $1$, since $\rreqs(\ip)$ is empty.
If a new id---determined by the function $\fnnrreqid$---is used by a node \dval{ip},
the id is also added to $\xi^\dval{ip}_N(\rreqs)$ (Pro.~\ref{pro:aodv}, Line~\ref{aodv:line38b}).
By \Prop{rreqs increase}, this id will never be deleted from $\xi(\rreqs)$.
Therefore, whenever the function $\fnnrreqid$ is called afterwards by the same node, the
return value will be strictly higher. In fact it will be increased by $1$ each time a new request is sent.
It follows that for each route request the pair $(\dval{oip}, \dval{rreqid})$ is unique.
\end{proof}
This pair $(\dval{oip}, \dval{rreqid})$ is stored in the local variables
$\rreqs$ maintained by each node that encounters the route request.

The following proposition paves the way for the conclusion that the role of the
  component $\dval{rreqid}$ in route request messages could just as well be taken over by the
  existing component $\dval{osn}$ of these messages.

\begin{prop}\label{prop:messagebroadcast}\rm~
\begin{enumerate}[($a$)]
\item A node's sequence number is greater than  its last used RREQ id, i.e.,
\[
\xiN{\dval{ip}}(\sn)  > \rreqid_{N}^{\dval{ip}}\ ,
\]
where
$\rreqid_{N}^{\dval{ip}}:=\max\{n\mid(\dval{ip},n)\in\xiN{\dval{ip}}(\rreqs)\}$and the maximum of the empty set is defined to be $0$.
\item A route request is uniquely determined by
the combination of $\dval{osn}$ and $\dval{oip}$.
\end{enumerate}
\end{prop}

\begin{proofNobox} ~
\begin{enumerate}[($a$)]
\item In the initial state $\xiN{\dval{ip}}(\sn)=1$ and $\rreqid_{N}^{\dval{ip}}=0$.
Both numbers are increased by $1$ if a route request is initiated; the {\sn} is increased first.
$\rreqid_{N}^{\dval{ip}}$ is not changed elsewhere; however, when generating
a route reply $\xiN{\dval{ip}}(\sn)$ might be increased
(cf.\ Pro.~\ref{pro:rreq}, Line~\ref{rreq:line12}).
\item When a route request is initiated, the value of the component $\dval{osn}$ in
  the initial RREQ message equals the (newly incremented) current value of $\sn$
  maintained by the initiating node, just like the component $\dval{rreqid}$ in
  the initial RREQ message equals the (newly incremented) current value of $\rreqid_{N}^{\dval{ip}}$
  of the initiating node.
  Now the statement follows
  since the value of $\sn$ is increased whenever a route request
  is initiated
  and $\dval{osn}$ and $\dval{oip}$ are passed on unchanged when
  forwarding a route request, just like $\dval{rreqid}$ and $\dval{oip}$.
  \hfill\mbox{\endbox}
\end{enumerate}
\end{proofNobox}

\noindent The following proposition states three properties about sending a route request. 
\begin{prop}\label{prop:starcast}\rm~
\begin{enumerate}[(a)]
\item\label{it:starcastii} If a route request is sent by a node $\ipc\in\IP$, the sender has stored the unique pair
  of the originator's IP address and the request id.
\begin{equation}\label{inv:starcast_iii}
	N\ar{R:\starcastP{\rreq{*}{\rreqidc}{*}{*}{*}{\oipc}{*}{\ipc}}}_{\dval{ip}}N' \ims (\oipc,\rreqidc)\in\xiN{\ipc}(\rreqs)
	\end{equation}
\item  If a route request is sent, the originator has stored the unique pair
  of the originator's IP address and the request id.
\begin{equation}\label{inv:starcast_rreqid}
	N\ar{R:\starcastP{\rreq{*}{\rreqidc}{*}{*}{*}{\oipc}{*}{*}}}_{\dval{ip}}N' \ims (\oipc,\rreqidc)\in\xiN{\oipc}(\rreqs)
	\end{equation}
\item\label{it:starcastiii} The sequence number of an originator appearing in a route request can never be greater than the
 	originator's own sequence number.
\begin{equation}\label{inv:starcast_v}
	  N\ar{R:\starcastP{\rreq{*}{*}{*}{*}{*}{\oipc}{\osnc}{*}}}_{\dval{ip}}N'
          \ims \osnc\leq\xiN{\oipc}(\sn)
	\end{equation}
\end{enumerate}
\end{prop}

\begin{proofNobox}
We have to check that the consequent holds whenever a route request is sent. In all the processes there
are only two locations where this happens, namely \Pro{aodv}, Line~\ref{aodv:line39}
and {Pro.~\ref{pro:rreq}, Line~\ref{rreq:line34}}.
\begin{enumerate}[(a)]
\item
\begin{description}
\item[\Pro{aodv}, Line~\ref{aodv:line39}:]
        A request with content
		$
		\xi(*\comma\rreqid\comma*\comma*\comma*\comma\ip\comma*\comma\ip)
		$
		 is sent.
        So
        $\ipc :=\xi(\ip)$, $\oipc:=\xi(\ip)$ and $\rreqidc:=\xi(\rreqid)$.
        Hence, using~\eqref{eq:uniqueidwithxi}, \ \plat{$\xiN{\ipc}=\xiN{\dval{ip}}= \xi$}.
	Right before broadcasting the request,
        $(\xi(\ip),\xi(\rreqid))$ is added to the set
        $\xi(\rreqs)$.
\item[Pro.~\ref{pro:rreq}, Line~\ref{rreq:line34}:]
	The information $(\xi(\oip),\xi(\rreqid))$  is added to $\xi(\rreqs)$ at Line~\ref{rreq:line8}.
	Moreover, the set of handled requests $\xi(\rreqs)$ as well as
	the values of $\oip$ and $\rreqid$ do not change between
	Line~\ref{rreq:line8} and \ref{rreq:line34}.
        Again \plat{$(\oipc,\rreqidc)\mathbin=(\xi(\oip),\xi(\rreqid))\in\xi(\rreqs)\mathbin=\xiN{\ipc}(\rreqs)$}.
\end{description}
\item 
\begin{description}
\item[\Pro{aodv}, Line~\ref{aodv:line39}:]
        A request with content
		$
		\xi(*\comma\rreqid\comma*\comma*\comma*\comma\ip\comma*\comma\ip)
		$
		 is sent.
        So $\oipc:=\xi(\ip)$ and hence, by \eqref{eq:uniqueidwithxi}, \plat{$\xiN{\oipc}= \xi$}.
        Moreover, $\rreqidc:=\xi(\rreqid)$.
	Right before broadcasting the request, the pair
        $(\xi(\ip),\xi(\rreqid))$ is added to the set
        $\xi(\rreqs)$.
\item[Pro.~\ref{pro:rreq}, Line~\ref{rreq:line34}:]
        A request with content $\xi(*\comma\rreqid\comma*\comma*\comma*\comma\oip\comma*\comma*)$ is sent.
        The values of the variables $\rreqid$ and $\oip$ do not change in
        Pro.~\ref{pro:rreq}; they stem, through Line~\ref{aodv:line8}
	of Pro.~\ref{pro:aodv}, from an incoming RREQ message
	(Pro.~\ref{pro:aodv}, Line~\ref{aodv:line2}).
        Now the claim follows immediately from
        the fact the each RREQ message received, 
	has been sent by some node (\Prop{preliminaries}\eqref{it:preliminariesi}), and 
	\hyperlink{induction-on-reachability}{induction on reachability}.
	\end{description}
\item
\begin{description}
\item[\Pro{aodv}, Line~\ref{aodv:line39}:]
	The sender is the originator, so
	$\oipc:=\xi(\ip)=\dval{ip}$ and $\osnc:=\xi(\sn)$. By \eqref{eq:uniqueidwithxi},
	 $\xiN{\oipc}=\xi$, which immediately implies $\osnc:=\xiN{\oipc}(\sn)$.
\item[Pro.~\ref{pro:rreq}, Line~\ref{rreq:line34}:]
        Here $\oipc:=\xi(\oip)$ and $\osnc:=\xi(\osn)$.
        The values of the variables $\oip$ and $\osn$ do not change in
        Pro.~\ref{pro:rreq}; they stem from Line~\ref{aodv:line8} of Pro.~\ref{pro:aodv}.
        By \Prop{preliminaries}\eqref{it:preliminariesi}, a transition
        labelled $\colonact{R}{\starcastP{\rreq{*}{*}{*}{*}{*}{\oipc}{\osnc}{*}}}$
        must have occurred before, say in state $N^\dagger$.  
        Thus, by induction and \Prop{invarianti_itemiii},
        \plat{$\osnc\leq \xi_{N^{{\dagger}}}^{\oipc}(\sn) \leq \xi_{N}^{\oipc}(\sn)$}.
	\endbox
\end{description}
\end{enumerate}
\end{proofNobox}

\subsubsection{Routing Table Entries}
\begin{prop}\rm\label{prop:dsn}~
\begin{enumerate}[(a)]
\item The sequence number of a destination appearing in a route reply can never be greater than the
  destination's own sequence number.
  \begin{equation}\label{eq:dsn_rrep}
    N\ar{R:\starcastP{\rrep{*}{\dipc}{\dsnc}{*}{*}}}_{\dval{ip}}N'
    \ims \dsnc\leq\xiN{\dipc}(\sn)
  \end{equation}
\item A known destination sequence number of a valid routing table entry can never be greater than the
  destination's own sequence number.
  \begin{equation}\label{eq:dsn}
    (\dval{dip},\dval{dsn},\kno,\val,*,*,*)\in\xiN{\dval{ip}}(\rt)
    \ims \dval{dsn}\leq\xiN{\dval{dip}}(\sn)
  \end{equation}
\end{enumerate}
\end{prop}

\begin{proofNobox}
We apply simultaneous induction to prove these invariants.
\begin{enumerate}[(a)]
\item We have to check that the consequent holds whenever a route reply is sent.
\begin{description}
	\item[Pro.~\ref{pro:rreq}, Line~\ref {rreq:line14}:]
		A route reply with sequence number $\dsnc:=\xiN{\dval{ip}}(\sn)$ is initiated.
		Moreover, by Line~\ref{rreq:line10}, $\dipc:=\xiN{\dval{ip}}(\dip)=\xiN{\dval{ip}}(\ip)=\dval{ip}$.
                So $\dsnc=\xiN{\dipc}(\sn)$.
	\item[Pro.~\ref{pro:rreq}, Line~\ref{rreq:line26}:]
		A route reply with $\dipc:=\xiN{\dval{ip}}(\dip)$ and
                $\dsnc:=\xiN{\dval{ip}}(\sqn{\rt}{\dip})=\sq[\dipc]{\dval{ip}}$ is initiated.
                By Line~\ref{rreq:line22} $\dsnc$ is a {\kno}wn sequence number, stemming from a valid entry for $\dipc$
                in the routing table of $\dval{ip}$. Hence by Invariant~\Eq{dsn}
                \plat{$\dsnc=\sq[\dipc]{\dval{ip}}\leq\xiN{\dipc}(\sn)$}.
	\item[Pro.~\ref{pro:rrep}, Line~\ref{rrep:line13}:]
		The RREP message has the form
		$\xiN{\dval{ip}}(\rrep{\hops\mathop{+}1}{\dip}{\dsn}{\oip}{\ip})$.
		Hence $\dipc:=\xiN{\dval{ip}}(\dip)$ and $\dsnc:=\xiN{\dval{ip}}(\dsn)$.
        The values of the variables $\dip$ and $\dsn$ do not change in
        Pro.~\ref{pro:rrep}; they stem, through Line~\ref{aodv:line12}
	of Pro.~\ref{pro:aodv}, from an incoming RREP message
	(Pro.~\ref{pro:aodv}, Line~\ref{aodv:line2}).
	By \Prop{preliminaries} this message was sent before, say by node \dval{sip} in state $N^\dagger$.
        By induction we have $\dsnc\leq\xiN[N^\dagger]{\dipc}(\sn)\leq\xiN{\dipc}(\sn)$, where the
        latter inequality is by \Prop{invarianti_itemiii}.
\end{description}
\item We have to examine all application calls of \hyperlink{update}{$\fnupd$}---entries resulting
  from a call of $\fninv$ are not valid. Moreover, without loss of generality we restrict
  attention to those applications of $\fnupd$ that actually
  modify the entry for \dval{dip}, beyond its precursors; if $\fnupd$
  only adds some precursors in the routing table, the invariant---which
  is assumed to hold before---is maintained.
\begin{description}
\item[Pro.~\ref{pro:aodv}, Lines~\ref{aodv:line10}, \ref{aodv:line14}, \ref{aodv:line18}:]
These calls yield entries with $\unkno$nown destination sequence numbers.
\item[Pro.~\ref{pro:rreq}, Line~\ref{rreq:line6}:] 
Here $\dval{dip}:=\xi(\oip)$ and $\dval{dsn}:=\xi(\osn)$. These values stem from an incoming RREQ
message, which must have been sent beforehand, say in state $N^\dagger$.
By Invariant~\eqref{inv:starcast_v}, with $\oipc:=\xi(\oip)=\dval{dip}$ and
$\osnc:=\xi(\osn)=\dval{dsn}$ we have
{$\dval{dsn}\leq\xiN[N^\dagger]{\dval{dip}}(\sn)\leq\xiN{\dval{dip}}(\sn)$}, where the
latter inequality is by \Prop{invarianti_itemiii}.
\item[Pro.~\ref{pro:rrep}, Line~\ref{rrep:line5}:]
Here $\dval{dip}:=\xi(\dip)$ and $\dval{dsn}:=\xi(\dsn)$. These values stem from an incoming RREP
message, which must have been sent beforehand, say in state $N^\dagger$.
By Invariant~\eqref{eq:dsn_rrep}, with $\dipc:=\xi(\dip)=\dval{dip}$ and
$\dsnc:=\xi(\dsn)=\dval{dsn}$ we have
{$\dval{dsn}\leq\xiN[N^\dagger]{\dval{dip}}(\sn)\leq\xiN{\dval{dip}}(\sn)$}.
\endbox
\end{description}
\end{enumerate}
\end{proofNobox}

\begin{prop}\label{prop:route to nhip}\rm
Whenever \dval{ip}'s routing table contains an entry with next hop $\dval{nhip}$, 
it also contains an entry for $\dval{nhip}$.\hspace{-5pt}
\begin{equation}
(*,*,*,*,*,\dval{nhip},*)\in\xiN{\dval{ip}}(\rt) \ims (\dval{nhip},*,*,*,*,*,*)\in\xiN{\dval{ip}}(\rt)
\end{equation}
\end{prop}
\begin{proofNobox}
As usual we only consider function calls of $\fnupd$ and assume that the update changes the routing table.
\begin{description}
\item[\Pro{aodv}, Lines~\ref{aodv:line10}, \ref{aodv:line14} and \ref{aodv:line18}:]
1-hop connections are inserted into the routing table. By Invariant~\eqref{eq:inv_vii}, the new entry has the form $(\dval{nhip},*,*,*,*,\dval{nhip},*)$. 
Therefore \dval{ip} has an entry for \dval{nhip}.
\item[\Pro{rreq}, Line~\ref{rreq:line6}:]
We assume that the entry $\xi(\oip,\osn,\kno,\val,\hops+1,\sip,*)$ is inserted into $\xi(\rt)$. So, $\dval{nhip} := \xi(\sip)$.
This information is distilled from a received route request message (cf.\ Lines~\ref{aodv:line2} and~\ref{aodv:line8} of Pro.~\ref{pro:aodv}).
Right after receiving the message, a route to $\xi(\sip)$ is created
or updated (Line~\ref{aodv:line10} of Pro.~\ref{pro:aodv}); hence an
entry for the next hop exists.
\item[\Pro{rrep}, Line~\ref{rrep:line5}:]
The update is similar to the one of \Pro{rreq}, Line~\ref{rreq:line6}. The only difference is that the information stems from an incoming RREP message and that a routing table entry to $\xi(\dip)$ (instead of $\xi(\oip)$) is established.\endbox
\end{description}
\end{proofNobox}

\newcommand{\fnupda}{\fnupd^{\RREP}}
\newcommand{\upda}[2]{\fnupda(#1\comma#2)}
\newcommand{\versions}{interpretations\xspace}
\newcounter{ambiguity}
\newcommand{\amb}{\refstepcounter{ambiguity}Ambiguity~\theambiguity}

\section{Interpreting the IETF RFC 3561 Specification}\label{sec:interpretation}
It is our belief that, up to the abstractions discussed in \Sect{abstractions}, the specification
presented in the previous sections reflects precisely the intention and the meaning of the IETF
RFC~\cite{rfc3561}. However, when formalising the AODV routing protocol, we came across some
ambiguities, contradictions and unspecified behaviour in the RFC\@. This is also reflected by the
fact that different implementations of AODV behave differently, although they all follow the lines
of the RFC\@. Of course a specification ``{\sf needs to be reasonably implementation
  independent\/}''\footnote{\url{http://www.ietf.org/iesg/statement/pseudocode-guidelines.html}} and
can leave some decisions to the software engineer; however it is our belief that any specification
should be clear and unambiguous enough to guarantee the same behaviour when given to different
developers.  As we will show, this is not the case for AODV\@.

In this section, we discuss and formalise many of the problematic behaviours found, as well as their
possible resolutions.  An \phrase{interpretation} of the RFC is given by the allocation of a
resolution to each of the ambiguities, contradictions and unspecified behaviours.  Each reading,
implementation, or formal analysis of AODV must pertain to one of its interpretations. The formal
specification of AODV presented in Sections~\ref{sec:types} and~\ref{sec:modelling_AODV} constitutes
one interpretation; the inventory of ambiguities and contradictions is formalised in
\SSect{interpreting} by specifying each resolution of each of the ambiguities and contradictions as
a modification of this formal specification, typically involving a rewrite of a few lines of code
only.  We also show which \versions give rise to routing loops or other unacceptable behaviour.
Beforehand, in \SSect{decreasingSQN}, we show how a decrease in the destination sequence number in a
routing table entry generally gives rise to unacceptable protocol behaviour; later on we use this
analysis to reject some of the possible interpretations of the RFC\@.  After we have presented the
ambiguities and their consequences, in \SSect{implementations} we briefly discuss five of the most
popular implementations of AODV\@ and demonstrate that the anomalies we discovered are not only
theoretically driven, but \emph{do} occur in practice. In particular, we show that three
implementations can yield routing loops.

\subsection{Decreasing Destination Sequence Numbers}\label{ssec:decreasingSQN}

In the RFC it is stated that a sequence number is
\begin{quote}\raggedright\small 
  ``{\tt
        A monotonically increasing number maintained by each originating node.}''\\ 
	\hfill\cite[Sect.~3]{rfc3561}
 \end{quote}
Based on this, it is tempting to assume
 that also any destination sequence
 number within a routing table entry should be
 increased monotonically. 
In fact this is also stated in the RFC: 
\hypertarget{6.1}{The sequence number for a particular destination
\begin{quote}\raggedright\small 
  ``{\tt is updated whenever a node receives new (i.e., not stale) information 
   about the sequence number from RREQ, RREP, or RERR messages that may 
   be received related to that destination. 
   {[\dots]} In order to ascertain that information about a destination is not 
   stale, the node compares its current numerical value for the sequence 
   number with that obtained from the incoming AODV message. 
   {[\dots]} If the result of 
   subtracting the currently stored sequence number from the value of 
   the incoming sequence number is less than zero, then the information 
   related to that destination in the AODV message MUST be discarded, 
   since that information is stale compared to the node's currently 
   stored information.}''
 	\hfill\cite[Sect.~6.1]{rfc3561}
 \end{quote}}
 This long-winded description simply says that all information distilled from any  AODV control message
that has a smaller sequence number for the destination under consideration, MUST be discarded. 
AODV should never decrease any destination sequence number, since this could create loops.
We illustrate this by \Fig{loopdec}.
 
\begin{figure}
\vspace{-2ex}
\begin{exampleFig}{Creating a loop when sequence numbers are decreased}{fig:loopdec}
\FigLine[slsr]%
  {The initial state;\\a connection between $d$ and $s$ has been established.}{fig/ex_loop_decrease_1}{}
  {Assumption:\\A sequence number inside $a$'s RT is decreased}{fig/ex_loop_decrease_2}{}
\FigLine[slsr]%
  {The topology changes;\\$a$ invalidates routes to $d$ and $s$.}{fig/ex_loop_decrease_3}{}
  {The topology changes again;\\$a$ broadcasts a new RREQ destined to $d$;\\node $s$ receives the RREQ and updates its RT.}{fig/ex_loop_decrease_4}{}
\FigLineHalf[sl]%
  {$s$ has information about a route to $d$;\\it unicasts a RREP back.\\$a$ updates its RT and creates a loop.}{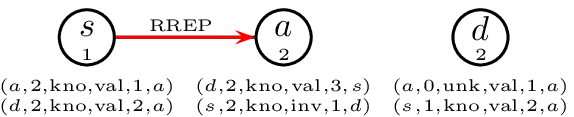}{}
\end{exampleFig}
\vspace{-3ex}
\end{figure}

Assume a linear topology with three nodes. 
In the past, node $d$ sent a request to establish a route to $s$. 
This RREQ message was answered by a RREP message of node $s$. After the route has been established,
the network is in the state of \Fig{loopdec}(a). In Part~(b) we assume that the sequence number of the 
routing table entry to $d$ of $a$'s routing table is decreased. Due to topology changes, node $a$ then 
looses connection to all neighbours and invalidates its routing table entries (Part (c)). In particular, it 
increments all sequence numbers of the routing table and sets the status flags to \inval. 
A possible error message sent by node $a$ is not received by any other node. 
After the link between $a$ and $s$ has appeared again, 
node $a$ wants to re-establish a route to $d$; it broadcasts  a new RREQ message (Part (d)).
The AODV control message  generated is $\rreq{0}{\dval{rreqid}}{d}{2}{\kno}{a}{2}{a}$, where 
$\dval{rreqid}$ is the unique id of the message. Since node $s$ has information about $d$, which is fresh enough, 
it generates the RREP message $\rrep{2}{d}{2}{a}{s}$ (Part (e)). 
Finally node $a$ receives the reply and establishes a route to $d$ via $a$. 
A loop has been created. 

Further on, we will discuss how sequence numbers might be
decreased when following the RFC literally or interpreting the RFC in
a wrong way.

\subsection{Interpreting the RFC}\label{ssec:interpreting}

In the following we discuss some ambiguities in the RFC, each giving rise to up to $6$ \versions
of AODV\@.
To resolve ambiguities, we often looked into real implementation, such as AODV-UU~\cite{AODVUU},
Kernel AODV~\cite{AODVNIST}
and AODV-UCSB~\cite{CB04} to determine the intended  version of AODV\@.
Additionally, we tried to derive unwanted behaviour from some of the possible interpretations.

\subsubsection[Updating Routing Table Entries]{\hypertarget{sss921}{Updating Routing Table Entries}}
\label{sssec:interpretation_update}
One of the crucial aspects of AODV is the maintenance of routing tables.
In this subsection we consider the update of routing table entries with new information.
In our specification we used the function \hyperlink{update}{$\fnupd$}
to specify the desired behaviour. Unfortunately, the RFC specification only gives hints 
how to update routing table entries; an exact and precise definition is missing.

\paragraph[Updating the Unknown Sequence Number in Response to a Route Reply]{\hypertarget{sss921a}{\amb: Updating the Unknown Sequence Number in Response to a Route Reply}}~\\
\hypertarget{6.7}{If a node receives a RREP message, it might have to update its routing table:
\begin{quote}\raggedright \small
``{\tt
  the existing entry is updated only in the following circumstances:
  \begin{enumerate}[(i)]
   \item the sequence number in the routing table is marked as invalid%
\index{sequence number!invalid}%
\footnote{The RFC \cite{rfc3561} uses the term ``invalid'' in relation to sequence
 numbers as synonym for ``unknown''. We use ``unknown'' ($\unkno\in\tSQNK$)
 only, in order to avoid confusion with the validity of the routing
 table entry in which the sequence number occurs ($\val,\inval\in\tFLAG$).}
\\ 
   \hspace{-2.0em}{[\dots]}{\rm''\hfill\cite[Sect.~6.7]{rfc3561}}
   \end{enumerate}
  }
 \end{quote}}
 In the same section it is also stated which actions occur if a route is updated:
 \begin{quote}\raggedright \small
   {\tt
  \begin{list}{-}{}
\index{route!active}%
	\item[{\rm``}-] the route is marked as active\footnote{The RFC
        uses the term ``active'' in relations to routes---actually
        referring to routing table entries---as a synonym for ``valid''.} [(\val)], 
	\item the destination sequence number is marked as valid [(\kno)], 
	\item the next hop in the route entry is assigned to be the node from which the RREP is received, [\dots]
     	\item the hop count is set to the value of the New Hop Count [%
          {\rm obtained by incrementing ``}the hop count value in
          the RREP by one, to account for the new hop through the
          intermediate node{\rm ''}],\\
	\hspace{-1.18em}{[\dots]}
	\item and the [new] destination sequence number is the Destination Sequence 
      Number in the RREP message.{\rm''\hfill\cite[Sect.~6.7]{rfc3561}}
   \end{list}
   }
 \end{quote}
{
\newcommand{\nrt}{\dval{nrt}}
\renewcommand{\rt}{\dval{rt}}
 \newcommand{\nr}{\dval{nr}}
To model this fragment of the RFC accurately, we define another update function, which 
adds a case to the \hyperlink{update}{original definition}:
\hypertarget{updatea}{
\[\begin{array}{r@{\hspace{0.5em}}c@{\hspace{0.5em}}l}
\fnupda : \tRT\times\tROUTE&\rightharpoonup&\tRT}\label{df:updatea}\\
\upda{\rt}{\route}&:=& \left\{
\begin{array}{@{\,}ll@{}}
\nrt\cup\{\nr\}&\mbox{if } \pi_{1}(\route)\in\kD{\rt} \wedge  \sqnf{\rt}{\route}=\unkno\\[1mm]
\upd{\rt}{\route}&\mbox{otherwise\ ,}
\end{array}
\right.
\end{array}\]
where, as in the definition of \hyperlink{update}{$\fnupd$},  $\nrt :=
\rt -\{\selr{\rt}{\pi_{1}(\route)}\}$ is the routing table without the
current entry in the routing table for the destination of $\route$ 
and  $\nr:=\addprec{\route}{\pi_{7}(\selr{\rt}{\pi_{1}(\route))}}$ is identical to~$\route$ except
that the precursors from the original entry are added. 
This function is now used in the process for RREP handling instead of $\fnupd$.
In particular, Lines 
\ref{rrep:line3},
\ref{rrep:line5} and
\ref{rrep:line25}
have to be changed in \Pro{rrep}; all other processes
(\Pro{aodv}--\Pro{rreq} and \Pro{rerr},\,\ref{pro:queues}) remain unchanged 
and use the original version of \fnupd .
}

Using this fragment of the RFC,
a sequence number of a routing table entry could be decreased. 
For example, an entry $(d,2,\unkno,\val,1,d,*)$ is replaced by $(d,1,\kno,\val,n\mathord+1,a,*)$
if the reply has the form $\rrep{n}{d}{1}{*}{a}$.\footnote{%
To see that this can actually happen, consider a variant of the example of \Fig{example2} in
\Sect{aodv} in which node $s$ starts out with a routing table entry $(d,2,\kno,\inval,1,d,*)$, which may
have resulted from a previous RREQ-RREP cycle, initiated by $s$, followed by an invalidation after the
link between $s$ and $d$ broke down.
Then in \Fig{example2}(e) this entry is updated to $(d,2,\unkno,\val,1,d,*)$, and in
\Fig{example2}(h) node $d$ sends a RREP message of the form $\rrep{0}{d}{1}{s}{d}$.
}
As indicated in \SSect{decreasingSQN}, this in turn can create routing loops.
This updating mechanism is in contradiction to 
\hyperlink{6.1}{the quote} from \cite[Sect.~6.1]{rfc3561} in \SSect{decreasingSQN}.
In view of the undesirability of routing loops, the only way to
resolve this contradiction is by ignoring (i) in \cite[Sect.~6.7]{rfc3561},
\hyperlink{6.7}{the statement quoted at the beginning of this paragraph}.

\paragraph[Updating with the Unknown Sequence Number]{\hypertarget{sss921b}{\amb: Updating with the Unknown Sequence Number}}~\\
Above we have discussed the update mechanism if a routing table entry with an unknown sequence number has to be updated. 
But what happens if the incoming AODV message carries an unknown number? This occurs regularly:
whenever a node receives a forwarded AODV control message from a $1$-hop neighbour (i.e., the neighbour is not the originator of the message), it creates a new or updates an existing routing table entry to that neighbour (cf.\ Lines~\ref{aodv:line10}, \ref{aodv:line14}, \ref{aodv:line18} of \Pro{aodv}).
For example, 
\hypertarget{6.5}{\begin{quote}\raggedright\small
``{\tt [w]hen a node receives a RREQ, it first creates or updates a route to
   the previous hop without a valid sequence number}''\hfill\cite[Sect.~6.5]{rfc3561}
 \end{quote}}
In case a new routing table entry is created, the sequence number is set to zero and the sequence-number-status flag is set to {\unkno} to signify that the sequence number corresponding to the neighbour is unknown. 
But, what happens if the routing table entry $(a,2,\kno,\val,2,b,\emptyset)$ 
of node $d$ is updated by 
$(a,0,\unkno,\val,1,a,\emptyset)$ as a consequence of the incoming RREQ
message $\rreq{1}{\dval{rreqid}}{x}{7}{\kno}{s}{2}{a}$, sent by node $a$? 
This situation is sketched in \Fig{updateunknownsqn}.\footnote{Only the routing table entry under consideration is depicted.}

\begin{exampleFig}{Updating routing table entries with the unknown sequence number}{fig:updateunknownsqn}
\FigLine[slsr]%
  {$d$ has established a route to $a$ with known sqn.\\~}{fig/ex_update_unknown_sqn_1}{}
  {The topology changes; $s$ looks for a route to $x$;\\$d$ receives the RREQ from $a$.}{fig/ex_update_unknown_sqn_2}{}
\end{exampleFig}

\noindent Following the RFC the routing table has to be updated. Unfortunately, it is not stated 
how the update is done. There are four reasonable updates---we call
them (2a), (2b), (2c) and (2d) to label them as resolutions of Ambiguity~\theambiguity:

\begin{enumerate}[(2a)]
\item\label{amb:2a} $(a\comma2\comma\kno\comma\val\comma 2\comma b\comma\emptyset)$: no update occurs (more precisely,
only an update of the lifetime of the routing table entry happens; this is not modelled in this paper).
To formalise this resolution, one skips the fifth option
(out of 6) in the definition of \hyperlink{update}{\fnupd} in Section~\ref{sssec:update}:
With this modification all our proofs in \Sect{invariants} remain
valid, which yields loop freedom and route correctness of this alternative interpretation of AODV\@.
It can be argued that the RFC rules out this resolution by including ``or
updates'' in the \hyperlink{6.5}{quote above}.
\item\label{amb:2b}
{
\newcommand{\nr}{\dval{nr}}
  $(a\comma0\comma\unkno\comma\val\comma1\comma a\comma\emptyset)$: all information is taken from the incoming AODV control message.  
To formalise this resolution, one 
changes the definition of \hyperlink{update}{$\fnupd$} by replacing $\nr'$ by $\nr$. 
Since this can decrease sequence numbers, 
routing loops might occur. Hence this update must not be used.
}
\item\label{amb:2c} $(a\comma2\comma\unkno\comma\val\comma1\comma a\comma\emptyset)$: the information from the routing table
and from the incoming AODV control message is
merged, by taking only the destination sequence number from the existing routing table
  entry and all other information from the AODV control message; as usual the sets of
  precursors are combined.
This is how our specification works. As we have shown in \Sect{invariants}, no loops
can occur. Moreover, node $d$ establishes an optimal route to $a$.
In case  $d$'s routing table would contain the tuple
$(a,1,\kno,\val,1,a)$, the sequence-number-status flag would also be
set to {\unkno}---this might be surprising, but it is consistent with
the RFC\@.
\item\label{amb:2d} $(a\comma2\comma\kno\comma\val\comma 1\comma a\comma\emptyset)$: the information from the routing table
  and from the incoming AODV control message is merged, by taking the destination sequence number
  and the sequence-number-status flag from the existing routing table entry and all other
  information from the AODV control message; as usual the sets of precursors are combined.
  To formalise this resolution, one takes
 {\newcommand{\s}{\dval{s}}%
  \newcommand{\nr}{\dval{nr}}%
\renewcommand{\dip}{\dval{dip}}%
\renewcommand{\dsn}{\dval{dsn}}%
\renewcommand{\flag}{\dval{flag}}%
\renewcommand{\hops}{\dval{hops}}%
\renewcommand{\nhip}{\dval{nhip}}%
\renewcommand{\pre}{\dval{pre}}%
  $\nr':=(\dip_{\nr}\comma\pi_{2}(\s)\comma\pi_{3}(\s)\comma\flag_{\nr}\comma\hops_{\nr}\comma\nhip_{\nr}\comma\pre_{\nr})$
  in the definition of \hyperlink{update}{\fnupd} in Section~\ref{sssec:update}.
In the case where $\sqn{\rt}{\pi_{1}(r)} = \pi_{2}(r)$ the routes $\nr$ and $\nr'$ are not equal anymore and 
  hence the function is not well defined. To achieve
  well-definedness, we create mutual exclusive cases by using the fourth and 
  fifth clause only if $\pi_{3}(r)=\kno$.
With this modification all results of \Sect{invariants}, except for \Prop{dsn},\footnote{The proof
  of \Prop{dsn} breaks down on the case Pro.~\ref{pro:aodv}, Lines~\ref{aodv:line10}, \ref{aodv:line14}, \ref{aodv:line18}.}
  remain valid, with the same proofs, which yields loop freedom and route correctness of this alternative interpretation.
  }
\end{enumerate}
One could also mix Resolution (2\ref{amb:2a}) with (2\ref{amb:2c}) or (2\ref{amb:2d}), for
instance by applying (2\ref{amb:2c}) for updates in response to a RREQ message and
(2\ref{amb:2a}) for updates in response to a RREP or RERR message. This could be
justified by the location of the \hyperlink{6.5}{quote above} in Sect.~6.5 of the RFC, which deals with
processing RREQ messages only. Furthermore, as a variant of (2\ref{amb:2c}) one
could skip the update of Line~\ref{aodv:line14} of \Pro{aodv} in the special case that
$\xi(\sip)=\xi(\dip)$, since in that case a sequence number for the previous hop is known.
Also for these variants, which we will not further elaborate here, the proofs of
\Sect{invariants} (with the exception of \Prop{dsn}) remain valid, and loop freedom and
route correctness hold.

When taking Resolutions (2\ref{amb:2a}) or (2\ref{amb:2d}), it is easy to check
  that for any routing table entry $r$ we always have $\pi_3(r)=\unkno \Leftrightarrow \pi_2(r)=0$.
  As a consequence, the sequence-number-status flag is redundant, and can be omitted from the
  specification altogether.\footnote{In Pro.~\ref{pro:rreq}, Line~\ref{rreq:line32},
  ``$\sqnf{\rt}{\dip}\mathbin=\unkno$'' should then be replaced by ``$\sqn{\rt}{\dip}=0$", and
  likewise, in Line~\ref{rreq:line22}, ``$\sqnf{\rt}{\dip}\mathbin=\kno$'' by ``$\sqn{\rt}{\dip}\neq0$".\label{foot:skipped}}
This is the way in which AODV-UU~\cite{AODVUU} is implemented: it skips the sequence-number-status
flag and follows (2\ref{amb:2d}).
Since Resolution (2\ref{amb:2b}) can lead to loops, and (2\ref{amb:2a}) and (2\ref{amb:2d}) do not make proper use of sequence-number-status flags, 
we assume that Resolution (2\ref{amb:2c}) is in line with the intention of the RFC\@.
In \SSect{skipunknownflag} we will discuss the relative merits of the Resolutions~(2\ref{amb:2a}), (2\ref{amb:2c}) and (2\ref{amb:2d}) and propose an improvement.

\paragraph[More Inconclusive Evidence on Dealing with the Unknown Sequence Number]{\hypertarget{sss921c}{\amb: More Inconclusive Evidence on Dealing with the Unknown Sequence Number}}~\\
Section~6.2 of the RFC
describes under which circumstances an update occurs.
\begin{quote}\raggedright\small
``\hypertarget{6.2}{\tt The route is only updated if  the new sequence number is either
   \begin{enumerate}[(i)]
        \item higher than the destination sequence number in the route table, or
        \item the sequence numbers are equal, but the hop count (of the
              new information) plus one, is smaller than the existing hop
              count in the routing table, or
   	\item the sequence number is unknown.{\rm''\hfill \cite[Sect.~6.2]{rfc3561}}
   \end{enumerate}
   }
 \end{quote}
Part (iii) is ambiguous. The most plausible reading appears to be that
``{\tt the sequence number}'' refers to the new sequence number, i.e., the one provided by an
incoming AODV control message triggering a potential update of the
node's routing table. This reading is incompatible with (2\ref{amb:2a}) above, and
thus supports only Resolutions~(2\ref{amb:2b}), (2\ref{amb:2c}) and (2\ref{amb:2d}).
An alternative reading is that
it refers to the sequence number
in the routing table, meaning that the corresponding
sequence-number-status flag has the value $\unkno$.
This reading of (iii) is consistent with \hyperlink{6.7}{the quote from Section 6.7
above}, and leads to routing loops in the same way.\footnote{%
It can be formalised by using \hyperlink{updatea}{$\fnupda$} instead of
{\fnupd} in all process \Pro{aodv}--\Pro{rrep}, and furthermore
skipping the fifth option in the definition of \hyperlink{update}{\fnupd} in \SSSect{update}.}
The remaining possibility is that Part (iii) refers to the sequence number
in the routing table, but only deals with the case that that number is
truly unknown, i.e.\ has the value $0$. This reading is consistent
with Resolution (2\ref{amb:2a}) above. However, it implies that the
   routing table may not be updated if the existing entry has a known
   sequence number whereas the route distilled from the incoming
   information does not. This is in contradiction the \hyperlink{6.5}{quote from
   Sect.~6.5 in the RFC above}. It is for this reason that we take the
   first reading of (iii) as our default.

An IETF Internet draft---published after the RFC---rephrases the 
above statement as follows:
\begin{quote}\small
{\tt
{\rm ``}A route is only updated if one of the following conditions is met:\\
{[\dots]}
   \begin{enumerate}[(iv)]
   \item the sequence number in the routing table is unknown.{\rm''\hfill \cite[Sect.~6.2]{dbis00}}
   \end{enumerate}}
\end{quote}
Since in \cite{dbis00} the sequence-number-status flag has been
dropped, the only ``unknown'' sequence number left is $0$, so this
quote takes the third reading above. We do not know, however, whether
this is meant to be a clarification of \cite{rfc3561}, or a
proposal to change it.

\paragraph[Updating Invalid Routes]{\hypertarget{sss921d}{\amb: Updating Invalid Routes}}~\\
Another closely related question that arose during formalising AODV is whether an invalid route should be updated in all cases. 
For example, should an entry $(a,3,\kno,\inval,4,b,\emptyset)$ of a routing table
be overwritten by $(a,1,\kno,\val,2,c,\emptyset)$?
 Of course this should not be allowed: if an invalid routing table entry were to be replaced by {\emph {any}}
valid entry---even with smaller sequence number---the protocol 
would not be loop free.

This time, the RFC~\cite{rfc3561} confirms this assumption:
\begin{quote}\raggedright\small
``{\tt
Whenever any fresh enough (i.e., containing a sequence number at
least equal to the recorded sequence number) routing information for
an affected  destination is received by a node that has marked that
route table entry as invalid, the node SHOULD update its route table
information according to the information contained in the update.}''\hfill \cite[Sect.~6.1]{rfc3561}
\end{quote}
However, it is somewhat less clear what should be done in case the sequence numbers are equal.
For example, should an entry $(a,3,\kno,\inval,2,b,\emptyset)$ of a routing table
be overwritten by $(a,3,\kno,\val,4,c,\emptyset)$?
According to the quote from Sect.~6.1 above the answer is yes, but according the
\hyperlink{6.2}{preceding quote} from Sect.\ 6.2 of the RFC, the answer is no.
Our formalisation follows Sect.~6.1 in this regard. To formalise the
alternative, one skips the fourth option in the definition of \hyperlink{update}{\fnupd} in \SSSect{update}.
This contradiction needs to be resolved in favour of Sect.~6.1: none
of the two options gives rise to routing loops, but the alternative
interpretation would result in a severely handicapped version of AODV,
in which many broken routes will never be repaired.
We illustrate this by the following example.
\begin{exampleFig}{Invalidating a route}{fig:invalidateroutes}
\FigLine[slsr]%
  {The initial state;\\a connection between $s$ and $d$ has been established.}{fig/ex_err_1}{}
  {The link breaks down;\\both nodes invalidate their entries.}{fig/ex_err_2}{}
\FigLine[slsr]%
  {The topology changes again;\\the link re-appears.}{fig/ex_err_3}{}
  {$s$ broadcasts a new RREQ destined to $d$;\\$d$ receives the RREQ message.}{fig/ex_err_4}{}
\FigLine[slsr]%
  {$d$ updates its RT as well as its sequence number;\\$s$ receives the RREP sent by $d$.}{fig/ex_err_5}{}
  {As a consequence $s$ updates its routing table.}{fig/ex_err_6}{}
\end{exampleFig}

\noindent
We assume a network with two nodes only. Node $s$ has already sent out a route request destined for $d$ and received 
a route reply message (\Fig{invalidateroutes}(a)).  Due to mobility the link between the nodes breaks. After nodes $s$ and 
$d$ have invalidated their routing table entries to each other (\Fig{invalidateroutes}(b)), the link 
becomes available  again. Node $s$ initiates a new route request
destined to $d$
(for instance because $s$ wants to send another data-packet to $d$). 
As usual, node $d$ receives the request (Part (d)), and,
depending on which version of AODV we follow,
may update its routing table. Figures~\ref{fig:invalidateroutes}(e) and (f) depict the standard and non-handicapped version of AODV
where first node $d$ updates its routing table with a valid routing table entry and sends a reply back to $s$. Then $s$ receives the reply 
and also updates its routing table. 
In the handicapped version of AODV neither node $s$ nor $d$ will update their routing tables---the messages would be send around
without any actual update being performed. At the end of the RREQ-RREP cycle the network would be in the same state as depicted in \Fig{invalidateroutes}(c).
Only if node $s$ initiates yet another route request---and therefore increases its own sequence number to~$4$, 
the resulting routing table of $d$ will contain a valid entry with destination $s$---$s$ would still end up with an invalid entry for $d$.
As long as node $d$ does not increase its own sequence number (e.g., due to the initiation of a route request),
node $s$ cannot re-establish a valid route.

\subsubsection[Self-Entries in Routing Tables]{\hypertarget{sss922}{Self-Entries in Routing Tables}}%
\label{sssec:interpretation_selfroutes}

In any practical implementation, when a node sends a message to itself, 
the message will be delivered to the corresponding application on the 
local node without ever involving a routing protocol and therefore without being ``seen'' 
by AODV or any other routing protocol.
Hence it ought not matter if any node using AODV creates a routing table entry 
to itself. However, as we will show later, there are situations where these \phrase{self-entries} yield loops.

\paragraph[(Dis)Allowing Self-Entries]{\hypertarget{sss922a}{\amb: (Dis)Allowing Self-Entries}}~\\
  In AODV, when a node receives a RREP message, it creates a routing table entry for the
  destination node if such an entry does not already exist \cite[Sect 6.7]{rfc3561}. If the destination node
  happens to be the processing node itself, this leads to the creation of a self-entry.
 The RFC does mention self-entries explicitly; it only refers to them at one location:
\begin{quote}\raggedright\small 
  ``{\tt 
  A node may change the sequence number in the routing table entry of a 
   destination only if: \\[1ex]
   -  it is itself the destination node {[\dots]}}''\hfill\hfill\cite[Sect.~6.1]{rfc3561}
\end{quote}
This points at least to the possibility of having self-entries.
On the other hand, some implementations, such as AODV-UU, disallow self-entries.
For our specification (Sections~\ref{sec:types} and~\ref{sec:modelling_AODV}) we have chosen to 
allow (arbitrary) self-entries since the RFC does {\em not} prohibits their occurrence.
We will refer to this resolution of Ambiguity 5 as \hypertarget{amb:5a}{(5a)}.

Looking at our specification, self-entries can only occur using the function \hyperlink{update}{$\fnupd$}. 
More precisely, we show the following.
\begin{proposition}\label{prop:selfentries}
There is only one location where self-entries can be established, namely Pro.~\ref{pro:rrep}, Line~\ref{rrep:line5}.
\end{proposition}
\begin{proofNobox}
No self-entries are established at the initialisation of the protocol (cf.\ \SSect{initial}).
Hence we have to consider merely all occurrences of update:
\begin{description}
	\item[\Pro{aodv}, Lines~\ref{aodv:line10}, \ref{aodv:line14}, \ref{aodv:line18}:]
                By \Cor{sipnotip}
		we have $\xi(\sip)\not=\dval{ip}$ and therefore no self-entry can be written in the routing table.
	\item[\Pro{rreq}, Line~\ref{rreq:line6}:] 
		An entry with destination $\xi(\oip)$ is updated/inserted.
          The value $\xi(\oip)$ stems from a received RREQ message
		(cf.\ Lines~\ref{aodv:line2} and~\ref{aodv:line8} of Pro.~\ref{pro:aodv}).
		A self-entry can only be established if
		$\xi(\oip)=\dval{ip}$.
		In that case, by Invariant~\eqref{inv:starcast_rreqid},
		$\xi((\oip,\rreqid))\in\xiN[N^\dagger]{\dval{ip}}(\rreqs)$, where $N^{\dagger}$ is the
		network expression at the time when the RREQ message was sent.
		(Here we use \Prop{preliminaries}.)  By \Prop{rreqs increase}, we obtain
		$\xi((\oip,\rreqid))\in\xiN[N_{\ref*{rreq:line4}}]{\dval{ip}}(\rreqs)$ and
		therefore Line~\ref{rreq:line4} evaluates to false and the update is never performed.
\item[\Pro{rrep}, Line~\ref{rrep:line5}:] Self-entries can occur;  an example is given below.
			\endbox
	\end{description}
\end{proofNobox}
We now show that self-entries can occur in our specifications. The
presented example (\Fig{selfentries}) is not the smallest possible
one; however, later on it will serve as a basis for showing how routing loops can occur.
The example consists of six nodes, where $5$ of them form a
circular topology, including one
link---between the nodes $d$ and $s$---that is is unstable and
unreliable. This link will disappear and re-appear several times in
the example.

\begin{exampleFig}{Self-entries}{fig:selfentries}
\FigLine[slsr]%
  {The initial state.}%
  {fig/ex_self-entry_1}%
  {}
  {$s$ broadcasts a new RREQ message destined to $d$;\\$c$ receives  the message and buffers it.}
  {fig/ex_self-entry_2}
  {\queue{d}{}\queue{s}{}\queue{x}{}\queue{a}{}\queue{b}{}\queue{c}{RREQ${}_{1s}$}}
\FigLine[slsr]%
  {The topology changes;\\$d$ moves into transmission range of $s$.}{fig/ex_self-entry_3}
  {\queue{d}{}\queue{s}{}\queue{x}{}\queue{a}{}\queue{b}{}\queue{c}{RREQ${}_{1s}$}}
  {$s$ broadcasts a new RREQ message destined to $x$.\\($x$ is either down or not in range of any node.)}{fig/ex_self-entry_4}
  {\queue{d}{RREQ${}_{2s}$}\queue{s}{}\queue{x}{}\queue{a}{}\queue{b}{}\queue{c}{RREQ${}_{2s}$\\RREQ${}_{1s}$}}
\FigLine[slsr]%
  {$d$ handles RREQ${}_{2s}$, forwards it and updates its RT;\\$s$ silently ignores RREQ${}_{2d}$ (after updating its RT);\\[-1.75348pt]\raisebox{-1.75348pt}[0pt][0pt]{node $a$ buffers the message.}}%
  {fig/ex_self-entry_5}
  {\queue{d}{}\queue{s}{}\queue{x}{}\queue{a}{RREQ${}_{2d}$}\queue{b}{}\queue{c}{RREQ${}_{2s}$\\RREQ${}_{1s}$}}
  {$c$ updates its RT and forwards RREQ${}_{1}$;\\$s$ silently ignores it; $b$ queues it.\\[-2pt]\mbox{}}{fig/ex_self-entry_6}
  {\queue{d}{}\queue{s}{}\queue{x}{}\queue{a}{RREQ${}_{2d}$}\queue{b}{RREQ${}_{1c}$}\queue{c}{RREQ${}_{2s}$}}
\FigLine[slxsr]%
  {The topology changes again;\\$d$ broadcasts a new RREQ message destined to $a$.}{fig/ex_self-entry_7}
  {\queue{d}{}\queue{s}{}\queue{x}{}\queue{a}{RREQ${}_{3d}$\\RREQ${}_{2d}$}\queue{b}{RREQ${}_{1c}$}\queue{c}{RREQ${}_{2s}$}}
  {$a$ handles RREQ${}_{2d}$ and forwards the message;\\node $d$ silently ignores it; $b$ queues it.}{fig/ex_self-entry_8}
  {\queue{d}{}\queue{s}{}\queue{x}{}\queue{a}{RREQ${}_{3d}$}\queue{b}{RREQ${}_{2a}$\\RREQ${}_{1c}$}\queue{c}{RREQ${}_{2s}$}}
  \FigNewline
\FigLine[xslxsr]%
  {$c$ handles RREQ${}_{2s}$ and forwards the message;\\node $s$ silently ignores it; $b$ queues it.}{fig/ex_self-entry_9}
  {\queue{d}{}\queue{s}{}\queue{x}{}\queue{a}{RREQ${}_{3d}$}\queue{b}{RREQ${}_{2c}$\\RREQ${}_{2a}$\\RREQ${}_{1c}$}\queue{c}{}}
  {$a$ handles RREQ${}_{3d}$ and unicasts RREP${}_{3a}$ back to $d$;\\$d$ handles the RREP and establishes an entry to $a$.}{fig/ex_self-entry_10}
  {\queue{d}{}\queue{s}{}\queue{x}{}\queue{a}{}\queue{b}{RREQ${}_{2c}$\\RREQ${}_{2a}$\\RREQ${}_{1c}$}\queue{c}{}}
\FigLine[xslxsr]%
  {$b$ updates its RT and forwards RREQ${}_{1}$;\\$c$ silently ignores it; $a$ queues it.}{fig/ex_self-entry_11}
  {\queue{d}{}\queue{s}{}\queue{x}{}\queue{a}{RREQ${}_{1b}$}\queue{b}{RREQ${}_{2c}$\\RREQ${}_{2a}$}\queue{c}{}}
  {$b$ updates its RT and forwards RREQ${}_{2}$;\\$c$ silently ignores it; $a$ queues it.}{fig/ex_self-entry_12}
  {\queue{d}{}\queue{s}{}\queue{x}{}\queue{a}{RREQ${}_{2b}$\\RREQ${}_{1b}$}\queue{b}{RREQ${}_{2c}$}\queue{c}{}}
  \multicolumn{2}{|@{}l|}{\raisebox{0pt}[12pt][6pt]{}
    (\alph{figexmp}) $b$ handles RREQ${}_{2c}$; since it has handled RREQ${}_{2a}$ before, the message is ignored.}\stepcounter{figexmp}\\
    \hline
  \FigLine[xslxsr]%
  {$a$ generates a route reply in response for RREQ${}_{1b}$;\\the RREP is sent to $d$.}{fig/ex_self-entry_13}
  {\queue{d}{RREP${}_{1a}$}\queue{s}{}\queue{x}{}\queue{a}{RREQ${}_{2b}$}\queue{b}{}\queue{c}{}}
  {The topology changes\footnotemark;\\$a$ handles and ignores RREQ${}_{2b}$.}{fig/ex_self-entry_14}
  {\queue{d}{RREP${}_{1a}$}\queue{s}{}\queue{x}{}\queue{a}{}\queue{b}{}\queue{c}{}}
\FigLineHalf[xsl]%
  {$d$ forwards RREP${}_{1}$;\\it is finally received and handled by $s$.}{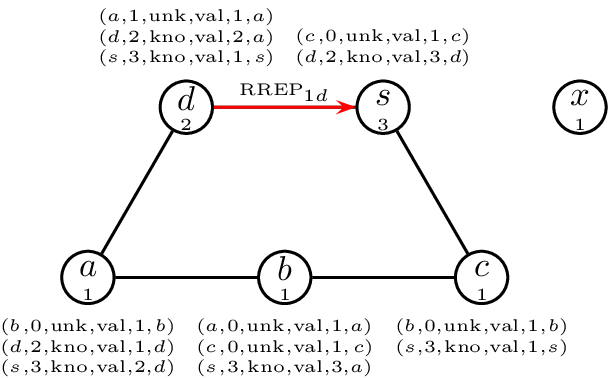}{}
\end{exampleFig}%
\footnotetext{The link between $d$ and $s$ can appear at any time between Parts (h) and (o).}

First, node $s$  broadcasts a new route request message RREQ${}_{1s}$ destined for node $d$---the
second index $s$ only indicates the sender of the  message; a message RREQ${}_{1c}$ belongs to the
same route discovery process. After the route discovery has been initiated, $d$ moves into
transmission range of $s$ (\Fig{selfentries}~(c)).\pagebreak[1]
Next, node $s$ sends a second route request; this time it is destined
for node $x$. In \Fig{selfentries}(e), node $d$ handles the request
destined to $x$; since it has not seen this route
request before and it is not the destination, the request is
forwarded. Node $s$ receives the forwarded message, updates its
routing table and silently ignores RREQ${}_{2d}$ afterwards; node $a$
receives RREQ${}_{2d}$ and stores it into its message buffer.
In \Fig{selfentries}(f) node $c$ handles RREQ${}_{1s}$, which is the first message in its message queue; it inserts an entry with destination $s$ in its routing table and forwards RREQ${}_{1}$. The link between nodes $d$ and $s$ disappears in Part (g). Moreover a third route discovery search is 
initiated: this time node $d$ is looking for a route to node $a$.  Node $a$ is the only node who receives the broadcast message and 
queues it in its buffer---it cannot handle RREQ${}_{3d}$ immediately since its message buffer queue is not empty. The topmost element
is RREQ${}_{2d}$, which has been received earlier in Part (e). Node $a$ now handles this request;it creates entries
for $d$ and $s$ in its routing table and forwards the message (\Fig{selfentries}(h)). In Part (i), node~$c$ forwards the second request. 
The message RREQ${}_{2c}$ is received by nodes $b$ and $s$; $s$ immediately handles the incoming message, updates its routing table and then 
silently ignores it. After that, node $a$ empties its message queue and handles the third request, initiated by $d$: since $a$ is the destination, it creates a route reply and unicasts it towards $d$. After a single hop, node $d$, the originator, receives the reply and establishes a routing table entry to $a$ (cf.\ \Fig{selfentries}(j)). All but the message buffer of node $b$ are now empty; the queue of $b$ contains the messages RREQ${}_{1c}$, RREQ${}_{2a}$ and RREQ${}_{2c}$. In Part~(k), $b$ handles the topmost message (RREQ${}_{1c}$). As usual, it updates its routing table and forwards the message to its neighbours $a$ and $c$; node $a$ buffers the message, whereas $c$ silently ignores it. In Part (l) the same procedure happens again---but this time 
node $b$ handles RREQ${}_{2}$ instead of RREQ${}_{1}$. The last message stored by node $b$ is RREQ${}_{2c}$, which is handled now. 
Since $b$ has handled a message belonging to the second route discovery search before, this message is ignored (Part (m)). 
The only node that has messages stored in its buffer is now node $a$. The first message in its queue is RREQ${}_{1b}$, a route request sent by its originator $s$ in \Fig{selfentries}(b) and destined for node $d$. Since node $a$ has an entry for $d$ in its routing table ($(d,2,\kno,val,1,d)$), it 
generates a route reply using $\rrep{1}{d}{2}{s}{a}$.
RREP${}_{1a}$ has to be sent to the next hop on the route to $s$, which stored in the routing table of $a$: this is node $d$, the 
destination of the original route request. In Part (o) the topology changes and the message queue of $a$ is emptied. Node $a$ handles RREQ${}_{2b}$ and silently ignores it since, in Part (h), 
the node  has already handled the second request. In the last step, node $d$ handles the reply generated by node $a$, updates its routing table and 
forwards the message to $s$, which establishes a route to $d$. When updating the routing table, node $d$ creates a self-entry since
RREP${}_{1a}$ unicasts information about a route to $d$.

Later on, in \SSSect{interpretation_invalidate}, we
continue this example to show that the combination of allowing
self-entries and literally following the RFC when invalidating routing table
entries yields loops.

By \Prop{selfentries} only Pro.~\ref{pro:rrep} has to be changed to disallow self-entries. 
There are two possibilities to accomplish this:

\begin{enumerate}[(5a)]
\setcounter{enumi}{1}
\item\label{amb:5b} If a node receives a route reply and would create a self-entry, it silently ignores the
message and continues with the main process \AODV. This resolution is implemented in \Pro{disallow_selfentries_rrep}.
A disadvantage of this process is that more replies are lost. In the above example (\Fig{selfentries}), node $s$ would never receive a reply 
as a consequence  of the very first request sent. 
Nevertheless, this resolution appears closest in spirit to the RFC, which
lists ``{\tt a forward route has been created or updated}'' as a
precondition for forwarding the RREP.
The invariants of \Sect{invariants} remain valid, with the very same proofs.
\begin{figure}[ht]
\vspace{-3ex}
  \algsetup{linenodelimiter=.,linenosize=\tiny}
  \begin{algorithm}[H]
    {\footnotesize
      \caption{RREP handling (Resolution (5b))}
      \label{pro:disallow_selfentries_rrep}
      \begin{algorithmic}[1]
        \input{processes/disallow_selfentries_rrep.tex}
	\end{algorithmic}
    }
  \end{algorithm}

\vspace{-3ex}
\end{figure}

\item\label{amb:5c} The alternative is that the node who would create a self-entry does forward the message without updating the its routing table. 
(\Pro{disallow_selfentries_rrep2}). This resolution bears the risk that the main-invariant (Invariant~\eqref{eq:inv_x}) is violated, since information is forwarded by a node without 
updating the node's routing table. However,  all
invariants established in \Sect{invariants} still hold, with
the very same proofs---except that the proof of Invariant~\eqref{inv:starcast_iv} requires
one extra case, which is trivial since $\dval{ip}_c = \dval{dip}_c$.
\begin{figure}[h]
\vspace{-2ex}
  \algsetup{linenodelimiter=.,linenosize=\tiny}
  \begin{algorithm}[H]
    {\footnotesize
      \caption{RREP handling (Resolution (5c))}
      \label{pro:disallow_selfentries_rrep2}
      \begin{algorithmic}[1]
        \input{processes/disallow_selfentries_rrep2.tex}
	\end{algorithmic}
    }
  \end{algorithm}

\vspace{-3ex}
\end{figure}
\end{enumerate}
Both resolutions by themselves do not yield weird or unwanted behaviour.

\paragraph[Storing the Own Sequence Number]{\hypertarget{sss922b}{\amb: Storing the Own Sequence Number}}
\begin{quote}\raggedright\small
``{\tt AODV depends on each node
   in the network to own and maintain its destination sequence number to
   guarantee the loop freedom of all routes towards that node.{\rm''\hfill \cite[Sect.~6.1]{rfc3561}}
   }
 \end{quote}
The RFC does not specify how own sequence numbers should be stored. 
Since the own sequence numbers are never mentioned in combination with destination sequence numbers that are stored in routing tables, 
it is reasonable to assume that the own sequence number should be stored in a separate data structure. 
However, there are implementations (e.g.\ Kernel AODV) that maintain a node's own sequence number in the node's routing table.
Of course, just storing a variable does not cause routing loops itself; but since the way of maintenance influences other design decisions, we list this ambiguity.

The resolution where the own sequence number is stored in a separate
variable---Resolution (6a)---has been modelled in the specification
presented in Sections~\ref{sec:types} and \ref{sec:modelling_AODV}.
The other resolution stores the own sequence number of node $s$ as a
self-entry in $s$'s routing table, i.e., as an entry with destination $s$. 
All other components of that routing table entry are more or less artificial,
so the self-entry could for instance have
the form $(s,\sn,\kno,\val,0,s,*)$, where $\sn$ is the maintained sequence number.

\hypertarget{amb:6b}{A variant of AODV in which a node's own sequence number is stored in its routing table---Resolution
(6b)---is obtained by adapting the specification of \Sect{modelling_AODV} as follows:
\begin{enumerate}[(i)]
\item The argument $\sn$ of the processes $\AODV$, $\PKT$, $\RREQ$, $\RREP$ and $\RERR$ is dropped.
\item In Pro.~\ref{pro:aodv}, Line~\ref{aodv:line39} and 
	Pro.~\ref{pro:rreq}, Line~\ref{rreq:line14a}, the occurrence
  of $\sn$ as argument of a \textbf{broadcast} or \textbf{unicast} is replaced by $\sqn{\rt}{\ip}$.
\item \label{iii}In the initial state each node \dval{ip} has a routing table containing exactly one
\index{self-entries!optimal}
  optimal self-entry: $$(\dval{ip},1,\kno,\val,0,\dval{ip},\emptyset)\in\xi(\rt) \wedge |\xi(\rt)|=1.$$
\item \hypertarget{iv}{In Pro.~\ref{pro:aodv}, Line~\ref{aodv:line36}
  the assignment $\assignment{\sn:=\inc{\sn}}$, incrementing the node's own sequence number,
  is replaced by}
  $$\begin{array}{l}
    \assignment{\rt:=\upd{\rt}{(\ip,\inc{\sqn{\rt}{\ip}},\kno,\val,0,\ip,\emptyset)}}.
  \end{array}$$
\item \hypertarget{v}{In Pro.~\ref{pro:rreq}, Line~\ref{rreq:line12}
  the assignment $\assignment{\sn:=\max(\sn,\dsn)}$, updating $\ip$'s own sequence number,
  is replaced by}\vspace{-1ex}
  $$\assignment{\rt:=\upd{\rt}{(\ip,\max(\sqn{\rt}{\ip},\dsn),\kno,\val,0,\ip,\emptyset)}}.$$
\end{enumerate}}

\begin{theorem}\rm\label{thm:storing own sn as selfentry}
The interpretation of AODV that follows Resolution~(6\hyperlink{amb:6b}{b}) and 
in all other ways our default specification of Sections~\ref{sec:types}
is loop free and route correct.
\end{theorem}

\begin{proofNobox}
All invariants established in \Sect{invariants} and their proofs remain valid, with $\sqn{\rt}{\ip}$
substituted for all occurrences of $\sn$ (in \Prop{invarianti_itemiii}, the proof of \Prop{msgsendingii},
and Propositions~\ref{prop:messagebroadcast}, \ref{prop:starcast}(\ref{it:starcastiii}) and
\ref{prop:dsn} and their proofs), and with the following modifications:
\begin{itemize}
\item \Prop{invarianti_itemiii} now follows from \Prop{dsn increase} and an inspection of the
  initial state.
\item \Prop{positive hopcount} now holds for non-\!-self-entries only:
  \begin{equation*}
  (\dval{dip},*,*,*,\dval{hops},*,*)\in\xiN{\dval{ip}}(\rt) \wedge \dval{dip}\neq\dval{ip} \ims \dval{hops}\geq1
  \end{equation*}
  The claim holds for the initial state, since there are only self-entries in the routing tables.
  Furthermore there are two more calls of $\fnupd$ to be checked, but they all deal with
  self-entries.
\item In the proof of Proposition~\ref{prop:starcastNew}(b), when calling \Prop{positive hopcount},
  we need to check that $\dval{dip}\neq\dval{ip}$. This follows by Pro.~\ref{pro:rreq}, Line~\ref{rreq:line20}.
\item In the proof of Proposition~\ref{prop:rte} two more calls of $\fnupd$ have to be checked, all trivial.
\item In the proof of \Prop{upd_well_defined} two more calls of $\fnupd$ have to be checked.
  The case Pro.~\ref{pro:aodv}, Line~\ref{aodv:line36}
  is trivial; the case Pro.~\ref{pro:rreq}, Line~\ref{rreq:line12} uses \Prop{invarianti_itemiii}.
\item \Prop{qual}(\ref{it:qual_upd}) now holds for updates with non-\!-self-entries only.
  The reason is that its proof depends on \Prop{positive hopcount}---this is in fact the only
  other place where we use \Prop{positive hopcount}.
\item In the proof of \Thm{state_quality} we now have to check the updates with self-entries
  explicitly, since they are no longer covered by \Prop{qual}(\ref{it:qual_upd})---this is the only use of \Prop{qual}(\ref{it:qual_upd}).
  There are two of them (both introduced above), and none of them can decrease the quality of routing tables.
\item In the proofs of \Prop{inv_nsqn} and \Thm{inv_a} two more calls of $\fnupd$ have to be checked.
  If any of those calls actually modifies the entry for \dval{dip}, beyond its precursors, then in
  the resulting routing table $\dval{nhip}=\dval{dip}=\dval{ip}$, and the precondition of the
  proposition or theorem is not met.
\item In the proofs of \Thm{route correctness}(a) and Propositions~\ref{prop:dsn}(b)
  and~\ref{prop:route to nhip} two more calls of $\fnupd$ have to be checked, all trivial.
\endbox
\end{itemize}
\end{proofNobox}
\Thm{storing own sn as selfentry} remains valid for any of the Resolutions
(2\ref{amb:2a},\,3c), (2\ref{amb:2c},\,3a) or (2\ref{amb:2d},\,3a), in combination with  (5\hyperlink{amb:5a}{a})--(5\ref{amb:5c}) and with (6\hyperlink{amb:6b}{b}), for the modifications in the
proofs of \Sect{invariants} induced by these resolutions are orthogonal.
\begin{cor}\label{cor:storing own sn as selfentry}\rm
Assume any interpretation of AODV that uses Resolution~(6\hyperlink{amb:6b}{b}) in combination with one of the
Resolutions (2\ref{amb:2a},\,3c), (2\ref{amb:2c},\,3a) or (2\ref{amb:2d},\,3a) and
any of the Resolutions~(5\hyperlink{amb:5a}{a})--(5\ref{amb:5c}); in all other ways it follows our default specification of Sections~\ref{sec:types}
and~\ref{sec:modelling_AODV}.
This interpretation is loop free and 
route correct.\makebox[0pt]{\hspace{-4in}\endbox}
\end{cor}

The following proposition shows that under Resolution (6\hyperlink{amb:6b}{b}), unless
combined with (2\ref{amb:2d}) and (5\hyperlink{amb:5a}{a}), the own sequence number of a
node is stored in an \emph{optimal} self-entry---with hop count $0$---and that such a self-entry cannot be invalided, nor
overwritten by data from an incoming message (as this would result in a non-optimal
self-entry---cf.\ \Prop{positive hopcount} and its proof).

\begin{prop}\label{prop:optimal}\rm~Assume an interpretation of AODV of the kind described in \Cor{storing own sn as selfentry}, 
except that it does not use a combination of Resolutions (2\ref{amb:2d}) and (5\hyperlink{amb:5a}{a}).\pagebreak[3]
\begin{enumerate}[(a)]
\item Each node \dval{ip} maintains an optimal self-entry at all times, i.e.,
  a valid entry for $\dval{ip}$ with hop count $0$ and \dval{ip} as next hop.\vspace{-1ex}
\begin{equation}\label{eq:inv_self}
(\dval{ip},*,\kno,\val,0,\dval{ip},*)\in\xiN{\dval{ip}}(\rt)
\end{equation}
\item Only self-entries of \dval{ip} have hop count $0$, and only self-entries of \dval{ip} have \dval{ip} itself as the next hop.\\
In other words, if $ (\dval{dip},*,*,*,\dval{hops},\dval{nhip},*)\in\xiN{\dval{ip}}(\rt)$ then
\begin{equation}\label{eq:inv_nonself}
\dval{dip} = \dval{ip} \iffs \dval{hops}= 0 \iffs \dval{nhip}=\dval{ip}
\end{equation}
\end{enumerate}
\end{prop}
\begin{proof}
We prove both invariants by simultaneous induction. In the initial state each routing table contains exactly one entry (see \eqref{iii} above),
which is a self-entry satisfying $\dval{dip}=\dval{ip}$, $\dval{hops}=0$ and $\dval{nhip}=\dval{ip}$.
By Remark~\ref{rem:remark} it suffices 
to look at the application calls of \hyperlink{update}{\fnupd} and \hyperlink{invalidate}{\fninv}.
If an update does not change the routing table entry (the last clause of \hyperlink{update}{\fnupd}),
both invariants are trivially preserved; hence we only examine the cases that an update actually occurs.
\begin{enumerate}[(a)]
\item\begin{description}
\item[Pro.~\ref{pro:aodv}, Line~\hyperlink{iv}{\ref*{aodv:line36}}; Pro.~\ref{pro:rreq},
  Line~\hyperlink{v}{\ref*{rreq:line12}}:]
  In these (new) cases the update yields a self-entry of the required form.
\item[Pro.~\ref{pro:aodv}, Lines~\ref{aodv:line10}, \ref{aodv:line14}, \ref{aodv:line18}:]
  By \Cor{sipnotip} the update does not result in a self-entry.
\item[Pro.~\ref{pro:rreq}, Line~\ref{rreq:line6}:]
  As in the proof of \Prop{selfentries} we conclude that
  $\dval{dip}=\xiN{\dval{ip}}(\oip)\neq\dval{ip}$, i.e.\ the inserted entry is not a self-entry.
\item[Pro.~\ref{pro:rrep}, Line~\ref{rrep:line5}:]
  The update has the form $\xiN{\dval{ip}}(\upd{\rt}{(\dip,\dsn,\kno,\val,\hops+1,\sip,\emptyset)})$.
  Assume, towards a contradiction, that it results in a self-entry, i.e.\ that $\xiN{\dval{ip}}(\dip)=\dval{ip}$.
  The values $\xiN{\dval{ip}}(\dip)$ and $\xiN{\dval{ip}}(\dsn)$ stem through
  Line~\ref{aodv:line12} of Pro.~\ref{pro:aodv} from a received RREP message, which by
  \Prop{preliminaries} was sent before, say in state $N^\dagger\!$.
  By Invariant~\eqref{eq:dsn_rrep}, with $\dval{dsn}_c:=\xiN{\dval{ip}}(\dsn)$ and $\dval{dip}_c:=\xiN{\dval{ip}}(\dip)=\dval{dip}$,
  we have
  \[
  \xiN{\dval{ip}}(\dsn)=\dval{dsn}_c\leq \xi_{N^\dagger}^{\dval{dip}_c}(\sqn{\rt}{\ip})
  \leq \xi_{N}^{\dval{dip}_c}(\sqn{\rt}{\ip})= \xi_{N}^{\dval{ip}}(\sqn{\rt}{\ip})\ ,\]
  the third step by  \Prop{invarianti_itemiii}.
  Since Invariant~\Eq{inv_self} holds before the $\fnupd$,
 $\dhp[\dval{ip}]{\dval{ip}}=0$ and $\sta[\dval{ip}]{\dval{ip}}=\val$.
  Hence, by the definition of \hyperlink{update}{\fnupd}, no update actually occurs.

  Thus, the update does not result in a self-entry.
  \item[Pro.~\ref{pro:aodv}, Line~\ref{aodv:line32}:]
  By construction of {\dests} in Line~\ref{aodv:line30}, for any
  $(\dval{rip},\dval{rsn})\in\xiN{\dval{ip}}(\dests)$ we have that
  $\fnnhop_{N_{\ref*{aodv:line30}}}^{\dval{ip}}(\dval{rip})=\fnnhop_{N_{\ref*{aodv:line30}}}^{\dval{ip}}(\dval{dip})$,
  where $\dval{dip}:=\xiN[N_{\ref*{aodv:line30}}]{\dval{ip}}(\dip)=\xiN[N_{\ref*{aodv:line23}}]{\dval{ip}}(\dip)$.
  By Line~\ref{aodv:line22} and Invariant~\Eq{inv_v} $\dval{dip}\neq\dval{ip}$.
  Hence by Invariant~\Eq{inv_nonself}, which holds at Line~\ref{aodv:line30},
  $\fnnhop_{N_{\ref*{aodv:line30}}}^{\dval{ip}}(\dval{dip})\neq\dval{ip}$.
  Thus $\fnnhop_{N_{\ref*{aodv:line30}}}^{\dval{ip}}(\dval{rip})\neq\dval{ip}$ and by
  Invariant~\Eq{inv_nonself} $\dval{rip}\neq\dval{ip}$.
  It follows that Line~\ref{aodv:line32} will never invalidate a self-entry.
\item[Pro.~\ref{pro:pkt}, Line~\ref{pkt2:line10}:]
  The proof is similar to the previous case, except that $\dval{dip}\neq\dval{ip}$
  follows from Line~\ref{pkt2:line4}.
\item[Pro.~\ref{pro:rreq}, Lines~\ref{rreq:line18}, \ref{rreq:line30}:]
  The proof is again the same, but with
  $\xiN{\dval{ip}}(\oip)$ taking the role of \dval{dip} and
  \plat{$\xiN{\dval{ip}}(\dip)$}.
  That $\xiN{\dval{ip}}(\oip)\neq\dval{ip}$ follows as in the case
  \Pro{rreq}, Line~\ref{rreq:line6} of the proof of \Prop{selfentries}.
\item[Pro.~\ref{pro:rrep}, Line~\ref{rrep:line18}:]
  This follows as in the previous case, except that $\xiN{\dval{ip}}(\oip)\neq\dval{ip}$
  follows from Line~\ref{rrep:line9}.
\item[Pro.~\ref{pro:rerr}, Line~\ref{rerr:line5}:]
  The proof is like the previous ones, but this time with
  $\fnnhop_{N_{\ref*{rerr:line2}}}^{\dval{ip}}(\dval{rip})=\xiN[N_{\ref*{rerr:line2}}]{\dval{ip}}(\sip)\neq\dval{ip}$
  following from \Cor{sipnotip}.
\endbox
\end{description}
\item The function \hyperlink{invalidate}{\fninv} neither changes the destination, the hop count nor the next hop; hence the invariant is preserved under function calls of \fninv.
Moreover, Invariant~\eqref{eq:inv_self} already shows $\dval{dip} = \dval{ip} \ims \dval{hops}= 0$ and $\dval{dip} = \dval{ip} \ims \dval{nhip}=\dval{ip}$.
\begin{description}
\item[Pro.~\ref{pro:aodv}, Line~\hyperlink{iv}{\ref*{aodv:line36}}; Pro.~\ref{pro:rreq},
  Line~\hyperlink{v}{\ref*{rreq:line12}}:]
  The update yields a self-entry satisfying
  $\dval{dip} = \dval{ip}$, $\dval{hops}= 0$ and $\dval{nhip}=\dval{ip}$.
\item[Pro.~\ref{pro:aodv}, Lines~\ref{aodv:line10}, \ref{aodv:line14}, \ref{aodv:line18}:]
  By \Cor{sipnotip} the update does not result in a self-entry, and indeed $\dval{hops}:=1$ and
  $\dval{nhip}:=\xiN{\dval{ip}}(\sip)\neq\dval{ip}$.
\item[Pro.~\ref{pro:rreq}, Line~\ref{rreq:line6}; Pro.~\ref{pro:rrep}, Line~\ref{rrep:line5}:]
  As observed under (a) above,
  the inserted entry can not be a self-entry.
  Moreover $\dval{hops}:= \xiN{\dval{ip}}(\hops)\mathord+1 >0$ and
  $\dval{nhip}:=\xiN{\dval{ip}}(\sip)\neq\dval{ip}$ by \Cor{sipnotip}.
\end{description}
\end{enumerate}
The above proof uses Invariant~\Eq{dsn_rrep} (\Prop{dsn}), which is not available under
Resolution~(2\ref{amb:2d}).  However, when using Resolutions~(5\ref{amb:5b}) or~(5\ref{amb:5c}), the call to
\Eq{dsn_rrep} can be avoided, because Line~\ref{rrep5b:line1} of Pro.~\ref{pro:disallow_selfentries_rrep}
or~\ref{pro:disallow_selfentries_rrep2} guarantees that the update of Pro.~\ref{pro:rrep},
Line~\ref{rrep:line5} does not yield a self-entry.
\end{proof}
When using Resolution (6\hyperlink{amb:6b}{b}) in combination with either Resolution~(2\ref{amb:2a}) or~(2\ref{amb:2c}),
the choice of one of the Resolutions (5\hyperlink{amb:5a}{a}), (5\ref{amb:5b}) or (5\ref{amb:5c}) doesn't make any difference in protocol
behaviour, as the optimal self-entry prevents ``accidental'' non-optimal self-entries to be written in the routing table.

\subsubsection[Invalidating Routing Table Entries]{\hypertarget{sss923}{Invalidating Routing Table Entries}}%
\label{sssec:interpretation_invalidate}

We have seen that decreasing a sequence number of a routing table entry yields  potential loops.
A similar effect occurs if the routing table entry is invalidated, but the sequence number is not incremented. 

Of course, invalidating routing table entries is closely related to route error message generation. 
\begin{quote}\raggedright\small
``{\tt A node initiates processing for a RERR message in three situations:
    \begin{enumerate}[(i)]
   	\item if it detects a link break for the next hop of an active [(\val)] 
             route in its routing table while transmitting data {[\dots]}, or 
	\item if it gets a data packet destined to a node for which it 
             does not have an active route  {[\dots]}, or
        \item if it receives a RERR from a neighbor for one or more
             active routes.{\rm''\\\hfill \cite[Sect.~6.11]{rfc3561}}
   \end{enumerate}
   }
 \end{quote}
Before the error message is transmitted, the routing table has to be updated:
\begin{quote}\raggedright\small
{\tt\label{pgquoteabove}
    \begin{enumerate}
   	\item[{\rm``}1.]  The destination sequence number of this routing entry, if it 
      exists and is valid, is incremented for cases (i) and (ii) above,
      and copied from the incoming RERR in case (iii) above.
      \item[2.] The entry is invalidated by marking the route entry as invalid\\
      \hspace{-1.65em}{[\dots]}{\rm ''\hfill \cite[Sect.~6.11]{rfc3561}}
   \end{enumerate}
   }
 \end{quote}

\paragraph[Invalidating Entries in Response to a Link Break or Unroutable Data Packet]
{\hypertarget{sss923a}{\amb: Invalidating Entries in Response to a Link Break or Unroutable Data Packet}}~\\
 Part 2.\ of the above quotation is clear, whereas Part 1.\ is ambiguous: Where does ``it'' refer to?
 Does the destination sequence number have to exist and be valid (\kno) or is it 
 the routing table entry that should exist and be valid (\val)?
 
 From a linguistic point of view, it is more likely that ``it'' refers to the destination sequence number:
 first, the sequence number is the first noun and subject in the sentence; second, the pronoun ``this'' already indicates that a 
 routing entry must exist; hence the condition of existence would be superfluous. Following this resolution, 
 the routing table entry $(d,1,\kno,\val,1,d)$ would be updated to $(d,2,\kno,\inval,1,d)$, but
 $(d,1,\unkno,\val,1,d)$ would yield $(d,1,\unkno,\inval,1,d)$.
 A formalisation of this resolution could be obtained by changing
 Line~\ref{aodv:line30} of Pro.~\ref{pro:aodv} into
 \vspace{-2ex}
 
 {\footnotesize\[\textbf{[\![}\dests\begin{array}[t]{@{\,}c@{\,}l@{}}
 :=&\{(\rip,\inc{\sqn{\rt}{\rip}})\,|\,\rip\in\akD{\rt}\ans
   \nhop{\rt}{\rip}=\nhop{\rt}{\dip} \ans \sqnf{\rt}{\dip}=\kno\}\\
   \phantom{:}\cup&\{(\rip,\sqn{\rt}{\rip})\,|\,\rip\in\akD{\rt}\ans
   \nhop{\rt}{\rip}=\nhop{\rt}{\dip} \ans \sqnf{\rt}{\dip}=\unkno\}\textbf{]\!]}\ .\end{array}\]}

\noindent
 Here, sequence numbers of \keyw{kno}wn routing table entries are incremented, whereas
 sequence numbers of \keyw{unk}nown entries remain the same. 
 Both kinds of sequence numbers are used later on to invalidate routing table entries and for further error handling.
 Similar changes need to be made for Pro.~\ref{pro:pkt},
 Line~\ref{pkt2:line9}; Pro.~\ref{pro:rreq},
 Line~\ref{rreq:line16}, \ref{rreq:line28}
 and Pro.~\ref{pro:rrep}, Lines~\ref{rrep:line16}.
 With this interpretation AODV is able to 
 create routing loops; \Fig{loop1} shows an example.

Part (a) shows a network, in which the node $s$ has already established a 
route to $d$. This was done by a single route discover process
(cf.\ the first example of \SSect{detExample}). 
Next, node $b$ tries to establish a route to $a$. To that end, 
it initiates and broadcasts a route request;
the RREQ message is forwarded by $d$ (Part (b)). 
Node $a$ receives the message and updates the routing table entry 
$(d,1,\kno,\val,1,d)$ to
$(d,1,\unkno,\val,1,d)$, following Resolution (2\ref{amb:2c}) 
of \SSSect{interpretation_update} (\hyperlink{sss921b}{Updating with the Unknown Sequence}
\hyperlink{sss921b}{Number}).
After the route has been established, the topology changes and 
all links to $a$ break down; the node itself notices that the link to $d$ 
is down, invalidates the route, and sends a RERR message, but its
 RERR message is not received by any node 
(Part (d)). Invalidating the routes uses the assumption that 
the routing table entry $(d,1,\unkno,\val,1,d)$ is updated to $(d,1,\unkno,\inval,1,d)$.
In Part (e), $a$ reconnects to $s$ and a new link between $s$ and $d$ occurs.
Last, node $a$ tries to reestablish a route to $d$, broadcasts a 
request with destination sequence number $1$ and 
immediately receives an answer by $s$.
Now, the routing table of $a$ contains an entry to $d$ with next hop 
$s$, and $s$ has a routing table entry to $d$ with next hop $a$. 
A packet which is sent to node $d$ by either of these two nodes would circulate in a loop forever.

\begin{exampleFig}{Creating a loop by not incrementing unknown sequence numbers}{fig:loop1}
\FigLine[xslxsr]%
  {The initial state;\\a connection between $s$ and $d$ has been established.}{fig/ex_loop1_1}{}
  {$b$ broadcasts a new RREQ destined to $a$;\\$d$ and $a$ receive the RREQ and update their RTs.}{fig/ex_loop1_2}{}
\FigLine[xslxsr]%
  {$a$ unicasts a RREP back.}{fig/ex_loop1_3}{}
  {The topology changes;\\$a$ invalidates routes to $b$, $d$, and $s$.}{fig/ex_loop1_4}{}
\FigNewline
\FigLine[xslxsr]%
 {The topology changes again.}{fig/ex_loop1_5}{}
 {$a$ broadcasts a new RREQ destined to $d$;\\node $s$ receives the RREQ and updates its RT.}{fig/ex_loop1_6}{}
\FigLineHalf[xsl]%
 {$s$ has information about a route to $d$;\\hence it unicasts a RREP back.}{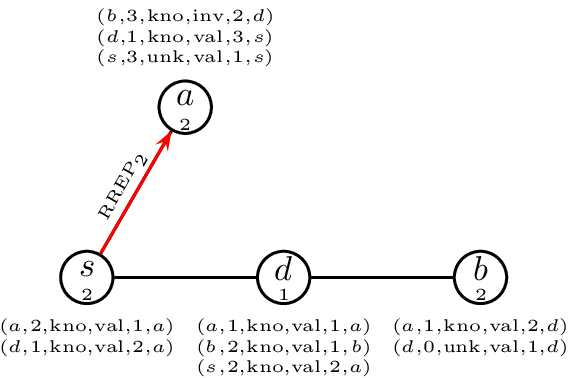}{}
\end{exampleFig}

In sum, the only acceptable reading of Part 1.\ above is the one
where ``it'' refers to ``routing entry'': before a RERR message is
sent, the destination sequence number of a routing table entry is
incremented, if such an entry exists and is valid.
This is the interpretation formalised in \Sect{modelling_AODV}.

\paragraph[Invalidating Entries in Response to a Route Error Message]
{\hypertarget{sss923b}{\amb: Invalidating Entries in Response to a Route Error Message}}~\\
\newcommand{\emptyspace}{\phantom{\rt - \{}}
The part ``{\tt and copied from the incoming RERR in case (iii) above}''
of the quote given on Page~\pageref{pgquoteabove} (from Sect.\ 6.11 of the RFC) is unambiguous.
It describes the replacement of an
existing destination sequence number in a routing table entry with another one, which may be strictly smaller.
This literal interpretation gives rise to a version of AODV without the requirement \highlight{$\sqn{\rt}{\rip}<\rsn$} in
Pro.~\ref{pro:rerr}, Line~\ref{rerr:line2} (cf.~Resolution (8a) below). 
However, replacing a sequence number with a strictly smaller one
contradicts the \hyperlink{6.1}{quote from Sect.\ 6.1 of the RFC} displayed in \SSect{decreasingSQN}.
To make the process of invalidation consistent with Sect.~6.1 of the RFC, one could
use Resolutions~(8b) or~(8c) instead. Resolution~(8b), which strictly follows Sect.~6.1, aborts the
invalidation attempt if the destination sequence number provided by the incoming RERR
message is smaller than the one already in the routing table. Resolution~(8c), on the other
hand, still invalidates in these circumstances, but prevents a decrease in the destination
sequence number by taking the maximum of the stored and the incoming number.

\begin{enumerate}[(8a)]
\item\label{amb:8a}
	Follow Section 6.11 of the RFC, in defiance of 6.1, i.e., {\em always\/} invalidate the \rte, and copy the destination sequence number from the error message to the corresponding entry in the routing table.%
	\footnote{It could be argued that this is not a reasonable interpretation of the RFC, since
          Section\ 6.1 should have priority over 6.11. However, this priority is not explicitly stated.} 
        This is formalised by skipping the requirement \highlight{$\sqn{\rt}{\rip}<\rsn$} in
        Pro.~\ref{pro:rerr}, Line~\ref{rerr:line2}.
\item\label{amb:8b}
	 Follow Section 6.11 only where it does not contradict 6.1, i.e., invalidate the \rte and
         copy the destination sequence number {\em only\/} if this does not give rise to a decrease
         of the destination sequence number in the routing table. This if formalised by
         replacing the requirement by \mbox{\highlight{$\sqn{\rt}{\rip}\leq\rsn$}}.
\item\label{amb:8c}
	Always invalidate the \rte (skip the requirement completely), but use a version of the function
        \hyperlink{invalidate}{$\fninv$} of \SSSect{invalidate} that uses
        $\highlight{\max(\pi_2(r),\dval{rsn})}$ instead of \highlight{\dval{rsn}},
        thereby updating the destination sequence number in the routing table to the maximum of its old value and the
        value contributed by the incoming RERR message.
\end{enumerate}
\noindent We now show that in combination with allowing self-entries (cf.~\SSSect{interpretation_selfroutes})
each of these resolutions gives rise to routing loops.
\Fig{loopsfromselfentries} continues the example of \Fig{selfentries},
and is valid for any of them.

\begin{exampleFig}{Combination of self-entries and an inappropriate $\fninv$ yields loops}{fig:loopsfromselfentries}
\FigLine[xslxsr]%
  {The initial state (same as \Fig{selfentries}(p)).}%
  {fig/ex_loop_from_selfentry_1}%
  {\queue{d}{}\queue{s}{}\queue{x}{}\queue{a}{}\queue{b}{}\queue{c}{}}
  {The topology changes.}
  {fig/ex_loop_from_selfentry_2}
  {\queue{d}{}\queue{s}{}\queue{x}{}\queue{a}{}\queue{b}{}\queue{c}{}}
\FigLine[xslxsr]%
  {A standard RREQ-RREP cycle starts;\\$d$ broadcasts a new RREQ message destined to $x$\\ nodes $a,s$ buffer the message.}%
  {fig/ex_loop_from_selfentry_3}%
  {\queue{d}{}\queue{s}{RREQ${}_{4d}$}\queue{x}{}\queue{a}{RREQ${}_{4d}$}\queue{b}{}\queue{c}{}}
  {$a$ and $s$ handle and forward the RREQ;\\$d$ silently ignores the messages;\\$b,c$ and $x$ store it.}
  {fig/ex_loop_from_selfentry_4}
  {\queue{d}{}\queue{s}{}\queue{x}{RREQ${}_{4s}$}\queue{a}{}\queue{b}{RREQ${}_{4a}$}\queue{c}{RREQ${}_{4s}$}}
\FigLine[xslxsr]%
  {$b$ forwards RREQ${}_{4}$;
   $a$ ignores it; $c$ stores it;\\
   $x$ replies to the RREQ with a RREP.}%
  {fig/ex_loop_from_selfentry_5}%
  {\queue{d}{}\queue{s}{RREP${}_{4x}$}\queue{x}{}\queue{a}{}\queue{b}{}\queue{c}{RREQ${}_{4b}$\\RREQ${}_{4s}$}}
  {$s$ forwards the RREP to $d$, which handles it;\\$c$ handles and forwards RREQ${}_{4s}$.}
  {fig/ex_loop_from_selfentry_6}
  {\queue{d}{}\queue{s}{RREQ${}_{4c}$}\queue{x}{}\queue{a}{}\queue{b}{RREQ${}_{4c}$}\queue{c}{RREQ${}_{4b}$}}
  \FigNewline
  \FigLine[xslxsr]%
  {All nodes have handled RREQ${}_{4}$ before;\\they silently ignore the messages in their queues;\\the topology changes.}%
  {fig/ex_loop_from_selfentry_7}%
  {\queue{d}{}\queue{s}{}\queue{x}{}\queue{a}{}\queue{b}{}\queue{c}{}}  
  {$d$ detects the link break, invalidates its self-entry,\\ and initiates a RERR message.}
  {fig/ex_loop_from_selfentry_8}
  {\queue{d}{}\queue{s}{RERR${}_{1d}$}\queue{x}{}\queue{a}{}\queue{b}{}\queue{c}{}}  
\FigLine[xslxsr]%
  {The topology changes;\\$s$ handles the RERR message.}%
  {fig/ex_loop_from_selfentry_9}%
  {\queue{d}{}\queue{s}{}\queue{x}{}\queue{a}{}\queue{b}{}\queue{c}{}}    
  {The topology changes again.}
  {fig/ex_loop_from_selfentry_10}
  {\queue{d}{}\queue{s}{}\queue{x}{}\queue{a}{}\queue{b}{}\queue{c}{}}    
\FigLine[xslxsr]%
  {$s$ broadcasts a new RREQ message destined to $d$}%
  {fig/ex_loop_from_selfentry_11}
  {\queue{d}{}\queue{s}{}\queue{x}{RREQ${}_{5s}$}\queue{a}{}\queue{b}{}\queue{c}{}}
  {$x$ replies to RREQ${}_{5s}$; $s$ handles the reply;\\a loop between $s$ and $x$ has been established.}
  {fig/ex_loop_from_selfentry_12}
  {\queue{d}{}\queue{s}{}\queue{x}{}\queue{a}{}\queue{b}{}\queue{c}{}}  
\end{exampleFig}
\vspace{-.1ex}

At the initial state (Part (a)), the message queues of all nodes are
empty. A couple of routes have been found and many routing
table entries are already set up. Node $d$ has among standard entries
also a self-entry, a valid entry to itself with sequence number and
hop count $2$. The example continues with node $x$ moving into the
transmission range of $s$, followed by a standard
RREQ-RREP cycle (\Fig{loopsfromselfentries}(b--g)) .

In Part (c), $d$ initiates a new route request for $x$. The generated message is received
by nodes $a$ and $s$; both nodes create reverse routes to $d$ and forward the
request (Part (d)). Node $x$ now handles the forwarded request and since this node is the
intended destination it generates and unicasts a route reply.  Meanwhile the broadcast
request still flows around in the network. In Part (e), node $b$ forwards it; in (f) the
message is handled by node $c$.  Here the route reply is also unicast back from $s$ to
$d$. The RREQ-RREP cycle ends with silently ignoring all remaining route request
messages of all message queues; this is due to the fact that all nodes have already
handled the request sent out by node $d$ in Part (d).

The last part of the example, which finally creates a routing loop, starts with a topology
change in \Fig{loopsfromselfentries}(g).  Node (d) detects the link break and invalidates
its link to $a$. It also invalidates its self-entry since the next hop on its
recorded route to $d$ is node $a$.  (At this point all treatments of the
  invalidation procedure contemplated in this section agree.)  As a
consequence of the link break node $d$
also casts a route error message to $s$. Due to
unreliable links (for example due to node $s$ moving around), the error message cannot be sent
forward; hence {\em only} node $s$ invalidates its entry for $d$. 
Before the entry is invalidated it is, by Line~\ref{aodv:line18} of \Pro{aodv}, updated to
 $(d,3,\unkno,\val,1,d)$. For the invalidation
we assume either of the Resolutions (8a), (8b) or (8c)---according to each the routing table entry
of $s$ to $d$ \emph{is} invalidated and the destination sequence number $3$ remains unchanged.
In Part (j) node $s$ moves into transmission range of $x$. We assume that it wants to send 
a data packet to $d$. Since its routing table entry for $d$ has been
invalidated, a new route request is sent out in Part (k).
The destination sequence number in this control message is set to $3$.
The message is received by node $x$, which immediately initiates a route reply since its
routing table contains a valid entry to $d$ with a sequence number
that is large enough. After node $s$ receives this message, a routing loop between $s$ and $x$ for destination $d$ has been established.
  
The problem in this example is that a routing table entry in invalidated without
increasing the destination sequence number.

As the above example shows, none of the above variants should be used in combination 
with non-optimal self-entries; thus either non-optimal self-entries should be
forbidden, or one should reject all plausible interpretations of the
invalidation process that are consistent with the combination of Sections 6.1 and 6.11 of the RFC\@.
However, the \hyperlink{6.2}{preceding quote from Sect.\ 6.2 of the RFC} suggests the interpretation
proposed in Sections~\ref{sec:types} and~\ref{sec:modelling_AODV}---Resolution (8f) below. Here we
invalidate the \rte and copy the destination sequence number only if this gives rise to an \emph{increase}
of the destination sequence number in the routing table. This if formalised by the requirement \highlight{$\sqn{\rt}{\rip}<\rsn$}
in~\Pro{rerr}, Line~\ref{rerr:line2}.
Another solution, Resolution (8d), is to still invalidate in this circumstances, but guarantee an
increase in the destination sequence number in the routing table by taking the maximum of
its incremented old value and the value stemming from the incoming RERR message. Finally, Resolution (8e) is a combination of~(8b) and (8d).

\begin{enumerate}[{(8}a{)}]
\setcounter{enumi}{3}
\item\label{amb:8d}
	Always invalidate the entry (skip the requirement \highlight{$\sqn{\rt}{\rip}<\rsn$}), but use a version of the function
        \hyperlink{invalidate}{$\fninv$} of \SSSect{invalidate} that uses
        \highlight{$\max(\inc{\pi_2(r)},\dval{rsn})$} instead of \highlight{\dval{rsn}}.
\item\label{amb:8e}
	Invalidate the \rte only if \highlight{$\sqn{\rt}{\rip}\leq\rsn$} and update the destination
        sequence number to \highlight{$\max(\inc{\pi_2(r)},\dval{rsn})$}.
	\footnote{The variant that invalidates only if \highlight{$\sqn{\rt}{\rip}\leq\rsn$} and
          updates to \highlight{$\max(\pi_2(r),\rsn)$} needs no separate consideration,
          since it is equivalent to Resolution (8b).}
\item\label{amb:8f}
	Invalidate the \rte only if \highlight{$\sqn{\rt}{\rip}<\rsn$}.%
	\footnote{Here, it does not matter whether we update to \highlight{\rsn},
          \highlight{$\max(\pi_2(r),\rsn)$} or to \highlight{$\max(\inc{\pi_2(r)},\rsn)$}; they are all equivalent.}
\end{enumerate}
In Sections~\ref{sec:types} and~\ref{sec:modelling_AODV} we have shown that our default
specification of AODV, implementing Resolution~(8f), is loop free and route correct.
We now show that the same holds when using Resolutions (8d) or (8e) instead.
\hypertarget{amb8modifications}{In fact, all invariants established in \Sect{invariants} and their proofs remain valid,
with the following modifications.}
\begin{itemize}
\item The proof of \Prop{dsn increase} simplifies, because not even the (modified) function $\fninv$
  can decrease a sequence number.
\item In \Prop{starcastrerr} the requirement $\rsnc=\sq[\ripc]{\dval{ip}}$ is weakened to
  $\rsnc\leq\sq[\ripc]{\dval{ip}}$. This change is harmless, since
  \Prop{starcastrerr} is applied in the proof of \Prop{inv_nsqn} only (at the end), where the
  weakened version is used anyway.
  
  The first case in the proof of \Prop{starcastrerr} is adapted to:
\begin{description}
\item[Pro.~\ref{pro:aodv}, Line~\ref{aodv:line33}:]
The set $\destsc$ is constructed in Line~\ref{aodv:line31a} as a subset of 
\plat{$\xiN[N_{\ref*{aodv:line31a}\!}]{\dval{ip}}(\dests)\mathbin=\xiN[N_{\ref*{aodv:line32}\!}]{\dval{ip}}(\dests)$}.
For each $(\ripc,\rsnc)\in\xiN[N_{\ref*{aodv:line32}}]{\dval{ip}}(\dests)$
one has $\ripc=\xiN[N_{\ref*{aodv:line30}}]{\dval{ip}}(\rip)\in\fnakD_{N_{\ref*{aodv:line30}}}^{\dval{ip}}$.
Then in Line~\ref{aodv:line32}, using the modified function {$\fninv$},
$\status{\xi(\rt)}{\ripc}$ be\-comes $\inval$ and 
$\sqn{\xi(\rt)}{\ripc}$ becomes $\max(\inc{\sqn{\xiN[N_{\ref*{aodv:line30}}]{\dval{ip}}(\rt)}{\ripc}},\rsnc)$.
Thus we obtain $\ripc\in\ikd{\dval{ip}}$ and
$\sq[\ripc]{\dval{ip}} \geq \rsnc$.
\end{description}
\item The proof of \Thm{state_quality} simplifies, because the (modified) function $\fninv$
  can never decrease the quality of routing tables.
\item The last case in the proof of \Prop{inv_nsqn} is adapted to:
\begin{description}
\item[Pro.~\ref{pro:rerr}, Line~\ref{rerr:line5}:]
        Let $N_{\ref*{rerr:line5}}$ and $N$ be the network expressions right before
        and right after executing Pro.~\ref{pro:rerr}, Line~\ref{rerr:line5}.
	The entry for destination $\dval{dip}$ can be affected 
	only if $(\dval{dip},\dval{dsn})\in\xiN[N_{\ref*{rerr:line2}}]{\dval{ip}}(\dests)$ for some $\dval{dsn}\in\tSQN$.
	In that case, by Line~\ref{rerr:line2},
        \plat{$(\dval{dip},\dval{dsn})\in\xiN[N_{\ref*{rerr:line2}}]{\dval{ip}}(\dests)$},
	\plat{$\dval{dip}\in\akd[N_{\ref*{rerr:line2}}]{\dval{ip}}$}, and
        \plat{$\fnnhop_{N_{\ref*{rerr:line2}}}^{\dval{ip}}(\dval{dip})=\xiN[N_{\ref*{rerr:line2}}]{\dval{ip}}(\sip)$}.
By the modified definition of \fninv,\\[1ex]
\mbox{}\hfill$\sq{\dval{ip}} = \max(\inc{\fnsqn_{N_{\ref*{rerr:line5}}}^{\dval{ip}}(\dval{dip})},\dval{dsn})$\hfill
and\hfill $\sta{\dval{ip}}=\inval$,\hfill so\hfill
\vspace{-1ex}
$$\nsq{\dval{ip}}\begin{array}[t]{@{~=~}l@{}} \sq{\dval{ip}}\decremented\\
\max(\inc{\fnsqn_{N_{\ref*{rerr:line5}}}^{\dval{ip}}(\dval{dip})}\decremented,\dval{dsn}\decremented)\\
\max(\fnsqn_{N_{\ref*{rerr:line5}}}^{\dval{ip}}(\dval{dip}),\dval{dsn}\decremented)\;.\vspace{-2ex}
\end{array}$$
Hence we need to show that
(i) $\fnsqn_{N_{\ref*{rerr:line5}}}^{\dval{ip}}(\dval{dip}) \leq \nsq{\dval{nhip}}$
and (ii) $\dval{dsn}\decremented \leq \nsq{\dval{nhip}}$.
\begin{enumerate}[\!\!(i)]
\item Since $\dval{dip}\in\akd[N_{\ref*{rerr:line2}}]{\dval{ip}}=\fnakD_{N_{\ref*{rerr:line5}}}^{\dval{ip}}$, we have\vspace{-1ex}
  $$\fnsqn_{N_{\ref*{rerr:line5}}}^{\dval{ip}}(\dval{dip}) =
    \fnnsqn_{N_{\ref*{rerr:line5}}}^{\dval{ip}}(\dval{dip}) \leq
    \fnnsqn_{N_{\ref*{rerr:line5}}}^{\dval{nhip}}(\dval{dip}) =
    \nsq{\dval{nhip}}$$
\noindent The inequality holds since the invariant is valid right before executing Line~\ref{rerr:line5}.
\item This case goes exactly as the corresponding case in \Sect{invariants}.
\endbox
\end{enumerate}
\end{description}
\end{itemize}\vspace{2ex}

When forbidding non-optimal self-entries---either by choosing one of the Resolutions~(5\ref{amb:5b}) or (5\ref{amb:5c}) of AODV
proposed on Page~\pageref{pro:disallow_selfentries_rrep}, or by storing the own sequence number in an optimal self-entry as
described in the \hyperlink{sss922b}{previous section}---all Variants
(8\ref{amb:8a})--(8\ref{amb:8f}) of the invalidation process described in this section behave exactly the same. Hence
all are loop free and route correct. This follows by the following invariants, which are established not for our
default specification of AODV, but for either of the resolutions without non-optimal self-entries,
still following (8\ref{amb:8f}) above.

\begin{proposition}\label{prop:no-self-invalidate}
Assume an interpretation of AODV that takes one of the Resolutions (2\ref{amb:2a},\,3c), (2\ref{amb:2c},\,3a) or
(2\ref{amb:2d},\,3a) in combination with (8\ref{amb:8f}) and any resolution of Ambiguities 5 and 6,
but not (5\hyperlink{amb:5a}{a}) and (6{a}) at the same time, and not (2\ref{amb:2d}) with (5\hyperlink{amb:5a}{a}) and (6\hyperlink{amb:6b}{b}); in all other ways it follows our default
specification of Sections~\ref{sec:types} and~\ref{sec:modelling_AODV}.
\begin{enumerate}[(1)]
\item Whenever Line~\ref{rerr:line2} of Pro.~\ref{pro:rerr} is executed by node \dval{ip} in state $N$ we have
\plat{$\sq[\dval{rip}]{\dval{ip}} < \dval{rsn}$} for all $(\dval{rip},\dval{rsn}) \in \xiN{\dval{ip}}(\dests)$ with
\plat{$\dval{rip}\in\akd{\dval{ip}}$} and \plat{$\nhp[\dval{rip}]{\dval{ip}}=\xiN{\dval{ip}}(\sip)$}.
\item Whenever node \dval{ip} makes a call $\inv{\rt}{\dests\!}$ in state $N\!$, then
  $\max(\!\inc{\sq[\dval{rip}]{\dval{ip}}},\dval{rsn})\mathbin=\mbox{}$ \plat{$\max(\sq[\dval{rip}]{\dval{ip}},\dval{rsn})=\dval{rsn}$}
  for all \plat{$(\dval{rip},\dval{rsn})\in\xiN{\dval{ip}}(\dests)$}.
\end{enumerate}
\end{proposition}

\pagebreak[3]
\begin{proofNobox}~
\newcommand{\starrsn}{\dval{rsn}^{\star}}
\newcommand{\starsip}{\dval{sip}^{\star}}
\begin{enumerate}[(1)]
\item Suppose Line~\ref{rerr:line2} of Pro.~\ref{pro:rerr} is executed in state $N$, and let
  $(\dval{rip},\dval{rsn}) \in \xiN{\dval{ip}}(\dests)$ with \plat{$\dval{rip}\in\akd{\dval{ip}}$}
  and \plat{$\dval{nhip}:=\nhp[\dval{rip}]{\dval{ip}}=\xiN{\dval{ip}}(\sip)$}.
        The values \plat{$\xiN{\dval{ip}}(\dests)$} and
        \plat{$\xiN{\dval{ip}}(\sip)$} stem from a received route 
	error message (cf.\ Lines~\ref{aodv:line2} and~\ref{aodv:line16} of Pro.~\ref{pro:aodv}).
	By \Prop{preliminaries}\eqref{it:preliminariesi}, a transition labelled
	$\colonact{R}{\starcastP{\rerr{\destsc}{\ipc}}}$ with $\destsc:=\xiN[N]{\dval{ip}}(\dests)$
	and $\ipc:=\xiN{\dval{ip}}(\sip)$ must have occurred before, say in state $N^\dagger\!$.
	By \Prop{ip=ipc}, the node casting this message is \plat{$\ipc=\xiN{\dval{ip}}(\sip)=\dval{nhip}$}.
By Invariant~\eqref{inv:starcast_rerr} we have
\plat{$\dval{rip}\in\fnikD_{N^\dagger}^{\dval{nhip}}$} and 
\plat{$\dval{rsn} = \fnsqn_{N^\dagger}^{\dval{nhip}}(\dval{rip})$}. 
Since (invalid) self-entries cannot occur,\footnote{When using Resolution (6\hyperlink{amb:6b}{b}), but not in
  combination with (5\hyperlink{amb:5a}{a}) and (2\ref{amb:2d}), invalid self-entries cannot occur by Invariant~\Eq{inv_self} in combination with
  \Prop{invarianti}(\ref{prop:invarianti_itemii}); otherwise under Resolutions (5\ref{amb:5b}) or (5\ref{amb:5c}) self-entries
  cannot occur at all.}
 it follows that $\dval{nhip}\neq\dval{rip}$.

Since \plat{$\dval{rip}\in\akd{\dval{ip}}$}, the last function call prior to state $N$ that created
or updated this valid routing table entry of node $\dval{ip}$, apart from an update of the
precursors only, must have been a call $\upd{*}{(\dval{rip},\starrsn,*,*,*,\starsip,*)}$,
where one of the first five clauses in the definition of \hyperlink{update}{$\fnupd$} was applied.
Let $N^\star$ be the state in which this call was made. Then $\starsip=\dval{nhip}$.
We consider all possibilities for this call.
\begin{description}
\item[Pro.~\ref{pro:aodv}, Lines~\ref{aodv:line10}, \ref{aodv:line14}, \ref{aodv:line18}:]
	The entry $\xiN[N^\star]{\dval{ip}}(\sip\comma0\comma\unkno\comma\val\comma1\comma\sip\comma\emptyset)$ is used for the update; 
	its next hop is $\xiN[N^\star]{\dval{ip}}(\sip)=\starsip=\dval{nhip}$ and its destination $\xiN[N^\star]{\dval{ip}}(\sip)=\dval{rip}$.
	This contradicts the conclusion that $\dval{nhip}\neq\dval{rip}$, and thus these cases cannot apply.
\item[Prop.~\ref{pro:rreq}, Line~\ref{rreq:line6}:]
 The update has the form $\xi_{N^\star}^{\dval{ip}}(\upd{\rt}{(\oip,\osn,\kno,\val,\hops+1,\sip,\emptyset)})$.
 Hence one of the first four clauses in the definition of $\fnupd$ was used, with
  $\xi_{N^\star}^{\dval{ip}}(\oip)=\dval{rip}$,
   $\xi_{N^\star}^{\dval{ip}}(\osn)=\sq[\dval{rip}]{\dval{ip}}$
 and 
 \plat{$\xi_{N^\star}^{\dval{ip}}(\sip)=\starsip=\dval{nhip}$}.
        The values $\xi_{N^\star}^{\dval{ip}}(\oip)$, $\xi_{N^\star}^{\dval{ip}}(\osn)$ and $\xi_{N^\star}^{\dval{ip}}(\sip)$ stem from a received
	RREQ message (cf.\ Lines~\ref{aodv:line2} and~\ref{aodv:line8}	of Pro.~\ref{pro:aodv}).
	 By \Prop{preliminaries}\eqref{it:preliminariesi}, a transition
        $\colonact{R}{\starcastP{\rreq{*}{*}{*}{*}{*}{\xi_{N^\star}^{\dval{ip}}(\oip)}{\xi_{N^\star}^{\dval{ip}}(\osn)}{\xi_{N^\star}^{\dval{ip}}(\sip)}}}$
        must have occurred before, say in state $N^\ddagger\!$.
  By \Prop{ip=ipc} the node casting this message is \plat{$\xi_{N^\star}^{\dval{ip}}(\sip)=\dval{nhip}$}.
  By Invariant~\eqref{inv:starcast_ii}, using that $\dval{nhip}\neq\dval{rip}$, we have\vspace{1pt}
  $\xi_{N^\star}^{\dval{ip}}(\osn) < \fnsqn_{N^\ddagger}^{\dval{nhip}}(\dval{rip})$
  or $\xi_{N^\star}^{\dval{ip}}(\osn) = \fnsqn_{N^\ddagger}^{\dval{nhip}}(\dval{rip})$ and
  \plat{$\fnstatus_{N^\ddagger}^{\dval{nhip}}(\dval{rip})=\val$}.
  In either case \plat{$\xi_{N^\star}^{\dval{ip}}(\osn) \leq \fnnsqn_{N^\ddagger}^{\dval{nhip}}(\dval{rip})$}.
  Since node $\dval{ip}$ handled the incoming RREQ message prior to the above-mentioned RERR message,
  the RREQ message was entered earlier in the FIFO queue of node \dval{ip} and hence transmitted earlier
  by node $\dval{nhip}$. So $N^\ddagger$ is prior to $N^\dagger\!$. We obtain
  \[\sq[\dval{rip}]{\dval{ip}}=\xi_{N^\star}^{\dval{ip}}(\osn)\leq\fnnsqn_{N^\ddagger}^{\dval{nhip}}(\dval{rip})
  \leq \fnnsqn_{N^\dagger}^{\dval{nhip}}(\dval{rip}) < \fnsqn_{N^\dagger}^{\dval{nhip}}(\dval{rip})=\dval{rsn}\ ,\]
where the first, second and last step have been established before; the third uses \Thm{state_quality},
and the penultimate step follows from the definition of net sequence numbers and
  \[\fnsqn_{N^\dagger}^{\dval{nhip}}(\dval{rip}) \geq
  \fnsqn_{N^\ddagger}^{\dval{nhip}}(\dval{rip})\geq \xi_{N^\star}^{\dval{ip}}(\osn)\geq 1\ ,\]
    which follows from \Prop{dsn increase} and Invariant~\eqref{inv:starcast_sqni}.
\item[Prop.~\ref{pro:rrep}, Line~\ref{rrep:line5}:]
                  The proof is similar to the one of Pro.~\ref{pro:rreq}, Line~\ref{rreq:line6}, the
                  main difference being that the information stems from an incoming RREP message; instead 
                  of $\oip$ and $\osn$ we use $\dip$ and $\dsn$, and instead of
                  Invariants~\eqref{inv:starcast_ii} and~\eqref{inv:starcast_sqni}
                  we use Invariants~\eqref{inv:starcast_iv} and~\eqref{inv:starcast_sqnii}.
\end{description}

\item We check all calls of $\fninv$.
\begin{description}
\item[Pro.~\ref{pro:aodv}, Line~\ref{aodv:line32}; Pro.~\ref{pro:pkt}, Line~\ref{pkt2:line10}; Pro.~\ref{pro:rreq}, Lines~\ref{rreq:line18}, \ref{rreq:line30}; Pro.~\ref{pro:rrep}, Line~\ref{rrep:line18}:]~\\
By construction of {\dests} (right before the invalidation call) if
$(\dval{rip},\dval{rsn})\mathbin\in\xiN{\dval{ip}}(\dests)$ then
\plat{$\dval{rsn}=\inc{\sq[\dval{rip}]{\dval{ip}}}$}.
\item[Pro.~\ref{pro:rerr}, Line~\ref{rerr:line5}:]
Immediately from (1).
\endbox
\end{description}

\end{enumerate}
\end{proofNobox}
By this proposition, Resolutions (8\ref{amb:8a})--(8\ref{amb:8f}) behave the same if non-optimal self-entries are forbidden.
Hence, by using \Cor{storing own sn as selfentry}, we obtain the following result.

\begin{cor}\rm\label{cor:loop free without self-entries}
Assume an interpretation of AODV that takes one of the Resolutions (2\ref{amb:2a},\,3c), (2\ref{amb:2c},\,3a) or
(2\ref{amb:2d},\,3a) in combination with any resolution of Ambiguities 5, 6 and 8,
but not (5\hyperlink{amb:5a}{a}) and (6{a}) 
at the same time and not (2\ref{amb:2d}) with (5\hyperlink{amb:5a}{a}) and (6\hyperlink{amb:6b}{b}); in all
other ways it follows our default specification of Sections~\ref{sec:types} and~\ref{sec:modelling_AODV}.
This interpretation is loop free and route correct.\endbox
\end{cor}

\subsubsection[Further Ambiguities]{\hypertarget{sss924}{Further Ambiguities}}\label{sssec:further}
\paragraph[Packet Handling for Unknown Destinations]{\hypertarget{sss924a}{\amb: Packet Handling for Unknown Destinations}}~

\noindent
In rare situations, often caused by node reboots, it may be possible that a node receives a data packet from another node for a destination for which it has no entry in its routing table at all. Such a situation cannot occur in our specification---this is a direct consequence of~\Prop{inv_nsqn}.
Nevertheless, since our specification given in \Sect{modelling_AODV} is intended to model {\em all} possible scenarios which might occur, 
we have to decide which rules AODV should follow. 
The RFC states that an error message should be generated
{\tt if [a node] gets a data packet destined to a node for which it does not have an active route}~\cite[Sect.~6.11]{rfc3561}.
It also states that the sequence number for the unreachable destination, to be listed in the error message, should be taken 
from the routing table and that
the {\tt neighboring node(s) that should receive the RERR are
  all those that belong to a precursor list of at least one of the unreachable
   destination(s)}. In this case neither the sequence number nor 
   the list of precursors are available.
There are two possible solutions:
\begin{enumerate}[(9a)]
\item\label{amb:9a} no error message is generated, since no information is available in the routing table---in
\Sect{modelling_AODV}, we follow that approach (\Pro{pkt}, Lines~\ref{pkt2:line21}--\ref{pkt2:line22});
\item\label{amb:9b} the error message is broadcast and the sequence number is set to unknown ($0$)---formalised in \Pro{pkt_unk_dest}. 
  This resolution makes sense only when using Resolutions (8\ref{amb:8c}) or (8\ref{amb:8d}) of the invalidation process:
  Resolutions (8\ref{amb:8b}), (8\ref{amb:8e}) and (8\ref{amb:8f}) would systematically ignore the broadcasted error message---so that there is no
  point in sending it---whereas with Resolution (8\ref{amb:8a}) this obviously leads to a decrease in destination
  sequence numbers and routing loops.
\end{enumerate}
  \algsetup{linenodelimiter=.,linenosize=\tiny}
  \begin{algorithm}[H]
    {\footnotesize
      \caption{Routine for packet handling (Resolution (9b))}
      \label{pro:pkt_unk_dest}
      \begin{algorithmic}[1]
\renewcommand{\algorithmicif}{$+$ \textbf{[}} %
\DEFPROCESS{\PKT}{\data\comma\dip\comma\oip\,\comma\,\ip\comma\sn\comma\rt\comma\rreqs\comma\queues}
\STATE\dots \COMMENT{Lines~\ref{pkt2:line2}--\ref{pkt2:line20} of Pro.~\ref{pro:pkt}}
\IF[route not in \rt]{$\dip\not\in\ikD{\rt}$}
	\broadcast{\rerr{\{(\dip,\,0)\}}{\ip}}\ .
	\aodv{\ip}{\sn}{\rt}{\rreqs}{\queues}
\ENDIFii
\STATE\dots \COMMENT{Lines~\ref{pkt2:line23}--\ref{pkt2:line24} of Pro.~\ref{pro:pkt}}

	\end{algorithmic}
    }
  \end{algorithm}

\noindent
In \Sect{invariants} we have shown that Resolution (9a) is loop free.
\hypertarget{amb9modifications}{Resolution (9b) is loop free as well, since all invariants of \Sect{invariants} and their proofs
  remain valid, with the following modifications:}

\begin{itemize}
\item \Prop{starcastrerr} (Invariant~\eqref{inv:starcast_rerr}) has to be weakened: it holds only for pairs $(\ripc,\rsnc)$ with \mbox{$\rsnc\mathbin>0$}.
  In the adapted proof there is an extra case to consider, but there $\rsnc=0$. This proposition is only used in the proof of \Prop{inv_nsqn}; we show below that the weaker form is sufficient.
\item In the proof of \Prop{invarianti}\eqref{prop:invarianti_dests} there is an extra case to consider, which is trivial.
\item In the last case of the proof of \Prop{inv_nsqn}, when using
  Invariant~\eqref{inv:starcast_rerr},
  there is an extra case to consider, namely that $\dval{dsn}=0$. In that case surely
  $\dval{dsn}\decremented =0\leq \nsq{\dval{nhip}}$, which we needed to establish.
\end{itemize}
The proof of \Prop{no-self-invalidate} is no longer valid when using Resolution (9b) of the packet
handling for unknown destinations, since it uses a version of Invariant~\eqref{inv:starcast_rerr} that no longer holds.
This indicates that Resolutions (8\ref{amb:8a}), (8\ref{amb:8b}) and (8\ref{amb:8c}) of the invalidation process are not necessarily
compatible with Resolution~(9b), even if non-optimal self-entries are forbidden.\footnote{It turns out
  that Resolutions (8\ref{amb:8b}) and (8\ref{amb:8c}) are compatible with (9b) after all; we skip the proof of this claim.}
On the other hand, it can be argued that in Resolution (9a) the originator node that initiated the
sending of the data packet might send more packets, which increases network traffic without delivering the data.

\paragraph[Setting the Own Sequence Number when Generating a RREP Message]{\hypertarget{sss924b}{\amb: Setting the Own Sequence Number when Generating a RREP Message}}~

\noindent
In the RFC, the way in which a destination of a route request updates its own sequence number before
initiating a route reply is described in two ways:
\begin{quote}\raggedright\small
``{\tt Immediately before a destination node originates a RREP in
      response to a RREQ, it MUST update its own sequence number to the
      maximum of its current sequence number and the destination
      sequence number in the RREQ packet.\rm''\\\hfill \cite[Sect.~6.1]{rfc3561}}
 \end{quote}
\begin{quote}\raggedright\small
``{\tt If the generating node is the destination itself, it MUST increment
   its own sequence number by one if the sequence number in the RREQ
   packet is equal to that incremented value.  Otherwise, the
   destination does not change its sequence number before generating the
   RREP message.\rm''\hfill \cite[Sect.~6.6.1]{rfc3561}}
 \end{quote}
In most cases these two descriptions yield the same result (because the destination
sequence number in the RREQ message is usually not more than 1 larger than the destination's own
sequence number). However, this is not guaranteed.

\begin{figure}[ht]
\vspace{-2ex}
\begin{exampleFig}{Failing to update the own sequence number
      before issuing a route reply}{fig:increaseSNfollowingRFC}
\FigLine[xslxsr]%
  {The initial state;\\$d$ established a route to $s$ via RREQ-RREP cycle.}%
  {fig/ex_increaseSNfollowingRFC_1}{}
  {The topology changes;\\ nodes $s$ and $d$ invalidate entries.}
  {fig/ex_increaseSNfollowingRFC_2}{}
\FigLine[xslxsr]%
  {The topology changes;\\$a$ broadcasts a new RREQ message destined to $x$.}%
  {fig/ex_increaseSNfollowingRFC_3}{}
  {The topology changes again;\\ node $s$ invalidates entries.}
  {fig/ex_increaseSNfollowingRFC_4}{}
\FigLine[xslxsr]%
  {The topology changes;\\$s$ broadcasts a new RREQ message destined to $d$.}%
  {fig/ex_increaseSNfollowingRFC_5}{}
  {$d$ unicasts a RREP back to $s$\\no update occurs at $s$.}
  {fig/ex_increaseSNfollowingRFC_6}{}
\end{exampleFig}
\vspace{-2ex}
\end{figure}

This is illustrated in \Fig{increaseSNfollowingRFC}.
In the initial state node $s$ has a route to $d$, with destination
sequence number ($2$) equal to $d$'s own sequence number; this is
default behaviour of AODV\@. Due to a link break between $s$
and $d$, node $s$ increments its destination sequence number for $d$
when invalidating the entry (\Fig{increaseSNfollowingRFC}(b)). Afterwards, in
\Fig{increaseSNfollowingRFC}(c), the link comes up again, and when
$d$ forwards a RREQ message (from another node $a$, destined to an
arbitrary node $x$ that is not in the vicinity) to its neighbour
$s$, node $s$ validates its $1$-hop route to $d$, without changing its
destination sequence number. These events (link break -- invalidation --
link coming back up) are repeated at least once (Part (d)),
resulting in a destination sequence number for $d$ at node $s$ that is at least $2$ higher than $d$'s own sequence number. Now, when $s$ searches
for a route to $d$ (\Fig{increaseSNfollowingRFC}(e)), $d$ will not update its own sequence number when
sending a route reply to $s$, so the route reply will have outdated
information (a too low sequence number) from the perspective of $s$,
and thus will be ignored by $s$. No matter how often $s$ sends a new
route request to $d$, it will never receive an answer that is good enough
to restore its routing table entry to $d$.

  \algsetup{linenodelimiter=.,linenosize=\tiny}
  \begin{algorithm}[H]
    {\footnotesize
      \caption{RREQ handling (Resolution (10b))\label{amb:10b}}
      \label{pro:rreq_inc_sqn}
      \begin{algorithmic}[1]
\DEFPROCESS{\RREQ}{\hops\comma\rreqid\comma\dip\comma\dsn\comma\dsk\comma\oip\comma\osn\comma\sip\,\comma\,\ip\comma\sn\comma\rt\comma\rreqs\comma\queues}
	\STATE\dots \COMMENT{Lines~\ref{rreq:line2}--\ref{rreq:line8} of Pro.~\ref{pro:rreq}}
	\PAR
		\IF[this node is the destination and the sequence number has to be updated]{$\dip=\ip \ans \inc{\sn}=\dsn$}
			\UPD{\sn:=\inc{\sn}}	\COMMENT{update the sqn of \ip}							
			\STATE\dots \COMMENT{Lines~\ref{rreq:line13}--\ref{rreq:line19} of Pro.~\ref{pro:rreq}}
		\ELSIF[this node is the destination and the sequence
	number need no update]{$\dip=\ip \ans \inc{\sn}\not=\dsn$}							
			\STATE\dots \COMMENT{Lines~\ref{rreq:line13}--\ref{rreq:line19} of Pro.~\ref{pro:rreq}}
		\ELSIF[this node is not the destination node]{$\dip\not=\ip$}
			\STATE\dots \COMMENT{Lines~\ref{rreq:line21}--\ref{rreq:line36} of Pro.~\ref{pro:rreq}}
	\ENDIFii
	\ENDPAR

	\end{algorithmic}
    }
  \end{algorithm}

In our specification we resolved this contradiction by following Sect.~6.1 of the RFC, in defiance
of Sect.~6.6.1. The alternative is obtained by modifying the RREQ handling process as indicated in Pro.~\ref{pro:rreq_inc_sqn}.
As the above example shows, this alternative leads to a severely
handicapped version of AODV, in which certain routes (in the example
the one from $s$ to $d$) can not be established.\footnote{\raggedright%
On the IETF MANET mailing list
(\url{http://www.ietf.org/mail-archive/web/manet/current/msg02589.html})
I.~Chakeres proposes a third resolution of this ambiguity, namely
``{\tt Immediately before a destination node issues a route reply in
  response to a RREQ, it MUST update its own sequence number to the
  maximum of its current sequence number and the destination
  sequence number in the RREQ packet plus one (1).}''\newline
As this is not a possible reading of the RFC, it ought to be construed as proposal for improvement
of AODV.}

\subsubsection{Further Assumptions}

During the creation of our specification (cf.\ \Sect{modelling_AODV}), we did not only come along 
some ambiguities, we also found some unspecified cases---we were forced to 
specify these situations on our own.

\paragraph{Recording and Invalidating the Truly Unknown Sequence Number}~\\
When creating a routing table entry to a new destination---not already present in the
routing table---for which no destination sequence number is known (i.e.\ in response to an AODV
control message from a neighbour; following Lines~\ref{aodv:line10},~\ref{aodv:line14}
and~\ref{aodv:line18} of \Pro{aodv}), the RFC does not stipulate how to fill in the
destination-sequence-number field in the new entry. It does say
\begin{quote}\raggedright\small``{\tt
   The sequence number is either determined from the
   information contained in the control packet, or else the valid
   sequence number field is set to false.}\makebox[0pt]{''}\\\hfill\hfill\cite[Sect.~6.1]{rfc3561}
 \end{quote}
Accordingly, the sequence-number-status flag in the entry is set to \unkno,
but that does not tell what to fill in for the destination sequence number itself.
Here, following the implementation AODV-UU~\cite{AODVUU}, we use the special value $0$,
\index{sequence number!truly unknown}%
indicating a \emph{truly unknown destination sequence number}.
As this value does not represent a regular sequence number, we do not increment it when
invalidating the entry.

\paragraph[Packet Handling]{Packet Handling}~\\
Even though not specified in the RFC, our model of AODV includes a
mechanism for handling data packets---this is necessary
to trigger any AODV activity.  A data packet injected at a node $s$ by
a client of the protocol (normally the application layer in the
protocol stack) for delivery at a destination $d$ towards which $s$ has
no valid route, is inserted in a queue of packets for $d$ maintained
by node~$s$. In case there is no queue for $d$ yet, such a queue is
created and a route discovery process is initiated, by means of a new
route request. As long as that process is pending, no new route
request should be issued when new packets for $d$ arrive; for it could
be that packets for $d$ are injected by the application layer at a
high rate, and sending a fresh route request for each of them would
flood the protocol with useless RREQ messages. For this reason we
await the route reply corresponding to the request, or anything else
that creates a route to $d$. Afterwards packets to $d$ can be send,
and the queue is emptied out. In case the route to $d$ is invalidated
before the queue is empty, it is appropriate to initiate a new route
discovery process, by generating a fresh route request.  To this end we
\index{request-required flag}%
created the ``request-required'' flag, one for each queue, that is set
when the route to the destination is invalidated, and unset when a
new route request has been issued. The only sensible way we see to
omit such a flag would be to use the non-existence of a queue of data
packets for $d$ as the trigger to initiate a route request when a data
packet for $d$ is posted at node $s$.  But for that to work one
would have to drop the entire queue of packets waiting for
transmission towards $d$ when the route to $d$ is invalidated, just as
packets are dropped when an intermediate node on the path towards $d$
loses its connection to the next hop.

\paragraph{Receiving a RREP Message}~\\
When an (intermediate) node receives a RREP message destined for a node $s$, it might happen that the node
has an invalid routing table entry for $s$ only.
The RFC does not consider this situation; however, this case \emph{can} occur and must be specified. For our
specification we decided that under these circumstances 
the AODV control message is lost and no error message is generated.

\subsection{Implementations}\label{ssec:implementations}

To show that
the ambiguities we found in the RFC and the associated problems are not only theoretically
driven, but {\em do} occur in practice, we analyse five different open source implementations of AODV:
\begin{itemize}
\item {\em AODV-UU}~\cite{AODVUU} is an RFC compliant implementation of AODV, developed at Uppsala University. {\url{http://aodvuu.sourceforge.net/}}
\item {\em Kernel AODV}~\cite{AODVNIST} is developed at NIST and is another RFC compliant implementation of AODV. {\url{http://w3.antd.nist.gov/wctg/aodv_kernel/}}
\item {\em AODV-UIUC}~\cite{Kawadia03} (University of Illinois at Urbana-Champaign) is an implementation that is based on an early draft (version 10) of AODV. {\url{http://sourceforge.net/projects/aslib/}}
\item {\em AODV-UCSB}~\cite{CB04} (University of California, Santa-Barbara) is another
  implementation based on an early draft (version 6).
  {\url{http://moment.cs.ucsb.edu/AODV/aodv-ucsb-0.1b.tar.gz}}
\item {\em AODV-ns2\/} is an AODV implementation in the ns2 network simulator~\cite{NS2}, originally developed
  by the CMU Monarch project and improved upon later by S. Das and E. Belding-Royer (the authors of the AODV RFC \cite{rfc3561}). It is based on an early draft (version 8) of AODV. It is frequently used by academic and industry researchers to simulate AODV. 
\url{http://ns2.sourcearchive.com/documentation/2.35~RC4-1/aodv_8cc-source.html}
\end{itemize}
Even though the latter three implementations of AODV are not RFC compliant, they {\em do} capture  the main aspects of the AODV protocol, as specified in the RFC~\cite{rfc3561}. As we have shown in the previous section, implementing the AODV protocol based on the RFC specification does not necessarily guarantee loop freedom. Therefore, we look at these five concrete AODV implementations to determine whether any of them is susceptible to routing loops.
AODV-UU, Kernel AODV and AODV-UIUC maintain an invalidation procedure that conforms to Resolution (8\ref{amb:8a}),
whereas AODV-UCSB and AODV-ns2 follow Resolution (8\ref{amb:8b}). Since both resolutions give rise to routing
loops when used in combination with non-optimal self-entries,
 we examine the code of these implementations to see if routing loops
such as the one described in \Fig{loopsfromselfentries}
occur. The results of this analysis are summarised in Table~\ref{tb:analysis}. 

\begin{table}[bht]
\vspace*{1ex}
\centering
{\begin{tabular}{@{}|l|p{0.68\columnwidth}|@{}}
\hline
{\bf Implementation} & {\bf Analysis}\\
\hline\hline
AODV-UU~\cite{AODVUU} 		&Loop free, since self-entries are explicitly excluded.\\
\hline
Kernel AODV~\cite{AODVNIST} 	&Loop free, due to optimal self-entries.\\
\hline
AODV-UIUC~\cite{Kawadia03} 	&{Yields routing loops, through sequence number reset.}\\
\hline
AODV-UCSB~\cite{CB04}		&{Yields routing loops, through sequence number reset.}\\
\hline
AODV-ns2		&Yields routing loops, since it implements Resolution (8\ref{amb:8b}) of the
                        invalidation procedure presented in \SSSect{interpretation_invalidate} and
                        does allow self-entries.\\
\hline
\end{tabular}}
\caption[Analysis of AODV implementations]
    {\em Analysis of AODV implementations}
\label{tb:analysis}
\end{table}

In AODV-UU, self-entries are never created because a check is always performed on an incoming RREP
message to make sure that the destination IP address is not the same as the node's own IP address,
just as in Resolution (5\ref{amb:5b}). By \Cor{loop free without self-entries},
this interpretation of the RFC is loop free. 

In Kernel AODV, an optimal self-entry is always maintained by every node in the network, just as
in Resolution (6\hyperlink{amb:6b}{b}).
By \Cor{loop free without self-entries}, this interpretation of the RFC is loop free.

Both AODV-UIUC and AODV-UCSB allow non-optimal self-entries to occur in nodes (Resolution (5\hyperlink{amb:5a}{a}) of
\hyperlink{sss922a}{Ambiguity 5}).
These are generated based on information contained in received RREP messages. While
self-entries are allowed, the processing of RERR messages in AODV-UIUC and AODV-UCSB does not adhere
to the RFC specification (or even the draft versions that these implementation are based upon).
Due to this non-adherence, we are unable to re-create the routing loop example of \Fig{loopsfromselfentries}.
However, if both AODV-UIUC and AODV-UCSB were to strictly follow the RFC specification
with respect to the RERR processing, loops would have been created.

Even though the routing loop example of \Fig{loopsfromselfentries} could not be recreated in AODV-UIUC or AODV-UCSB,
  both implementations allow a decrease of destination sequence numbers in \rtes to occur, by
  following Resolution (2\ref{amb:2b}).\footnote{AODV-ns2 follows Resolution
    (2\ref{amb:2a}), whereas AODV-UU follows (2\ref{amb:2d}). Kernel AODV is not compliant with the
    RFC in this matter and operates differently.} This
  gives rise to routing loops in the way described in \SSect{decreasingSQN}.

In AODV-ns2, self-entries are allowed to
occur in nodes. Unlike AODV-UIUC and AODV-UCSB, the processing of RERR messages
follows the RFC specification. However, whenever a node generates a RREQ message,
sequence numbers are incremented by two instead of by one as specified in the RFC\@.
We have modified the AODV-ns2 code such that sequence numbers are incremented by one whenever
a node generates a RREQ message, and are able to replicate the routing loop example
presented in \cite{AODVloop}\footnote{The example in \cite{AODVloop} is a simplification of the one in \Fig{loopsfromselfentries},
but is based on the interpretation of AODV without the sequence-number-status flag, following
Resolution~(2\ref{amb:2d}). The example of
\Fig{loopsfromselfentries} itself works equally well in the presence of that flag.} in the ns2 simulator, with the results
showing the existence of a routing loop between nodes $s$ and $x$. However, even if the code
remains unchanged and sequence numbers are incremented by two, AODV-ns2 can still
yield loops; the example is very similar to the one presented and only varies in subtle details.

In sum, we discovered not only that three out of five AODV implementations can produce routing loops, 
but also that there are essential differences between the various implementations in various
aspects of protocol behaviour. This is due to different interpretations of the RFC\@.

\subsection{Summary}\label{ssec:Interpretation_summary}
The following table summarises the ambiguities we discovered,
as well as their consequences.
The resolutions coloured red lead to unacceptable protocol behaviour, such as routing loops.
The white and green resolutions are all acceptable readings of the RFC\@; the green ones have been
chosen in our default specification of Sections~\ref{sec:types} and~\ref{sec:modelling_AODV}.
The section numbers refer to the RFC \cite{rfc3561}.
\begin{center}
\newcounter{tabcounti}
\newcounter{tabcountii}
\setcounter{tabcounti}{1}
\setcounter{tabcountii}{1}
\definecolor{lightred}{rgb}{1.0,0.8,0.8}
\definecolor{lightgreen}{rgb}{0.8,1.0,0.8}
\newcommand{\no}{\cellcolor{lightred}}
\newcommand{\yes}{\cellcolor{lightgreen}}
\newcommand{\nri}{\arabic{tabcounti}}
\newcommand{\nrii}{\nri\alph{tabcountii}\stepcounter{tabcountii}.}
{\small
\setlength\doublerulesep{4em}
\begin{longtable}{@{\,}l@{\,}|p{0.43\textwidth}|p{0.43\textwidth}|@{}}
\hline
\multicolumn{3}{|c|}{\cellcolor[gray]{0.8}\bf \hyperlink{sss921}{Updating Routing Table Entries}}\\
\hline
\multicolumn{3}{c}{}\\[-2.1ex]
\hline
\multicolumn{3}{|@{\,}l|}{\bf \nri. \hyperlink{sss921a}{Updating the Unknown Sequence Number in Response to a Route Reply}}\\
\hline
\nrii&\no the destination sequence number (DSN) is copied from the RREP message (Sect 6.7)&decrement of sequence numbers and loops\\
\cline{2-3}
\nrii&\yes routing table is not updated when the information that it has is ``fresher''
(Sect.\ 6.1) & does not cause loops; used in our specification\\
\cline{2-3}
\multicolumn{3}{c}{}\\[-2.1ex]
\hline
\multicolumn{3}{|@{\,}l|}{
    \stepcounter{tabcounti}\setcounter{tabcountii}{1}%
    \bf \nri. \hyperlink{sss921b}{Updating with the Unknown Sequence Number} (Sect.\ 6.5)}\\
\hline
\nrii&no update occurs & does not cause loops, but  opportunity to improve routes is missed
\\
\cline{2-3}
\nrii&\no overwrite any routing table entry by an update with an unknown DSN& decrement of sequence numbers and loops\\
\cline{2-3}
\nrii&\yes use the new entry with the old DSN & does not cause loops; used in our specification\\
\cline{2-3}
\nrii& use the new entry with the old DSN and DSN-flag & does not cause loops\\
\cline{2-3}
\multicolumn{3}{c}{}\\[-2.1ex]%
\hline
\multicolumn{3}{|@{\,}l|}{
    \stepcounter{tabcounti}\setcounter{tabcountii}{1}%
    \bf \nri. \bf \hyperlink{sss921c}{More Inconclusive Evidence on Dealing with the
    Unknown Sequence Number} (Sect.\ 6.2)}\\\hline
\nrii&\yes update when \emph{incoming} sequence number is unknown &
    supports Interpretations 2b or 2c above; used in our specification \\
\cline{2-3}
\nrii&\no update when \emph{existing} sequence number is \emph{marked
    as} unknown & decrement of sequence numbers and loops;
    implies 1a and 2a\\
\cline{2-3}
\nrii& update when no \emph{existing} sequence number is known
    & supports Interpretation 2a above \\
\cline{2-3}
\multicolumn{3}{c}{}\\[-2.1ex]
\hline
\multicolumn{3}{|@{\,}l|}{
    \stepcounter{tabcounti}\setcounter{tabcountii}{1}%
    \bf \nri. \hyperlink{sss921d}{Updating Invalid Routes}}\\
    \hline
\cline{2-3}
\nrii& \yes update an invalid route when the new route has the same
    sequence number (Sect.\ 6.1) & does not cause loops; used in our specification\\
\cline{2-3}
\nrii& \no do not update an invalid route when the new route has the same
    sequence number (Sect.\ 6.2) & results in handicapped version of AODV,
    in\newline which many broken routes will never be repaired.\\
\cline{2-3}
\multicolumn{3}{c}{}\\[-1.8ex]
\hline
\multicolumn{3}{|c|}{\cellcolor[gray]{0.8}\bf \hyperlink{sss922}{Self-Entries in Routing Tables}}\\
\hline
\multicolumn{3}{c}{}\\[-2.1ex]%
\hline
\multicolumn{3}{|@{\,}l|}{
    \stepcounter{tabcounti}\setcounter{tabcountii}{1}%
    \bf \nri. \hyperlink{sss922a}{(Dis)Allowing Self-Entries}}\\\hline
\nrii&\yes allow (arbitrary) self-entries& loop free if used with appropriate \fninv; used in our specification\\
\cline{2-3}
\nrii&disallow (non-optimal) self-entries; \newline if self-entries would be created, ignore message&does not cause loops\\
\cline{2-3}
\nrii&disallow (non-optimal) self-entries;\newline if self-entries would be created, forward message&does not cause loops\\
\cline{2-3}
\multicolumn{3}{c}{}\\[-2.1ex]
\hline
\multicolumn{3}{|@{\,}l|}{
    \stepcounter{tabcounti}\setcounter{tabcountii}{1}%
    \bf \nri. \hyperlink{sss922b}{Storing the Own Sequence Number}}\\\hline    
\nrii&\yes store sequence number as separate value&does not cause loops; used in our specification\\
\cline{2-3}
\nrii&store sequence number inside routing table&does not cause loops\\
\cline{2-3}
\multicolumn{3}{c}{}\\[-1.8ex]
\hline
\multicolumn{3}{|c|}{\cellcolor[gray]{0.8}\bf \hyperlink{sss923}{Invalidating Routing Table Entries}}\\
\hline
\multicolumn{3}{c}{}\\[-2.1ex]
\hline
\multicolumn{3}{|@{\,}l|}{
    \stepcounter{tabcounti}\setcounter{tabcountii}{1}%
    \bf \nri. \hyperlink{sss923a}{Invalidating Entries 
    in Response to a Link Break or Unroutable Data Packet} 
    \mbox{\small (Sect.\ 6.11)}}\\\hline    
\nrii&\yes``it'' refers to routing table entry&does not cause loops; used in our specification\\
\cline{2-3}
\nrii&\no ``it'' refers to DSN&loops\\
\cline{2-3}
\pagebreak[4]
\multicolumn{3}{c}{}\\[-2.1ex] \hline
\multicolumn{3}{|@{\,}l|}{
    \stepcounter{tabcounti}\setcounter{tabcountii}{1}%
    \bf \nri. \hyperlink{sss923b}{Invalidating Entries in Response to a Route Error Message}}\\\hline    
\nrii&\no copy DSN from RERR message (Sect.\ 6.11)&decrement of sequence numbers and loops\newline
(when allowing self-entries (Interpretation 5a))\\
\cline{2-3}
\nrii&\no no action if the DSN in the routing table is larger than the one in the RERR
  mess.\ (Sect.\ 6.1 \& 6.11)&loops (when allowing self-entries) \\
\cline{2-3}
\nrii&\no take the maximum of the DSN of the routing
    table and the one from the RERR message &
  loops (when allowing self-entries)\\
\cline{2-3}
\nrii& take the maximum of the increased DSN of the routing table and the one
  from the RERR mess.
 &does not cause loops\\
\cline{2-3}
\nrii& combine 8b and 8d
 &does not cause loops\\
\cline{2-3}
\nrii& \yes only invalidate if the DSN in the routing table is smaller than the one from
 the RERR message &does not cause loops; used in our specification\\
\cline{2-3}
\multicolumn{3}{c}{}\\[-2.1ex]
\hline
\multicolumn{3}{|c|}{\cellcolor[gray]{0.8}\bf \hyperlink{sss924}{Further Ambiguities}}\\
\hline
\multicolumn{3}{c}{}\\[-2.1ex]
\hline
\multicolumn{3}{|@{\,}l|}{
    \stepcounter{tabcounti}\setcounter{tabcountii}{1}%
    \bf \nri. \hyperlink{sss924a}{Packet Handling for Unknown Destinations} (Sect.\ 6.11)}\\\hline    
\nrii&\yes do nothing&the sender is not informed and keeps sending;\newline used in our specification\\
\cline{2-3}
\nrii& broadcast RERR message with unknown DSN&loop free if used with adequate \fninv\\
\cline{2-3}
\multicolumn{3}{c}{}\\[-2.1ex]
\hline
\multicolumn{3}{|@{\,}l|}{
    \stepcounter{tabcounti}\setcounter{tabcountii}{1}%
    \bf \nri. \hyperlink{sss924a}{Setting the Own Sequence Number when Generating a RREP Message}}\\\hline    
\nrii& \yes taking max (Sect.~6.1)&used in our specification\\
\cline{2-3}
\nrii& \no taking the ``conditional increment'' (Sect.~6.6.1)&loss of RREP message\\
\cline{2-3}
\multicolumn{3}{c}{}\\[-1ex]
\caption[Different interpretations and consequences of ambiguities in the RFC]
    {\em Different interpretations and consequences of ambiguities in the RFC}
\label{tb:diff_interaction_of_rfc}
\end{longtable}
}
\end{center}
\vspace{-3.8ex}
The above classification of ambiguities and their resolutions can be used to calculate the number of possible readings of the RFC\@.
The table shows that the resolution for Ambiguity 3 is uniquely determined by the choice of resolutions for Ambiguities 1 and 2; 
except for the case of taking (1a) in combination with (2a); here Resolutions (3b) and (3c) are possible.
Hence Ambiguity 3 only adds one new variant.
In sum we have 
$[(2\times 4)+1]\times 2 \times 3 \times 2 \times 2 \times 6
\times 2 \times 2 = 5184$
possible interpretations of the AODV RFC\@.
Only $\big([(1\times 3)+0]\times 1 \times [(3 \times 2 \times 1 \times 3
\times 2) + (5 \times 1 \times 5)] \times 1\big)  -  5 = 178$
are loop free and without major flaws. (Here the first ``5'' refers to all
resolutions of Ambiguities 5 and 6 except for the combination of (5a) and (6a);
the second ``5'' refers to the first 3 resolutions of Ambiguity 8 and both resolutions of Ambiguity 9,
except for the combination of (8a) and (9b); and the last ``5'' deducts the combinations of (6b) with
(2d), (5a) and one of the second ``5'').

All these ambiguities, missing details and misinterpretations of the RFC show that the specification of
a reasonably rich protocol such as AODV cannot be described by simple
(English) text; is  {\em has to be done} using formal methods in a precise way.

\newcommand{\Fi}{\mbox{F$_1$}\xspace}
\newcommand{\Fii}{\mbox{F$_2$}\xspace}
\newcommand{\Fiii}{\mbox{F$_3$}\xspace}
\section{Formalising Temporal Properties of Routing Protocols}\label{sec:properties}

Our formalism enables verification of correctness properties. While
some properties, such as loop freedom and route correctness, are invariants on routing
tables, others require reasoning about the temporal order of transitions.
Here we use Linear-time Temporal Logic (LTL)~\cite{Pnueli77}
\index{Linear-time Temporal Logic (LTL)}%
 to specify and discuss 
two of such properties, namely \emph{route discovery} and \emph{packet delivery}.

Let us briefly recapitulate the syntax and semantics of LTL\@. The logic is built from a set of
\phrase{atomic propositions}. Such propositions stand for  facts that may hold
at some point (in some state) during a protocol run. An example is ``two nodes are connected in the
(current) topology''.

LTL formulas are interpreted on  paths in a transition system, where
each state is labelled with the atomic propositions that hold in that state.  A
\phrase{path} is an alternating sequence of states and transitions, starting from a state and either
being infinite or ending in a state,
such that each transition in the sequence goes from the state before to the state after it.
An atomic proposition $p$ holds on a path $\pi$ if $p$ holds in the first state of $\pi$.

\index{G@{$\mathbf{G}$ (globally)}}%
\index{F@{$\mathbf{F}$ (eventually)}}%
LTL \cite{Pnueli77} uses the temporal operators $\mathbf{G}$ and $\mathbf{F}$. 
The formulas $\mathbf{G} \phi$ and $\mathbf{F} \phi$ mean that $\phi$ holds
\emph{globally} in all states on a path, and \emph{eventually} in some state, respectively. Here a formula $\phi$ is deemed to
\emph{hold in a state on a path $\pi$} iff it holds for the remainder of $\pi$ when starting from that
state. In later work on LTL, two more temporal operators were added---the \emph{next-state} and the
\emph{until} operator; these will not be needed here.
LTL formulas can be combined by the logical connectives \emph{conjunction} $\wedge$,
\emph{disjunction} $\vee$, \emph{implication} $\Rightarrow$ and \mbox{\emph{negation} $\neg$}.
An LTL formula holds for a transition system iff it holds for all \emph{complete} paths in the
\index{path!complete}%
system starting from an initial state. A path is complete iff it leaves no
transitions undone without a good reason; in the original work on temporal logic \cite{Pnueli77} the
complete paths are exactly the infinite ones, but in \SSect{progress} we will propose a different
concept of completeness (cf.\ Definition~\ref{df:complete path}).

Below we will apply LTL to the transition system $\mathcal{T}$ generated by the structural
operational semantics of \awn from an arbitrary \awn specification, and from our specification of
AODV in particular. Here we use two kinds of atomic propositions. 
\index{atomic propositions}%
\label{pg:typesofpredicates}%
The first kind are predicates
on the states (or network expressions) $N$ that are fully determined by the (local) values of all
variables maintained by the nodes in the network, as well as by the current topology, i.e.\ by
$\xiN{\dval{ip}}$, $\zetaN{\dval{ip}}$ and $\RN{\dval{ip}}$ for all $\dval{ip}\in\IP$.
The second kind are predicates on transitions $N\ar{\ell}N'$ that are fully determined
either by the label $\ell$ of the transition, or by transition-labels appearing in the derivation
from the structural operational semantics of \awn of a $\tau$-transition---compare the
$\colonact{R}{\starcastP{m}}$-transitions in \SSect{transition invariants}.

To incorporate the transition-based atomic propositions into the framework of temporal logic, we
  perform a translation of the transition-labelled transition system
  $\mathcal{T}$ into a state-labelled transition system~$\mathcal{S}\!\!$, and apply LTL to the latter.
  A suitable translation, proposed in \cite{DV95}, introduces new
  states halfway the existing transitions, thereby splitting a transition $\ell$ into $\ell;\tau$, and
  attaches transition labels, or predicates evaluated on transitions, to the new mid-way states.
  Since we also have state-based atomic propositions, we furthermore declare any atomic proposition
  that holds in state $N'$ to also hold for the new state midway a transition $N\ar{\ell}N'$.

Below we use LTL to formalise properties that say that whenever a precondition $\phi^{\it pre}$ holds
in a reachable state, the system will eventually reach a state satisfying the postcondition $\phi^{\it post}$.
Such a property is called an \phrase{eventuality property} in \cite{Pnueli77}; it is formalised by the LTL formula
\begin{equation}\label{eq:basic ltl property}
\mathbf{G} \big(\phi^{\it pre} \ims \mathbf{F}\phi^{\it post}\big) \;.
\end{equation}
However, sometimes we want to guarantee such a property only when a side condition $\psi$
keeps being satisfied from the state where $\phi^{\it pre}$ holds until $\phi^{\it post}$ finally holds.
There are three ways to formalise this:
\begin{equation}\label{eq:ltl property with side condition}
\mathbf{G} \big((\phi^{\it pre} \wedge \mathbf{G}\psi) \Rightarrow \mathbf{F}\phi^{\it post}) \qquad
\mathbf{G} \big((\phi^{\it pre} \wedge \psi\mathbf{W}\phi^{\it post}) \Rightarrow \mathbf{F}\phi^{\it post}) \qquad
\mathbf{G} \big(\phi^{\it pre} \Rightarrow \mathbf{F}(\phi^{\it post}\vee \neg\psi)\big)\,.
\end{equation}
The first formula is derived from \Eq{basic ltl property} by adding to the precondition $\phi^{\it pre}$
the requirement that $\psi$ is valid as well, and remains valid ever after. If that precondition is
not satisfied, nothing is required about $\phi^{\it post}$. One might argue that this precondition
is too strong: it requires the side condition to be valid forever, even after $\phi^{\it post}$ has
occurred. The second formula addresses this issue by weakening the precondition $\phi^{\it pre} \wedge \mathbf{G}\psi$.
It uses a binary temporal operator $\mathbf{W}$---the \emph{weak until} operator---that can be
expressed in terms of $\mathbf{G}$ and the (strong) until operator.
The meaning of an expression $\psi \mathbf{W}\phi$ is that either $\psi$ holds forever, or at some
point $\phi$ holds and until then $\psi$ holds. In other words, $\psi$ holds until we reach a state
where $\phi$ holds, or forever if the latter never happens.

Although the precondition of the second formula is weaker than the one of the first,
as a whole the two formulas are equivalent: they are satisfied by all runs of the system, except
those for which\vspace{-2pt}
\begin{itemize}\parskip 0pt \topsep 0pt \itemsep 0pt
\item[$-$] at some point $\phi^{\it pre}$ holds,
\item[$-$] and from that point onwards $\psi$ remains valid,
\item[$-$] yet never a state occurs satisfying $\phi^{\it post}$.
\vspace{-2pt}
\end{itemize}
\pagebreak[3]
Both formulas are also equivalent to the third formula in \Eq{ltl property with side condition}.
It can be understood to say that once $\phi^{\it pre}$ holds, we will eventually reach a state where
$\phi^{\it post}$ holds, except that we are off the hook (in the sense that nothing further is
required) when (prior to that) we reach a state where $\psi$ fails to hold.
It is this last form that we will use further on.

\subsection{Progress, Justness and Fairness}\label{ssec:progress}

In Sections~\ref{ssec:route_discovery} and \ref{ssec:packet_delivery}, we will
formalise properties that say that under certain conditions some desired activity
will eventually happen, or some desired state will eventually be reached. As a particularly simple instance
of this, consider the transition systems in Figures~\ref{fig:PJF}(a)--(c), where
the double-circled state satisfies a desired property $\phi$.  The formula $\mathbf{G} (a \Rightarrow
\mathbf{F}\phi)$ says that once the action $a$ occurs, eventually we will reach a state where $\phi$
holds.  In this section we investigate reasons why this formula might not hold, and formulate
assumptions that guarantee that it does.

\newcommand{\Figprogress}{Figure~\ref{fig:PJF}(a)\xspace}
\newcommand{\Figjustness}{Figure~\ref{fig:PJF}(b)\xspace}
\newcommand{\Figfairness}{Figure~\ref{fig:PJF}(c)\xspace}

\paragraph{Progress.}\label{par:progress}
The first thing that can go wrong is that the process in \Figprogress performs $a$, thereby reaching the state $s$,
and subsequently remains in the state $s$ without ever performing the internal action $\tau$ that leads to
the desired state $t$, satisfying $\phi$. If there is the possibility of remaining in a state
even when there are enabled internal actions, no useful temporal property about
processes will ever be guaranteed. We therefore make an assumption that rules out this type of behaviour.
\begin{equation}\tag{$P_{1}$}\label{eq:progress1}
\mbox{\textit{A process in a state that admits an internal transition $\tau$ will eventually perform a transition.}}
\end{equation}
\index{progress property}%
\eqref{eq:progress1} is called a \emph{progress} property. It guarantees that the process depicted in
\Figprogress satisfies the LTL formula $\mathbf{G} (a \Rightarrow \mathbf{F}\phi)$.
We do not always assume progress when only external transitions are possible.%
\footnote{A transition is external iff it is not internal, i.e.\  iff
its label is different from $\tau$.} For instance, the process of
\Figprogress, when in its initial state $r$, will not necessarily perform the $a$-transition,
and hence need not satisfy the formula $\mathbf{F}\phi$. The reason is that external transitions
could be synchronisations with the environment, and the environment may not be ready to synchronise.
This can happen for instance when $a$ is the action $\receive{m}$.
However, for our applications it makes sense to distinguish two kinds of external transitions: those whose
execution requires cooperation from the environment in which the process runs, and those who do not.
The latter kind could be called \phrase{output transitions}. As far as progress properties go, output
transitions can be treated just like internal transitions:\!
\begin{equation}\tag{$P_{2}$}\label{eq:progress2}
\mbox{\textit{A process in a state that admits an output transition will eventually perform a transition.}}
\end{equation}
Whether a transition is an output transition is completely determined by its label; hence we also
  speak of \phrase{output actions}.
In case $a$ is an output action, which can happen independent of the environment, the
formula $\mathbf{F}\phi$ does hold for the process of \Figprogress.

We formalise \eqref{eq:progress1} and \eqref{eq:progress2} through a suitable definition of a
complete path.  In early work on
temporal logic, formulas were interpreted on Kripke structures: transition systems with unlabelled transitions,
subject to the condition of \emph{totality}, saying that each state admits at least one outgoing
transition. In this context, the complete paths are defined 
to be all infinite paths of the transition system. When giving up totality, it is customary to deem complete also those paths that end in a
state from which no further transitions are possible \cite{DV95}. Here we go a step further, and
\index{path!complete}%
(for now) define a path to be \emph{complete} iff it is either infinite or ends in a state from which no further
\emph{internal or output} transitions are possible. This definition exactly captures the
progress properties \eqref{eq:progress1} and \eqref{eq:progress2} proposed above. (Dropping all
progress properties amounts to defining \emph{each} path to be complete.) Below we will
restrict the notion of a complete path to also capture a forthcoming justness property.
\index{justness property}%

\begin{figure}[t]
	\centering
		\begin{minipage}[b]{0.45\linewidth}
			\centering
				\includegraphics[scale=1.2]{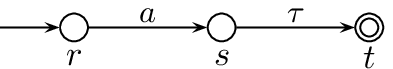}\\
				(a)\ Progress
		\end{minipage}
		\hspace{0.03\linewidth}
		\begin{minipage}[b]{0.45\linewidth}
			\centering
				\includegraphics[scale=1.2]{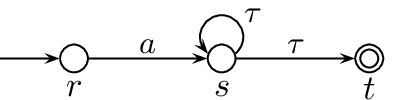}\\
				(c)\ Fairness
		\end{minipage}\\[3ex]
		\begin{minipage}[b]{0.93\linewidth}
			\centering
				\includegraphics[scale=1.2]{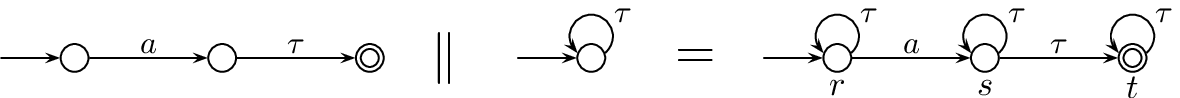}\\
				(b)\ Justness
		\end{minipage}
	\caption{Progress, Justness and Fairness}
	\label{fig:PJF}
  \end{figure}

It remains to be determined which transitions generated by the structural operational semantics of \awn
should be classified as output transitions. In the transition system for (encapsulated) network expressions
generated by the rules of~\Tab{sos network}, only five types of transition labels occur:
$\textbf{connect}(\dval{ip},\dval{ip}')$,
$\textbf{disconnect}(\dval{ip},\dval{ip}')$,
$\colonact{\dval{ip}\hspace{-.5pt}}{\hspace{-.5pt}\textbf{newpkt}(\dval{d}\hspace{-.5pt},\hspace{-.5pt}\dval{dip})}$,
$\colonact{\dval{ip}\hspace{-.5pt}}{\hspace{-.5pt}\deliver{\dval{d}}}$
and $\tau$. 
These are all actions to be considered, since we regard
(LTL-\hspace{-.16pt})properties on network expressions only.
The actions
$\textbf{connect}(\dval{ip}\hspace{-.5pt},\dval{ip}')$,
$\textbf{disconnect}(\dval{ip}\hspace{-.5pt},\dval{ip}')$ and
$\colonact{\dval{ip}}{\textbf{newpkt}(\dval{d},\dval{dip})}$
are entirely triggered by the environment of the network, and thus cannot be classified as output
actions. Transitions labelled $\tau$ are internal. For transitions labelled
$\colonact{\dval{ip}}{\deliver{\dval{d}}}$ two points of view are possible.
It could be that the action $\colonact{\dval{ip}}{\deliver{\dval{d}}}$ is seen as attempt of the
network to synchronise with its client in delivering a message; the synchronisation will then happen
only when both the network and the client are ready to engage in this activity.
A possible scenario would be that \Pro{pkt} gets stuck in Line~\ref{pkt2:line3} because the client
is not ready for such a synchronisation (the same happens in \Pro{newpkt}, Line~\ref{newpkt:line3}).
This interpretation of our formalisation of AODV would give
rise to deadlock possibilities that violate useful properties we would like the protocol to have,
such as the forthcoming \emph{route discovery} and \emph{packet delivery} properties.
We therefore take the opposite point of view by classifying
$\colonact{\dval{ip}}{\deliver{\dval{d}}}$ as an output action.
Hereby we disallow a deadlock
when attempting a \textbf{deliver}-action, since the environment of the
network cannot prevent delivery of data packets.
As a consequence, finite complete paths of AODV can end only in
states $N$ where all message queues are empty, all nodes \dval{ip} are
are either in their initial state or about to call the process \AODV,%
\footnote{More precisely these positions are at the beginning of
\Pro{aodv}, Line~\ref{aodv:line27},
\Pro{rreq}, Lines~\ref{rreq:line3}, \ref{rreq:line26a}, \ref{rreq:line35},
\Pro{rrep}, Lines~\ref{rrep:line8}, \ref{rrep:line14}, \ref{rrep:line21a}, \ref{rrep:line26},
and
in the middle of 
Lines~\ref{aodv:line33}, \ref{aodv:line40} (\Pro{aodv}),
\ref{newpkt:line3}, \ref{newpkt:line5} (\Pro{newpkt}),
\ref{pkt2:line3}, \ref{pkt2:line7}, \ref{pkt2:line14}, \ref{pkt2:line20}, \ref{pkt2:line22} (\Pro{pkt}),
\ref{rreq:line14}, \ref{rreq:line19}, \ref{rreq:line31} (\Pro{rreq}),
\ref{rrep:line20} (\Pro{rrep}),
\ref{rerr:line6} (\Pro{rerr}).
}
and
for all destinations \dval{dip} for which \dval{ip} has a (non-empty) queue of data packets we have
$\dval{dip}\notin\akd{\dval{ip}}$ and \plat{$\fD{\xiN{\dval{ip}}(\queues)}{\dval{dip}}=\pen$}.
This follows since our specification of AODV is input-enabled, is non-blocking, and avoids 
livelocks.

In the remainder of this paper we will only use LTL-formulas to check (encapsulated) network expressions. 
However, when defining output transitions also on partial networks, parallel processes and sequential processes, 
it is easy to carry over our mechanism to arbitrary expressions of AWN.
On the level of partial network expressions $\colonact{R}{\starcastP{m}}$ counts as an output action, as its occurrence cannot be
prevented by other nodes in the network. Similarly, on the level of sequential and parallel
processes
$\broadcastP{m}$,
$\groupcastP{D}{m}$,
$\unicast{\dval{dip}}{m}$,
$\neg\unicast{\dval{dip}}{\dval{m}}$ and
$\deliver{\dval{d}}$ are output actions, but
$\send{m}$ is not, for it requires synchronisation with
$\receive{m}$.
The remaining actions ($\listen{m}$, $\receive{m}$) are not considered output actions.

\paragraph{Justness.}\label{par:justness}
Now suppose we have two concurrent systems that work independently in parallel, such as two
completely disconnected nodes in our network. One of them is modelled by the transition system of
\Figprogress, and the other is doing internal transitions in perpetuity.  The parallel
composition is depicted on the left-hand side of \Figjustness. According to our structural operational
semantics, the overall transition system resulting from this parallel composition is the one
depicted on the right. In this transition system, the LTL formula $\mathbf{G} (a \Rightarrow
\mathbf{F}\phi)$ is not valid, because, after performing the action $a$, the process may do an
infinite sequence of internal transitions that stem from the other component in the parallel
composition, instead of the transition to the desired success state.  Yet the formula $\mathbf{G} (a
\Rightarrow \mathbf{F}\phi)$ does hold intuitively, because no amount of internal activity in the
remote node should prevent our own node from making progress.  That this formula does not
hold can be seen as a pitfall stemming from the use of interleaving semantics.  The intended
behaviour of the process is captured by the following \emph{justness} property:\footnote{In the
\index{fairness}%
  literature \emph{justness} is often used as a synonym for \hyperlink{weakfairness}{\emph{weak fairness}},
  defined on Page~\pageref*{fairness}---see, e.g., \cite{MP92}.
  In this paper we introduce a different concept of justness: 
  fairness is a property of schedulers that repeatedly choose between several tasks, 
  whereas
  justness is a property of parallel-composed transition systems, guaranteeing progress of all components.
}
\begin{equation}\tag{$J$}\label{eq:justness}
\parbox{0.9\textwidth}{\textit{A component in a parallel composition in a state that admits an internal or output
  transition will eventually perform a transition.}}
\end{equation}
Progress can be seen as a special case of justness, obtained by regarding a system as a parallel composition of one component only.
We will formalise the justness requirement \eqref{eq:justness} by fine-tuning our definition of a complete path.

Any path $\piN$ starting from an \awn network expression $[M]$ is derived through the
structural operational semantics of \Tab{sos network} from a path $\piPN$ starting from the
partial network expression $M$. All states occurring in $\piN$ have the form $[M']$ for some partial
network expression $M'$, and in $\piPN$ such a state is replaced by $M'$. Moreover, some transition
labels $\tau$ in $\piN$ are replaced by $\colonact{R}{\starcastP{m}}$ in $\piPN$, and transition
labels $\colonact{\dval{ip}}{\textbf{newpkt}(\dval{d},\dval{dip})}$ are replaced by
$\colonact{\{\dval{ip}\}\neg K}{\listen{\newpkt{\dval{d}}{\dval{dip}}}}$. To indicate the relationship
between $\piN$ and $\piPN$ we write $\piN=[\piPN]$. It might be that $\piPN$ is
not uniquely determined by $\piN$; if this happens, the partial network expression $M$ admits
different paths  that upon encapsulating become indistinguishable.

In the same way, any path $\piPN$ starting from a partial network expression $M$ that happens to be a
parallel composition of $n$ node expressions derives through the structural operational semantics of
\Tab{sos network} from $n$ paths $\pin_1,\ldots,\pin_n$ starting from each of these node expressions.
In this case we write $\piPN = \pin_1 \| \cdots \| \pin_n$. Here it could be that $\piPN$ is infinite, yet
some (but not all) of the $\pin_i$ are finite. As before, it might be that the $\pin_i$ are
not uniquely determined by $\piPN$.

Zooming in further, any path $\pin$ starting from a node expression $\dval{ip}:P:R$ derives through the
structural operational semantics of \Tab{sos node} from a path $\piP$ starting from the parallel
process expression $P$. As transitions labelled $\textbf{connect}(\dval{ip},\dval{ip}')$ or
$\textbf{disconnect}(\dval{ip},\dval{ip}')$ occurring in $\piN$, $\piPN$ and $\pin$ do not occur in $\piP$, it can be
that $\piP$ is finite even though $\pin$ is infinite. We write $\pin=\dval{ip}:\piP:*\;$ (without
filling in the $R$, since it may change when following $\pin$).

Finally, any path $\piP$ of a parallel process expression $P$ that is the parallel composition of $m$
sequential process expressions derives through the structural operational semantics of
\Tab{sos} from $m$ paths $\pi_1,\ldots,\pi_{m}$ starting from each of these sequential process expressions.
In this case we write $\piP = \pi_1 \parl \cdots \parl \pi_{m}$.\linebreak[2]
Again it may happen that $\piP$ is infinite, yet some (but not all) of the $\pi_i$ are finite.

\begin{definition}\label{df:complete path}\rm
\index{path!complete}%
A path starting from any \awn expression (i.e.\ a sequential or parallel process expression, a node
expression or (partial) network expression) \emph{ends prematurely} if it is finite and from its
last state an internal or output transitions is possible.
\begin{itemize}
\item 
A path $\pi_{i}$ starting from a sequential process expression is \emph{complete} if it does not end
prematurely---hence is infinite or ends in a state from which no further internal or output
transitions are possible.
\item 
A path $\piP$ starting from a parallel process expression is \emph{complete} if it does not end
prematurely and can be written as $\pi_1 \parl \cdots \parl \pi_{m}$ where each of the $\pi_i$ is complete.
\item 
A path $\pin$ starting from a node expression is \emph{complete} if it does not end
prematurely and can be written as $\dval{ip}:\piP:*$ where $\piP$ is complete.
\item 
A path $\piPN$ starting from a partial network expression is \emph{complete} if it does not end
prematurely and can be written as $\pin_{1} \| \cdots \| \pin_{n}$ where each of the $\pin_i$ is complete.
\item 
A path $\piN$ starting from a network expression is \emph{complete} if it does not end
prematurely and can be written as $[\piPN]$ where $\piPN$ is complete.
\end{itemize}
\end{definition}

Note that if $\piN=[\piPN]$ and $\piN$ ends prematurely, then also $\piPN$ ends prematurely.
This holds because any internal or output action enabled in the last state of $\piN$ must stem from an
internal or output action enabled in the last state of $\piPN$. For this reason the requirement
``it does not end prematurely'' is redundant in the above definition of complete path starting from a
network expression. For the same reason this requirement is redundant in the definition of a
complete path for node expressions or partial network expressions, but not in the definition for
parallel process expressions. The reason for including this requirement in each part of the
definition above, is to establish a general pattern that ought to lift smoothly to languages other
than \awn.

This definition of a complete path captures our (progress and) justness requirement, and
ensures that the formula $\mathbf{G} (a \Rightarrow \mathbf{F}\phi)$ holds for the process of \Figjustness.
For example, the infinite path $\pi$ starting from $r$ that after the $a$-transitions keeps looping through the
$\tau$-loop at $s$ can only be derived as $\pi_1 \| \pi_2$, where $\pi_1$ is a finite path ending
right after the $a$-transitions. Since $\pi_1$ fails to be complete (because its ends prematurely,
by its end state admitting a $\tau$-transition), $\pi$ is defined to be incomplete as well, and hence
does not count when searching for a complete path that fails to satisfy the formula.

\paragraph{Fairness.}\label{par:fairness}
With the justness requirement \Eq{justness}\footnote{Remember that \Eq{justness} implies the
progress requirements \Eq{progress1} and \Eq{progress2}.} embedded in our semantics of LTL, the processes of 
\Figprogress--(b) satisfy the formula $\mathbf{G} (a \Rightarrow \mathbf{F}\phi)$.
Yet, the process of \Figfairness does not satisfy this formula.
The reason is that in state $s$ a choice is made between two internal transitions.
One leads to the desired state satisfying $\phi$, whereas the other gives the process a chance to
make the decision again. This can go wrong in exactly one way, namely if the
$\tau$-loop is chosen each and every time.

\index{fairness}%
For some applications it is warranted to make a \emph{global fairness assumption}, saying that in
verifications we may simply assume our processes to eventually escape from a loop such as in
\Figfairness and do the right thing. A process-algebraic verification approach based on such an
assumption is described in~\cite{BBK87a}. Moreover, a global fairness assumption is
incorporated in the weak bisimulation semantics employed in~\cite{Mi89}.

An alternative approach, which we follow here, is to explicitly declare certain choices to be fair,
while leaving open the possibility that others are not. To see which choices come into question, we
search for all occurrences of the choice operator $+$ in our AODV specification in
Processes~\ref{pro:aodv}--\ref{pro:queues}. A nondeterministic choice occurs in
Lines~\ref{aodv:line22} and~\ref{aodv:line34} of \Pro{aodv} and in Lines~\ref{queues:line3}
and~\ref{queues:line8} of \Pro{queues}. All other occurrences of the $+$-operator are of the form
$[\varphi_{1}]p+\ldots+[\varphi_{n}]q$
where the guards $\varphi_{i}$ are mutually exclusive; these are deterministic choices,
where in any reachable state at most one of the alternatives is enabled.

Considering Lines~\ref{aodv:line2},~\ref{aodv:line22} and~\ref{aodv:line34} of \Pro{aodv}, the
process $\AODV$ running on a node in a network can be seen as a scheduler that needs to schedule three kinds of tasks.
Lines~\ref{aodv:line2}--\ref{aodv:line20} deal with handling an incoming message. This task is enabled when
there is a message in the message queue of that node. Lines~\ref{aodv:line22}--\ref{aodv:line33}
deal with sending a data packet towards a destination $\dval{dip}$. This task is enabled when there is a
queued data packet for destination $\dval{dip}$, i.e.\ $\dval{dip}\in\qD{\xi(\queues)}$, and moreover
a valid route to \dval{dip} exists, i.e.\ $\dval{dip}\in\akD{\xi(\rt)}$.
As data queues for multiple destinations $\dval{dip}$ may have formed, each time when sending a
data packet is scheduled a choice is made which destination to serve. Finally,
Lines~\ref{aodv:line34}--\ref{aodv:line40} deal with the initiation of a route discovery process for
destination $\dval{dip}$. It is enabled when the guard of Line~\ref{aodv:line34} evaluates to {\tt true}.
No matter which of these tasks is chosen, the chosen instance always terminates in a finite amount
of time,\footnote{Here we use that each of these tasks consists of finitely many actions,
  of which only the initial one could be blocking.
  The task of handling an incoming message could fail to terminate if the message
  received is not of the form specified in any of the guards of
  Lines~\ref{aodv:line4},~\ref{aodv:line6},~\ref{aodv:line8},~\ref{aodv:line12} or~\ref{aodv:line16};
  in this case a deadlock would occur in Line~\ref{aodv:line3}. However, using
  \Prop{preliminaries}(\ref{it:preliminariesi}), this will never happen, as all messages sent have the required form.}
after which the \AODV-scheduler needs to make another choice.

\hypertarget{weakfairness}{\label{fairness}%
\index{fairness!weak fairness property}%
For each of these tasks we postulate a \emph{weak fairness} property. It requires that if this
task, from some point onwards, is perpetually enabled, it will eventually be scheduled.
A weak fairness property is expressed in LTL as the requirement
$\mathbf{G}(\mathbf{G}\psi \ims \mathbf{F}\phi)$; here $\psi$ is the condition that states that
the task is enabled, whereas $\phi$ states that it is being
executed.\footnote{These properties were introduced and formalised in LTL in \cite{GPSS80} under the
  name ``responsiveness to insistence''. They were deemed ``the minimal fairness requirement'' for any scheduler.}
The property says that if the condition $\psi$ holds uninterruptedly from some time point onwards,
then eventually $\phi$ will hold.
This is the first formula of~\Eq{ltl property with side condition} with
$\phi^{\it pre}=\mbox{\tt true}$ and $\phi^{\it post}=\phi$.
Hence a logically equivalent formula is $\mathbf{GF}(\phi \vee \neg \psi)$.
Another equivalent formula expressing weak fairness is
$\mathbf{F}\mathbf{G}\psi \ims \mathbf{G}\mathbf{F}\phi$.
It says that if, from some point onwards, a task is perpetually enabled, it will be
scheduled infinitely often.\footnote{or is scheduled in the final state of the system. This 
    possibility needs to be added because, unlike in \cite{Pnueli77,GPSS80}, we allow complete
    paths to be finite.}}

\index{fairness!strong fairness property}%
Sometimes a \emph{strong fairness} property is needed, saying that if a task is
enabled infinitely often,\footnote{or in the final state of the system}
but allowing interruptions during which it is not enabled,
it will eventually be scheduled. Such a property is expressed in LTL as
$\mathbf{G}(\mathbf{GF}\psi \ims \mathbf{F}\phi)$,\footnote{These properties were introduced and
  formalised in LTL in \cite{GPSS80} under the  name ``responsiveness to persistence''.}
or equivalently $\mathbf{G}\mathbf{F}\psi \ims \mathbf{G}\mathbf{F}\phi$.
We do not need strong fairness properties in this paper.

Our first fairness property (\Fi) requires that if the guard of Pro.~\ref{pro:aodv},
Line~\ref{aodv:line22} evaluates to {\tt true} from some state onwards, for a particular value of
\dval{dip}, then eventually Line~\ref{aodv:line22} (or equivalently Line~\ref{aodv:line23} or \ref{aodv:line24})
will be executed, for that value of \dval{dip}.  Naturally, such a property needs to be required for
each node \dval{ip} in the network, and for each possible destination \dval{dip}.
Later, we will formulate a \emph{packet delivery} property, saying that under
certain circumstances a data packet will surely be delivered to its destination. Without the fairness
property (\Fi) there is no hope on such a property being satisfied by AODV\@. It could be
that a node \dval{ip} with a valid route to \dval{dip} has a queued data packet for \dval{dip}, but
will never send it, because it is constantly busy processing messages---that is, executing
Line~\ref{aodv:line2} instead of Line~\ref{aodv:line22}. 
Alternatively, it could be that the node has a constant supply of data packets for another
destination $\dval{dip}'$, and always chooses to send a packet to $\dval{dip}'$ instead of
to \dval{dip}.

Fairness property (\Fi) can be formalised as an instance of the template
$\mathbf{G}(\mathbf{G}\psi \ims \mathbf{F}\phi)$ by taking $\psi$ to be the formula that
says that the guard in Line~\ref{aodv:line22} is satisfied, and $\phi$ a formula that holds
after Line~\ref{aodv:line22} has been executed.
We take $\psi$ to be the atomic proposition
$\dval{dip}\in\keyw{qD}^{\dval{ip}}\cap\fnakD^{\dval{ip}}$, which we define to hold for
state $N$ iff $\dval{dip}\in \qD{\xiN{\dval{ip}}(\queues)} \cap \akd{\dval{ip}}$.
Other atomic propositions used below are defined along the same lines.
In order to formulate $\phi$ we use the atomic proposition
$\unicast{*}{\pkt{*}{\dval{dip}}{\dval{ip}}}$, which is
defined to hold when node \dval{ip} tries to unicast a data packet with destination \dval{dip}.
Thus we require, for all $\dval{ip},\dval{dip}\in\IP$, that\vspace{-1ex}
\vspace{-1ex}\begin{equation}
\mathbf{G}\big(\mathbf{G}(\dval{dip}\in\keyw{qD}^{\dval{ip}}\cap\fnakD^{\dval{ip}}) \ims
\mathbf{F}\big(\unicast{*}{\pkt{*}{\dval{dip}}{\dval{ip}}}\big)\big).
\tag{\Fi}
\vspace{-0.8ex}
\end{equation}
(\Fi) says that whenever the node \dval{ip} perpetually has queued packets for the destination
\dval{dip} as well as a valid route to \dval{dip}, it will eventually forward a data packet
originating from \dval{ip} towards \dval{dip}---i.e.\ Line~\ref{aodv:line24} will be executed.
In classifying this property as a weak fairness property, we count a task as enabled when its guard
is valid, notwithstanding that the task cannot be started during the time \AODV\ is working on a competing task.

Our second fairness property (\Fii) demands fairness for the task starting with Line~\ref{aodv:line34} of Pro.~\ref{pro:aodv}.
We require, for all $\dval{ip},\dval{dip}\in\IP$, that
\vspace{-0.8ex}
\[
\mathbf{G}\big(\mathbf{G}(\dval{dip}\in\keyw{qD}^{\dval{ip}} - \fnakD^{\dval{ip}} \ans \fnfD^{\dval{ip}}(\dval{dip})=\nonpen) \ims
\mathbf{F}\big(\broadcastP{\rreq{*}{*}{\dval{dip}}{*}{*}{\dval{ip}}{*}{\dval{ip}}}\big)\big).
\tag{\Fii}\label{eq:F2}
\vspace{-2.7ex}\]
\pagebreak[3]

(\Fii) says that whenever \dval{ip} perpetually has queued packets for \dval{dip} but no valid route to \dval{dip},
and the request-required flag at \dval{ip} for destination \dval{dip} is set to $\nonpen$, indicating
that a new route discovery process needs to be initiated, then node \dval{ip} does issue a request for
a route from {\dval{ip}} to {\dval{dip}}---so Line~\ref{aodv:line39} will be executed.

We do not formalise a fairness property saying that
Line~\ref{aodv:line2} of Pro.~\ref{pro:aodv} will be executed eventually.
Since the \textbf{receive}-action of Line~\ref{aodv:line2} of \Pro{aodv}
has to synchronise with the \textbf{send}-action in Line~\ref{queues:line5} of \Pro{queues}
it suffices to formalise a fairness property for \QMSG.


Process $\QMSG$ can be understood as scheduling two tasks:
(1) store an incoming message at the end of the message queue,
and (2) pop the top message from the queue and send it to $\AODV$ for handling.
The reason that (1) occurs twice in the specification (Lines~\ref{queues:line1}--\ref{queues:line2}
as well as~\ref{queues:line7}--\ref{queues:line8}) is that we require our node to be input
enabled, meaning that (1) must be possible in every state.

Our third and last fairness property (\Fiii) guards against starvation of task (2).  It says that if
the guard of Line~\ref{queues:line3} of Pro.~\ref{pro:queues} evaluates to {\tt true} from some
state onwards, then eventually Line~\ref{queues:line5} of Pro.~\ref{pro:queues} will be executed.
In order to formulate this property we use the atomic propositions $\keyw{msgs}^{\dval{ip}}\neq[\,]$,
which holds in state $N$ iff $\xiN{\dval{ip}}(\keyw{msgs})\neq[\,]$, and $\dval{ip}:\send{*}$, saying that the
process $\QMSG$ running on node \dval{ip} performs a \textbf{send}-action. We need to explicitly
annotate this activity with the name of node \dval{ip}, as---unlike for \textbf{unicast} and
\textbf{broadcast}---this information cannot be derived from the message being sent.
We require, for all $\dval{ip}\in\IP$, that\vspace{-1ex}
\begin{equation}\label{eq:F3}
\mathbf{G}\big(\mathbf{G}(\keyw{msgs}^{\dval{ip}}\neq[\,]) \ims
\mathbf{F}\big(\dval{ip}:\send{*}\big)\big).
\tag{\Fiii}
\end{equation}
(\Fiii) says that whenever node \dval{ip} perpetually has a non-empty queue of incoming messages, eventually one of
these messages will be handled.
Just as for the first task of the process $\AODV$, there is no need to specify a fairness property for task (1):
our justness property forbids any component from stopping when it can do a \textbf{*cast}-action,
and our structural operational semantics requires each component within transmission range of a
component doing a \textbf{*cast} to receive the transmitted message.

To say that a run of AODV is fair amounts to requiring the corresponding complete path to
satisfy properties (\Fi)--(\Fiii) for all values of \dval{ip} and \dval{dip}.  In order to require
fairness for all runs of AODV we 
augment the specification of AODV with a fairness component.
Henceforth, our specification of AODV consists of two parts: (A) the \awn specification of
\Sect{modelling_AODV}, which by the operational semantics 
of \awn generates a labelled transition
system $L$, and (B) a \phrase{fairness specification}, consisting of a collection of LTL formulas.
The latter narrows down the complete paths in $L$ to the ones that satisfy those formulas.%
\footnote{
Formally, we require the labelled transition system $L$ and the fairness specification to be consistent with each other.
By this we mean that one cannot reach a state in $L$ from where, given a sufficiently uncooperative
environment, it is impossible to satisfy the fairness specification---in other words \cite{Lamport00},
`the automaton can never ``paint itself into a corner.''\;\!' In~\cite{Lamport00} this requirement is
called \emph{machine closure}, and demands that any finite path in $L$, starting from an initial
state, can be extended to a path satisfying the fairness specification.
Since we deal with a reactive system here, we need a more complicated consistency requirement,
taking into account all possibilities of the environment to allow or block transitions that
are not fully controlled by the specified system itself. This requirement can best be explained in
terms of a two player game between a \emph{scheduler} and the \emph{environment}.

Define a \emph{run} of $L$ as a path that starts from an initial state. Thus a 
\emph{finite run} is an alternating sequence of states and transitions, starting from an initial
state and ending in a state, such that each transition in the sequence goes from the state before to
the state after it. Moreover, a \emph{complete run} is a finite or infinite path starting from an initial state.
The game begins with any finite run $\pi$ of $L$, chosen by the environment.
In each turn, first the environment selects a set $\textit{next}(p)$ of transitions starting in
the last state $N$ of $\pi$; this set has to include all internal and output transitions starting
from $N$, but can also include further transitions starting in $N$. If $\textit{next}(p)$ is empty, the game ends; otherwise the scheduler selects a
transition from this set, which is, together with its ending state,  appended to $\pi$, and a new turn starts with the prolonged
finite run. The \emph{result} of the game is the finite run in which the game ends, or---if it
does not---the infinite run that arises as the limit of all finite runs encountered during the game.
So the result of the game always is a complete run. The game is \emph{won} by the scheduler
iff the result satisfies the fairness specification. Now $L$ is \emph{consistent} with a fairness
specification iff there exists a winning strategy for the scheduler.

Our AODV specification and our fairness properties (\Fi)--(\Fiii) are constructed in such a way that they are consistent.
}
\newpage

There are many ways in which we could alter our \awn specification of AODV so as to ensure
that (\Fi)--(\Fiii) are satisfied and thus need not be required as an extra part of our
specification. For example, \Pro{aodv} could be modified in a way such that the three 
different activities (Lines~\ref{aodv:line2}--\ref{aodv:line20}, Lines \ref{aodv:line22}--\ref{aodv:line33} and Lines~\ref{aodv:line34}--\ref{aodv:line40}) 
are prioritised. The process could first initiate all route discovery processes, then handle all queued data packets (for which a valid route is known) and 
finally handle a fixed number of received messages (less if there are not enough messages in the queue). 
After the messages have been handled, the modified process would loop back and start initiating 
route discovery processes again.
 However, for the purpose of protocol specification we do not want to commit to any
particular method of ensuring fairness. Therefore we state fairness as an extra requirement without
telling how it should be implemented.

When we later claim that an LTL formula $\phi$ holds for AODV, as specified by (A) and (B) together,
this is equivalent to the claim that $\psi\Rightarrow\phi$ holds for AODV as specified by (A) alone,
where $\psi$ is the conjunction of all LTL formulas that make up the fairness specification (B).

\subsection{Route Discovery}\label{ssec:route_discovery}

\index{route discovery property}%
An important property that every routing protocol ought to satisfy is that if a route 
discovery process is initiated in a state where the source is connected to the destination
and during this process no (relevant) link breaks, then the source will eventually discover a route to the destination.
In case of AODV a route discovery process is initiated when a route request is issued.
So for any  pair of IP addresses $\dval{oip},\dval{dip}\in\IP$ the following should hold:\vspace{-.5pt}
\[
\mathbf{G} \left(\begin{array}{@{}c@{}}   
 \big(\textbf{connected}^*(\dval{oip},\dval{dip}) \ans \broadcastP{\rreq{*}{*}{\dval{dip}}{*}{*}{\dval{oip}}{*}{\dval{oip}}}\big) \\
 \ims  \mathbf{F} \big(\dval{dip}\in\fnakD^{\dval{oip}} \ors \textbf{disconnect}(*,*)\big) 
\end{array}\right).
\]  
Here, the predicate
$\textbf{connected}^*(oip,dip)$\label{pg:connected*}
holds in state $N$ iff there exist nodes
$\dval{ip}_0,\ldots,\dval{ip}_n$ such that $\dval{ip}_0\mathop=\dval{oip}$, $\dval{ip}_n\mathop=\dval{dip}$
and \plat{$\dval{ip}_{i}\mathop\in \RN{\dval{ip}_{i-1}}$} for $i\mathop=1,\ldots,n$.
The latter condition describes the fact that $\dval{ip}_{i}$ is in range of $\dval{ip}_{i-1}$.\footnote{%
Since the connectivity graph of \awn is always symmetric, this condition suffices to guarantee that both the RREQ message and the RREP message reach their destinations.}
All other predicates follow the description of Page~\pageref{pg:typesofpredicates}:
$\broadcastP{\rreq{*}{*}{\dval{dip}}{*}{*}{\dval{oip}}{*}{\dval{oip}}}$ 
models that node \dval{oip} issues a request for a route from {\dval{oip}} to {\dval{dip}};
the predicate
$\dval{dip}\in\fnakD^{\dval{oip}}$ holds in state $N$ iff
$\dval{dip}\in\akd{\dval{oip}}$, i.e.\ \dval{oip} has found a valid route to \dval{dip},
and $\textbf{disconnect}(*,*)$ is the action of disconnecting any two nodes.
By means of the last disjunct, the property does not require a route to be found
once any link in the network breaks.\footnote{Here $\neg\,\textbf{disconnect}(*,*)$ is the side condition
  $\psi$ of \Eq{ltl property with side condition}.}

The following theorem might be a surprise.

\begin{theorem}\rm
AODV does not satisfy the property route discovery.
\end{theorem}

We show this by an example (\Fig{route_lost}).
In particular, we show that a route reply may be dropped.
This problem has been raised before, back in Oct 2004.\footnote{\url{http://www.ietf.org/mail-archive/web/manet/current/msg05702.html} shows the same shortcoming using a $4$-node linear topology.}
We discuss modifications of AODV to solve this problem in \SSect{modifyRREP}.
\Fig{route_lost} shows a network consisting of $3$ nodes in
a linear topology. Two nodes ($a$ and $s$) are both searching for a
route to  destination $d$.\footnote{In \cite{MSWIM12} we present a version of this example in a
    non-linear 4-node topology with symmetry between the two nodes that search for a route to $d$.}
First, node $a$ broadcasts a route request, RREQ${}_1$ (\Fig{route_lost}(b)). As usual all recipients update their routing tables.
Since node $s$ still has no information about $d$, it also initiates a
route request, RREQ${}_{2}$.
After $a$ has forwarded that request (\Fig{route_lost}(c)), $d$ initiates a route reply as a consequence
of RREQ${}_{1}$. When node $a$ receives this reply, it updates its own routing table (\Fig{route_lost}(d)).
Finally, node $d$ reacts on the second route request received 
(RREQ${}_{2}$) and sends yet another route reply.
Node $a$ receives RREP${}_{2}$, but does \emph{not} forward it. This is%
\linebreak\mbox{}\vspace{-13.6pt}\pagebreak[3]
\begin{exampleFig}{Route discovery fails}{fig:route_lost}
\FigLine[lsr]%
  {The initial state.}{fig/ex_lost_rrep2_1}{}
  {$a$ broadcasts a new RREQ message destined to $d$;\\ all nodes receive the RREQ and update their RTs.}{fig/ex_lost_rrep2_2}{}
 \FigLine[slsr]%
  {$s$ broadcasts a new RREQ destined to $d$;\\$a$ forwards it.}{fig/ex_lost_rrep2_3}{}
  {$d$ handles RREQ${}_1$ and unicasts a RREP to $a$.}{fig/ex_lost_rrep2_4}{}
\FigLine[sl]%
  {$d$ handles RREQ${}_2$ and unicasts a RREP to $a$.}{fig/ex_lost_rrep2_5}{}
  {This ends the work of AODV; $s$ will never get an answer for its RREQ.}{}{}
\end{exampleFig}
\noindent
because RREP${}_{2}$ does not contain any fresher information about
destination $d$, in comparison with the information in node
$a$'s existing routing table entry for  $d$. As a result,
RREP${}_{2}$ is dropped at node $a$, and node $s$ never receives a
route reply for its route request.
Looking at our model (Process~\ref{pro:rrep}), the node does not forward a request since
Line~\ref{rrep:line3} evaluates to false whereas Line~\ref{rrep:line25} evaluates to true.
\endbox

At first glance, it seems that this behaviour can be fixed by a repeated route request.
If node $s$ would initiate and broadcast another route request, node $a$ would receive it and
generate a route reply immediately. The AODV RFC specifies that a node can broadcast another route
request if it has not received a route reply within a pre-defined time. However, a repeated route
request does not guarantee the receipt of a route reply. It is easy to construct an example similar
to \Fig{route_lost} where, instead of a linear topology with $3$ nodes, we use a linear
topology with $n+2$ nodes, where $n$ is the maximum number of repeated route requests.

But the situation is even worse. Even in a $4$-node topology an infinite stream of repeated route
requests cannot guarantee route discovery. \Fig{route_lost_rreq_resend} illustrates this fact.

In the initial state, node $a$ has established a route to $d$ via a standard RREQ-RREP cycle, initiated by~$a$.
Subsequently, in Part (b), node $b$ searches for a route to $x$ (an
arbitrary node that is not connected 
to any of the nodes we consider). After $d$ forwards the RREQ
message destined for $x$, node $a$ creates a valid route to $d$ with an unknown sequence number that
equals $d$'s own sequence number.\footnote{This examples hinges on our choice of Resolution (2c) of Ambiguity~2. Taking Resolutions
  (2a) or (2d) would avoid this problem; another solution would be following the suggestion of
  I.~Chakares in Footnote~\ref{Chakares-increment-when-issuing-RREP} on Page~\pageref{Chakares-increment-when-issuing-RREP}.
  We will propose a more thorough solution, that also tackles the problem of \Fig{route_lost},
  in \SSect{modifyRREP}.}
Now $s$ initiates\vspace{-5pt}\linebreak\mbox{}\vspace{-13.6pt}\pagebreak[3]
\begin{exampleFig}{Route discovery also fails with repeated request resending}{fig:route_lost_rreq_resend}
\FigLine[xslxsr]%
  {The initial state;\\$a$ established a route to $d$ by a RREQ-RREP-cycle.}{fig/ex_lost_rrep3_1}{}
  {$b$  broadcasts a new RREQ destined to $x$;\\the request travels through the network.}{fig/ex_lost_rrep3_2}{}
  \FigNewline
 \FigLine[xslxsr]%
  {$s$ broadcasts a new RREQ destined to $d$.}{fig/ex_lost_rrep3_3}{}
  {$d$ sends a route reply for $s$ back to $a$;\\$a$ drops the reply.}{fig/ex_lost_rrep3_4}{}
\end{exampleFig}
\noindent
a route request, searching for a route to
$d$. Since node $a$ does not have a known sequence number for $d$ it may not generate an intermediate route
reply (Pro.~\ref{pro:rreq}, Line~\ref{rreq:line22} evaluate to \keyw{false}). Hence it forwards the
route request (Part~(c)), and node $d$ answers with a RREP message (Part (d)). However, node $a$ will
not update its routing table entry for $d$, because it already has an entry with the same sequence
number and the same hop count (Line~\ref{rrep:line3} of Pro.~\ref{pro:rrep} evaluates to false
whereas Line~\ref{rrep:line25} evaluates to true).
As a consequence, $a$ does not forward the route reply to $s$, and $s$ will not create a route to $d$.
Repeating the route request by $s$ will not help, as the same events will be repeated.%

\hypertarget{ReplyIssued}{
Both counterexamples show a failure in forwarding a route reply back to the originator 
of the route discovery process. This travelling back can be seen as the second step of 
a route discovery process. The first step consists of the route request travelling from the
  originator to either the destination or to a node that has a valid route to the destination (with
  known sequence number) in its routing table.
  The following property states that this step
  always succeeds: whenever a route request is issued in a state where the source is connected to the destination
  and subsequently no link break occurs, then some node will eventually send a route reply back towards the source.\pagebreak[1]
\[
\mathbf{G} \left(\begin{array}{@{}c@{}}   
 \big(\textbf{connected}^*(\dval{oip},\dval{dip}) \ans \broadcastP{\rreq{*}{*}{\dval{dip}}{*}{*}{\dval{oip}}{*}{\dval{oip}}}\big) \\
 \ims  \mathbf{F} \big(\unicast{\rrep{*}{\dval{dip}}{*}{\dval{oip}}{*}}{*} \ors \textbf{disconnect}(*,*)\big) 
\end{array}\right).
\]  
This property does hold for AODV\@. Namely, Pro.~\ref{pro:rreq} is structured in such a way that
upon receipt of a RREQ message, either a matching RREP is sent or the RREQ is forwarded.
So if a route reply is never generated, then the route request floods the network and reaches all nodes connected to the originator of the request,
which by assumption includes the destination---this would cause a RREP to be sent.}

\subsection{Packet Delivery}\label{ssec:packet_delivery}

\index{packet delivery property}%
The property of \emph{packet delivery} says that if a client injects a
packet, it will eventually be delivered to the destination. However, in a WMN
it is not guaranteed that this property holds, since nodes can get
disconnected, e.g., due to node mobility. A useful formulation has to be weaker.
A higher-layer communication protocol should
guarantee
 packet delivery only if an end-to-end route
exists long enough. More precisely, such a protocol should
guarantee delivery of a packet injected by a client at node
$\dval{oip}$ with destination $\dval{dip}$, when $\dval{oip}$ is
connected to $\dval{dip}$ and afterwards no link in the network is
disconnected. This means that for 
all $\dval{oip},\dval{dip}\in\IP$, and any data packet $\dval{dp}\in\tDATA$, the following should
hold:
\begin{equation}\label{eq:PD1}
    \mathbf{G}\left(
\begin{array}{@{}c@{}}
    \big(\textbf{connected}^*(\dval{oip},\dval{dip}) \ans \dval{oip}:{\bf newpkt}(\dval{dp},\dval{dip})\big)\\
    \ims \mathbf{F} \big(\dval{dip}:\deliver{\dval{dp}}\ors\textbf{disconnect}(*,*)\big)
\end{array}
\right).
\tag{PD${}_{1}$}
\end{equation}

\noindent
Here $\colonact{\dval{oip}}{{\bf newpkt}(\dval{dp},\dval{dip})}$ models injection of a new data
packet \dval{dp} at $\dval{oip}$, and $\colonact{\dval{dip}}{\deliver{\dval{dp}}}$ that the destination receives it.
This formulation of packet delivery does not specify any particular
route, but merely requires that $\dval{dp}$ will eventually be delivered. The property does not
require a packet to arrive once any link in the network breaks down.

For a routing protocol like AODV, this form of packet delivery is a much too strong
requirement. The example of \Fig{packet_delivery} shows why it does not hold.

\begin{figure}
\vspace{-2ex}
\begin{exampleFig}{Packet delivery property PD$_1$ fails}{fig:packet_delivery}
\FigLine[lr]%
  {The initial state;\\ $s$ has established a route to $d$.}{fig/ex_lost_packetdelivery2_1}{}
  {The topology changes.}{fig/ex_lost_packetdelivery2_2}{}
\FigLine[lr]%
  {$s$ transfers a packet to $a$, for delivery at $d$.}{fig/ex_lost_packetdelivery2_3}{}
  {$a$ drops the packet and sends a RERR message to $s$.}{fig/ex_lost_packetdelivery2_4}{}
\end{exampleFig}
\vspace{-2ex}
\end{figure}

In the initial state node $s$ has, through a standard RREQ-RREP cycle, established a route to $d$.
Afterwards, the link between $a$ and $d$ breaks, and a new link between $s$ and $d$ is established.
Subsequently, say in state $S$, the application layer injects a data packet \dval{dp} destined for $d$ at node $s$.
Based on the information in its routing table, $s$ transfers the packet to $a$. However, the packet
is dropped by $a$ when $a$ fails to forward the packet to $d$.
To be precise, the reachable state $S$ satisfies $\textbf{connected}^*(s,d) \ans \dval{s}:{\bf newpkt}(\dval{dp},d)$
but there is a path from $S$ that does not feature any state with $d:\deliver{\dval{dp}}$
or $\textbf{disconnect}(*,*)$.

This failure of \Eq{PD1} is normal behaviour of a routing protocol. A higher layer in the 
network stack (e.g. the transport or the application layer) may
use an acknowledgement and retransmission protocol on top of its use of a routing
protocol, and this combination might guarantee \Eq{PD1}. 
For the routing protocol itself, it suffices that a packet will eventually be delivered if 
the client (higher-layer protocol)  injects the same data packet 
again and again, until the packet has reached the destination.
This gives rise to the following weaker form of packet delivery:
\begin{equation}
    \mathbf{G}\left(\begin{array}{@{}c@{}}\label{eq:PD2}
    \big(\textbf{connected}^*(\dval{oip},\dval{dip}) \ans \dval{oip}:{\bf newpkt}(\dval{dp},\dval{dip})\big)\\
    \ims\rule{0pt}{11pt} \mathbf{F} \big(\dval{dip}:\deliver{\dval{dp}} \ors\textbf{disconnect}(*,*)
    \ors \neg\mathbf{F}\big(\dval{oip}:{\bf newpkt}(\dval{dp},\dval{dip})\big)\big)
\end{array}
\right).
\tag{PD${}_{2}$}
\end{equation}
This is the property \Eq{PD1}, but under the side condition $\psi =
\mathbf{F}\big(\dval{oip}:{\bf newpkt}(\dval{dp},\dval{dip})\big)$ that is required to hold after
the initial injection of the data packet and until the packet is delivered---see
\Eq{ltl property with side condition}. This side condition says that one will keep injecting copies
of the same data packet, i.e.\ every state for which $\psi$ holds is followed by one where such a packet is injected.
In \Eq{PD2}, the clause $\dval{oip}:{\bf newpkt}(\dval{dp},\dval{dip})$ in the
precondition is redundant, as it is implied by the side condition $\psi$. Moreover, by the
equivalence of \Eq{ltl property with side condition}, \Eq{PD2} can also be formulated as
\[
    \mathbf{G}\left(\begin{array}{@{}c@{}}
 \big(\textbf{connected}^*(\dval{oip},\dval{dip}) \ans \mathbf{GF}\big(\dval{oip}:{\bf newpkt}(\dval{dp},\dval{dip})\big)\big)\\
    \ims\rule{0pt}{11pt} \mathbf{F} \big(\dval{dip}:\deliver{\dval{dp}} \ors\textbf{disconnect}(*,*)\big)
\end{array}
\right).
\]

\noindent
Here, $\mathbf{GF}\big(\dval{oip}:{\bf newpkt}(\dval{dp},\dval{dip})\big)$ states that the
injection of the data packet \dval{dp} at node \dval{oip} will be repeated infinitely
often.\footnote{Due to the existence of finite complete paths, the formula
$\mathbf{GF}\big(\dval{oip}:{\bf newpkt}(\dval{dp},\dval{dip})\big)$ also holds for complete paths
    whose final state satisfies $\dval{oip}:{\bf newpkt}(\dval{dp},\dval{dip})$.
However, in our specification of AODV such complete paths do not occur.}
If during that time no two nodes get disconnected, the packet will eventually be delivered at its
destination \dval{dip}.

Continuing the example of \Fig{packet_delivery}, in Part (d), node $a$ sends a route
error message to $s$, as a result of which $s$ invalidates its routing table entry for $d$.
If now a new data packet destined for $d$ is injected at $s$, node $s$ initiates a new route
discovery process and finds the 1-hop connection. As a result of this, the packet will be delivered
at $d$, as required by \Eq{PD2}.

\Eq{PD2} appears to be a reasonable packet delivery property for a routing protocol like AODV\@.
Yet, it is still too strong for our purposes. A failure of \Eq{PD2} can occur easily in the following
scenario: node \dval{oip} has a packet for node \dval{dip}, and initiates a route discovery process
by issuing a route request, while setting the request-required flag for the route towards \dval{dip} to $\pen$.
The route request reaches \dval{dip}, but the corresponding route reply is lost on the way back to
\dval{oip}, due to a link break. From that moment onwards the topology remains stable and a route
from \dval{oip} to \dval{dip} exists. We may even assume that it would be found if only \dval{oip}
does a second route request. However, such a second route request will never happen because the
request-required flag keeps having the value $\pen$ in perpetuity.

This failure of \Eq{PD2} is a flaw of our model rather than of AODV\@. A more realistic model would
specify that the request-required flag cannot keep the value $\pen$ forever. After a timeout,
either the flag should revert to $\nonpen$, so that a new route request will be made, or the entire
queue of data packets destined to \dval{dip} will be dropped, so that a newly injected packet will
start a fresh queue, which is initialised with a request-required flag $\nonpen$.
Such modelling requires timing primitives; however, since we abstract from timing issues, we
did not build such a feature into our packet handling routine.

To compensate for this omission, we add a precondition to the packet delivery property, namely
that if \dval{oip} perpetually has queued packets for \dval{dip} but no valid route to \dval{dip},
then eventually the request-required flag at \dval{oip} for destination \dval{dip} will be set to $\nonpen$:
\[
\mathbf{G}\big(\mathbf{G}(\dval{dip}\in\keyw{qD}^{\dval{oip}} - \fnakD^{\dval{oip}}) \ims
\mathbf{F}( \fnfD^{\dval{oip}}(\dval{dip})=\nonpen)\big)
\]
\index{packet delivery property}%
Adding this precondition to \Eq{PD2} yields \Eq{PD4}, our final \emph{packet delivery} property:
\begin{equation}
\begin{array}{@{}l@{}}
\phantom{\ims\ }\mathbf{G}\big(\mathbf{G}(\dval{dip}\in\keyw{qD}^{\dval{oip}} - \fnakD^{\dval{oip}}) \ims
\mathbf{F}\big( \fnfD^{\dval{oip}}(\dval{dip})=\nonpen\big)\\
\ims    \mathbf{G}\left(\begin{array}{@{}c@{}}
\textbf{connected}^*(\dval{oip},\dval{dip})\\
    \ims\rule{0pt}{11pt} \mathbf{F} \big(\dval{dip}:\deliver{\dval{dp}} \ors\textbf{disconnect}(*,*)
    \ors \neg\mathbf{F}\big(\dval{oip}:{\bf newpkt}(\dval{dp},\dval{dip})\big)\big)
\end{array}
\right).
\end{array}
\tag{PD$_3$}\label{eq:PD4}
\end{equation}
This property ought to be satisfied by a protocol like AODV\@.
Nevertheless,

\begin{theorem}\rm
AODV does not satisfy the property packet delivery.
\end{theorem}
\Fig{route_lost_rreq_resend} presents an example where an infinite stream of repeated route
request does not result in route discovery, let alone in packet delivery.

\begin{figure}
\vspace{-2ex}
\begin{exampleFig}{Precursor maintenance limits packet delivery ratio}{fig:lost_due_precs}
\FigLine[lr]%
  {$d$ broadcasts a new RREQ message destined to $b$;\\the RREQ floods the network; $s$ creates a route to~$d$.}{fig/ex_FailureOfPrecursors_2}{}
  {$b$ handles RREQ$_{1}$ and unicasts a reply back to d.}{fig/ex_FailureOfPrecursors_3}{}
 \FigLine[lr]%
  {The topology changes;\\$s$ receives a data packet destined to $d$.}{fig/ex_FailureOfPrecursors_5}{}
  {$a$ tries to forward data packet to $d$;\\packet delivery fails.}{fig/ex_FailureOfPrecursors_6}{}
\end{exampleFig}
\vspace{-2ex}
\end{figure}

\Fig{lost_due_precs} shows yet another counterexample against packet delivery, this time when the
route discovery property is satisfied.
Initially, node $d$ requests a route to $b$ (\Fig{lost_due_precs}(a)). As a result, $a$ creates a routing table entry for $d$,
with an empty set of precursors.\footnote{In fact, in this example all lists of precursors are empty.}
In Part (b), the reply is sent from node $b$ to node $d$. Afterwards, in Part (c), the link
between $a$ and $d$ breaks. From here on the topology remains stable, and $\textbf{connected}^*(s,d)$ holds.
In Part (c) the application layer injects a packet at $s$ for delivery at $d$. Since $s$ already has a routing table entry for $d$, no new route
request needs to be initiated, and the packet can be sent right away. Unfortunately, the packet is
dropped when $a$ fails to forward it to~$d$. Node $a$ invalidates its entry, but has no precursors
for the route to $d$ to send an error message to.\footnote{The same behaviour occurs when node $a$
  detects the link break earlier, for instance by using Hello messages.} As a consequence, $s$ will not learn about the
broken link, and all subsequent packets travelling from $s$ to $d$ will be dropped at $a$
(Pro.~\ref{pro:pkt}, Lines~\ref{pkt2:line15}--\ref{pkt2:line20}).


\section{Analysing AODV---Problems and Improvements}\label{sec:analysingAODV}

In this section we point at shortcomings of the AODV protocol and discuss possible solutions.
The solutions are again modelled in our process algebra. This makes it easy to ensure that the presented
improvements are unambiguous and still satisfy the invariants discussed in the \Sect{invariants}. In particular we show that
all variants of AODV presented in the remainder of this section are loop free and satisfy the route correctness property.

More precisely we propose five changes to the AODV protocol.

In \SSect{skipbroadcastid} we show that the route request identifier (RREQ ID)
is redundant and can be dropped
from the specification of AODV without changing the behaviour of the protocol in any way.
This is a small  improvement, but reduces the size of message headers.

In Sections~\ref{ssec:modifyRREP}--\ref{ssec:groupandbroadcast} we address three deficiencies of
AODV that each cause a failure of the packet delivery property discussed in \Sect{properties}.
The first two deal with failures of the route discovery property, which is a necessary precondition to
ensure packet delivery. 

In \SSect{modifyRREP} we discuss a known problem of AODV, namely that a node
fails to forward a RREP message that does not contain new information. This leads to a failure of
route discovery because the information can be new to the nodes to which the message ought to be
forwarded.

In \SSect{skipunknownflag} we discuss failures of route discovery that depend on the convention for
routing table updates in response to an AODV control message from a neighbour (cf.~\hyperlink{sss921b}{Ambiguity 2})
and analyse conventions that are not prone to such failures.

In \SSect{groupandbroadcast} we show how error messages may fail to reach nodes that need to be
informed of a link break. This may cause a failure of packet delivery even when route discovery is
guaranteed. This problem can be solved by always broadcasting error messages.

Finally, in \SSect{forwardRREQ}, we show that AODV inadvertently establishes sub-optimal routes, i.e., even when there is a shorter route
towards a destination, AODV will use (much) longer paths to send packets. This problem can
be avoided by modifying the process $\RREQ$ for handling message requests.

\subsection{Skipping the RREQ ID}\label{ssec:skipbroadcastid}
AODV does not need the route request identifier. This number, in combination with the IP address of the originator, is used
to identify every RREQ message in a unique way. However, we have shown that
the combination of the originator's IP address and its sequence number
is just as suited to uniquely determine the route request to which the message belongs
(cf.\ \Prop{messagebroadcast}(b)).
Hence, the route request identifier field is not required. This can then reduce the size of the RREQ message.

In detail, the following changes have to be made:
\begin{itemize}
\item The set {\tRREQID} (including the variable \rreqid) and the function $\fnnrreqid$ are skipped.
\item The variable {\rreqs} is now of type $\pow(\tIP\times\tSQN)$.
\item The function $\rreqID$ to generate route requests has now the type

\centerline{
$\rreqID:\NN \times\tIP \times \tSQN \times \tSQNK\times \tIP \times \tSQN \times \tIP \rightarrow \tMSG\ .$
}
All the parameters are the same, except that the request identifier is left out.
\item The modified basic routine (\Pro{aodv}) is given by \Pro{aodv_no_rreqid}.
\begin{figure}[th]
\vspace{-3ex}
  \algsetup{linenodelimiter=.,linenosize=\tiny}
  \begin{algorithm}[H]
    {\footnotesize
      \caption{The basic routine, not using $\rreqid$}
      \label{pro:aodv_no_rreqid}
      \begin{algorithmic}[1]
        \input{processes/aodv_no_rreqid.tex}
	\end{algorithmic}
    }
  \end{algorithm}

\vspace{-3ex}
\end{figure}
\item In \Pro{rreq}, the occurrences of $\rreqid$ in Lines ``0'' and~\ref{rreq:line34} are dropped; all other occurrences (Lines~\ref{rreq:line2}, \ref{rreq:line4} and~\ref{rreq:line8}) are replaced by $\osn$.
\end{itemize}

The statements and proofs of Sections~\ref{sec:invariants} and \ref{sec:interpretation} are all valid, but need the following modifications. 
\begin{itemize}
\item Whenever the function $\rreqID$ is used, the second parameter (\rreqid) has to be dropped. 
\item Propositions~\ref{prop:invarianti_itemi} and
  \ref{prop:messagebroadcast}(a) use the variable $\rreqid$; they can be dropped.
The statement that a route request is uniquely determined by the pair
$(\oip,\osn)$, the replacement of \Prop{invarianti_itemi}, is already
stated and proven in \Prop{messagebroadcast}(b).
\item The statement of Invariant~\eqref{inv:starcast_iii} in \Prop{starcast} changes into
      \begin{equation}
      \renewcommand{\rreq}[7]{\rreqID(#1\comma#2\comma#3\comma#4\comma#5\comma#6\comma#7)}
	N\ar{R:\starcastP{\rreq{*}{*}{*}{*}{\oipc}{\osnc}{\ipc}}}_{\dval{ip}}N' \ims (\oipc,\osnc)\in\xiN{\ipc}(\rreqs)
      \end{equation}
      and likewise for Invariant~\eqref{inv:starcast_rreqid}. In the proof,
      ``content
		$
		\xi(*\comma\rreqid\comma*\comma*\comma*\comma\ip\comma*\comma\ip)
		$''
       changes into ``content
		$
		\xi(*\,\comma*\comma*\comma*\comma\ip\comma\osn\comma\ip)
		$''.
       All other occurrences of ``$\rreqid$'' change into ``$\osn$'', and ``$\rreqidc$'' into ``$\osnc$''.
\item In the proof of \Prop{selfentries} ``$\rreqid$'' changes into ``$\osn$''.
\end{itemize}

\subsection{Forwarding the Route Reply}\label{ssec:modifyRREP}

In AODV's route discovery process, a RREP message from the destination
node is unicast back along a route towards the
originator of the RREQ message. Every intermediate node on the
selected route will process the RREP message and, in most cases,
forward it towards the originator node. However, there is a
possibility that the RREP message is discarded at an intermediate
node, which results in the originator node not receiving a
reply. The discarding of the RREP message is due to the RFC
specification of AODV \cite{rfc3561} stating that an intermediate node
only forwards the RREP message if it is not the originator node
\emph{and} it has created or updated a {\rte} to the destination
node described in the RREP message:
\begin{quote}\raggedright\small 
``{\tt If the current node is not the node indicated by the Originator IP
   Address in the RREP message AND a forward route has been created or
   updated as described above, the node consults its route table entry
   for the originating node to determine the next hop for the RREP
   packet, and then forwards the RREP towards the originator using the
   information in that route table entry.}''\hfill\cite[Sect.~6.7]{rfc3561}
\end{quote}
The latter requirement means that
if a valid {\rte} to the destination node already exists, and is
not updated when processing the RREP message, then the intermediate
node will not forward the message. In \Sect{properties} we have illustrated this problem
with two examples (Figures~\ref{fig:route_lost} and~\ref{fig:route_lost_rreq_resend}),
also showing that this leads to a failure of route discovery.

A solution to this problem is to require intermediate nodes to forward 
\emph{all\/} RREP messages that they receive.
In the example presented in \Fig{route_lost}, the intermediate node $a$ will forward 
RREP${}_{2}$, after RREP${}_{2}$ was received in Part (e). As a result, node $s$ will establish a route to $d$.
Likewise, in \Fig{route_lost_rreq_resend}(d), node $a$ will forward  RREP${}_{2}$ and again $s$ will establish a route to $d$.

To implement this behaviour one can simply drop the
Lines~\ref{rrep:line3} and \ref{rrep:line25}--\ref{rrep:line26} of \Pro{rrep} (RREP handling), 
keeping Lines~\ref{rrep:line5}--\ref{rrep:line23a} only.

This solution guarantees the forwarding of the RREP message. However, it might be the case that outdated information is forwarded and, as a consequence, non-optimal information is stored in the routing tables. This is shown by the example presented in \Fig{non_optimal_forwarding}.

The example assumes a linear topology with $5$ nodes.  In Part (b), node $s$ receives a data packet destined to node $d$; it initiates a route discovery process. 
The request is forwarded by nodes $a$, $b$ and $c$ until it reaches the destination $d$. 
Node $d$ then generates a route reply and unicasts the message to $c$ (Part (c)).
After the RREP message is (successfully) sent, \Fig{non_optimal_forwarding}(d), a link between $a$ and $d$ is established and node $d$ broadcasts a new RREQ message, destined to $a$. 
This message is received by nodes $a$ and $c$. In principle node $c$ would later forward the request; however, this forwarding and the subsequent actions do not add anything to the example and therefore we drop this bit. 
\begin{exampleFig}{Always forwarding RREP messages}{fig:non_optimal_forwarding}
\FigLine[slxsr]%
  {The initial state.}{fig/ex_nonoptimal_due_to_forward_1}{}
  {$s$  broadcasts a new RREQ message destined to $d$;\\the request floods the network.}{fig/ex_nonoptimal_due_to_forward_2}{}
\FigLine[xslxsr]%
  {Node $d$ generates and send a RREP message to $c$.}{fig/ex_nonoptimal_due_to_forward_3}{}
  {The topology changes;\\$d$  broadcasts a new RREQ message destined to $a$.\footnotemark}{fig/ex_nonoptimal_due_to_forward_5}{}
\FigLine[xslxsr]%
  {$a$ unicasts RREP$_{2}$ back to $d$;\\$b$ forwards RREP${}_{1}$.}{fig/ex_nonoptimal_due_to_forward_6}{}
  {Due to the modification, $a$ forwards RREP${}_{1}$.}{fig/ex_nonoptimal_due_to_forward_7}{}
  \end{exampleFig} 
\footnotetext{The message RREQ${}_{2}$ is also sent to node $c$. Since it does not change the example, we suppress this message.}
After $a$ has initiated a route reply as a consequence of RREQ${}_{2}$, which is sent back to $d$, it receives RREP${}_{1}$ from $b$---the reply generated by node $d$ and destined to $s$. In the original version of AODV, as presented in Sections~\ref{sec:types} and \ref{sec:modelling_AODV}, the reply would be dropped, since $a$ does not update its routing table. In the modified version, $a$ creates a message by $\rrep{3}{d}{1}{s}{a}$, which is sent to node $s$ (\nhop{\rt}{\oip}). Note, that $a$ does not update its own routing table. As a consequence of this message, node $s$ updates its routing table and creates an entry to $d$ with sequence number $1$ and hop count $4$ (Part (f)).

Although this information is not incorrect, it is outdated. Any data
packet sent from $s$ to $d$ would be forwarded to $a$ and then
immediately to the destination, thanks to
node $a$ having fresher information (in its routing table the sequence number belonging to $d$ is $2$). As a general rule,
it makes sense to use the newest available information on 
the route to the destination node: if an intermediate node's  routing table contains an entry for the destination node
that is valid and fresher than that in the received RREP message, the intermediate node ought to
update the contents of the RREP message to reflect this.
To achieve this one can replace Line~\ref{rrep:line13} of \Pro{rrep} by
\[
\unicast{\nhop{\rt}{\oip}}{\rrep{\dhops{\rt}{\dip}}{\dip}{\sqn{\rt}{\dip}}{\oip}{\ip}}\ .
\]
In case the received reply contained fresher information, the routing table was already updated.
The full modified RREP handling is shown in \Pro{rrepmodifiedforward}. 
Note that Lines~\ref{rrepmod:line8} and \ref{rrepmod:line21}
(Lines~\ref{rrep:line11} and~\ref{rrep:line21} in the original
process) are also changed. The reason for this change is that information should only be forwarded when the intermediate node has a
\emph{valid}
route to the destination of the route discovery process. Assume for example the situation given in
\Fig{non_optimal_forwarding}(d). As before, node $a$ sends RREP${}_{2}$; but just before
RREP${}_{1}$ is handled by $a$, the unreliable link between $a$ and $d$ breaks and $a$ invalidates
its routing table for $d$, i.e.,\ it changes into $(d,3,\kno,\inval,1,d)$. 
Under such circumstances a route reply should not be forward, since any data packet reaching the intermediate node (in the example $a$) would be dropped.

All invariants presented in Sections~\ref{sec:invariants} and \ref{sec:interpretation} remain valid. However, a few proofs need adaptation. 
\begin{itemize}
\item In \Prop{starcastNew}(b), the case dealing with \Pro{rrep} now reads as follows:
\begin{description}
	\item[Pro.~\ref{pro:rrepmodifiedforward}, Line~\ref{rrepmod:line12}:\footnotemark] 
                \footnotetext{This line corresponds to Pro.~\ref{pro:rrep}, Line~\ref{rrep:line13} in the original specification.}
	The message has the form $\rrep{\xi(\dhops{\rt}{\dip})}{*}{*}{*}{*}$.
        By \Prop{positive hopcount} $\xi(\dhops{\rt}{\dip})>0$, so the antecedent does not hold.
\end{description}
  \algsetup{linenodelimiter=.,linenosize=\tiny}
  \begin{algorithm}[H]
    {\footnotesize
      \caption{RREP handling (Forwarding the Route Reply)}
      \label{pro:rrepmodifiedforward}
      \begin{algorithmic}[1]
\DEFPROCESS{\RREP}{\hops\comma\dip\comma\dsn\comma\oip\comma\sip\,\comma\,\ip\comma\sn\comma\rt\comma\rreqs\comma\queues}
		\UPD{\rt:=\upd{\rt}{(\dip\comma\dsn\comma\kno\comma\val\comma\hops+1\comma\sip\comma\emptyset)}}		\label{rrepmod:line1}
		\PAR																																									\label{rrepmod:line2}
		\IF[this node is the originator of the corresponding RREQ]{$\oip = \ip$}	
			\COMLINE{a packet may now be sent; this is done in the process \AODV}
			\aodvL{\ip}{\sn}{\rt}{\rreqs}{\queues}	
		\ELSIF[this node is not the originator; forward RREP]{$\oip \not= \ip$}
			\PAR
				\IF[valid route to \oip\ and to \dip]{$\oip\in\akD{\rt}\ans \dip\in\akD{\rt}$}														\label{rrepmod:line8}
					\COMLINE{add next hop towards $\oip$ as precursor and forward the route reply}						
					\UPD{\rt := \addprecrt{\rt}{\dip}{\{\nhop{\rt}{\oip}\}}}																					\label{rrepmod:line10}
					\UPD{\rt := \addprecrt{\rt}{\nhop{\rt}{\dip}}{\{\nhop{\rt}{\oip}\}}}																	\label{rrepmod:line11}
					\STARTPRIO	
						\unicast{\nhop{\rt}{\oip}}{\rrep{\dhops{\rt}{\dip}}{\dip}{\sqn{\rt}{\dip}}{\oip}{\ip}}\ .   								\label{rrepmod:line12}
						\aodvL{\ip}{\sn}{\rt}{\rreqs}{\queues}	
					\PRIO
						\COMspec{If the transmission is unsuccessful, a RERR message is generated}
						\UPD{\dests:=\{(\rip,\inc{\sqn{\rt}{\rip}})\,|\,\rip\in\akD{\rt}\ans \nhop{\rt}{\rip}=\nhop{\rt}{\oip}\}}			
						\UPD{\rt:=\inv{\rt}{\dests}}																								
												\UPD{\queues:=\setrrf{\queues}{\dests}}													
						\UPD{\pre:=\bigcup\{\precs{\rt}{\rip}\,|\,(\rip,*)\in\dests\}}											
						\UPD{\dests:=\{(\rip,\rsn)\,|\,(\rip,\rsn)\in\dests\ans \precs{\rt}{\rip}\not=\emptyset\}}				
						\groupcast{\pre}{\rerr{\dests}{\ip}}\ .\ \aodv{\ip}{\sn}{\rt}{\rreqs}{\queues} 					
					\ENDPRIO
				\ELSIF[no valid route to \oip\ or to \dip]{$\oip\not\in\akD{\rt}\ors\dip\not\in\akD{\rt}$}					\label{rrepmod:line21}									
					\aodvL{\ip}{\sn}{\rt}{\rreqs}{\queues}
				\ENDIFii																													
			\ENDPAR																								
		\ENDIFii	
	\ENDPAR
	

	\end{algorithmic}
    }
  \end{algorithm}

\item In \Prop{msgsendingii}(b), the case dealing with \Pro{rrep} now reads as follows:
\begin{description}
	\item[Pro.~\ref{pro:rrepmodifiedforward}, Line~\ref{rrepmod:line12}:] 
		Here, $\dsnc:=\xi(\sqn{\rt}{\dip})$. The last routing table update happened in Line~\ref{rrepmod:line1}.
		The update uses $\xi(\dsn)$, which stems, through Line~\ref{aodv:line12}
		of Pro.~\ref{pro:aodv}, from an incoming RREP message (Pro.~\ref{pro:aodv}, Line~\ref{aodv:line2}). 
		For this incoming RREP message the invariant holds, i.e.\ $\xi(\dsn)\geq1$. By \Prop{dsn increase}, the sequence number is increased monotonically, and hence
		$
		\dsnc:=\xi(\sqn{\rt}{\dip})\geq\xi(\dsn)\geq1
		$.
\end{description}
\item The case of \Prop{msgsending}(b) dealing with RREP handling now becomes
\begin{description}
	\item[Pro.~\ref{pro:rrepmodifiedforward}, Line~\ref{rrepmod:line12}:]
		The message has the form
		$\xi(\rrep{\dhops{\rt}{\dip}}{\dip}{\sqn{\rt}{\dip}}{\oip}{\ip}).$
		Hence $\hopsc:=\xi(\dhops{\rt}{\dip})$,
                $\dipc:=\xi(\dip)$, $\dsnc:=\xi(\sqn{\rt}{\dip})$,
                $\ipc:=\xi(\ip)=\dval{ip}$ and $\xiN{\ipc}=\xi$.
                Line~\ref{rrepmod:line1} guarantees that $\dipc=\xi(\dip)\in\kd{\ipc}$.
                Since the sequence number and the hop count are taken from the routing table, we get immediately\\
                \mbox{}\hfill$\begin{array}[b]{r@{~=~}l@{~=~}l}
				\sq[\dipc]{\ipc}& \sqn{\xi(\rt)}{\xi(\dip)} & \dsnc\\
				\dhp[\dipc]{\ipc} & \dhops{\xi(\rt)}{\xi(\dip)} & \hopsc\ .
				\end{array}$\hfill\mbox{}\\
		With exception of its precursors, which are irrelevant here, the routing table
		does not change between Lines~\ref{rrepmod:line8} and \ref{rrepmod:line12}.
		So, by Line~\ref{rrepmod:line8},  $\dipc=\xi(\dip)\in\akD{\xi(\rt)}$ and therefore\\
		\mbox{}\hfill$\begin{array}[b]{r@{~=~}l@{~=~}l}
		\sta[\dipc]{\ipc} & \status{\xi(\rt)}{\xi(\dip)} & \val\;.
		\end{array}$\hfill\mbox{}
\end{description}
\item The 7th case of the proof of \Prop{nhop_well_defined} is changed to
\begin{description}
\item[Pro.~\ref{pro:rrepmodifiedforward}, Line~\ref{rrepmod:line11}:] By Line~\ref{rrepmod:line8} $\xi(\oip)\in\akD{\xi(\rt)}$ and  $\xi(\dip)\in\akD{\xi(\rt)}$.
\end{description}
\item In \Prop{dhops_well_defined}, the following case needs to be added:
\begin{description}
\item[Pro.~\ref{pro:rrepmodifiedforward}, Line~\ref{rrepmod:line12}:]
By Line~\ref{rrepmod:line8} $\xi(\dip)\in\akD{\xi(\rt)}$.
\end{description}
\item In the proof of \Prop{upd_well_defined}, the cases for \Pro{rrep}, Lines~\ref{rrep:line3} and~\ref{rrep:line25} are skipped.
\item In \Prop{addpreRT_well_defined}, the last case changes into
\begin{description}
  \item[Pro.~\ref{pro:rrepmodifiedforward}, Line~\ref{rrepmod:line11}:]
 By Line~\ref{rrepmod:line8}, $\xi(\dip)\in\akD{\xi(\rt)}\subseteq\kD{\xi(\rt)}$, so a routing table entry for $\xi(\dip)$
 exists. Using \Prop{route to nhip}, this implies that $\nhop{\xi(\dip)}{\xi(\rt)}\in\kD{\xi(\rt)}$.
\end{description}
\item In \Thm{route correctness}(c), the case dealing with \Pro{rrep} becomes
\begin{description}
\item[Pro.~\ref{pro:rrepmodifiedforward}, Line~\ref{rrepmod:line12}:]
		The proof is the same as for Pro.~\ref{pro:rreq}, Line~\ref{rreq:line26}.
\end{description}
\item The last case of \Prop{dsn}(a) is changed to
\begin{description}
	\item[Pro.~\ref{pro:rrepmodifiedforward}, Line~\ref{rrepmod:line12}:]
		A route reply with $\dipc\mathbin{:=}\xiN{\dval{ip}}(\dip)$ and
                $\dsnc\mathbin{:=}\xiN{\dval{ip}}(\sqn{\rt}{\dip})\mathbin=\sq[\dipc]{\dval{ip}}$ is initiated. By
                Invariant~\eqref{eq:dsn} $\dsnc=\sq[\dipc]{\dval{ip}}\leq \xiN{\dipc}(\sn)$.
\end{description}
\end{itemize}
Surely, always forwarding (unicasting) replies increases the number of messages in the network. 
However, as illustrated by the examples of Figures~\ref{fig:route_lost} and~\ref{fig:route_lost_rreq_resend},
the policy to always forward the route reply
significantly increases the probability of a route discovery process being successful.
As a consequence, the probability of the originator re-issuing a route request to establish a route is much smaller. 
Such a re-sending would yield another broadcast cycle, which,  with respect to network load, is much
more expensive than the extra unicast of RREP messages.

\subsection{Updating with the Unknown Sequence Number}
\label{ssec:skipunknownflag}

In this section we evaluate the resolutions of \hyperlink{sss921b}{Ambiguity 2} 
of \Sect{interpretation}.
We have already discarded Resolution (2\ref{amb:2b}), as it leads to routing loops.
The alternatives, (2\ref{amb:2a}), (2\ref{amb:2c}) and (2\ref{amb:2d}),
have been shown to satisfy the loop freedom and route correctness property. 

\begin{wrapfigure}[9]{r}{0.35\textwidth}
 \vspace{-2ex}
\centering
\includegraphics[scale=1]{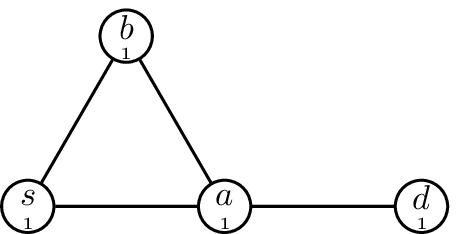}
\caption{$4$-node topology missing route optimisation opportunity}\label{fig:4nodetopAmb2}
\end{wrapfigure}
A disadvantage of Resolution (2\ref{amb:2a}) is that it misses opportunities to improve the
routes between two neighbouring nodes. It can lead to situations in which a node $s$ knows that node $d$
can be reached using $a$ as next hop, but at the same time does not know that there is a
valid 1-hop route to $a$ itself:  assume the topology given in \Fig{4nodetopAmb2}. The link between the nodes 
$s$ and $a$ is unreliable---messages sent via this link might get lost and the neighbouring nodes might detect that this link is broken. 
Let us further assume that $s$ has established a route to $a$;\pagebreak[3] the corresponding routing table entry might be 
$(a,1,\kno,\val,2,b)$ (the RREQ message from $s$ to $a$ got lost) or $(a,2,\kno,\inval,1,a)$ (a
1-hop connection was established, but the link broke down).
Next, node $d$ searches for a route to $s$. The generated RREQ message is received by $a$ and forwarded to $s$. 
Node $s$ creates a routing table entry to $d$ ($(d,2,\kno,\val,2,a)$) and tries to update its entry to $a$. 
However, by use of Resolution (2\ref{amb:2a}), neither of the above mentioned entries would be changed. 

This strongly gives the impression that information is not used in
an optimal way.

Resolutions (2\ref{amb:2c}) and (2\ref{amb:2d}) do not suffer from this drawback.
However, they have their own problems.
Resolution (2\ref{amb:2d}) gives rise to non-optimal routes, as illustrated in \Fig{sub_optimal_routes}.
In the initial state (\Fig{sub_optimal_routes}(a)), a route between
$a$ and $d$ is established through a standard RREQ-RREP cycle.
Then, in Part (b)), the connection between $a$ and $d$ breaks down. $a$ and $d$ detect the link break and
invalidate their routing table entries for each other, thereby increasing the destination sequence numbers.
Subsequently, the connection between $a$ and $d$ comes back up, and node $b$ (connected to $d$)
initiates a route request for a node $x$, which is not to be found in the vicinity (\Fig{sub_optimal_routes}(c)).
\begin{exampleFig}{Resolution (2d) gives rise to non--optimal routes}{fig:sub_optimal_routes}
\FigLine[xslxsr]%
  {The initial state;\\$a$ established a route to $d$ by a RREQ-RREP cycle.}{fig/ex_skipSQN-flag1_1}{}
  {The link between $a$ and $d$ breaks down;\\$a$ and $d$ invalidate their entries to each other.}{fig/ex_skipSQN-flag1_2}{}
 \FigLine[xslxsr]%
  {The link reappears;\\a RREQ from $b$ floods the network.}{fig/ex_skipSQN-flag1_3}{}
  {The topology changes again.}{fig/ex_skipSQN-flag1_4}{}
 \FigLine[xslxsr]%
  {$s$ broadcasts a new RREQ destined for $d$.}{fig/ex_skipSQN-flag1_5}{}
  {$s$ receives RREPs from $a$ and $d$.}{fig/ex_skipSQN-flag1_6}{}
  \end{exampleFig}
\noindent As a consequence, when $a$ receives the forwarded RREQ message from $d$, it validates its routing
table entry for $d$, the destination sequence number being higher than $d$'s own sequence number.
In Parts (d) and (e), a direct link between $s$ and $d$ appears, and $s$ searches for a route to $d$.
Its RREQ message is answered both by $a$, which knows a route to $d$, and by $d$ itself (\Fig{sub_optimal_routes}(f)).
Regardless which of the two RREP messages arrives first, $s$ establishes a route to $d$ of
length 2 via $a$, since the RREP message from $a$ carries a higher destination sequence number for
$d$ than the RREP message from $d$ itself. 
This anomaly pleads against the use of Resolution (2\ref{amb:2d}).

Although Resolution (2\ref{amb:2c}) seems to be the intention of the RFC
(cf.\ \hyperlink{sss921b}{Ambiguity 2}), it gives rise to route discovery failures as illustrated
in \Fig{route_lost_rreq_resend}. This situation is so common, and the lack of route discovery is
such a severe problem, that for the original AODV Resolution (2\ref{amb:2c}) can be judged
worse than (2\ref{amb:2a}) and (2\ref{amb:2d}), and should not be used.
The problem is a combination of the use of Resolution (2\ref{amb:2c}) and AODV's failure to
forward route replies. Once the latter problem is satisfactory addressed, for instance by following
our proposal in \SSect{modifyRREP}, the problem of \Fig{route_lost_rreq_resend} is solved, and
Resolution (2\ref{amb:2c}) is back in the race.
Nevertheless, the following example shows a remaining problem, that pertains to both
Resolutions~(2\ref{amb:2c}) and (2\ref{amb:2d}). (The sequence-number-status flags in the routing
table entries of \Fig{no_route_reply} conform to Resolution~(2\ref{amb:2c})---however, they play no role
in this example.)

\begin{exampleFig}{Failure in generating a route reply (Resolutions (2c) and (2d))}{fig:no_route_reply}
\FigLine[xslxsr]%
  {The initial state;\\$s$ established a route to $d$ by a RREQ-RREP cycle.}{fig/ex_skipSQN-flag2_1}{}
  {The link between $s$ and $d$ breaks down;\\$s$ and $d$ invalidate their entries to each other.}{fig/ex_skipSQN-flag2_2}{}
 \FigLine[xslxsr]%
  {The link reappears;\\a RREQ from $a$ floods the network.}{fig/ex_skipSQN-flag2_3}{}
  {The topology changes;\\$s$ and $d$ invalidate their entries to each other.}{fig/ex_skipSQN-flag2_4}{}
 \FigLineHalf[xsl]%
  {$s$ broadcasts a new RREQ message destined to $d$;\\$d$'s reply cannot be sent to $s$.}{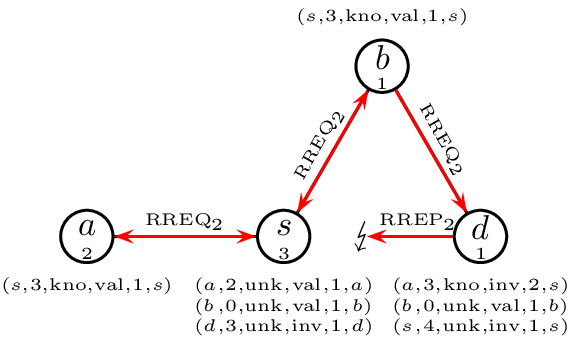}{}
\end{exampleFig}

\noindent
In the initial state (\Fig{no_route_reply}(a)), a route between $s$ and $d$ is established through a
standard RREQ-RREP cycle.
Then, in Part (b), the connection between $s$ and $d$ breaks down. $d$ detects the link break and
invalidates its routing table entry for $s$, thereby increasing the destination sequence number.
In \Fig{no_route_reply}(c), the connection between $s$ and $d$ comes back up, and $a$ initiates a route request
for a node $x$ (which is not to be found in the vicinity).
As a consequence, when $d$ receives the forwarded route request from $s$, it validates its
routing table entry for $s$, the destination sequence number being higher than $s$'s own sequence number.
In \Fig{no_route_reply}(d), the connection breaks down and the entry becomes again invalid. The destination sequence
number of the entry is now $2$ higher than $s$'s own sequence number.
Moreover, a node $b$ appears in the network, and gets connected to $s$ and $d$.
From this point onwards the topology remains stable and the predicate
$\textbf{connected}^*(s,d)$ (cf.\ Page~\pageref{pg:connected*}) holds.
In Part (e), $s$ searches for a route to $d$. 
Even though this increases $s$'s own sequence number, it is
still smaller than the destination sequence number for $s$ at $d$. When the route request reaches
$d$ (via $b$), $d$ tries to update its own routing table entry for $s$.
However, $d$ already has an invalid entry for $s$ with a higher sequence number.
As a result, no update occurs and the route from $d$ to $s$ remains invalid.
Therefore $s$ does not get a reply.


Since each of the Resolutions (2\ref{amb:2a}-\ref{amb:2d}) turned out to have serious disadvantages,
we now propose an alternative---Resolution (2e)---that does not share these disadvantages.
The intuition is that when a node changes an invalid route into a valid one, while keeping the
sequence number from the routing table (as done in Resolutions (2\ref{amb:2c}-\ref{amb:2d})), it
needs to undo the increment of the sequence number performed upon invalidation of the route.
This involves decrementing destination sequence numbers, a
practice that goes strongly against the spirit of the RFC\@. Nevertheless, since the net
  sequence number stays the same, we are able to show that all our
invariants are maintained, which constitutes a formal proof of loop freedom and route correctness.
So in this special case decrementing destination sequence numbers turns out to be harmless.

{
\renewcommand{\rt}{\dval{rt}}
  \newcommand{\nrt}{\dval{nrt}}
  \newcommand{\nr}{\dval{nr}}
  \newcommand{\s}{\dval{s}}
  \newcommand{\ns}{\dval{ns}}
Resolution (2e) is a variant of Resolution (2\ref{amb:2d}), defined through a modification in the
definition of \hyperlink{update}{$\fnupd$}. 
The 5th clause ($\nrt\cup\{\nr'\}\ \mbox{if }
\pi_{1}(\route)\in\kD{\rt} \wedge  \pi_3(\route)=\unkno$) is split into two parts (depending on the validity of the route).
\hypertarget{mod_update}{
\[\begin{array}{@{}l@{\hspace{0.5em}}c@{\hspace{0.45em}}l@{}}
\upd{\rt}{\hspace{-1pt}\route}&:=& \left\{
\begin{array}{@{\,}ll@{}}
\rt\cup\{\route\} & \mbox{if }  \pi_{1}(\route)\not\in\kD{\rt}\\[1mm]
\nrt\cup\{\nr\}&\mbox{if }  \pi_{1}(\route)\in\kD{\rt} \wedge  \sqn{\rt}{\pi_{1}(\route)}<\pi_{2}(\route)\\[1mm]
\nrt\cup\{\nr\}&\mbox{if }  \pi_{1}(\route)\in\kD{\rt} \wedge \sqn{\rt}{\pi_{1}(\route)}=\pi_{2}(\route) \wedge \dhops{\rt}{\pi_{1}(\route)}>\pi_{5}(\route)\\
&\hspace{25.4mm}{}\wedge \pi_3(\route)=\kno\\[1mm]
\nrt\cup\{\nr\}&\mbox{if }  \pi_{1}(\route)\in\kD{\rt} \wedge \sqn{\rt}{\pi_{1}(\route)}=\pi_{2}(\route) \wedge \status{\rt}{\pi_{1}(\route)}=\inval\\
&\hspace{25.4mm}{}\wedge \pi_3(\route)=\kno\\[1mm]
\nrt\cup\{\nr''\}&\mbox{if } \pi_{1}(\route)\in\kD{\rt} \wedge  \pi_3(\route)=\unkno \wedge \status{\rt}{\pi_{1}(\route)}=\val\\
\nrt\cup\{\nr'''\}&\mbox{if } \pi_{1}(\route)\in\kD{\rt} \wedge  \pi_3(\route)=\unkno \wedge \status{\rt}{\pi_{1}(\route)}=\inval\\
\nrt\cup\{\ns\}&\mbox{otherwise\ ,}
\end{array}
\right.\vspace{-1ex}
\end{array}\]
}
where (in the terminology of \SSSect{update})
$\nr'':=(\dip_{\nr}\comma\pi_{2}(\s)\comma\pi_3(\s)\comma\flag_{\nr}\comma\hops_{\nr}\comma\nhip_{\nr}\comma\pre_{\nr})$ and
$\nr''':=(\dip_{\nr}\comma\pi_{2}(\s)\decremented\comma\pi_3(\s)\comma\flag_{\nr}\comma\hops_{\nr}\comma\nhip_{\nr}\comma\pre_{\nr})$.
We illustrate the behaviour of this modification using a similar example as in the Section about \hyperlink{sss921b}{Ambiguity 2}:
as a consequence of the incoming RREQ message $\rreq{1}{\dval{rreqid}}{x}{7}{\kno}{s}{2}{a}$ the routing table entry $(a,2,\kno,\val,2,b,\emptyset)$ of node $d$ is now updated 
to $(a,2,\kno,\val,1,a,\emptyset)$---the same behaviour as in  Resolution (2\ref{amb:2d})---but 
the entry $(a,2,\kno,\inval,2,b,\emptyset)$ is updated to $(a,1,\kno,\val,1,a,\emptyset)$.
}

Any of the interpretations and variants of AODV using Resolution (2\ref{amb:2c})---our default resolution of Ambiguity 2---that
have been shown loop free in this paper, remain loop free when using Resolution (2e) instead%
---the invariants, proofs and proof modifications of Section~\ref{sec:invariants}--\ref{ssec:modifyRREP} remain
valid, with the following modifications:
\begin{itemize}
\item \Prop{dsn increase} is reformulated as:
  \begin{quote}
  In each node's routing table, the
  {\em net} sequence number for a given destination increases monotonically.
  That is, for $\dval{ip},\dval{dip}\mathbin\in\IP$, if $N \ar{\ell} N'$
  then $\nsq{\dval{ip}}\leq\fnnsqn_{N'}^\dval{ip}(\dval{dip})$.
  \end{quote}
  For the proof, note that the modified {\fnupd} cannot decrease a net sequence number, so again
  the only function that can decrease a net sequence number is \hyperlink{invalidate}{$\fninv$}. 
  When invalidating routing table entries using the function $\inv{\rt}{\dests}$, sequence numbers
  are copied from {\dests} to the corresponding entry in \rt. 
  It is sufficient to show that for all \plat{$(\dval{rip},\dval{rsn})\in\xiN{\dval{ip}}(\dests)$}
  $\sq[\dval{rip}]{\dval{ip}}\leq\dval{rsn}\decremented$, as all other sequence numbers in routing table
  entries remain unchanged.
\begin{description}
  \item[Pro.~\ref{pro:aodv}, Line~\ref{aodv:line32}; Pro.~\ref{pro:pkt}, Line~\ref{pkt2:line10}; Pro.~\ref{pro:rreq}, Lines~\ref{rreq:line18}, \ref{rreq:line30}; Pro.~\ref{pro:rrep}, Line~\ref{rrep:line18}:]~\\
  The set {\dests} is constructed immediately before the invalidation procedure. For $(\dval{rip},\dval{rsn})\in\xiN{\dval{ip}}(\dests)$, we have
  $\sq[\dval{rip}]{\dval{ip}} = \inc{\sq[\dval{rip}]{\dval{ip}}}\decremented = \dval{rsn}\decremented.$
\item[Pro.~\ref{pro:rerr}, Line~\ref{rerr:line5}:] When constructing {\dests} in Line~\ref{rerr:line2}, the condition $\xiN[N_{\ref*{rerr:line2}}]{\dval{ip}}(\sqn{\rt}{\rip})<\xiN[N_{\ref*{rerr:line2}}]{\dval{ip}}(\rsn)$ is taken into account, which immediately yields the claim for \plat{$(\dval{rip},\dval{rsn})\in\xiN{\dval{ip}}(\dests)$}.
\end{description}
\item We also need a weakened version of the old \Prop{dsn increase}:
  \begin{equation}\label{eq:invalidate no decrease}
  \mbox{An application of \hyperlink{invalidate}{$\fninv$} never decreases a sequence number in a routing table.}
  \end{equation}
  The proof is contained in proof of the old \Prop{dsn increase} and does not rely on the (modified) function $\fnupd$.
\item The reference to \Prop{dsn increase} in the proof of \Prop{rte} is replace by a reference to \Eq{invalidate no decrease}.
\item In the beginning of the proof of \Prop{qual}(a) the inequality
  $$\nsqn{\dval{rt}}{\dval{dip}} \leq \sqn{\dval{rt}}{\dval{dip}}=\dval{dsn}_{\dval{rt}} = \nsqn{\dval{rt}'}{\dval{dip}}$$
  turns into an equality \qquad$\nsqn{\dval{rt}}{\dval{dip}} =\dval{dsn}_{\dval{rt}} = \nsqn{\dval{rt}'}{\dval{dip}},$\\
  and follows, by the new definition of $\fnupd$, without the step that references \Eq{sqn_vs_nsqn}.
\item To adapt the proof of \Thm{inv_a} we need some new auxiliary invariants. The first states
  that the sequence number of an invalid \rte can never be 1.
\begin{equation}
  \dval{dip}\in\ikd{\dval{ip}} \ims\sq{\dval{ip}}\neq 1\ .
\end{equation}
 \begin{proofNobox}
Invalid \rtes only arise by applications of \hyperlink{invalidate}{$\fninv$} on valid \rtes;
furthermore, only calls of $\fninv$ can change the sequence number of an invalid \rte while keeping the route invalid.
  Hence it suffices to check all calls of \fninv. 
An application $\xiN{\dval{ip}}(\inv{\rt}{\dests})$ invalidates the \rte to \dval{rip} and changes its sequence
  number into \dval{rsn} for any pair \plat{$(\dval{rip},\dval{rsn})\in\xiN{\dval{ip}}(\dests)$}.
\begin{description}
\item[Pro.~\ref{pro:aodv}, Line~\ref{aodv:line32}; Pro.~\ref{pro:pkt}, Line~\ref{pkt2:line10}; Pro.~\ref{pro:rreq}, Lines~\ref{rreq:line18}, \ref{rreq:line30}; Pro.~\ref{pro:rrep}, Line~\ref{rrep:line18}:]~\\
By construction of {\dests} (immediately before the invalidation call) 
$(\dval{rip},\dval{rsn})\in\xiN{\dval{ip}}(\dests)$ implies $\dval{rsn}\mathop=\inc{\sqn{\xiN{\dval{ip}}(\rt)}{\dval{rip}}}$.
By \hyperlink{inc}{definition of $\fninc$}
we have $\inc{n}\neq 1$, for $n\in\NN$.
\item[Pro.~\ref{pro:rerr}, Line~\ref{rerr:line5}:]
Let $(\dval{rip},\dval{rsn})\mathbin\in\xiN[N_3]{\dval{ip}}(dests)$; then
$(\dval{rip},\dval{rsn})\mathbin\in\xiN[N_2]{\dval{ip}}(dests)$, and $\xiN[N_2]{\dval{ip}}(dests)$
stems from a received RERR message that must have been sent beforehand, say by a node $\ipc$ in
state $N^\dagger$. By \Prop{starcastrerr}, \plat{$\dval{rip}\in\ikd[N^\dagger]{\ipc}$} and
$\dval{rsn}=\fnsqn_{N^\dagger}^\dval{ip}(\dval{rip})$. So by
\hyperlink{induction-on-reachability}{induction on reachability}, $\dval{rip}\neq 1$.
\endbox
\end{description}
  \end{proofNobox}
 As an immediate corollary of this invariant we obtain that
  \begin{equation}\label{eq:invalidate not 1}
  \dval{dip}\in\ikd{\dval{ip}} \ims \inc{\nsq{\dval{ip}}}=\sq{\dval{ip}}\ .
  \end{equation}
\item Define the \emph{upgraded sequence number} of destination \dval{dip} at node \dval{ip} by
$$
  \usq{\dval{ip}} = \left\{\begin{array}{@{}ll@{}}\inc{\sq{\dval{ip}}} &
  \mbox{if}~\dval{dip}\in\akd{\dval{ip}} \ans \dhp{\dval{ip}}=1\\
  \sq{\dval{ip}} & \mbox{otherwise.}\end{array}\right.
  $$
By this definition we immediately get the following inequation.
\begin{eqnarray}\label{eq:usqn_inequation}
\usq{\dval{dip}}\geq\sq{\dval{dip}}\geq\usq{\dval{dip}}\decremented
\end{eqnarray}  
  After \Thm{state_quality} has been established, we obtain the following invariant,
  saying that in each routing table, the upgraded
  sequence number for any given destination increases monotonically:
  for $\dval{ip},\dval{dip}\mathbin\in\IP$ and a reachable network expression $N$,
  \begin{equation}\label{eq:usq increase}
  N \ar{\ell} N' \ims \usq{\dval{ip}}\leq\fnusqn_{N'}^\dval{ip}(\dval{dip})\ .
  \end{equation}
  \begin{proofNobox}
  We distinguish four cases.
  \begin{enumerate}[(i)]
  \item  Neither $\dval{dip}\in\akd{\dval{ip}} \ans \dhp{\dval{ip}}=1$ nor
  $\dval{dip}\in\akd[N']{\dval{ip}} \ans \fndhops_{N'}^\dval{ip}(\dval{dip})=1$ holds.
  Then \plat{$\usq{\dval{ip}}=\sq{\dval{ip}}\leq\fnsqn_{N'}^\dval{ip}(\dval{dip})=\fnusqn_{N'}^\dval{ip}(\dval{dip})$},
  where the inequality follows just as in the proof of \Prop{dsn increase}, taking into account that
  the modified 5th clause of \hyperlink{update}{$\fnupd$} cannot apply, as, using
  \Prop{upd_well_defined}, it would result in a routing table entry with
  $\dval{dip}\in\akd[N']{\dval{ip}} \ans \fndhops_{N'}^\dval{ip}(\dval{dip})=1$.
  \item Both $\dval{dip}\in\akd{\dval{ip}} \ans \dhp{\dval{ip}}=1$ and
  $\dval{dip}\in\akd[N']{\dval{ip}} \ans \fndhops_{N'}^\dval{ip}(\dval{dip})=1$ hold.
  Then
  \vspace{-1ex}
\begin{eqnarray*}
&\usq{\dval{ip}}=\inc{\sq{\dval{ip}}}=\inc{\nsq{\dval{ip}}}\leq\\
&\inc{\fnnsqn_{N'}^\dval{ip}(\dval{dip})}=\inc{\fnsqn_{N'}^\dval{ip}(\dval{dip})}=\fnusqn_{N'}^\dval{ip}(\dval{dip})\ ,
\end{eqnarray*}
  where the inequality follows by \Thm{state_quality}.
  \item 
   $\dval{dip}\in\akd[N']{\dval{ip}} \ans \fndhops_{N'}^\dval{ip}(\dval{dip})=1$ holds, but $\dval{dip}\in\akd{\dval{ip}} \ans \fndhops_{N}^\dval{ip}(\dval{dip})=1$ does not.
  Then
  \vspace{-1ex}
  \begin{eqnarray*}
&\usq{\dval{ip}}=\sq{\dval{ip}}\leq\inc{\nsq{\dval{ip}}}\leq\\
& \inc{\fnnsqn_{N'}^\dval{ip}(\dval{dip})}=\inc{\fnsqn_{N'}^\dval{ip}(\dval{dip})}=\fnusqn_{N'}^\dval{ip}(\dval{dip})\ ,
\end{eqnarray*}
  where the first inequality is by \Eq{invalidate not 1} 
  in case that $\dval{dip}\in\ikd{\dval{ip}}$ and by $\fnsqn_{N}^\dval{ip}(\dval{dip})=\fnnsqn_{N}^\dval{ip}(\dval{dip})$ otherwise;
  the second inequality follows by \Thm{state_quality}.
  \item $\dval{dip}\in\akd{\dval{ip}} \ans \dhp{\dval{ip}}=1$ holds,
  but $\dval{dip}\in\akd[N']{\dval{ip}} \ans \fndhops_{N'}^\dval{ip}(\dval{dip})=1$ does not. We consider two subcases.
  \begin{itemize}
    \item $\dval{dip}\not\in\akd[N']{\dval{ip}}$, which is equivalent to $\dval{dip}\in\ikd[N']{\dval{ip}} \ors \dip\not\in\kd[N']{\dval{dip}}$. By \Prop{destinations maintained}, $\dip\not\in\kd[N']{\dval{dip}}$ is not possible, hence
    $\dval{dip}\in\ikd[N']{\dval{ip}}$. Then, again using \Thm{state_quality} and \Eq{invalidate not 1},
  \vspace{-1ex}
  \begin{eqnarray*}
&\usq{\dval{ip}}=\inc{\sq{\dval{ip}}}=\inc{\nsq{\dval{ip}}}\leq\\
&\inc{\fnnsqn_{N'}^\dval{ip}(\dval{dip})}=\fnsqn_{N'}^\dval{ip}(\dval{dip})=\fnusqn_{N'}^\dval{ip}(\dval{dip})\ .
\end{eqnarray*}
    \item $\dval{dip}\in\akd[N']{\dval{ip}}$ and $\fndhops_{N'}^\dval{ip}(\dval{dip})\neq 1$. Then, by \Prop{positive hopcount},
  $\fndhops_{N'}^\dval{ip}(\dval{dip})\geq 2  >1= \dhp{\dval{ip}}$.
 As $N$ changes into $N'$, by \Thm{state_quality} the quality of the route to \dval{dip}
  cannot decrease: $\xiN{\dval{ip}}(\rt)\rtord\xiN[N']{\dval{ip}}(\rt)$.
  Yet the hop count strictly increases, so the net sequence number must strictly increase as well:
  $\nsq{\dval{ip}}<\fnnsqn_{N'}^\dval{ip}(\dval{dip})$.
  From this we get
  \vspace{-1ex}
  \begin{eqnarray*}
  &\usq{\dval{ip}}=\inc{\sq{\dval{ip}}}=\inc{\nsq{\dval{ip}}}\leq\\
  &\fnnsqn_{N'}^\dval{ip}(\dval{dip})=\fnsqn_{N'}^\dval{ip}(\dval{dip})=\fnusqn_{N'}^\dval{ip}(\dval{dip})\ .
  \end{eqnarray*}
  
   \vspace*{-4.8ex}\endbox
  \end{itemize}
  \end{enumerate}
  \end{proofNobox}

\item In the proof of \Thm{inv_a}, the \hyperlink{729Pro3Line4}{case of Pro.~\ref*{pro:rreq}, Line~\ref*{rreq:line6}}, where we
  ``assume that the first line holds'', we may no longer appeal to \Prop{dsn increase}.
  Instead we consider two sub cases.
\begin{itemize}
\item First, let $\dhp{\dval{nhip}}=1$.\\
Since  $\dval{nhip}\not=\dval{dip}$ we have, by \Prop{rte}(\ref{it_c}), $\dhp{\dval{ip}}\neq 1$ and hence,
by  \Prop{positive hopcount}, $\dhp{\dval{ip}}\geq 2>1=\dhp{\dval{nhip}}$. Hence to conclude that 
$\xiN{\dval{ip}}(\rt)\rtsord \xiN{\dval{nhip}}(\rt)$, it suffices to show that 
$\nsq{\dval{nhip}}\geq\nsq{\dval{ip}}$. Using the Equations \eqref{eq:usqn_inequation} (twice) and 
\eqref{eq:usq increase}, we get 
\vspace{-1ex}
\begin{eqnarray*}
&\nsq{\dval{nhip}}=\sq{\dval{nhip}} \geq \usq{\dval{nhip}}\decremented \geq\\
&\fnusqn_{N^\dagger}^{\dval{nhip}}(\dval{dip})\decremented\geq \fnsqn_{N^\dagger}^{\dval{nhip}}(\dval{dip})\decremented
\geq\xi(\osn)=\nsq{\dval{ip}}\ ,
\end{eqnarray*}
where the last inequality follows from $\fnsqn_{N^\dagger}^{\dval{nhip}}(\dval{dip}) > \xi(\osn)$, which holds in the circumstances considered (cf. 
Page~\pageref{pg:729Pro3Line4}).
\item If $\dhp{\dval{nhip}}\not=1$ the net sequence number is strictly increased:
\begin{eqnarray*}
&\nsq{\dval{nhip}}=\sq{\dval{nhip}} = \usq{\dval{nhip}} \geq\\
&\fnusqn_{N^\dagger}^{\dval{nhip}}(\dval{dip})\geq \fnsqn_{N^\dagger}^{\dval{nhip}}(\dval{dip})
>\xi(\osn)=\nsq{\dval{ip}}\ .
\end{eqnarray*}
(As before for we use Equations~\eqref{eq:usq increase},  \eqref{eq:usqn_inequation} and $\fnsqn_{N^\dagger}^{\dval{nhip}}(\dval{dip}) > \xi(\osn)$ for the proof of this inequality.)
Hence $\xiN{\dval{ip}}(\rt)\rtsord \xiN{\dval{nhip}}(\rt)$. 
\end{itemize}
\item The proof of \Prop{dsn}(b) needs to be modified; to this end we strengthen the statement as in
  (b) below, and establish (c) by simultaneous induction.
\begin{enumerate}[(a)]
\item[(b)]  The net destination sequence number of a routing table entry can never be greater than the
  destination's own sequence number.
  \begin{equation}\label{eq:ndsn}
     \nsq{\dval{ip}}\leq\xiN{\dval{dip}}(\sn)
  \end{equation}
\item[(c)] The sequence number of a destination appearing in a route error message can never be
  more than 1 greater than the destination's own sequence number.
  \begin{equation}\label{eq:dsn_rerr}
  N\ar{R:\starcastP{\rerr{\destsc}{\ipc}}}_{\dval{ip}}N'\ans (\dval{rip}_{c},\dval{rsn}_{c})\in\destsc
    \ims \dval{rsn}_{c}\decremented\leq\xiN{\dval{rip}_{c}}(\sn)
  \end{equation}
\end{enumerate}
  \begin{proofNobox}~
\begin{enumerate}[(a)]
\item[(b)] The statement holds in the initial states.
  By \Prop{invarianti_itemiii}, any update of $\xiN{\dval{dip}}(\sn)$ is harmless.
  Hence we have to examine all application calls of \hyperlink{update}{$\fnupd$} and
  \hyperlink{invalidate}{$\fninv$}, restricting attention to those calls that actually
  modify the entry for \dval{dip}, beyond its precursors.
\begin{description}
\item[Pro.~\ref{pro:aodv}, Lines~\ref{aodv:line10}, \ref{aodv:line14}, \ref{aodv:line18}:]
With Resolution 2(e) these calls maintain $\nsq{\dval{dip}}$.
\item[Pro.~\ref{pro:rreq}, Line~\ref{rreq:line6}; Pro.~\ref{pro:rrep}, Line~\ref{rrep:line5}:]
These updates yield a valid \rte with a known sequence number. The proof is unchanged from the
one of \Prop{dsn}(b).
\item[Pro.~\ref{pro:aodv}, Line~\ref{aodv:line32}; Pro.~\ref{pro:pkt}, Line~\ref{pkt2:line10};
  Pro.~\ref{pro:rreq}, Lines~\ref{rreq:line18}, \ref{rreq:line30}; Pro.~\ref{pro:rrep}, Line~\ref{rrep:line18}:]
By construction of {\dests} (immediately before the invalidation call) 
\plat{$(\dval{rip},\dval{rsn})\in\xiN{\dval{ip}}(\dests)$} implies $\dval{rsn}=\inc{\sqn{\xiN{\dval{ip}}(\rt)}{\dval{rip}}}$.
Hence the call maintains $\nsq{\dval{dip}}$.
\item[Pro.~\ref{pro:rerr}, Line~\ref{rerr:line5}:]
Let $(\dval{rip},\dval{rsn})\mathbin\in\xiN[N_3]{\dval{ip}}(dests)$; then
$(\dval{rip},\dval{rsn})\mathbin\in\xiN[N_2]{\dval{ip}}(dests)$, and $\xiN[N_2]{\dval{ip}}(dests)$
stems from a received RERR message that must have been sent beforehand, say by a node $\ipc$ in
state $N^\dagger$. By Invariant~\Eq{dsn_rerr} we have
$\nsq[\dval{rip}]{\dval{ip}} = \dval{rsn}\decremented \leq \xiN[N^\dagger]{\dval{rip}}(\sn) \leq \xiN{\dval{rip}}(\sn)$.
\end{description}
\item[(c)] Immediately from \Prop{starcastrerr} and Invariant~\Eq{ndsn}.
Here the weakened form of \Prop{starcastrerr} proposed in the
\hyperlink{amb8modifications}{proof modifications} for Resolutions (8\ref{amb:8d}--\ref{amb:8e}) is sufficient.
The same holds for the weakened form of \Prop{starcastrerr} proposed in the
\hyperlink{amb9modifications}{proof modifications} for Resolution (9b).
\endbox
\end{enumerate}
  \end{proofNobox}
\end{itemize}

In case Resolution (2e) is chosen, the sequence-number-status flag becomes redundant and can be
skipped, just as for Resolutions (2\ref{amb:2a}) and (2\ref{amb:2d})---see Footnote~\ref{foot:skipped}.
Moreover, with Resolution (2e) it would make sense to record the net sequence number in \rtes rather
than the sequence number, because only the net sequence number is monotonically increasing.
This means that sequence numbers of \rtes are not incremented upon invalidation, but instead a node
that may initiate a route reply bases its actions on the incremented value of the destination
sequence number in the received RREQ message.

\subsection{From Groupcast to Broadcast}\label{ssec:groupandbroadcast}

\begin{quote}\raggedright\small
``{\tt 
   For each valid route maintained by a node as a routing table entry,
   the node also maintains a list of precursors that may be forwarding
   packets on this route.  These precursors will receive notifications
   from the node in the event of detection of the loss of the next hop
   link.}''\hfill \cite[Sect.~2]{rfc3561}
\end{quote}
\noindent This notification is modelled by means of a groupcast mechanism. 
It sends error messages pertaining to certain routes to the precursors collected for those routes only.
The idea is to reduce the number of messages received and handled. However, 
precursor lists are incomplete. They are updated only when a RREP message is sent
(Lines~\ref{rreq:line24}, \ref{rreq:line25} of Pro.~\ref{pro:rreq} and Lines~\ref{rrep:line12a},
\ref{rrep:line12b} of Pro.~\ref{pro:rrep}). The following example, showing a standard RREQ-RREP cycle, illustrates that all nodes 
not handling a route reply have no information about precursors; even those nodes
that handle the RREP message may have insufficient information. 
As a consequence, some nodes are not informed of a link break and will use a broken route; hence packets are lost. 

\setlength{\leftSpecial}{57pt}
\begin{exampleFig}[Precursor lists are incomplete]{Precursor lists are incomplete\footnotemark}{fig:prec lists}
\FigLine[speciallr]%
  {The initial state.}{fig/ex_precursorlists_1}{}
  {$s$ broadcasts a new RREQ destined to $d$;\\the request floods the network.}{fig/ex_precursorlists_2}{}
\FigLine[]%
  {The RREP message is sent from $d$ to $s$;\\only $a$ updates
    one of its precursor lists}{fig/ex_precursorlists_3}{}
  {The topology changes;\\ $a$, $d$, $s$ detect link breaks; no RERR message is sent.}{fig/ex_precursorlists_4}{}
\end{exampleFig}
\footnotetext{This is the only example where precursor lists are shown; they are the last component of an entry.}

The example is a standard RREQ-RREP cycle. Within the network given in \Fig{prec lists}(a), a data packet is inserted at node $s$, destined for $d$.
Consequently, $s$ issues a route discovery process.
In Part (b) the RREQ message floods the network. While handling RREQ messages no (non-empty) precursor list is set or changed. 
In fact, it is not possible to detect the precursors for a route to the originator of the route request when handling and forwarding a RREQ message;
the necessary information is not available.
Thus, whenever a link break is detected during a route request process, no RERR message is
sent, except when a node has information from previous control messages.
In \Fig{prec lists}(c), the reply is sent from node $d$ to~$s$.
When node $a$ forwards the RREP message, its adds $s$ to its list of precursors of the route to
$d$ (Pro.~\ref{pro:rrep}, Line~\ref{rrep:line12a}). However, it fails to add $d$ as a precursor
of the route to $s$.
In Part (d), the links between nodes $s$, $a$ and $b$ break down. 
Although nodes $a$, $b$ and $s$ detect the link break, they do not send
error messages---so nodes $c$ and $d$ will not be informed about the broken routes.
If these nodes receive packets for~$s$ they will keep sending them via $a$ or $b$, without ever
learning that none of those packets ever reaches $s$. In detail, when node $b$ (or $a$) receives
  a data packet for $s$ from $c$ (or $d$) it drops the packet (\Pro{pkt}, Line~\ref{pkt2:line16})
  and composes a error message reporting the broken link to $s$ (\Pro{pkt}, Line~\ref{pkt2:line20}).
  However, this error message is send to the list of precursors for its route to $s$, which in our example is still empty.
A variant of this example that constitutes a counterexample to the packet delivery property of
\Sect{properties} was already presented in \Fig{lost_due_precs}.

As already remarked in \SSect{rrep}, the failure of node $a$ to add $d$ as a
precursor of its route to $s$ can be remedied by the
addition of the line \mbox{$\assignment{\rt := \addprecrt{\rt}{\oip}{\{\nhop{\rt}{\dip}\}}}$}
to Pro.~\ref{pro:rrep}, right after Line~\ref{rrep:line12b}.
One can even go a step further and also add the line 
$$\mbox{$\assignment{\rt := \addprecrt{\rt}{\nhop{\rt}{\oip}}{\{\nhop{\rt}{\dip}\}}}$},$$
which is the equivalent to Pro.~\ref{pro:rrep}, Line~\ref{rrep:line12b}.
However, the problem at node $b$ cannot be fixed using precursors.
 
A possible solution is to abandon precursors and to replace every groupcast by a broadcast. 
At first glance  this strategy seems to need more bandwidth, but this is not the case. 
Sending error messages to a set of precursors is implemented at the link layer
by broadcasting the message anyway;
a node receiving such a message then checks the header to determine whether it is one of 
the intended recipients. Instead of analysing the header only, a node can just as well
read the message and decide whether the information contained in the message is 
of use. To be more precise: an error message is useful for a node 
if  the node has established a route to one of the nodes listed in the message, and 
the next hop to a listed node is the sender of the error message. 
In case a node finds useful information inside the message, it should update its routing table and
distribute another 
error message. This is exactly what happens in the route error process (\Pro{rerr}). 

In the specification given in Sections~\ref{sec:types} and~\ref{sec:modelling_AODV}, 
the last entry of a routing table entry can be dropped; yielding small adaptations in functions and function calls (for example 
Line~\ref{aodv:line10} of Pro.~\ref{pro:aodv} should be $\assignment{rt:=\upd{\rt}{(\sip,0,\unkno,\val,1,\sip)}}$). 
Next to these small adaptations, the following changes need to be implemented:

\begin{description}
\item [Pro.~\ref{pro:aodv}, Line~\ref{aodv:line33};
	Pro.~\ref{pro:pkt}, Line~\ref{pkt2:line14};
	Pro.~\ref{pro:rreq}, Lines~\ref{rreq:line19}, \ref{rreq:line31};
	Pro.~\ref{pro:rrep}, Line~\ref{rrep:line20}:]~\\
        The command $\groupcastP{\pre}{\rerr{\dests}{\ip}}$ is replaced by $\broadcastP{\rerr{\dests}{\ip}}$.
\item[\bf Pro.~\ref{pro:aodv}, Lines~\ref{aodv:line31},~\ref{aodv:line31a};
	Pro.~\ref{pro:pkt}, Lines~\ref{pkt2:line12}, \ref{pkt2:line13};
	Pro.~\ref{pro:rreq}, Lines~\ref{rreq:line17}, \ref{rreq:line17a}, \ref{rreq:line24}, \ref{rreq:line25}, \ref{rreq:line29}, \ref{rreq:line29a};]~\\
\hspace*{-27.25pt}{\bf Pro.~\ref{pro:rrep}, Lines~\ref{rrep:line12a}, \ref{rrep:line12b}, \ref{rrep:line17}, \ref{rrep:line17a}:}
	These lines are dropped without replacement.
\item [Pro.~\ref{pro:pkt}, Line~\ref{pkt2:line20}:]
        The command $\groupcastP{\precs{\rt}{\dip}}{\rerr{\{(\dip,\sqn{\rt}{\dip})\}}{\ip}}$
        is replaced by $\broadcastP{\rerr{\{(\dip,\sqn{\rt}{\dip})\}}{\ip}}$.
\item[Pro.~\ref{pro:rerr}, Lines~\ref{rerr:line3}--\ref{rerr:line6}:] The error forwarding is replaced by
\begin{algorithmic}[8]%
  {\small\PAR
		\IF[the RERR needs to be forwarded]{$\dests\not=\emptyset$}
			\broadcast{\rerr{\dests}{\ip}}\ . \aodv{\ip}{\sn}{\rt}{\rreqs}{\queues}
		\ELSIF[no valid route via broken link]{$\dests=\emptyset$}
			\aodvL{\ip}{\sn}{\rt}{\rreqs}{\queues}
		\ENDIFii
		\ENDPAR
}\end{algorithmic}
\end{description}
All invariants and statements of Sections~\ref{sec:invariants} and~\ref{sec:interpretation} remain valid; 
the necessary proof adaptations are marginal and straightforward.

\subsection{Forwarding the Route Request}\label{ssec:forwardRREQ}
\newcommand{\alreadyrep}{\keyw{handled}}
\newcommand{\true}{\keyw{true}}
\newcommand{\false}{\keyw{false}}
\newcommand{\mrreq}[9]{\rreqID(\ensuremath{\nosp{#1}\comma\nosp{#2}\comma\nosp{#3}\comma\nosp{#4}\comma\nosp{#5}\comma\nosp{#6}\comma\nosp{#7}\comma\nosp{#8}\comma\nosp{#9}})}
\newcommand{\mrreqPL}[9]{\mrreqPaux{\RREQ(\ensuremath{\nosp{#1}\comma\nosp{#2}\comma\nosp{#3}\comma\nosp{#4}\comma\nosp{#5}\comma\nosp{#6}\comma\nosp{#7}\comma\nosp{#8}\comma\nosp{#9}}}}
\newcommand{\mrreqPaux}[6]{\ensuremath{#1\,\comma\,\nosp{#2}\comma\nosp{#3}\comma\nosp{#4}\comma\nosp{#5}\comma\nosp{#6}})}

In AODV's route discovery process, a destination node (or an
intermediate node with an active route to the destination) will
generate a RREP message in response to a received RREQ message. 
The RREQ message is then dropped and not forwarded. This termination
of the route discovery process at the destination can lead to
other nodes inadvertently creating non-optimal routes to the source
node~\cite{MK10}, where route optimality is defined in terms of a
metric, for example hop count. In~\cite{MK10} it is shown
that during the route discovery process in AODV, the only nodes that
generally discover optimal routes to the source and
destination nodes are those lying on the selected route between
the source node and the destination node (or the intermediate node) generating the reply.
All other network nodes receiving the RREQ message (in particular those located
``downstream'' of the destination node) may inadvertently be
creating non-optimal routes to the source node due to the unavoidable
receipt of RREQ messages over other routes. These ``poorly selected
paths have significantly higher routing-metric costs and their
duration can extend to minute time scales''~\cite{MK10}. 

We illustrate this by the example in \Fig{non-optimal-route}.
There, node $s$ wants to find a route to node $d$. It generates and
 broadcasts a RREQ message that is received by its neighbour nodes $d$
 and~$b$ (\Fig{non-optimal-route}(a)). Since node $d$ is the destination, it responds with a
 RREP message; the received RREQ message is not
 forwarded. On the other hand, $b$ continues to forward its
 received RREQ message, which eventually arrives at $a$ (Part (b)). At node
 $a$, a routing table entry is created for the source $s$, with a hop count
 of six. This is clearly not optimal, as $a$ is only two hops away from $s$.
 Due to the discarding of the RREQ message at node $d$, node $a$ is
 prevented from discovering its optimal route to $s$, via node $d$.
 In a next step, the RREQ message would also reach $d$ via $a$, but this message is then silently ignored by $d$.

\begin{exampleFig}{Non-optimal route selection}{fig:non-optimal-route}
\FigLine[0]%
  {$s$  broadcasts a new RREQ message destined to $d$.}{fig/ex_bad_route2_2}{}
  {The request floods the network;\\it reaches $a$, which establishes a non-optimal route.}{fig/ex_bad_route2_3}{}
  \end{exampleFig}  
  
A possible modification to solve this problem is to allow the
destination node to continue to forward the RREQ message. This will
then enable node $a$ in \Fig{non-optimal-route} to discover its
optimal route to~$s$. A route request is only stopped if it has been handled before. 
The forwarded RREQ message
from the destination node  needs to be modified to include a Boolean flag {\alreadyrep}\footnote{The
AODV RFC provides a field {\tt Reserved} as part of a RREQ message~\cite[Sect. 5.1]{rfc3561}, which is more or less designed to cater for such extensions.} that indicates
a RREP message has already been generated and sent in response to the
former message. In case the flag is set to {\true}, it prevents other nodes (with valid route to
the destination) from sending a RREP message in response to their
reception of the forwarded RREQ message.

The entire specification of this variant differs only in 
eight lines from the original.
\Pro{aodv} needs only slight adaptations.
First the newly introduced flag needs to be introduced in Lines~\ref{aodv:line8} and \ref{aodv:line11};
these lines now read 
$\cond{\msg = \mrreq{\hops}{\rreqid}{\dip}{\dsn}{\dsk}{\oip}{\osn}{\sip}{\,\highlight{\alreadyrep}}}$ and
$\mrreqPL{\hops}{\rreqid}{\dip}{\dsn}{\dsk}{\oip}{\osn}{\sip}{\,\highlight{\alreadyrep}}{\ip}{\sn}{\rt}{\rreqs}{\queues}$, respectively.
The broadcast in Line~\ref{aodv:line40} needs also be equipped with the flag. Since the route request is initiated, the 
flag is set to $\false$:
\[
\broadcastP{\mrreq{\hops+1}{\rreqid}{\dip}{\dsn}{\dsk}{\oip}{\osn}{\ip}{\highlight{\false}}}
\]
All other changes happen in the process {\RREQ}.
The new process {\RREQ} is given in Process \ref{pro:rreqmodified}.
  \algsetup{linenodelimiter=.,linenosize=\tiny}
  \begin{algorithm}[H]
    {\footnotesize
      \caption{The modified RREQ handling}
      \label{pro:rreqmodified}
      \begin{algorithmic}[1]
\DEFPROCESS{\RREQ}{\hops\comma\rreqid\comma\dip\comma\dsn\comma\dsk\comma\oip\comma\osn\comma\sip\comma\highlight{\alreadyrep}\,\comma\,\ip\comma\sn\comma\rt\comma\rreqs\comma\queues}
	\IF[the RREQ has been received previously]{$(\oip\comma\rreqid)\in\rreqs$}																						\label{rreqmod:line1}
		\aodvL{\ip}{\sn}{\rt}{\rreqs}{\queues} \COM{silently ignore RREQ, i.e. do nothing}																				\label{rreqmod:line2}
	\ELSIF[the RREQ is new to this node]{$(\oip\comma\rreqid)\not\in\rreqs$}																							\label{rreqmod:line3}
		\UPD{\rt:=\upd{\rt}{(\oip\comma\osn\comma\kno\comma\val\comma\hops+1\comma\sip\comma\emptyset)}}
		\COMMENT{update the route to \oip\ in \rt}																																		\label{rreqmod:line4}
		\UPD{\rreqs:=\rreqs\cup\{(\oip\comma\rreqid)\}}		\COMMENT{update \rreqs\ by adding $(\oip\comma\rreqid)$}								\label{rreqmod:line5}
		\PAR																																																	\label{rreqmod:line6}
			\IF[the request has not yet been handled]{$\highlight{\alreadyrep=\false}$}																					\label{rreqmod:line7}
			\PAR																																																\label{rreqmod:line8}
				\IF[this node is the destination node]{$\dip=\ip$}																														\label{rreqmod:line9}
					\UPD{\sn:=\max(\sn,\dsn)}	\COMMENT{update the sqn of \ip}																								\label{rreqmod:line10}
					\COMLINE{unicast a RREP towards \oip\ of the RREQ \highlight{ and forward the request}}													\label{rreqmod:line11}
					\STARTPRIO
						\unicast{\nhop{\rt}{\oip}}{{\rrep{$0$}{\dip}{\sn}{\oip}{\ip}}}\ . 																								\label{rreqmod:line12}
						\color{purple}\broadcast{\mrreq{\hops+1}{\rreqid}{\dip}{\dsn}{\dsk}{\oip}{\osn}{\ip}{\true}}\ .\color{black}  							\label{rreqmod:line13}
						\aodvL{\ip}{\sn}{\rt}{\rreqs}{\queues}																																	\label{rreqmod:line14}
					\PRIO
						\COMspec{If the transmission is unsuccessful, a RERR message is generated}																\label{rreqmod:line15}
						\UPD{\dests:=\{(\rip,\inc{\sqn{\rt}{\rip}})\,|\,\rip\in\akD{\rt}\ans \nhop{\rt}{\rip}=\nhop{\rt}{\oip}\}}												\label{rreqmod:line16}
						\UPD{\rt:=\inv{\rt}{\dests}}																																					\label{rreqmod:line17}			
						\UPD{\queues:=\setrrf{\queues}{\dests}}																															\label{rreqmod:line18}
						\UPD{\pre:=\bigcup\{\precs{\rt}{\rip}\,|\,(\rip,*)\in\dests\}}																										\label{rreqmod:line19}
						\UPD{\dests:=\{(\rip,\rsn)\,|\,(\rip,\rsn)\in\dests\ans \precs{\rt}{\rip}\not=\emptyset\}}																\label{rreqmod:line20}
						\groupcast{\pre}{\rerr{\dests}{\ip}}\ . \aodv{\ip}{\sn}{\rt}{\rreqs}{\queues}																				\label{rreqmod:line21}
					\ENDPRIO																			
				\ELSIF[this node is not the destination node]{$\dip\not=\ip$}																										\label{rreqmod:line22}
					\PAR																																														\label{rreqmod:line23}
						\IF[fresh enough valid route to \dip]{$\!\!\dip\mathbin\in\akD{\rt} \wedge \dsn \mathbin\leq  \sqn{\rt}{\!\dip} \wedge\sqnf{\rt}{\!\dip}\mathbin=\kno\!\!$}		\label{rreqmod:line24}
						\COMLINE{update \rt\ by adding precursors}																														\label{rreqmod:line25}
						\UPD{\rt := \addprecrt{\rt}{\dip}{\{\sip\}}}																																\label{rreqmod:line26}
						\UPD{\rt := \addprecrt{\rt}{\oip}{\{\nhop{\rt}{\dip}\}}}																												\label{rreqmod:line27}
						\COMLINE{unicast a RREP towards the \oip\ of the RREQ  \highlight{and forward the request}}										\label{rreqmod:line28}
						\STARTPRIO
							\unicast{\nhop{\rt}{\oip}}{\rrep{\dhops{\rt}{\dip}}{\dip}{\sqn{\rt}{\dip}}{\oip}{\ip}}\ .\\																\label{rreqmod:line29}
							\color{purple}\broadcast{\mrreq{\hops+1}{\rreqid}{\dip}{\dsn}{\dsk}{\oip}{\osn}{\ip}{\true}}\ . \color{black}						\label{rreqmod:line30}
							\aodvL{\ip}{\sn}{\rt}{\rreqs}{\queues}																																\label{rreqmod:line31}
						\PRIO
							\COMspec{If the transmission is unsuccessful, a RERR message is generated}															\label{rreqmod:line32}
							\UPD{\dests:=\{(\rip,\inc{\sqn{\rt}{\rip}})\,|\,\rip\in\akD{\rt}\ans \nhop{\rt}{\rip}=\nhop{\rt}{\oip}\}}											\label{rreqmod:line33}
							\UPD{\rt:=\inv{\rt}{\dests}}																																				\label{rreqmod:line34}			
							\UPD{\queues:=\setrrf{\queues}{\dests}}																														\label{rreqmod:line35}		
							\UPD{\pre:=\bigcup\{\precs{\rt}{\rip}\,|\,(\rip,*)\in\dests\}}																									\label{rreqmod:line36}
							\UPD{\dests:=\{(\rip,\rsn)\,|\,(\rip,\rsn)\in\dests\ans \precs{\rt}{\rip}\not=\emptyset\}}															\label{rreqmod:line37}
							\groupcast{\pre}{\rerr{\dests}{\ip}}\ . 																																\label{rreqmod:line38}
							\aodv{\ip}{\sn}{\rt}{\rreqs}{\queues}																																\label{rreqmod:line38a}
						\ENDPRIO
						\ELSIF[no fresh enough valid route]{$\dip\mathbin{\not\in}\akD{\rt} \vee \sqn{\rt}{\!\dip} <  \dsn \vee\sqnf{\rt}{\!\dip}\mathbin=\unkno$}					\label{rreqmod:line39}
							\COMLINE{no further update of \rt}																																\label{rreqmod:line40}
							\broadcast{\mrreq{\hops+1}{\rreqid}{\dip}{\max(\sqn{\rt}{\dip}\comma\dsn)}{\dsk}{\oip}{\osn}{\ip}{{\color{purple}\false}}}\ .							\label{rreqmod:line41}
							\aodvL{\ip}{\sn}{\rt}{\rreqs}{\queues}																																\label{rreqmod:line41a}
						\ENDIFii
					\ENDPAR																																												\label{rreqmod:line42}
				\ENDIFii
			\ENDPAR																																														\label{rreqmod:line43}
			\ELSIF[the request has been answered before]{$\highlight{\alreadyrep=\true}$}																			\label{rreqmod:line44}
				\COMLINE{the request is just forwarded (the RREQ was not handled before)}																			\label{rreqmod:line45}
				\color{purple}\broadcast{\mrreq{\hops+1}{\rreqid}{\dip}{\dsn}{\dsk}{\oip}{\osn}{\ip}{\true}}\ .														\label{rreqmod:line46}
				\aodv{\ip}{\sn}{\rt}{\rreqs}{\queues}\color{black}																															\label{rreqmod:line46a}
			\ENDIFii
		\ENDPAR
	\ENDIFii

	\end{algorithmic}
    }
  \end{algorithm}

The changes introduce the new flag and a case distinction based on that (Lines~\ref{rreqmod:line7} and \ref{rreqmod:line44}), 
as well as three new broadcasts (Lines~\ref{rreqmod:line13}, \ref{rreqmod:line30} and \ref{rreqmod:line46}).
For example, after initiating a route reply at the destination (Pro.~\ref{pro:rreqmodified}, Line \ref{rreqmod:line12}),
the route request message is forwarded.
In case the \textbf{unicast} of Line \ref{rreqmod:line12} is not successful (Lines~\ref{rreqmod:line15}--\ref{rreqmod:line21}),
no forwarding is necessary, since it does not make sense to establish a link back to the originator of the RREQ message---the intermediate node just detected that this link is broken.

The proofs of Sections~\ref{sec:invariants} and \ref{sec:interpretation} are still valid, 
but need, as usual, some modifications. 
The newly introduced flag does not have any effect on the proofs---except that some line numbers change
and the additional flag is added to all calls of \keyw{rreq}.
The only real modification is that in Propositions~\ref{prop:starcastNew}(a),~\ref{prop:msgsendingii}(a),
\ref{prop:msgsending}(a) and \ref{prop:starcast}(a)--(c), as well as in \Thm{route correctness}(b),
three new \textbf{broadcast}-actions need to be examined. However, these cases are identical to
the case Pro~\ref{pro:rreq}, Line~\ref{rreq:line34} (corresponding to Line~\ref{rreqmod:line41} of Pro.~\ref{pro:rreqmodified}).

If this modification is applied to the example presented earlier, $a$ now establishes an optimal route.
The example is illustrated in detail in \Fig{good_routes}.
As before node $s$ issues a route discovery in \Fig{good_routes}(a). 
The RREQ message is received by $b$ and $d$. 
Following the standard behaviour of AODV, node $b$ forwards the RREQ and the destination $d$ unicasts a RREP message back to $s$. 
Additionally, node $d$ also broadcasts the modified request  RREQ${}_{m}$ (Part (b)).
It is received by $s$ and $a$. Since $s$ initiated the request, it silently ignores the modified RREQ message;
node $a$ establishes an optimal route. 
Subsequently, both the original and the modified request are sent through the network (\Fig{good_routes}(c)).
The flooding is terminated as soon as every node has handled one of
the RREQ messages---both have the same route request identifier.
In the example, the node in the lower
right corner receives both RREQ messages, forwards one of them (here RREQ) and silently ignores the other.
\vspace{-2pt}
\begin{exampleFig}{Forwarding route requests}{fig:good_routes}
\FigLine[0]
  {$s$ broadcasts a new RREQ destined to $d$.}{fig/ex_good_route2_2}{}
  {The request reaches $d$, which unicasts a RREP back\\ and forwards a modified RREQ.}{fig/ex_good_route2_3}{}
\FigLine[0]
  {The (modified) RREQ floods the network\\(until a node receives both messages).}{fig/ex_good_route2_4}{}
  {The RREQ terminates as soon as all nodes have\\handled RREQ or RREQ${}_{m}$.}{fig/ex_good_route2_5}{}
\end{exampleFig}
\vspace{-1pt}

An intermediate node answering the route request on behalf of the destination will also forward the RREQ message.
The destination will receive the modified message and establish a route to the originator.
By this, a bidirectional route 
between the source and the destination is established. 

This finishes our list of improvements. 
All presented improvements are ``orthogonal'', i.e., they can be  combined without problems; the properties 
of Sections~\ref{sec:invariants} and \ref{sec:interpretation} remain valid.
In case new shortcomings are found, our specification
  as well as the proofs can easily be changed, as illustrated in this section.

\section{Related Work}\label{sec:related work}
\newcommand{\denote}[1]{\llbracket#1\rrbracket} 
\newcommand{\uppaal}{\textsc{Uppaal}}
\newcommand{\aodvii}{AODV\@{v2}\xspace}

\subsection{Process Algebras for Wireless Mesh Networks}

Several process algebras modelling broadcast communication have been
proposed before:
the Calculus of Broadcasting Systems (CBS) \cite{CBS91,CBS},
the $b\pi$-calculus \cite{bpi},
CBS\# \cite{NH06},
the Calculus of Wireless Systems (CWS) \cite{CWS},
the Calculus of Mobile Ad Hoc Networks (CMAN) \cite{CMAN},
the Calculus for Mobile Ad Hoc Networks (CMN) \cite{CMN},
the $\omega$-calculus \cite{SRS10},
rooted branching process theory (RBPT) \cite{RBPT},
$bA\pi$ \cite{bApi}
and the broadcast psi-calculi \cite{BHJRVPP11}.
The latter eight of these were specifically designed to model MANETs.
However, we believe that none of these process calculi provides all  features needed to fully model
routing protocols such as AODV\@, namely data handling, (conditional) unicast and (local) broadcast. 
Moreover, all these process algebras lack the feature of guaranteed receipt of message. 
Due to this, it is not possible to analyse properties such as route discovery.
We will elaborate on this in the following.

\paragraph{Modelling Broadcast Communication}
All these languages, as well as ours, feature a form of broadcast
communication between nodes in a network, in which a single message
emitted by one node can be received by multiple other nodes. In terms
of operational semantics, this is captured by rules like
$$\frac{M \ar{\broadcastP{m}} M' \quad N \ar{\receive{m}} N'}
  {\rule[13pt]{0pt}{1pt} M \| N \ar{\broadcastP{m}} M' \| N'}
\qquad\qquad
  \frac{ M \ar{\receive{m}} M' \quad N \ar{\broadcastP{m}} N'}
  {\rule[13pt]{0pt}{1pt} M \| N \ar{\broadcastP{m}} M' \| N'}$$
that stem from \cite{CBS91} and can be found in the operational semantics of each of these
languages, except for $bA\pi$.
In such a rule the broadcast action in the conclusion is simply
inherited from the broadcasting argument of the parallel composition,
so that it remains available for the parallel composition with another receiver.
In order to guarantee associativity of the parallel composition, i.e.\
$$\big(\broadcastP{m}.P \| \receive{m}.Q\big) \| \receive{m}.R
 = \broadcastP{m}.P \| \big(\receive{m}.Q \| \receive{m}.R\big)$$
one also needs a rule like
$$\frac{ M \ar{\receive{m}} M' \quad N \ar{\receive{m}} N'}
  {\rule[13pt]{0pt}{1pt} M \| N \ar{\receive{m}} M' \| N'}\;\;.\footnotemark$$
\footnotetext{CMN lacks such a rule, resulting in a non-associative parallel composition.
This in turn leads to a failure of \cite[Lemma 3.2]{CMN}, saying that structural congruence
  ``respects transitions'', i.e.\ is a strong bisimulation. This mistake is propagated in \cite{GR13}.
The $\omega$-calculus lacks this rule as well, but there associativity
of the parallel composition is enforced by closing the transition relation under structural congruence.
Nevertheless, the same problem returns in the form of a failure of \cite[Theorem 9]{SRS10}, stating that strong
bisimilarity is a congruence. Namely, ${\bf r}(x).{\it nil}:\{g\} \mid {\bf r}(x).{\it nil}:\{g\}$, the parallel
composition of two connected nodes, each doing a receive action, is
bisimilar to ${\bf r}(x).{\bf r}(x).{\it nil}:\{g\}$, a single node doing two receive actions;
yet in the context ${\bf b}(u).{\it nil}:\{g\} \mid \_\!\_\,$, involving a
fully connected node broadcasting a value $u$, only the former can do a broadcast transition to a process
that can do no further receive actions.}

\paragraph{Lossy Broadcast versus Enforced Synchronisation without Blocking}

The languages CMAN, CMN, RBPT, the $\omega$-calculus, $bA\pi$ and the broadcast psi-calculi
\index{lossy broadcast}
model \emph{lossy} communication, which allows, as a nondeterministic
possibility, any node to miss a message of another node, even when the
two nodes are within transmission range. The corresponding
operational rules are
$$\frac{M \ar{\broadcastP{m}} M'}
  {\rule[13pt]{0pt}{1pt} M \| N \ar{\broadcastP{m}} M' \| N}
\qquad\qquad
  \frac{N \ar{\broadcastP{m}} N'}
  {\rule[13pt]{0pt}{1pt} M \| N \ar{\broadcastP{m}} M \| N'}
\;.$$
In such a language it impossible to formulate valid properties of
modelled protocols like ``if there is a path from \dval{ip} to
\dval{dip}, and the topology does not change, then a packet for
\dval{dip} submitted at \dval{ip} will eventually be delivered at
\dval{dip}'' (cf.~\SSect{packet_delivery}).  Namely, there is never
a guarantee that any message arrives.

In the operational semantics, the only alternative to the lossy rules
above appears to be enforced synchronisation of a broadcast action of
one component in a parallel composition with some (in)activity of the other.
This approach is followed in CBS, $b\pi$, CBS\# and CWS, as well as in AWN\@.
In CBS, $b\pi$, CBS\# and CWS, and in the optional augmentation of \awn
presented in \SSect{non-blocking}, any node within transmission range
\emph{must} receive a message $m$ sent to it, provided the node
is \emph{ready} to receive it, i.e., in a state that admits a
transition $\receive{m}$. This proviso makes all these calculi
\phrase{non-blocking}, meaning that no sender can be delayed in
transmitting a message simply because one of the potential recipients
is not ready to receive it.
The default version of \awn (\SSect{networks}) lacks this proviso and
hence does allow blocking. However,
in applications of our language we model nodes in such a way that any
message can be received at any time (cf.~\SSect{message_queue}). Nodes
with this property are called \phrase{input enabled}, a concept introduced in the work on
\emph{IO-automata} \cite{LT89}. For such applications our models are non-blocking.

\index{discard actions}%
In CBS \cite{CBS91}, actions $\textbf{discard}(m)$ are used to model situations
where a process cannot receive a message $m$. The definitions are such
that $N \ar{\textbf{discard}(m)} N'$ only if $N'=N$---so the discard actions
do not correspond with any state-change---and $N
\ar{\textbf{discard}(m)}N'$ if an only if $N \nar{\textbf{receive}(m)}$.
Now the rules for broadcast are augmented by
$$\frac{M \ar{\receive{m}} M' ~~ N \ar{\textbf{discard}(m)} N'}
  {\rule[13pt]{0pt}{1pt} M \| N \ar{\receive{m}} M' \| N'}
\quad~
  \frac{M \ar{\textbf{discard}(m)} M' ~~ N \ar{\textbf{discard}(m)} N'}
  {\rule[13pt]{0pt}{1pt} M \| N \ar{\textbf{discard}(m)} M' \| N'}
\quad~
  \frac{ M \ar{\textbf{discard}(m)} M' ~~ N \ar{\receive{m}} N'}
  {\rule[13pt]{0pt}{1pt} M \| N \ar{\receive{m}} M' \| N'}$$
$$\frac{M \ar{\broadcastP{m}} M' \quad N \ar{\textbf{discard}(m)} N'}
  {\rule[13pt]{0pt}{1pt} M \| N \ar{\broadcastP{m}} M' \| N'}
\qquad\qquad
  \frac{ M \ar{\textbf{discard}(m)} M' \quad N \ar{\broadcastP{m}} N'}
  {\rule[13pt]{0pt}{1pt} M \| N \ar{\broadcastP{m}} M' \| N'}\;.$$
This way, in a parallel composition $M\|N$, a broadcast action of one
component can never be blocked by the other component; it synchronises
either with a receive or a discard of the other component, depending
on whether the other component is ready to receive a message or not.
In $b\pi$ and CBS\# the same approach is followed,
except that in CBS\# messages are annotated with their sender
(i.e.\ read $m$ as a sender-message pair in all rules above). This
way, one can say that a node $N$ (in a certain state) can receive
a message from one sender, but not from another.

At the expense of the use of negative premises, it is possible
to eliminate the discard action, and replace a premise
$N \ar{\textbf{discard}(m)} N'$ by $N \nar{\receive{m}}$
(cf.~Section~\ref{ssec:non-blocking}).
Another variant of the same idea, applied in \cite{CBS} and CWS \cite{CWS},
is to simply replace the discard transitions
$N \ar{\textbf{discard}(m)} N $ by \mbox{$N \ar{\textbf{receive}(m)} N$}.
Thus, whenever a node $N$ is unable to receive a particular message from a
particular other node, its operational semantics introduces a discard
transition \mbox{$N \ar{\textbf{receive}(m)} N$} that
essentially allows the message to be received and completely ignored.

\paragraph{Local Broadcast with Arbitrary Dynamic Topologies}

\index{global broadcast}%
CBS models global broadcast communication, where all processes (or nodes) are able to receive
any broadcast message. The language $b\pi$ allows processes to join
\emph{groups}, and receive all messages sent by members of that group.
In order to join a group a process needs to have knowledge of the name
of this group. This appears less suitable for the specification of
wireless mesh networks, where nodes may receive messages from unknown
other nodes as soon as they move into transmission range.
The other languages allow arbitrary network topologies, and feature a
\index{local broadcast}%
local broadcast, which can be received only by nodes within
transmission range of the sender. CWS deals with static topologies,
whereas in CBS\#, CMAN, CMN, the $\omega$-calculus, RBPT, $bA\pi$ and
the broadcast psi-calculus, as in our approach, the topology is subject to arbitrary changes.

\paragraph{Guaranteed Receipt of Messages Broadcast within Transmission Range}

The syntax of CBS\# and CWS does not permit the construction of
meaningful nodes that are always ready to receive a message.  Hence
our model is the first that assumes that any message \emph{is} received by a
potential recipient within range.  It is this feature that allows us
to evaluate whether a protocol satisfies the \emph{packet delivery}
property of \SSect{packet_delivery}.  \emph{Any} routing protocol
formalised in any of the other formalisms would automatically fail to
satisfy such a property.

\paragraph{Modelling Connectivity}

To model connectivity of nodes in the current topology, our node expressions have the form
$\dval{ip}:P:R$, where $P$ is a process running on the node, $\dval{ip}$ is the node's address, and
$R$ is the current transmission range, given as the set of addresses of nodes that can receive
messages sent by this node. Changes in the transmission range occur through
{\bf connect} and {\bf disconnect} actions, which can occur at any time (cf.~\SSect{networks}).
This follows CMAN \cite{CMAN}, where our $l\mathord{:}p\mathord{:}\sigma$ is denoted as $\lfloor p \rfloor^\sigma_l$,
with $l$ being a \emph{location}, which plays the role of the node's address. In CWS nodes have
the form $n[P]^c_{l,r}$, where $n$ is the node address (our $\dval{ip}$), $c$ denotes the broadcast
channel (e.g.\ a frequency) to which the node is tuned, $l$ is the physical location and
$r\in \mathord{\rm I\mkern-2.5mu R}$ the \emph{radius} or transmission range of the node.
A global function $d$ is postulated that given two locations $l$ and $l'$ returns the distance $d(l,l')$ between them;
comparing this value with the radius of a node at $l$ determines whether broadcast messages from
that node reach a node at $l'$. In comparison with CWS our formalism could be said to use only one
possible channel. CMN uses the same syntax as CWS, except that the channel is replaced by a
\emph{mobility tag} $\mu$, telling whether the node is mobile or stationary.
In the latter case, the physical location $l$ of the node is subject to chance.
The $\omega$-calculus has node expressions of the form $P:G$, where $P$ is a process and $G$ the set
of \emph{groups} the node belongs to. Each group is a clique in the graph expressing the network
topology, and two nodes can communicate iff they belong to a common group.
Contrary to these approaches, in CBS\#, RBPT and the broadcast $\psi$-calculi node expressions do
not contain connectivity information. Instead, connectivity is modelled in the semantics only, by
labelling transitions with (information about) the topologies that support them.
In CBS\# node expressions have the form $n[P,S]$ where $n$ is the \emph{location} or identifier of
a node, $P$ a process and $S$ the node's memory, storing values that could have been received.
RBPT node expressions simply have the form $\denote{P}_l$, denoting a process $P$ at the location $l$.
Broadcast psi differs from the above calculi in that it makes no distinction between processes and
node (or network) expressions. Consequently, nodes are not equipped with an address and connectivity
cannot be expressed as a relation between nodes. Instead it is expressed as a relation between channel
expressions occurring in processes.\linebreak[3]
In $bA\pi$ nodes have the form $\lfloor p \rfloor_l$, as in CMAN, but without any connectivity information.
The operational semantics differs from those of the other calculi, in that a broadcast action
results in messages sitting as separate components in the parallel composition among the nodes.
Connectivities of the form $\{l\mapsto m\}$, saying that node $l$ can receive message send by node $m$,
also occur as separate entities in this parallel composition, and can react with messages to
guide them in appropriate directions.

\begin{table}
\newcommand{\lb}{lossy broadcast}
\newcommand{\es}{enforced synchr.}
\newcommand{\esg}{\begin{tabular}[t]{@{}l@{}}\es \\ with guar.\ receipt\end{tabular}}
\newcommand{\gb}{\multicolumn{4}{l|}{global broadcast}}
\newcommand{\gr}{\multicolumn{4}{l|}{subscription-based broadcast}}
\newcommand{\lbs}{local bc. & static topology}
\newcommand{\lbd}{local bc.& dynamic top.}
\newcommand{\syn}{node}
\newcommand{\netw}{network}
\newcommand{\sem}{op.\ sem.}
\newcommand{\sym}{symmetric}
\newcommand{\asym}{asymmetric}

\begin{center}
\begin{tabular}{|l@{~}r@{~~}r@{~}|l|l|l|l|l|l|}
\hline
\multicolumn{3}{|c|}{Process algebra} & Message loss & \multicolumn{2}{c}{Type of broadcast} &
\multicolumn{3}{|c|}{Connectivity model}
\\
\hline\hline
CBS   & \cite{CBS} & '91 & \es & \gb & \sym
\\
$b\pi$ & \cite{bpi} & '99 & \es & \gr & \sym
\\
\cline{5-8}
CBS\# & \cite{NH06} & '06 & \es & \lbd & $n[P,S]$ & \sem & \sym
\\
CWS   & \cite{CWS} & '06 & \es & \lbs & \plat{$n[P]^c_{l,r}$} & \syn & \sym
\\
CMAN  & \cite{CMAN} & '07 & \lb & \lbd & $\lfloor p \rfloor^\sigma_l$ & \syn & \sym
\\
CMN   & \cite{CMN} & '07 & \lb & \lbd & \plat{$n[P]^\mu_{l,r}$} & \syn & \sym
\\
$\omega$ & \cite{SRS10} & '07 & \lb & \lbd & $P:G$ & \syn & \sym
\\
RBPT & \cite{RBPT} & '08 & \lb & \lbd & $\denote{P}_l$ & \sem & \asym
\\
$bA\pi$  & \cite{bApi} & '09 & \lb & \lbd & $\lfloor p \rfloor_l$ & \netw & \asym
\\
$b\psi$ & \cite{BHJRVPP11} & '11 & \lb & \lbd & $P$ & \sem & \asym
\\
\awn & here & '11 & \esg & \lbd & $\dval{ip}\mathord:P\mathord:R$ & \syn & asym./sym.
\\
\hline
\end{tabular}
\end{center}
\caption{Process algebras modelling broadcast communication}
\label{tab:broadcastPAs}
\end{table}

The above comparison between the various formalisms is summarised in \Tab{broadcastPAs},
of which the last three columns are largely taken from \cite{RBPT}.
The sixth column tells whether connectivity information is stored in the syntax of node or network
expressions, or whether it appears in the structural operation semantics of the language only.
The last column indicates whether the formalism assumes the connectivity relation between nodes to
be symmetric. In this regard there are two versions of \awn; in \cite{ESOP12} the default version is
asymmetric, whereas here, in view of the application to AODV, we made the symmetric version the default.

\paragraph{Operational Semantics of Local Broadcast with Enforced Synchronisation}

Whereas CBS\# and CWS enrich receive actions of messages with their
senders---to indicate that a message can be received from one sender
but not from another, based on the topology---in our operational
semantics this administrative burden is shifted to the broadcast
actions---they are annotated with the range of possible receivers.
This enables us to model groupcast and unicast actions, which are not
treated in CBS\# and CWS, in the same way as broadcast actions.
However, the price to be payed for this convenience is that our
actions $\listen{m}$, which  are synchronisations of (non)receive
actions of multiple components, need to be annotated with the locations
of all these components. Moreover, this set of locations is
partitioned into the ones that are in and out of transmission range
of the message $m$. It does not appear possible to model our groupcast
in the style of CBS\# and CWS\@.

\paragraph{Conditional Unicast}

Our novel \phrase{conditional unicast}
operator chooses a continuation process dependent on whether the
message can be delivered.  This operator is essential for the correct
formalisation of AODV\@ and other network protocols.  In practice such an operator may be
implemented by means of an acknowledgement mechanism; however, this is typically
done at the link layer, from which the AODV specification
\cite{rfc3561}, and hence our formalism, abstracts.  One could
formalise a conditional unicast as a standard unicast in the
scope of a priority operator \cite{CLN01}; however, our operator
allows an operational semantics within the
de Simone format.
Of the other process algebras of \Tab{broadcastPAs}, only the $\omega$-calculus, $bA\pi$ and the
broadcast psi-calculi model unicast at all, next to broadcast; they do not have anything comparable to the
conditional unicast.

\paragraph{Data Structures}

Although our treatment of data structures follows the classical
approach of universal algebra, and is in the spirit of formalisms like
$\mu$CRL \cite{GP95}, we have not seen a process algebra that freely
mixes in imperative programming constructs like variable assignment.
Yet this helps to properly capture AODV and other routing protocols.
This mixture should make the syntax of {\awn} on the level of sequential processes
easy to read for anybody who has some experience in programming, thus making it easier to implement protocol specifications written in {\awn}.

\paragraph{Other Process Algebras for WMNs}

In \cite{GR13} CMN is extended with mechanisms for unicast and multicast/groupcast communication; the
paper focuses on power-consumption issues.
Process calculi in the same spirit as the ones above, but focusing on
security aspects and trust, appear in \cite{GHK09,MS13}.
Probabilistic and stochastic calculi for WMNs, based on similar design principles as the
process algebras discussed above, are proposed in \cite{SG10,GTMF11,GHMR11,LM11,SG12,BGHMR13,CH13,GMRHK13}.
An extended and improved version of CWS appears in \cite{LS10}.
Extensions of CWS with time appear in \cite{MBS11,MM14,LM11,CHM13,WL12}; these process algebras
focus on the MAC-layer rather than the network layer of the TCP/IP reference model.
In \cite{God09} a variant of CMAN is proposed that limits mobility. A variant of CMAN that incorporates another mobility model
appears in \cite{GN09}.
In \cite{GFM09} the process algebra RBPT is enriched with specifications of sets of topologies into
Computed Network Theory (CNT). This facilitates the equational axiomatisation of RBPT\@.
\index{encapsulation operator}%
\index{abstraction operator}%
In \cite{GFM10} RBPT and CNT are extended with encapsulation and abstraction operators;
a simple abstraction of AODV has been shown to be loop free in this
framework by means of equational reasoning \cite{GFM11}.

\subsection{Modelling, Verifying and Analysing AODV and Related Protocols}

Our complete formalisation of AODV, presented here, has grown from
elaborating a partial formalisation of AODV in \cite{SRS10}.
The features of our process algebra were
largely determined by what we needed to enable a complete and accurate
formalisation of this protocol. 
The same formalism has been used to model the Dynamic MANET On-demand (DYMO) Routing Protocol (also known as AODVv2)~\cite{EHWripe12}.
By this we did not only derive an unambiguous specification for draft-version 22 (as we did for the RFC of AODV);
we were also able to verify that some of the problems discovered for AODV have
been addressed and solved. However, we showed that other limitations still exist, e.g., the
establishment of non-optimal routes (cf.\ \SSect{forwardRREQ}).
We conjecture that \awn is also applicable to a wide range of other wireless protocols, such as 
the Dynamic Source Routing (DSR) protocol~\cite{rfc4728},
the Lightweight Underlay Network Ad-hoc Routing (LUNAR)
protocol~\cite{lunar,Tschudin04},
the Optimized Link State Routing (OSLR) protocol~\cite{rfc3626}
or the Better Approach To Mobile Adhoc Networking (B.A.T.M.A.N.)~\cite{batman}.
\index{timing}%
The specification and the correctness of  the latter three, however, rely heavily on timing aspects; hence
an {\awn}-extension with time appears necessary (see also \Sect{conclusion}).

\paragraph{Test-bed Experiments and Simulation}
While process algebra can be used to formally model and verify the correctness of network routing protocols, test-bed experiments and simulations are complementary tools that can be used to quantitatively evaluate the performance of the protocols. While test-bed experiments are able to capture the full complex characteristics of the wireless medium and its effect on the network routing protocols \cite{MaltzEtAl01,PirzadaEtAl09}, network simulators~\cite{NS2,QUALNET} offer the ease and flexibility of evaluating and comparing the performance of different routing protocols in a large-scale network of hundreds of nodes, coupled with the added advantage of being able to repeat and reproduce the experiments \cite{DasEtAl00,PerkinsBDM01,JacquetEtAl02}.

\paragraph{Loop Freedom}

Loop freedom is a crucial property of network protocols, commonly claimed to hold for AODV \cite{rfc3561}.
Merlin and Segall \cite{MS79} were amongst the first to use sequence numbers to guarantee loop freedom of a routing protocol.
We have shown that several \emph{interpretations} of AODV---consistent ways to revolve
the ambiguities in the RFC---fail to be loop free, while proving loop freedom of others.

A preliminary draft of AODV has been shown to be not loop free by Bhargavan et al.\
in~\cite{BOG02}. Their counterexamples to loop freedom have to do with timing issues:
the premature deletion of invalid routes, and a too quick restart of a node after a reboot.
Since then, AODV has changed to such a degree that these examples do not apply to
the current version \cite{rfc3561}. However, similar examples, claimed to apply to the current
version, are reported in \cite{Garcia04, Rangarajan05}. All these papers propose repairs that avoid
these loops through better timing policies.
In contrast, the routing loops documented in \cite{AODVloop} as well as in \Sect{interpretation} of
this paper are time-independent.

Previous attempts to prove loop freedom of AODV have been reported
in \cite{AODV99,BOG02,ZYZW09}, but none of these proofs are complete and valid for the current version of AODV \cite{rfc3561}:
\begin{itemize}
\item 
The proof sketch given in~\cite{AODV99} uses the fact that when a loop in a route to a destination
$Z$ is created, all nodes $X_i$ on that loop must have route entries for destination $Z$
with the same destination sequence number. ``{\sl Furthermore, because the destination sequence
numbers are all the same, the next hop information must have been derived at every node
$X_i$ from the same RREP transmitted by the destination $Z$\/}''~\cite[Page 11]{AODV99}. The latter is not true at all:
some of the information could have been derived from RREQ messages, or from a RREP message transmitted
by an intermediate node that has a route to $Z$. More importantly, the nodes on the loop
may have acquired their information on a route to $Z$ from different RREP or RREQ
messages, that all carried the same sequence number. This is illustrated by our
loop created in \Fig{loopsfromselfentries} (\SSSect{interpretation_invalidate}).

\item Based on an analysis of an early draft of AODV\footnote{Draft version 2 is analysed, dated November 1998; 
the RFC can be seen as version 14, dated July 2001.}~\cite{BOG02} suggests three improvements. The \
modified version is then proved to be loop free, using the following invariant
(written in our notation):
 \[\begin{array}{cl}
\multicolumn{2}{l}{\mbox{if $\dval{nhip} = \nhp{\dval{ip}}$, then }}\\
\mbox{(1)}& \sq{\dval{ip}} \leq \sq{\dval{nhip}}, \mbox{ and}\\ 
\mbox{(2)}& \sq{\dval{ip}} = \sq{\dval{nhip}} \ims \dhp{\dval{ip}} < \dhp{\dval{nhip}}\ . 
\end{array}
\]
This invariant does not hold for this  modified version of AODV nor for the
  current version, documented in the RFC\@. It can happen that in a state $N$ where
  $\sq{\dval{ip}} = \sq{\dval{nhip}}$, node $\dval{ip}$ notices that the link to $\dval{nhip}$ is broken.
  Consequently, \dval{ip} invalidates its route to \dval{dip}, which has \dval{nhip} as its next hop.
  According to recommendation (\textbf{A1}) of \cite[Page 561]{BOG02}), node $\dval{ip}$
  increments its sequence number for the (invalid) route to \dval{dip}, resulting in a state $N'$
  for which $\keyw{sqn}_{N'}^{\dval{ip}}(\dval{dip}) > \keyw{sqn}_{N'}^{\dval{nhip}}(\dval{dip})$,
  thereby violating the invariant.

  Note that the invariant of \cite{BOG02} does not restrict itself to the case that the routing
  table entry for \dval{dip} maintained by \dval{ip} is \emph{valid}. Adapting the invariant with such
  a requirement would give rise to a valid invariant, but one whose verification poses some
  problems, at least for the current version of AODV\@. These problems led us, in this paper, to use
  \emph{net sequence numbers\/} instead (cf.\ \SSect{quality}).

Recommendation (\textbf{A1}) is assumed to be in
  effect for the (improved) version of AODV analysed in \cite{BOG02}, although it was not in effect for
  the draft of AODV existing at the time. Since then, recommendation (\textbf{A1}) has been
  incorporated in the RFC\@. Looking at the proofs in \cite{BOG02}, it turns out that Lemma~20(1) of
  \cite{BOG02} is invalid. This failure is surprising, given that according to \cite{BOG02} Lemma~20
  is automatically verified by SPIN\@. A possible explanation might be that this lemma \emph{is} obviously
  valid for the version of AODV prior to the recommendations of \cite{BOG02}.
\end{itemize}

\paragraph{Model Checking}

Bhargavan et al.\ \cite{BOG02} not only found problems in an early draft of AODV, they were
also among the first to apply model checking techniques to AODV, thereby demonstrating the
feasibility and value of automated verification of routing protocols.
For their studies they use the model checker SPIN.

Musuvathi et al.\ \cite{MPCED02} introduced the model checker CMC
primarily to search for coding errors in implementations of protocols
written in C\@. They use AODV (draft version 10) as an example and were able to discover
$34$ distinct errors in three different implementations: mad-hoc (version 1), Kernel AODV (version
1.5) and AODV-UU (version 0.5).\footnote{For our analysis in \SSect{implementations}, we use version 0.9.5 of AODV-UU and version 2.2.2. of Kernel AODV\@. 
 We did not analyse mad-hoc since it is no longer actively supported.}
They also found a problem with the specification itself: they discovered that routing loops can be created when 
sequence numbers are just copied from an incoming RERR message, without checking the value. 
We have discovered the same problem (see \hyperlink{sss923b}{Ambiguity 8} in \SSSect{interpretation_invalidate})
and proposed the same solution as they do, namely 
introducing a check prior to invalidating routes (in our specification the check is \highlight{$\sqn{\rt}{\rip}
< \rsn$} in Line~\ref{rerr:line2} of \Pro{rerr}).
However, the routing loops found in \cite{MPCED02} crucially depend on the use of an unordered
message queue, in which messages can overtake each other after being sent. Our loop, on the other hand, manifests
itself even when using FIFO queues, as specified in the RFC\@.
Although \cite{MPCED02} testifies that both the bug and the fix where accepted by the protocol authors,
the proposed solution is not incorporated in the current standard~\cite{rfc3561}. 

Chiyangwa and Kwiatkowska~\cite{CK05} use the timing features of the model checker {\uppaal} to
study the relationship between the timing parameters and the performance of route discovery in AODV,
and find some route discovery failures.

Using the model checkers SPIN and \uppaal,
\cite{WPP04} demonstrates that the ad-hoc protocol LUNAR satisfies a strong variant of the
packet delivery property for a number of routing scenarios.

All this related work show that model checking can be used as a diagnostic tool for MANETs and WMNs.
Although model checking generally lacks the ability to verify protocols for an arbitrary and changing topology, it can be 
efficiently used to check specific scenarios (topologies) and to reveal problems in the specification in an early stage of 
protocol development; even before anybody starts to 
verify interesting properties by pen-and-paper proofs or with support of interactive theorem provers.

We believe that model checking as a diagnostic tool can complement the process-algebraic approach presented in this paper.
Having the ability of model checking  specifications written in {\awn} will allow the confirmation
and detailed diagnostics of suspected errors which arise during modelling. The availability of an
executable model will become especially useful in the evaluation of proposed improvements.
A first step to this complementation was taken in~\cite{FGHMPT11} and further elaborated in~\cite{TACAS12}.
In~\cite{FGHMPT11}, we generated a (``time-free'') {\uppaal} model of AODV from our {\awn} specification, confirmed some of the problems discovered
by Chiyangwa and Kwiatkowska~\cite{CK05}, and show their independence of time.
In~\cite{TACAS12} we continued the analysis of AODV by model-checking techniques by an exhaustive exploration of 
AODV's behaviour in all network topologies up to 5 nodes. We were able to automatically locate problematic and undesirable behaviours.
In that paper, we moreover sketched possible modifications of AODV, which also were subjected to rigorous analysis by means of model checking.
In these experiments we created an environment
in which we can test a range of different topologies in a systematic manner. This will allow us
to do a fast comparison between standard protocols (e.g.\ given by RFCs) and proposed variations in
contexts known to be problematic.

\paragraph{Statistical Model Checking}
Unfortunately, current state-of-the art (exhaustive) model checkers are unable to handle protocols of the complexity needed for WMN routing in realistic settings: network size (usually dozens, sometimes even hundreds of nodes) and topology changes yield an explosion in the state space.
Another limitation of (exhaustive) model checking is that a quantitative analysis is often not possible: finding a shortcoming in a protocol is great but does not show how often the shortcoming actually occurs. Statistical model checking (SMC)~\cite{Younes04,SenVA05} is a complementary approach that can overcome these problems. It combines ideas of model checking and simulation with the aim of supporting quantitative analysis as well as addressing the size barrier.
Among others, SMC has been used to analyse AODV and DYMO.

\cite{NFM2013} first develops timed models for AODV and DYMO. These models are based on the {\uppaal} models created from our {\awn} specifications.
The paper then carries out a systematic analysis across all small networks. In contrast to simulation and test bed studies, the analysis based on quality and quantity enables the examination of reasons for observed differences in performance between AODV and DYMO\@.
\cite{NFM2013} then examines the feasibility of SMC w.r.t.\ scalability; the results imply that networks of realistic size (up to 100 nodes) can be analysed. 

For small networks it is possible to analyse all topologies. This gives a good overall view of the performance and behaviour in any situation. 
For large networks this is not feasible, and so the selection of topologies as well as their dynamic behaviour becomes something of a `stab in the dark'. 
The Node Placement Algorithm for Realistic Topologies (NPART)~\cite{MilicMalek2009} is a tool that
allows the generation of arbitrary-sized topologies and transmission ranges; it has been shown that
the generated topologies have graph characteristics similar to realistic wireless multihop ones. 
\cite{QEST2013}  proposes a topology-based mobility model that abstracts from physical behaviour and models mobility as probabilistic changes in the topology. It is demonstrated how this model can be instantiated to cover the main aspects of the random walk and the random waypoint mobility model. The model is not a stand-alone model, but intended to be used in combination with protocol models.  As one application a brief analysis of the Ad-hoc On demand Distance Vector (AODV) routing protocol is given.

A more thorough (quantitative) analysis of AODV based on this  topology-based mobility model is performed in~\cite{FORMATS2013}.
Here, variants of AODV, such as always forwarding route replies (see \SSect{modifyRREP}), are analysed as well.
The paper makes surprising observations on the behaviour of AODV\@. For example, it is shown that some optional features (D-flag) should not be combined with others (resending). Another observation of~\cite{FORMATS2013} is that a well-known shortcoming occurs more often than expected and has a significant effect on the success of route establishment.

\paragraph{Other Approaches}
Next to process algebra and model checking other approaches have been used to analyse 
WMNs. A frequently used approach is \textit{coloured Petri nets} (CPNs)~\cite{Jensen97}.

The idea to use CPNs to model routing protocols for MANETs was first employed in \cite{Xiong2002}:
the paper proposes a \emph{topology approximation} (TA) mechanism for modelling mobility
and, based on this, presents a CPN model of AODV\@.
Using this formal model the network behaviour for a network with 5 nodes is simulated.

Mandatory parts of DYMO are modelled as a hierarchy of CPNs in \cite{EspensenEtAl08}.
The paper analyses draft-version 10  and identifies and resolves some ambiguities in specification.
Moreover, it points at problematic behaviour; six of these findings have been reported to the 
IETF MANET Working Group mailing list, and have been resolved by the DYMO developers in version 11 of the DYMO specification.
The model presented in \cite{EspensenEtAl08} has a complex net structure, comprising 4 levels of hierarchy and 14 modules.
A much smaller model of DYMO, which even covers some optional features, is presented in \cite{BY09}.\footnote{A detailed comparison between the models given in \cite{EspensenEtAl08} and \cite{BY09} is given in \cite[Sect. 4]{BY09}.} Reducing the size of the model also reduces the 
state space, so larger networks can be analysed. 
Experiments performing test runs on small topologies confirm specified behaviour.
However, similar to model checking, networks with a few nodes only can be analysed.

\textit{Graph Transformation Systems} are used in~\cite{SWJ08} to model DYMO (version 10),
but without the feature of route reply by intermediate nodes.
The paper provides a semi-algorithm, based on graph rewriting, which was used to verify
loop freedom for this version of DYMO\@.

Other formal approaches are \textit{algebraic techniques} involving a.o.\ semirings and matrices.
Sobrinho was the first 
who brought algebraic reasoning into the realm of hop-by-hop 
routing~\cite{S02}. He uses algebraic properties to argue about 
the relationship between routing algorithms and Dijkstra's shortest path 
algorithm. This approach has been further elaborated for the analysis 
of path vector protocols like the Border Gate Protocol BGP~\cite{S03,GS05}.
Similar algebraic reasoning has been performed in~\cite{HM12} to present algebraic versions of the algorithms of Dijkstra 
and Floyd-Warshall. \cite{HM11} presents first steps towards an algebraic characterisation of AODV using these algebraic techniques.

\section{Conclusion and Future Work}\label{sec:conclusion}

In this paper we have proposed \awn, a novel process algebra that can be used to 
model, verify and analyse (routing) protocols for Wireless Mesh Networks (WMNs). 
The applicability of the process algebra has been demonstrated by a 
careful analysis of the Ad hoc On-Demand Distance Vector (AODV) Routing Protocol. 
To the best of our knowledge it is by far the most detailed analysis of 
a routing protocol for WMNs. 

The introduced process algebra \awn
covers major aspects of WMN routing
protocols, for example the crucial aspect of data handling, such as maintaining
routing table information. 
Amongst others, the assignment primitive, which is used to manipulate data, 
turns {\awn} into an easy to read language---its syntax is close to the syntax of 
programming languages.
Key operators of {\awn} are \phrase{local broadcast} and \phrase{conditional unicast}.
Local broadcast allows a node to send messages to all 
its immediate neighbours as implemented by the physical and 
data link layer. 
Conditional unicast models an abstraction of an acknowledgment-of-receipt mechanism that is typical for unicast communication but absent in broadcast communication,
as typically implemented by the link layer of relevant wireless standards such as IEEE 802.11.
AWN can capture the bifurcation depending on the success of the unicast;
it allows error handling in response to failed communications while abstracting from link layer
implementations of the communication handling.

The unique set of features and primitives of {\awn} allows the creation of accurate and concise
models of relatively complex and practically relevant network protocols in a simple language.
We have demonstrated this by giving a complete and accurate model of 
the core functionality of AODV, a widely used protocol of practical relevance. 
We currently do not model optional
features such as local route repair, expanding ring search, gratuitous route reply and multicast.
We also abstract from all timing issues.
In addition to modelling the complete set of core functionalities of the
AODV protocol, our model also covers the interface to higher protocol layers via the injection
and delivery of application layer data,
as well as the forwarding of data packets at intermediate nodes.
Although this is not part of the AODV protocol specification, it is necessary
for a practical model of any reactive routing protocol, where protocol activity is triggered via the
sending and forwarding of data packets.

Process algebras are standard tools to describe interactions, communications and
synchronisations between a collection of independent agents, processes or network nodes.
They  provide algebraic laws that facilitate formal reasoning.
To demonstrate the strength of formal reasoning we performed a 
careful analysis of AODV, in particular with respect to the loop-freedom property. 
By establishing invariants that remain valid in a network running AODV, 
we have shown that our model is in fact loop free. 
In contrast to protocol evaluation using simulation, test-bed experiments or model checking, 
where only a finite number of specific network scenarios can be considered, our reasoning with {\awn} is generic and the proofs hold for any possible network scenario in terms of topology and traffic pattern. None of the experimental protocol evaluation approaches can deliver this high degree of assurance about protocol behaviour.
We have also shown that, in contrast to common belief, sequence numbers do not 
guarantee loop freedom, even if they are increased monotonically over time and 
incremented whenever a new route request is generated. 

Our analysis of AODV uncovered several ambiguities in the RFC,  the de facto standard of AODV\@.
In this paper we have analysed {\em all\/}  interpretations of the AODV RFC\@ that stem from the 
ambiguities revealed. 
It turned out that several interpretations can yield unwanted behaviour such as routing loops. 
We also found that implementations of 
AODV behave differently in crucial aspects of protocol behaviour, although they all follow the lines of the RFC\@. 
As pointed out, this is often caused by ambiguities, contradictions or unspecified behaviour in the RFC\@.
Of course a specification ``{\sf needs to be reasonably implementation independent\/}''%
\footnote{\url{http://www.ietf.org/iesg/statement/pseudocode-guidelines.html}}
and can leave some decisions to the software engineer; however it is our belief that any specification should be clear and 
unambiguous enough to guarantee the same behaviour when given to different developers. 
As demonstrated, this is not the case for AODV, and likely
not for many other RFCs provided by the IETF.

Finding ambiguities and unexpected behaviour is not
uncommon for RFCs, since the currently predominant practice is an informal 
protocol specification via English prose.
This shows that the specification of a
reasonably rich protocol such as AODV cannot be described precisely
and unambiguously by simple (English) text only; formal methods are
indispensable for this purpose.
We believe that formal specification languages and
analysis techniques---offering rigorous verification and ana\-lysis techniques---are
now able to capture the full syntax and semantics of reasonably rich
IETF protocols. These are an indispensable augmentation to natural
language, both for specifying protocols such as AODV, AODVv2 and HWMP,
and for verifying their essential properties.

Our analysis of AODV also uncovered several shortcomings of the protocol, including 
a failure in route discovery, and the creation of non-optimal routes.
In this paper, we have not only listed the shortcomings, we have proposed 
(small) modifications of AODV to overcome these deficiencies. 
All proposed variants have been carefully analysed as well, in particular with respect to 
loop freedom. By this we have shown how proofs based on {\awn} can relatively easily be 
adapted to protocol variants. 

\index{timing}%
A further analysis of AODV will require an extension of {\awn} with time and probability:
the former to cover aspects such as AODV's handling (deletion) of
stale routing table entries and the latter to model the probability
associated with lossy links. We expect that the resulting algebra will
be also applicable to a wide range of other wireless protocols.

Next to this on-going work, we also aim at a complementation of {\awn} by model checking.
Having the ability of automatically deriving a model for model checkers such as {\sc Uppaal} from an \awn specification allows the confirmation and detailed
diagnostics of suspected errors in an early phase of protocol development. 
Surely, model checking is limited to particular topologies, but finding shortcomings in some topologies is  useful to identify problematic behaviour. 
These shortcomings  can be eliminated, even before a more thorough and general analysis using {\awn}. 

\vspace{1cm}
\begin{flushright}
``Time is the nurse and breeder of all good.''\\
{\footnotesize W.~Shakespeare, The Two Gentlemen of Verona}
\end{flushright}
\newpage

\addtocontents{toc}{\protect\renewcommand\protect\baselinestretch{.5}}

\newpage
\phantomsection \label{references}
\bibliographystyle{eptcsini}
\addcontentsline{toc}{section}{References}
\bibliography{aodv}

\newpage
\phantomsection \label{listofalgo}
\addcontentsline{toc}{section}{\listalgorithmname}
\listofalgorithms

\newpage
\phantomsection \label{listoffig}
\addcontentsline{toc}{section}{\listfigurename}
\listoffigures

\newpage
\phantomsection \label{listoftab}
\addcontentsline{toc}{section}{\listtablename}
\listoftables

\newpage
\phantomsection \label{index}
\addcontentsline{toc}{section}{Index}
\printindex

\end{document}